\documentclass[11pt]{article}

\usepackage[nottoc]{tocbibind}

\usepackage{amsfonts}
\usepackage{amssymb}
\usepackage{amstext}
\usepackage{amsmath}
\usepackage{mathtools}

\usepackage{amsthm} 

\usepackage{caption}
\usepackage{subcaption}

\usepackage{color}
\usepackage{nameref}
\definecolor{ForestGreen}{rgb}{0.1333,0.5451,0.1333}
\definecolor{DarkRed}{rgb}{0.8,0,0}
\definecolor{Red}{rgb}{0.9,0,0}
\usepackage[linktocpage=true,
	pagebackref=true,colorlinks,
	linkcolor=DarkRed,citecolor=ForestGreen,
	bookmarks,bookmarksopen,bookmarksnumbered]
	{hyperref}
\usepackage{cleveref}

\usepackage{thmtools,thm-restate} 

\usepackage[numbers,sort&compress]{natbib}

\usepackage{graphicx}
\usepackage{graphics}
\usepackage{colordvi}
\usepackage{xspace}
\usepackage{algorithm}
\usepackage{algorithmicx}
\usepackage{url}
\usepackage{enumitem}

        {\hspace*{\fill}$\Box$\par\vspace{4mm}}

\newcommand{\enc}[1]{\langle #1 \rangle}

\usepackage[left=1in,top=1in,right=1in,bottom=1in]{geometry} 
\usepackage{mathptmx} 
\usepackage[scaled=.90]{helvet}
\usepackage{courier}
\usepackage{booktabs}
\usepackage{threeparttable}
\usepackage[small,compact]{titlesec}
\setlist[itemize]{noitemsep,nolistsep} 
\setlist[enumerate]{noitemsep,nolistsep} 

\makeatletter
\renewcommand{\paragraph}{%
  \@startsection{paragraph}{4}%
  {\z@}{1ex \@plus 1ex \@minus .2ex}{-1em}%
  {\normalfont\normalsize\bfseries}%
}
\makeatother

\makeatletter
\def\thmt@refnamewithcomma #1#2#3,#4,#5\@nil{%
  \@xa\def\csname\thmt@envname #1utorefname\endcsname{#3}%
  \ifcsname #2refname\endcsname
    \csname #2refname\expandafter\endcsname\expandafter{\thmt@envname}{#3}{#4}%
  \fi
}
\makeatother

\declaretheorem[numberwithin=section,refname={Theorem,Theorems},Refname={Theorem,Theorems}]{theorem}
\declaretheorem[numberlike=theorem,refname={Lemma,Lemmas},Refname={Lemma,Lemmas}]{lemma}

\declaretheorem[numberlike=theorem,refname={Corollary,Corollaries},Refname={Corollary,Corollaries}]{corollary}

\declaretheorem[numberlike=theorem,refname={property,properties},Refname={Property,Properties}]{property}
\declaretheorem[numberlike=theorem,refname={observation,observations},Refname={Observation,Observations}]{observation}
\declaretheorem[numberlike=theorem,refname={assumption,assumptions},Refname={Assumption,Assumptions}]{assumption}
\declaretheorem[numberlike=theorem,refname={Claim, Claims},Refname={Claim, Claims}]{claim}

\declaretheorem[numberlike=theorem]{definition}

\newtheorem{invariant}[theorem]{Invariant}
\newtheorem{assume}{Assumption}
\newtheorem{fact}{Fact}
\newtheorem{Rule}[theorem]{Rule}
\newcommand{\size}[1]{\ensuremath{\left|#1\right|}}
\newcommand{\ceil}[1]{\ensuremath{\left\lceil#1\right\rceil}}
\newcommand{\floor}[1]{\ensuremath{\left\lfloor#1\right\rfloor}}
\newcommand{\Dexp}[1]{\dexp\{#1\}}
\newcommand{\Tower}[2]{\operatorname{tower}^{(#1)}\{#2\}}
\newcommand{\logi}[2]{\operatorname{log}^{(#1)}{#2}}
\newcommand{\norm}[1]{\lVert #1\rVert}
\newcommand{\abs}[1]{\lvert #1\rvert}
\newcommand{\paren}[1]{\left ( #1 \right ) }
\newcommand{\union}{\cup}
\newcommand{\band}{\wedge}
\newcommand{\bor}{\vee}
\newcommand{\prof}{\ensuremath{\operatorname{profit}}}
\newcommand{\pay}{\ensuremath{\operatorname{pay}}}

\newcommand{\scale}[2]{\scalebox{#1}{#2}}

\newcommand{\fig}[1]{
\begin{figure}[h]
\rotatebox{0}{\includegraphics{#1}}
\end{figure}}

\newcommand{\figcap}[2]
{
\begin{figure}[h]
\rotatebox{270}{\includegraphics{#1}} \caption{#2}
\end{figure}
}

\newcommand{\scalefig}[2]{
\begin{figure}[h]
\scalebox{#1}{\includegraphics{#2}}
\end{figure}}

\newcommand{\scalefigcap}[3]{
\begin{figure}[h]
\scalebox{#1}{\rotatebox{0}{\includegraphics{#2}}} \caption{#3}
\end{figure}}

\newcommand{\scalefigcaplabel}[4]{
\begin{figure}[h]
\begin{center}
\label{#4} \scalebox{#1}{\includegraphics{#2}}\caption{#3}
\end{center}
\end{figure}}

\renewcommand{\phi}{\varphi}
\newcommand{\Sum}{\displaystyle\sum}
\newcommand{\half}{\ensuremath{\frac{1}{2}}}

\newcommand{\poly}{\operatorname{poly}}
\newcommand{\dist}{\mbox{\sf dist}}

\newcommand{\reals}{{\mathbb R}}

\newcommand{\expct}[1]{\text{\bf E}_\left [#1\right]}
\newcommand{\expect}[2]{\text{\bf E}_{#1}\left [#2\right]}
\newcommand{\prob}[1]{\text{\bf Pr}\left [#1\right]}
\newcommand{\pr}[2]{\text{\bf Pr}_{#1}\left [#2\right ]}

\newcommand{\notsat}[1]{\overline{\text{SAT}(#1)}}
\newcommand{\notsats}[2]{\overline{\text{SAT}_{#1}(#2)}}

\newcommand{\sndp}{\mbox{\sf SNDP}}
\newcommand{\ecsndp}{\mbox{\sf EC-SNDP}}
\newcommand{\vcsndp}{\mbox{\sf VC-SNDP}}
\newcommand{\kec}{k\mbox{\sf -edge connectivity}}
\newcommand{\kvc}{k\mbox{\sf -vertex connectivity}}
\newcommand{\sskec}{\mbox{\sf single-source}~k\mbox{\sf -edge connectivity}}
\newcommand{\sskvc}{\mbox{\sf single-source}~k\mbox{\sf -vertex connectivity}}
\newcommand{\subkvc}{\mbox{\sf subset}~k\mbox{\sf -vertex connectivity}}
\newcommand{\oneec}{1\mbox{\sf -edge connectivity}}
\newcommand{\onevc}{1\mbox{\sf -vertex connectivity}}
\newcommand{\kvcssp}{k\mbox{\sf -vertex-connected spanning subgraph problem}}
\newcommand{\BLM}{{\sf Bounded-Length Multicut}\xspace}
\newcommand{\bydef}{\stackrel{{\triangle}}{=}}
\newcommand{\eblue}{E^{\mbox{\scriptsize{blue}}}}
\newcommand{\ered}{E^{\mbox{\scriptsize{red}}}}

\newcommand{\induce}[1]{\ensuremath{{\sf im}(#1)}\xspace}
\newcommand{\sinduce}[1]{\ensuremath{{\sf sim}(#1)}\xspace}
\newcommand{\sinducesigma}[2]{\ensuremath{{\sf sim}_{#1}(#2)}\xspace}

\newcommand{\SMP}{{\sf SMP}\xspace}
\newcommand{\UDP}{{\sf UDP}\xspace}
\newcommand{\semiind}{{\sf Semi-Induced Matching}\xspace}
\newcommand{\LSMP}{{\sf Limit-SMP}\xspace}
\newcommand{\LUDP}{{\sf Limit-UDP}\xspace}
\newcommand{\ETH}{{\sf ETH}\xspace}
\newcommand{\SAT}{{\sf SAT}\xspace}
\newcommand{\p}{{\bf p}}
\newcommand{\e}{{\sc e}}
\newcommand{\val}{{\sf val}\xspace}

\newcommand{\wall}{\overrightarrow{R}}
\newcommand{\awall}{\widehat{R}}
\newcommand{\pairs}{\mathcal{S}}

\newcommand{\edp}{{\sf edp}\xspace}
\newcommand{\tildeEDP}{\widetilde{\sf edp}\xspace}
\newcommand{\vdp}{{\sf vdp}\xspace}
\newcommand{\mis}{{\sf mis}\xspace}
\newcommand{\NDP}{{\sf VDP}\xspace}
\newcommand{\VDP}{{\sf VDP}\xspace}
\newcommand{\EDP}{{\sf EDP}\xspace}

\newcommand{\NN}{\mathcal{N}}

\newcommand{\reminder}[1]{{\it\bf \color{red}**#1**}}

\newcommand{\QValue}{{\sc QueryValue}\xspace}
\newcommand{\QSet}{{\sc QuerySet}\xspace}

\renewcommand{\t}{\text{slack}}
\newcommand{\s}{\text{status}}
\newcommand{\N}{\mathcal{N}}
\newcommand{\B}{\mathcal{B}}
\newcommand{\C}{\mathcal{C}}
\renewcommand{\P}{\mathcal{P}}
\newcommand{\eps}{\delta}
\newcommand{\dl}{\delta}
\newcommand{\dd}{D}
\newcommand{\G}{\mathcal{G}}
\newcommand{\V}{\mathcal{V}}
\newcommand{\E}{\mathcal{E}}

\ifdefined\ShowComment

\def\danupon#1{\marginpar{$\leftarrow$\fbox{D}}\footnote{$\Rightarrow$~{\sf #1 --Danupon}}}
\def\sayan#1{\marginpar{$\leftarrow$\fbox{S}}\footnote{$\Rightarrow$~{\sf #1 --Sayan}}}
\def\monika#1{\marginpar{$\leftarrow$\fbox{M}}\footnote{$\Rightarrow$~{\sf #1 --Monika}}}

\else

\def\danupon#1{}
\def\sayan#1{}
\def\monika#1{}

\fi

\newboolean{short}
\setboolean{short}{true} 

\newcommand{\shortOnly}[1]{\ifthenelse{\boolean{short}}{#1}{}}
\newcommand{\longOnly}[1]{\ifthenelse{\boolean{short}}{}{#1}}

\title{New Deterministic Approximation Algorithms for Fully Dynamic Matching\thanks{A preliminary version of this paper will appear in STOC 2016.}}

\author{Sayan Bhattacharya\thanks{The Institute of Mathematical Sciences, Chennai, India. Email: bsayan@imsc.res.in}
	\and Monika Henzinger\thanks{Faculty of Computer Science, University of Vienna, Austria. Email: monika.henzinger@univie.ac.at. This work was done in part while the author was visiting the Simons Institute for the Theory of Computing. The research leading to this work has received funding from the European Union's Seventh Framework Programme (FP7/2007-2013) under grant agreement number 317532 and from the European Research Council under the European Union's Seventh Framework Programme (FP7/2007-2013)/ERC grant agreement number 340506.}       
    \and Danupon Nanongkai\thanks{KTH Royal Institute of Technology, Sweden. Email: danupon@gmail.com. Supported by Swedish Research Council grant 2015-04659.}   
}

\begin{document}

\setcounter{tocdepth}{3}

\begin{titlepage}
\maketitle
\pagenumbering{roman}

\begin{abstract}
We present two deterministic dynamic algorithms for the maximum matching problem.  (1) An algorithm that maintains a $(2+\epsilon)$-approximate maximum matching in general graphs with $O(\poly(\log n, 1/\epsilon))$ update time. (2) An algorithm that maintains an $\alpha_K$ approximation of the {\em value} of the maximum matching  with $O(n^{2/K})$ update time in bipartite graphs, for every sufficiently large constant positive integer $K$.  Here, $1\leq \alpha_K < 2$ is a constant determined by the value  of $K$. 
Result (1) is the first deterministic algorithm that can maintain an $o(\log n)$-approximate maximum matching with polylogarithmic update time, improving the seminal result of Onak et al. [STOC 2010]. Its approximation guarantee almost matches the guarantee of the best {\em randomized} polylogarithmic update time algorithm  [Baswana et al. FOCS 2011]. 
Result (2) achieves a better-than-two approximation with {\em arbitrarily small polynomial} update time on bipartite graphs. Previously the best update time for this problem was $O(m^{1/4})$ [Bernstein et al. ICALP 2015], where $m$ is the current number of edges in the graph. \end{abstract}

\newpage
\pagenumbering{gobble}
\clearpage

\setcounter{tocdepth}{3}
\tableofcontents

\newpage
\part{EXTENDED ABSTRACT}

\end{titlepage}

\newpage

\pagenumbering{arabic}

\section{Introduction}
\label{sec:intro}

\renewcommand{\poly}{\text{poly}}
\renewcommand{\E}{\mathcal{E}}


In this paper, we consider the {\em dynamic maximum cardinality matching} problem. In this problem an algorithm has to quickly maintain an (integral) maximum-cardinality matching or its approximation, when the $n$-node input graph is undergoing edge insertions and deletions. We consider two versions of this problem:
In the {\em matching version}, the algorithm has to output the change in the (approximate) matching, if any, after each edge insertion and deletion. In the {\em value version}, the algorithm only has to output the value of the matching. (Note that an algorithm for the matching version can be used to solve the value version within the same time.)
When stating the running time below, we give the time {\em per update}\footnote{In this discussion, we ignore whether the update time is amortized or worst-case as this is not the focus of this paper. The update time of our algorithm is amortized.}. 
If not stated otherwise, these results hold for both versions. 
%
%
%
%

The state of the art for maintaining an {\em exact} solution for the value version of this problem is a randomized $O(n^{1.495})$-time algorithm \cite{Sankowski07}. This  is complemented by various hardness results which rules out polylogarithmic update time \cite{AbboudW14,HenzingerKNS15-oMv,KopelowitzPP16}. 
%
%
As it is  desirable for dynamic algorithms to have polylogarithmic update time, the recent work has focused on achieving this goal by allowing {\em approximate solutions}.
The first paper that achieved this is by Onak and Rubinfeld \cite{OnakR10},  which gave a randomized $O(1)$-approximation $O(\log^2 n)$-time algorithm and a deterministic $O(\log n)$ approximation $O(\log^2 n)$-time algorithm.
As stated in the two open problems in \cite{OnakR10}, this seminal paper opened up the doors for two research directions: \\

\begin{enumerate}
\item Designing a (possibly randomized) {\em polylogarithmic time} algorithm with smallest approximation ratio. \\
\item  Designing a {\em deterministic  polylogarithmic time} algorithm with constant approximation ratio. \\
\end{enumerate}

\smallskip
\noindent The second question is motivated by the fact that {\em randomized}  dynamic approximation algorithms only fulfill their approximation guarantee when used by an {\em oblivious} adversary, i.e., an adversary that gives
the next update {\em without} knowing  the outputs of the algorithm resulting from earlier updates. This limits the usefulness of randomized  dynamic algorithms. In contrast,
{\em deterministic}  dynamic algorithms fulfill their approximation guarantee against {\em any} adversary, even non-oblivous ones. Thus, they can be used, for example, as a ``black box'' by any other (potentially static) algorithm, while this is not generally the case for {\em randomized} dynamic algorithms. This motivates the search for deterministic fully dynamic approximation algorithms, even though a randomized algorithm with the same approximation guarantee might exists. 

\medskip
\noindent Up to date, the best answer to the first question is the randomized $2$ approximation $O(\log n)$ update time algorithm from~\cite{BaswanaGS15}. 
It remains  elusive to design  a {\em better-than-two} approximation factor with polylogarithmic update time. Some recent works have focused on achieving such approximation factor with lowest update time possible. The current best update time is $O(m^{1/4}/\epsilon^{2.5})$ \cite{BernsteinS16,BernsteinS15}, which  is deterministic and guarantees a $(3/2 + \epsilon)$ approximation factor.
%
%

\medskip
\noindent  For the second question, deterministic polylogarithmic-time $(1+\epsilon)$-approximation algorithms were known for the special case of low arboricity graphs \cite{NeimanS13,KopelowitzPP16,PelegS16}. On general graphs, the paper~\cite{Arxiv} achieved a deterministic $(3+\epsilon)$-approximation polylogarithmic-time algorithm by maintaining a fractional matching; this algorithm however works only for the value version. {\em No} deterministic $o(\log n)$ approximation algorithm with polylogarithmic update time was known for the matching version. 
(There were many deterministic constant approximation algorithms with $o(m)$ update time for the matching version (e.g. \cite{NeimanS13,Arxiv,GuptaP13,BernsteinS16,BernsteinS15}). The fastest among them requires $O(m^{1/4}/\epsilon^{2.5})$ update time \cite{BernsteinS16}.)

\medskip
\noindent {\bf Our Results.}
We make progress on both versions of the problem as stated in Theorems~\ref{th:old} and~\ref{th:new:approx}.

\begin{theorem}\label{thm:2-approx}\label{th:old}
For every $\epsilon \in (0,1)$, there is a deterministic algorithm that maintains a $(2+\epsilon)$-approximate maximum  matching in a graph in  $O(\text{poly} (\log n, 1/\epsilon))$ update time, where $n$ denotes the number of nodes.
\end{theorem}

Theorem~\ref{th:old} answers Onak and Rubinfeld's second question positively.
In fact, our approximation guarantee almost matches the best ($2$-approximation) one provided by a randomized algorithm \cite{BaswanaGS15}.\footnote{By combining our result with the techniques of~\cite{CrouchS14} in a standard way, we also obtain a deterministic $(4+\epsilon)$-approximation $O(\poly\log n \poly(1/\epsilon)\log W)$-time for the dynamic maximum-weight matching problem, where $W$ is the ratio between the largest and smallest edge weights. 
}  
Our algorithm for Theorem~\ref{th:old} is obtained by combining previous techniques \cite{Arxiv,GuptaP13,PelegS16} with two new ideas that concern fractional matchings. First, we dynamize the {\em degree splitting process} previously used in the parallel and distributed algorithms literature \cite{ipl} and use it to reduce the size of the support of the fractional matching maintained by the algorithm of \cite{Arxiv}. This helps us maintain an approximate integral matching cheaply using the result in~\cite{GuptaP13}. 
This idea alone already leads to a $(3+\epsilon)$-approximation deterministic algorithm. 
Second, we improve the approximation guarantee further to $(2+\epsilon)$ by proving a new structural lemma that concerns the ratio between (i) the maximum (integral) matching in the support of a maximal fractional matching and (ii) the maximum (integral) matching in the whole graph. It was known that this ratio is at least $1/3$. We can improve this ratio to $1/2$ with a fairly simple proof (using Vizing's theorem~\cite{vizing2}). We note that this lemma can be used to improve the analysis of an algorithm in~\cite{Arxiv} to get the following result: There is  a deterministic algorithm that maintains a $(2+\epsilon)$ approximation to the size of the maximum matching in a general graph in $O(m^{1/3}/\epsilon^2)$ amortized update time.


\begin{theorem}\label{thm:better than two}\label{th:new:approx}


For every sufficiently large positive integral constant $K$, we can  maintain an $\alpha_K$-approximation to the value of the maximum matching\footnote{We can actually maintain an approximate {\em fractional} matching with the same performance bounds.} 
in a bipartite graph $G = (V, E)$, where $1 \leq \alpha_K < 2$. The algorithm is deterministic and has an amortized update time of  $O(n^{2/K})$.
\end{theorem}
We consider Theorem~\ref{th:new:approx} to be a step towards achieving a polylogarithmic time (randomized or deterministic) fully dynamic algorithm with an approximation ratio less than 2, i.e., towards answering Onak and Rubinfeld's first question. This is because, firstly, it shows that on bipartite graphs the better-than-two approximation factor can be achieved with arbitrarily small polynomial update time, as opposed to the previous best $O(m^{1/4})$ time of \cite{BernsteinS15}. Secondly, it rules out a natural form of hardness result and thus suggests that a polylogarithmic-time algorithm with better-than-two approximation factor exists on bipartite graphs. 
More precisely, the known hardness results (e.g. those in \cite{Patrascu10,AbboudW14,KopelowitzPP16,HenzingerKNS15-oMv}) that rule out a polylogarithmic-time $\alpha$-approximation algorithm (for any $\alpha>0$) are usually in the form {\em ``assuming some conjecture, there exists a constant $\delta>0$ such that for any constant $\epsilon>0$, there is no $(1-\epsilon)\alpha$-approximation algorithm that has $n^{\delta-\epsilon}$ update time''}; 
for example, for dynamically {\em $2$-approximating} graph's diameter, this statement was proved for $\alpha=2$ and $\delta=1/2$ in \cite{HenzingerKNS15-oMv}, implying that any better-than-two approximation algorithm for this problem will require an update time close to $n^{1/2}$. Our result in \ref{thm:better than two} implies that a similar statement cannot be proved for $\alpha=2$ for the bipartite matching problem since, for any constant $\delta>0$, there is a $(2-\epsilon)$-approximation algorithm with update time, say, $O(n^{\delta/2})$ for some $\epsilon>0$.


To derive an algorithm for Theorem~\ref{th:new:approx}, we use  the fact that in a bipartite graph the size of the maximum fractional matching is the same as the size of the maximum integral matching. Accordingly,  a {\em maximal fractional matching} (which gives $2$-approximation) can be {\em augmented} by a fractional {\em $b$-matching}, for a carefully chosen capacity vector $b$, to obtain  a better-than-two approximate fractional matching. 
The idea of ``augmenting a bad solution'' that we use here is inspired by the approach in the streaming setting by Konrad~et~al.~\cite{KonradMM12}. But the way it is implemented is different as \cite{KonradMM12} focuses on using augmenting paths  while we use fractional $b$-matchings.


\medskip
\noindent {\bf Organization.} In Section~\ref{sec:prelim},   we  define some basic concepts and notations that will be used throughout the rest of this paper. In Section~\ref{main:sec:general},  we give an overview of our algorithm for Theorem~\ref{th:old}. In Section~\ref{main:sec:bipartite}, we highlight the main ideas behind our algorithm for Theorem~\ref{th:new:approx}. Finally, we conclude with some open problems in Section~\ref{sec:open}. All the missing details can be found in Part II and Part III of the paper.

\subsection{Notations and preliminaries}
\label{sec:prelim}
 Let $n = |V|$ and $m = |E|$ respectively denote the number of nodes and edges in the input graph $G = (V, E)$. Note that $m$ changes with time, but $n$ remains fixed.  Let $\text{deg}_v(E')$ denote the number of edges in a subset $E' \subseteq E$ that are incident upon a node $v \in V$.  An (integral) matching $M \subseteq E$  is a subset of edges that do not share any common endpoints. The {\em size} of a matching is the number of edges contained in it. We are also  interested in the concept of a {\em fractional matching}. Towards this end, we first define the notion of a {\em fractional assignment}. A fractional assignment $w$ assigns a weight $w(e) \geq 0$ to every edge $e \in E$. We let $W_v(w) = \sum_{(u, v) \in E} w(u,v)$ denote the total weight received by a node $v \in V$ under $w$ from its incident edges. Further, the {\em support} of $w$ is defined to be the subset of edges $e \in E$ with $w(e) > 0$. Given two fractional assignments $w, w'$, we define their addition $(w+w')$ to be a new fractional assignment that assigns a weight $(w+w')(e) = w(e) + w'(e)$ to every edge $e \in E$. We say that a fractional assignment $w$ forms a {\em fractional matching} iff we have $W_v(w) \leq 1$ for all nodes $v \in V$. Given any subset of edges $E' \subseteq E$, we define $w(E') = \sum_{e \in E'} w(e)$. We define the {\em size} of a fractional matching $w$ to be  $w(E)$. Given any subset of edges $E' \subseteq E$, we let $\text{Opt}_{f}(E')$ (resp. $\text{Opt}(E')$) denote the maximum possible size of a fractional matching with support $E'$ (resp. the maximum possible size of an integral matching $M' \subseteq E'$).  Theorem~\ref{main:th:structure:bipartite}  follows from the half-integrality of the matching polytope in general graphs and its total unimodularity in bipartite graphs.

\begin{theorem}
\label{main:th:structure:bipartite}
\label{main:main:th:structure}
Consider any subset of edges $E' \subseteq E$ in the graph $G = (V, E)$. We have: $\text{Opt}(E') \leq \text{Opt}_f(E') \leq (3/2) \cdot \text{Opt}(E')$. Further, if the graph $G$ is bipartite, then we have: $\text{Opt}_f(E') = \text{Opt}(E')$.
\end{theorem}

 Gupta and Peng~\cite{GuptaP13} gave a dynamic algorithm that maintains a $(1+\epsilon)$-approximate maximum matching in $O(\sqrt{m}/\epsilon^2)$ update time. A simple modification of their algorithm gives the following result.  

\begin{theorem}\cite{GuptaP13}
\label{th:gupta:peng}
If  the maximum degree in  a dynamic graph  never exceeds some threshold $d$, then we can maintain a $(1+\epsilon)$-approximate maximum matching  in $O(d/\epsilon^2)$ update time. 
\end{theorem}

We say that a fractional matching $w$ is   {\em $\alpha$-maximal}, for $\alpha \geq 1$, iff $W_u(w) + W_v(w) \geq 1/\alpha$ for every edge $(u,v) \in E$. Using  LP-duality and complementary slackness conditions, one can show  the following result.

\begin{lemma}
\label{lm:approx:maximal}
 $\text{Opt}_f(E) \leq 2\alpha \cdot w(E)$ for every $\alpha$-maximal fractional matching $w$ in a graph $G = (V, E)$. 
 \end{lemma}

\section{General Graphs}
\label{main:sec:general}

We give a dynamic algorithm for maintaining an approximate maximum matching in a general graph. We consider the following dynamic setting. Initially, the input graph is empty. Subsequently, at each time-step, either an edge is inserted into the graph, or an already existing edge is deleted from the graph. The node-set of the graph, however, remains unchanged. Our main result in this section is stated in Theorem~\ref{th:old}. Throughout this section, we will use the notations and concepts introduced in Section~\ref{sec:prelim}.

\subsection{Maintaining a large fractional matching}
\label{sec:sodapaper}

Our algorithm for Theorem~\ref{th:old} builds upon an existing dynamic data structure that maintains a large fractional matching. This data structure   was developed in~\cite{Arxiv}, and can be described as follows. Fix a small constant $\epsilon > 0$. Define $L = \lceil \log_{(1+\epsilon)} n \rceil$, and partition the node-set $V$ into $L+1$ subsets $V_0,  \ldots, V_L$. We say that the nodes belonging to the subset $V_i$ are in ``level $i$''. We denote the level of a node  $v$ by $\ell(v)$, i.e.,  $v \in V_i$ iff $\ell(v) = i$. We next define the ``{\em level of an edge}'' $(u,v)$ to be $\ell(u,v) = \max(\ell(u), \ell(v))$. In other words, the level of an edge is the maximum level of its endpoints. We let $E_i = \{ e \in E : \ell(e) = i\}$ denote the set of edges at level $i$, and define the subgraph $G_i = (V, E_i)$. Thus, note that the edge-set $E$ is partitioned by the subsets $E_0,  \ldots, E_L$. For each  level $i \in \{0, \ldots, L\}$, we now define a fractional assignment $w_i$ with support $E_i$. The fractional assignment $w_i$ is uniform, in the sense that it assigns the same weight $w_i(e) = 1/d_i$, where $d_i = (1+\epsilon)^i$, to every edge $e \in E_i$ in its support. In contrast, $w_i(e) = 0$ for every edge $e \in E \setminus E_i$.  Throughout the rest of this section, we  refer to this structure as a {\em ``hierarchical partition''}.

\begin{theorem}~\cite{Arxiv}
\label{th:runtime}
We can maintain a hierarchical partition dynamically in $O(\log n/\epsilon^2)$ update time. The algorithm ensures that the fractional assignment $w = \sum_{i=0}^L w_i$ is a $(1+\epsilon)^2$-maximal matching in $G = (V, E)$. Furthermore, the algorithm ensures that $1/(1+\epsilon)^2 \leq W_v(w) \leq 1$ for all nodes $v \in V$ at levels $\ell(v) > 0$.  
\end{theorem}

\begin{corollary}
\label{cor:approx}
The  fractional matching $w$ in Theorem~\ref{th:runtime} is a $2(1+\epsilon)^2$-approximation to  $\text{Opt}_f(E)$. 
\end{corollary}

\begin{proof}
Follows from Lemma~\ref{lm:approx:maximal} and Theorem~\ref{th:runtime}.
\end{proof}

\begin{corollary}
\label{cor:lm:deg:1}
Consider the hierarchical partition in Theorem~\ref{th:runtime}. There, we have  $\text{deg}_v(E_i) \leq d_i$ for all nodes $v \in V$ and levels $0 \leq i \leq  L$. 
\end{corollary}

\begin{proof}
The corollary holds since
$1 \geq W_v(w) \geq W_v(w_i) =  \sum_{(u, v) \in E_i} w_i(u, v) 
 = (1/d_i) \cdot \text{deg}_v(E_i)$.
\end{proof}

Accordingly, throughout the rest of this section, we refer to $d_i$ as being the {\em degree threshold} for level $i$.

\subsection{An overview of our approach}
\label{sec:overview} We will now explain the main ideas that are needed to prove Theorem~\ref{th:old}. Due to space constraints, we will focus on getting a constant approximation in $O(\text{poly} \log n)$ update time. See the full version of the paper for the complete proof of Theorem~\ref{th:old}.  
 First, we maintain a hierarchical partition as per Theorem~\ref{th:runtime}. This gives us a $2(1+\epsilon)^2$-approximate maximum fractional matching (see Corollary~\ref{cor:approx}). Next, we give a dynamic data structure that deterministically {\em rounds} this fractional matching into an integral matching without losing too much in the approximation ratio. The main challenge is to ensure that the data structure has $O(\text{poly} \log n)$ update time, for otherwise one could simply use any deterministic rounding algorithm that works well in the static setting.

\subsubsection{An ideal skeleton} 
\label{sub:sec:ideal}

Our dynamic rounding procedure, when applied on top of the data structure used for Theorem~\ref{th:runtime}, will output a low-degree subgraph that approximately preserves the size of the maximum matching. We will then extract a large integral matching from this subgraph using Theorem~\ref{th:gupta:peng}. To be more specific, recall that $w$ is the fractional matching maintained in Theorem~\ref{th:runtime}. We will   maintain a subset of edges $E' \subseteq E$  in $O(\text{poly} \log n)$ update time that satisfies two properties. 
\begin{eqnarray}
 \text{There is a fractional matching } w' \text{ with support } E'  
 \text{ such that }  w(E) \leq c \cdot w'(E') \text{ for some constant } c \geq 1. \label{main:eq:condition:1} \\
 \label{main:eq:condition:2} \text{deg}_v(E') = O(\text{poly} \log n) \text{ for all nodes } v \in V.
\end{eqnarray}  
Equation~\ref{main:eq:condition:1}, along with Corollary~\ref{cor:approx} and Theorem~\ref{main:th:structure:bipartite},  guarantees that the subgraph  $G' = (V, E')$ preserves the size of the maximum  matching in $G = (V, E)$ within a constant factor. Equation~\ref{main:eq:condition:2}, along with Theorem~\ref{th:gupta:peng},  guarantees that we can maintain a matching $M' \subseteq E'$ in $O(\text{poly} \log n/\epsilon^2)$ update time such that $\text{Opt}(E') \leq (1+\epsilon) \cdot |M'|$. Setting  $\epsilon$ to be some small constant (say, $1/3$),  these two observations together imply that we can maintain a $O(1)$-approximate maximum matching $M' \subseteq E$   in $O(\text{poly} \log n)$ update time.

To carry out this scheme, we note that in the hierarchical partition the degree thresholds $d_i = (1+\epsilon)^i$ get smaller and smaller as we get into lower levels (see Corollary~\ref{cor:lm:deg:1}).  Thus, if most of the value of $w(E)$ is coming from the lower levels (where the maximum degree is already small), then we can easily satisfy equations~\ref{main:eq:condition:1},~\ref{main:eq:condition:2}. Specifically, we fix a level $0 \leq L' \leq L$ with degree threshold $d_{L'} = (1+\epsilon)^{L'} = \Theta(\text{poly} \log n)$, and define the edge-set $Y = \bigcup_{j=0}^{L'} E_j$. We also define  $w^+ = \sum_{i > L'} w_i$ and $w^- = \sum_{i \leq L'} w_i$. Note that $w(E) = w^+(E) + w^-(E)$.  Now, consider two possible  cases. \\

\noindent {\em Case 1. $w^-(E) \geq (1/2) \cdot w(E)$.}  In other words, most of the value of $w(E)$ is coming from the levels $[0, L']$. By Corollary~\ref{cor:lm:deg:1}, we have $\text{deg}_v(Y) \leq \sum_{j=0}^{L'} d_j \leq (L'+1) \cdot d_{L'} = \Theta(\text{poly} \log n)$ for all nodes $v \in V$. Thus,  we can simply set $w' = w^+$ and $E' = Y$ to satisfy equations~\ref{main:eq:condition:1},~\ref{main:eq:condition:2}. \\

\noindent {\em Case 2. $w^+(E) > (1/2) \cdot w(E)$.} In other words, most of the value of $w(E)$ is coming from the levels $[L'+1, L]$. To deal with this case, we introduce the concept of an {\em ideal skeleton}. See Definition~\ref{def:ideal:split}. Basically, this is a subset of edges $X_i \subseteq E_i$ that scales {\em down} the degree of every node  by a factor of $d_{i}/d_{L'}$. We will later show how to maintain a structure akin to an ideal skeleton in a dynamic setting.  Once this is done, we can easily construct a new fractional assignment $\hat{w}_i$ that scales {\em up} the weights of the surviving edges in $X_i$ by the same factor  $d_i/d_{L'}$. Since $w_i(e) = 1/d_i$ for all edges $e \in E_i$, we set $\hat{w}_i(e) = 1/d_{L'}$ for all edges $e \in X_i$. To ensure that $X_i$ is the support of the fractional assignment $\hat{w}_i$, we set $\hat{w}_i(e) = 0$ for all edges $e \in E \setminus X_i$.  Let $X = \cup_{i > L'} X_i$ and $\hat{w} = \sum_{i > L'} \hat{w}_i$. It is easy to check that this transformation preserves the weight received by a node under the fractional assignment $w^+$, that is, we have $W_v(\hat{w}) = W_v(w^+)$ for all nodes $v \in V$.   Accordingly, Lemma~\ref{main:lm:structure} implies  that if we set $w' = \hat{w}$ and $E' = X$, then equations~\ref{main:eq:condition:1},~\ref{main:eq:condition:2} are satisfied.

\begin{definition}
\label{def:ideal:split}
Consider any  level $i >L'$. An {\em ideal skeleton} at level $i$ is a subset of edges $X_i \subseteq E_i$ such that $\text{deg}(v, X_i) = (d_{L'}/d_i) \cdot  \text{deg}(v, E_i)$ for all nodes $v \in V$.  Define a fractional assignment $\hat{w}_i$ on support $X_i$ by setting $\hat{w}_i(e) = 1/d_{L'}$ for all $e \in X_i$. For every other edge $e \in E \setminus X_i$,  set $\hat{w}_i(e) = 0$. Finally,  define the edge-set $X = \bigcup_{i > L'} X_i$ and the fractional assignment $\hat{w} = \sum_{j > L'} \hat{w}_j$. 
\end{definition}

\begin{lemma}
\label{main:lm:structure}
We have: $\text{deg}_v(X) = O(\text{poly} \log n)$ for all nodes $v \in V$, and $\hat{w}(E) = w^+(E)$. The edge-set $X$ and the fractional assignment $\hat{w}$ are defined as per Definition~\ref{def:ideal:split}.
\end{lemma}
\begin{proof}
Fix any node $v \in V$. Corollary~\ref{cor:lm:deg:1} and Definition~\ref{def:ideal:split} imply that:  $\text{deg}_v(X) = \sum_{j > L'} \text{deg}_v(X_j) = \sum_{j > L'} (d_{L'}/d_j) \times \text{deg}_v(E_j) \leq \sum_{j > L'} (d_{L'}/d_j)  d_j = (L-L')  d_{L'} = O(\text{poly} \log n)$.

To prove the second part,  consider any level $i > L'$. Definition~\ref{def:ideal:split} implies that  $W_v(\hat{w}_i) = (1/d_{L'}) \cdot \text{deg}_v(X_i) = (1/d_i) \cdot \text{deg}_v(E_i) = W_v(w_i)$. Accordingly, we infer that: $W_v(\hat{w}) = \sum_{i > L'} W_v(\hat{w}_i) = \sum_{i > L'} W_v(w_i) = W_v(w^+)$. Summing over all the nodes, we get: $\sum_{v \in V} W_v(\hat{w}) = \sum_{v \in V} W_v(w^+)$. It follows that $\hat{w}(E) = w^+(E)$. 
\end{proof}

\renewcommand{\H}{\mathcal{H}}

\subsubsection{A degree-splitting procedure} 
\label{sub:sec:split}

It remains to show to how to maintain an ideal skeleton. To gain some intuition, let us first consider the problem in a static setting. Fix any level $i > L'$, and let $\lambda_i = d_i/d_{L'}$. An ideal skeleton at level $i$ is simply a subset of edges $X_i \subseteq E_i$ that scales down the degree of every node (w.r.t. $E_i$) by a factor $\lambda_i$. Can we compute such a subset $X_i$ in $O(|E_i| \cdot \text{poly} \log n)$ time? Unless we manage to solve this problem in the static setting, we cannot expect to get a dynamic data structure for the same problem  with $O(\text{poly} \log n)$ update time.  The SPLIT($E_i$) subroutine described below answers this question in the affirmative, albeit for $\lambda_i = 2$. Specifically, in linear time the subroutine outputs a subset of edges where the degree of each  node is halved. If $\lambda_i > 2$, then to get an ideal skeleton we need to repeatedly invoke this subroutine $\log_2 \lambda_i$ times: each invocation of the subroutine reduces the degree of each node by a factor of two, and hence in the final output the degree of each node is reduced by a factor of $\lambda_i$.\footnote{To highlight the main idea, we assume that $\lambda_i$ is a power of $2$.} This leads to a total runtime of $O(|E_i| \cdot \log_2 \lambda_i) = O(|E_i| \cdot \log n)$ since $\lambda_i = d_i/d_{L'} \leq d_i \leq n$.  \\

\noindent {\bf The SPLIT($\E$) subroutine, where $\E \subseteq E$.} To highlight the main idea, we assume that (1) $\text{deg}_v(\E)$ is even for every node $v \in V$, and (2) there are an even number of edges in $\E$. Hence, there exists an Euler tour on $\E$ that visits each edge exactly once. We construct such an Euler tour in $O(|\E|)$ time and then collect alternating edges of this Euler tour in a set $\mathcal{H}$. It follows that (1) $\mathcal{H} \subseteq \E$ with $|\mathcal{H}| = |\E|/2$, and (2) $\text{deg}_v(|\mathcal{H}|) = (1/2) \cdot \text{deg}_v(|\E|)$ for every node $v \in V$. The subroutine returns the set of edges $\mathcal{H}$. In other words, the subroutine runs in $O(|\E|)$ time, and returns a subgraph that halves the degree of every node. 


\subsection{From ideal to approximate skeleton}
\label{main:sec:runtime}

We now shift our attention to maintaining an ideal skeleton in a dynamic setting. Specifically, we focus on the following problem: We are given an input graph $G_i = (V, E_i)$, with $|V| = n$, that is undergoing a sequence of edge insertions/deletions. The set $E_i$ corresponds to the level $i$ edges in the hierarchical partition (see Section~\ref{sec:sodapaper}). We always have $\text{deg}_v(E_i) \leq d_i$ for all nodes $v \in V$ (see Corollary~\ref{cor:lm:deg:1}). There is a parameter $1 \leq \lambda_i = d_i/d_{L'} \leq n$. In $O(\text{poly} \log n)$ update time, we want to maintain a subset of edges $X_i \subseteq E_i$ such that $\text{deg}_v(X_i) = (1/\lambda_i) \cdot \text{deg}_v(E_i)$ for all nodes $v \in V$.   The basic building block of our dynamic algorithm will be the (static) subroutine  SPLIT($\E$) from Section~\ref{sub:sec:split}.

 Unfortunately, we will not  be able to achieve our initial goal, which was to reduce the degree of {\em every} node  by {\em exactly} the  factor $\lambda_i$ in a dynamic setting. For one thing, there might be some  nodes $v$ with $\text{deg}_v(E_i) < \lambda_i$. It is clearly not possible to reduce their degrees by  a factor of $\lambda_i$ (otherwise their new degrees will be between zero and one). Further, we will need to introduce some {\em slack} in our data structures if we want to ensure polylogarithmic update time.

 We now describe the structures that will be actually maintained by our dynamic algorithm. We maintain a partition of the node-set $V$ into two subsets: $B_i \subseteq V$ and $T_i = V \setminus B$. The nodes in $B_i$ (resp. $T_i$) are called ``big'' (resp. ``tiny''). We also maintain a subset of nodes $S_i \subseteq V$ that are called ``spurious''.  Finally, we maintain a subset of edges $X_i \subseteq E_i$. Fix two parameters $\epsilon, \delta \in (0, 1)$. For technical reasons that will become clear later on, we require that:
 \begin{equation}
 \label{eq:parameter}
 \epsilon = 1/100,  \text{ and } \delta = \epsilon^2/L
 \end{equation}
  We ensure that the following properties are satisfied.
\begin{eqnarray}
\text{deg}_v(E_i) \geq \epsilon d_i/L \text{ for all nodes } v \in B_i \setminus S_i.  \label{main:eq:algo:big} \\
\text{deg}_v(E_i) \leq 2\epsilon d_i/L \text{ for all nodes } v \in T_i \setminus S_i. \label{main:eq:algo:tiny} \\
|S_i|  \leq \delta \cdot |B_i| \label{main:eq:algo:spurious} 
\end{eqnarray}
\begin{eqnarray}
 \frac{(1-\epsilon)}{\lambda_i} \cdot \text{deg}_v(E_i) \leq \text{deg}_v(X_i) \leq \frac{(1+\epsilon)}{\lambda_i} \cdot \text{deg}_v(E_i) \nonumber \\
 \text{ for all nodes } v \in B_i \setminus S_i. \label{main:eq:algo:1} \label{main:eq:algo:deg:big} 
\end{eqnarray}
\begin{eqnarray}
\text{deg}_v(X_i) \leq (1/\lambda_i) \cdot (2\epsilon d_i/L) \text{ for all nodes } v \in T_i \setminus S_i. \label{main:eq:algo:deg:tiny} \\
\text{deg}_v(X_i) \leq (1/\lambda_i) \cdot d_i \text{ for all nodes } v \in S_i. \label{main:eq:algo:2}  \label{main:eq:algo:deg:spurious}
\end{eqnarray}

Equation~\ref{main:eq:algo:big} implies that  all the non-spurious big  nodes  have large degrees in $G_i = (V, E_i)$. On a similar note, equation~\ref{main:eq:algo:tiny} implies that all the non-spurious tiny nodes have small degrees in $G_i$. Next,  by equation~\ref{main:eq:algo:spurious}, the number of spurious nodes is negligibly small in comparison with the number of big nodes.  By equation~\ref{main:eq:algo:deg:big}, the degrees of the non-spurious big nodes  are scaled by a factor that is very close to $\lambda_i$. Thus, the non-spurious big nodes satisfy an approximate version of the degree-splitting property required by Definition~\ref{def:ideal:split}. 

Moving on, by equation~\ref{main:eq:algo:deg:tiny}, the degrees of the non-spurious tiny nodes in $X_i$ are at most  $(1/\lambda_i) \cdot (2\epsilon d_i/L) = 2 \epsilon d_{L'}/L$. Since each edge in $X_i$ receives weight $1/d_{L'}$ under the assignment $\hat{w}_i$ (see Definition~\ref{def:ideal:split}), we infer that $W_v(\hat{w}_i) = (1/d_{L'}) \cdot \text{deg}_v(X_i) \leq 2 \epsilon/L$ for all nodes $v \in T_i \setminus S_i$. Since there are at most $(L-L')$ relevant levels in a hierarchical partition, we infer that: 
\begin{equation}
\label{eq:less:90}
\sum_{i > L' : v \in T_i \setminus S_i} W_v(\hat{w}_i) \leq L \cdot (2\epsilon/L) = 2\epsilon
\end{equation}
Since for a non-spurious tiny node $v \in T_i \setminus S_i$ we have $\text{deg}_v(E_i) \leq 2\epsilon d_i/L$ (see equation~\ref{main:eq:algo:tiny}) and $W_v(w_i) = (1/d_i) \cdot \text{deg}_v(E_i) \leq 2\epsilon/L$,  an exactly similar  argument gives us:
 \begin{equation}
\label{eq:less:91}
\sum_{i > L' : v \in T_i \setminus S_i} W_v(w_i) \leq L \cdot (2\epsilon/L) = 2\epsilon
\end{equation}
Equations~\ref{eq:less:90},~\ref{eq:less:91} have the following implication: The levels where $v$ is a non-spurious tiny node contribute a negligible amount towards the weights $W_v(w^+)$ and $W_v(\hat{w})$ (see Section~\ref{sub:sec:ideal}). Hence, although we are no longer guaranteed that the degrees of these nodes will be scaled down exactly by the factor $\lambda_i$, this should not cause too much of a problem -- the sizes of the fractional assignments $w^+(E)$ and $\hat{w}(E)$ should still be close to each other as in Section~\ref{sub:sec:ideal}.

Finally, Corollary~\ref{cor:lm:deg:1} states that the degree of a node in $E_i$ is at most $d_i$. Hence, according to the definition of an ideal skeleton (see Definition~\ref{def:ideal:split}), the degree of a node in $X_i$ ought not to exceed $(1/\lambda_i) \cdot d_i = d_{L'}$. Equation~\ref{main:eq:algo:deg:spurious} ensures that the spurious nodes satisfy this property. 

If the set of edges $X_i$ satisfies the conditions described above, then we say that we have an {\em approximate-skeleton} at our disposal. This is formally stated as  follows.

\begin{definition}
\label{def:approx:split}
Fix any level $i > L'$, and suppose that there is a partition of the node-set $V$ into two subsets $B_i \subseteq V$ and $T_i = V \setminus B_i$. Further, consider another subset of nodes $S_i \subseteq V$ and a subset of edges $X_i \subseteq E_i$. The tuple $(B_i, T_i, S_i, X_i)$ is an {\em approximate-skeleton} iff it satisfies equations~(\ref{main:eq:algo:big}) -- (\ref{main:eq:algo:deg:spurious}).  
\end{definition}

One may object at this point that we have deviated from the concept of an ideal skeleton (see Definition~\ref{def:ideal:split}) so much that it will impact the approximation ratio of our final algorithm. To address this concern, we now state the following theorem whose proof appears in Section~\ref{sec:new:101}.

\begin{theorem}
\label{th:approx:split}
For each level $i > L'$, consider an approximate skeleton as per Definition~\ref{def:approx:split}. Let $X = \bigcup_{i > L'} X_i$ denote the set of edges from these approximate-skeletons.  Let $Y = \bigcup_{i \leq L'} E_i$ denote the set of edges from the remaining levels in the hierarchical partition. Then we have: 
\begin{enumerate}
\item There is a fractional matching $w'$ on support $X \cup Y$ such that $w(E) \leq  O(1) \cdot w'(X \cup Y)$. Here, $w$ is the fractional matching given by  Theorem~\ref{th:runtime}.
\item $\text{deg}_v(X \cup Y) = O(\text{poly} \log n)$ for all  $v \in V$.
\end{enumerate}
In other words, the set of edges $X \cup Y$ satisfies equations~\ref{main:eq:condition:1},~\ref{main:eq:condition:2}. 
\end{theorem}

As per the discussion immediately after equations~\ref{main:eq:condition:1},~\ref{main:eq:condition:2}, we infer the following guarantee.

\begin{corollary}
\label{cor:approx:split}
Suppose that for each level $i > L'$ there is a dynamic algorithm that maintains an approximate-skeleton  in $O(\text{poly} \log n)$ update time. Then we can also maintain a $O(1)$-approximate maximum matching in the input graph $G$ in $O(\text{poly} \log n)$ update time.  
\end{corollary}

It  remains to show how to maintain an approximate skeleton efficiently in a dynamic setting. Accordingly, we state the following theorem whose proof appears in Section~\ref{main:sec:main:algo}.

\begin{theorem}
\label{th:approx:skeleton:maintain}
Consider any level $i > L'$. In $O(\text{poly} \log n)$ update time, we can maintain an approximate-skeleton at level $i$ as per Definition~\ref{def:approx:split}. 
\end{theorem}

Corollary~\ref{cor:approx:split} and Theorem~\ref{th:approx:skeleton:maintain} imply that we can maintain a $O(1)$-approximate maximum matching in a dynamic graph in $O(\text{poly} \log n)$ update time.

\subsection{Maintaing an approximate skeleton: Proof of Theorem~\ref{th:approx:skeleton:maintain}}
\label{main:sec:main:algo}

Fix a level $i > L'$. We will show how to efficiently maintain an approximate skeleton at level $i$ under the assumption that $\lambda_i = 2$. In the full version of the paper, if $\lambda_i > 2$, then we iteratively apply the algorithm presented here $O(\log_2 \lambda_i) = O(\log_2 (d_i/d_{L'})) = O(\log n)$ times, and each iteration reduces the degrees of the nodes by a factor of two. Hence, after the final iteration, we get a subgraph that is an approximate skeleton as per Definition~\ref{def:approx:split}. 

We maintain the set of edges $E_{B_i} = \{ (u, v) \in E_i : \{u, v\} \cap B_i \neq \emptyset\}$ that are incident upon the big nodes. Further, we associate a ``status bit'' with each node $v \in V$, denoted by $\text{{\sc Status}}[v] \in \{0, 1\}$. We ensure that they satisfy two conditions: (1) If $\text{{\sc Status}}[v] = 1$, then $\text{deg}_v(E_i) \geq \epsilon d_i/L$ (which is the threshold for non-spurious big nodes in equation~\ref{main:eq:algo:big}). (2) If $\text{{\sc Status}}[v] = 0$, then $\text{deg}_v(E_i) \leq 2\epsilon d_i/L$ (which is the  threshold for non-spurious tiny nodes in equation~\ref{main:eq:algo:tiny}). Whenever an edge incident upon $v$ is inserted into (resp. deleted from) $E_i$, we update the status bit of $v$ {\em in a lazy manner} (i.e., we flip the bit only if one of the two conditions is violated). We define an ``epoch'' of a node $v$ to be the time-interval between any two consecutive flips of the bit $\text{{\sc Status}}[v]$.  Since there is a gap of $\epsilon d_i/L$ between the thresholds in equations~\ref{main:eq:algo:big},~\ref{main:eq:algo:tiny}, we infer that:
\begin{eqnarray}
\text{In any  epoch of a node } v, \text{ at least } \epsilon d_i/L \text{ edge }
\text{insertions/deletions incident upon } v \text{ takes place in } E_i. \label{eq:status}
\end{eqnarray}
Our dynamic algorithm runs in {\em ``phases''.} In the beginning of a phase, there are no spurious nodes, i.e., we have $S_i = \emptyset$. During a phase, we handle the insertion/deletion of an edge in $E_i$ as follows.

\subsubsection{Handling the insertion/deletion of an edge} Consider the insertion/deletion of an edge $(u,v)$ in $E_i$. To handle this event, we first update the set $E_i$ and the status bits of $u, v$. If  $\{u, v \} \cap B_i \neq \emptyset$ , then we also update the edge-set  $E_{B_i}$. Next, for every  endpoint $x \in \{u, v\} \setminus S_i$, we check if the node $x$ violates any of the equations~\ref{main:eq:algo:big},~\ref{main:eq:algo:tiny},~\ref{main:eq:algo:deg:big}. If yes,  then we set $S_i \leftarrow S_i \cup \{x\}$. Finally, we check if $|S_i| > \delta \cdot |B_i|$, and if yes, then we terminate the phase by calling the subroutine TERMINATE-PHASE(.). 

\subsubsection{The subroutine TERMINATE-PHASE(.)} We scan through the nodes in $S_i$. For each such node $v \in S_i$, if $\text{{\sc Status}}[v] = 1$, then we set $B_i \leftarrow B_i \cup \{v\}$, $T_i \leftarrow T_i \setminus \{v\}$, and ensure that all the edges $(u, v) \in E_i$ incident upon $v$ are included in $E_{B_i}$. Else if $\text{{\sc Status}}[v] = 0$, then we set $T_i \leftarrow T_i \cup \{v\}$, $B_i \leftarrow B_i \setminus \{v\}$, and ensure that all the edges $(u, v) \in E_i$ incident upon $v$ are excluded from $E_{B_i}$. Finally, we set $S_i \leftarrow \emptyset$ and  $X_i \leftarrow \text{SPLIT}(E_{B_i})$ (see Section~\ref{sub:sec:split}). From the next edge insertion/deletion in $E_i$, we begin a new phase. 

\subsubsection{Correctness.} At the start of a phase, clearly all the properties hold. This fact, along with the observation that an edge is never inserted into $X_i$ during the middle of a phase, implies that equations~\ref{main:eq:algo:deg:tiny},~\ref{main:eq:algo:deg:spurious} hold all the time. Whenever a node violates equations~\ref{main:eq:algo:big},~\ref{main:eq:algo:tiny},~\ref{main:eq:algo:deg:big}, we make it spurious. Finally, whenever equation~\ref{main:eq:algo:spurious} is violated, we terminate the phase. This ensures that all the properties hold all the time.

\subsubsection{Analyzing the amortized update time.} Handling an edge insertion/deletion in $E_i$ in the middle of a phase needs $O(1)$ update time. Just before a phase ends, let $b$ and $s$ respectively denote the number of big and spurious nodes. Since a phase ends only when equation~\ref{main:eq:algo:spurious} is violated, we have $s \geq \delta \cdot b$. In the subroutine TERMINATE-PHASE(.), updating the edge-set $E_{B_i}$ requires $O(s \cdot d_i)$ time, since we need to go through all the $s$ nodes in $S$, and for each such node, we need to check all the edges incident upon it (and a node can have at most $d_i$ edges incident upon it by Corollary~\ref{cor:lm:deg:1}). At this stage, the set $B_i$ consists of at most $(s+b)$ nodes, and so the set $E_{B_i}$ consists of at most $(s+b) d_i$ edges. Hence, the call to the subroutine SPLIT($E_{B_i}$) takes $O((s+b) d_i)$ time. Accordingly, the total time taken to terminate the phase is $O((s+b) d_i) = O((s+s/\delta) d_i) = O(s d_i/\delta)$. We thus reach the following conclusion: The total time spent on a given phase is equal to $O(s d_i/\delta)$, where $s$ is the number of spurious nodes at the end of the phase. Since $S_i = \emptyset$ in the beginning of the phase, we can also interpret $s$ as being the number of nodes that {\em becomes spurious} during the phase. Let $C$ denote a counter that is initially set to zero, and is incremented by one each time some node becomes spurious. From the preceding  discussion, it follows that the total update time of our algorithm, across all the phases, is at most $O( C d_i/\delta)$.  Let $t$ be the total number of edge insertions/deletions in $E_i$. We will show that $C = O(t L/(\epsilon^2 d_i))$. This will imply an amortized update time of $O((1/t) \cdot C d_i/\delta) = O(L/\epsilon^2 \delta) = O(\text{poly} \log n)$.  The last equality holds due to equation~\ref{eq:parameter}.

Note that during a phase a node $v$ becomes spurious because of one of two reasons: (1) It violated equations~\ref{main:eq:algo:big} or~\ref{main:eq:algo:tiny}.  In this case, the node's status bit is flipped. Hence, by equation~\ref{eq:status}, between any two such events, at least $\epsilon d_i/L$ edge insertions/deletions occur incident upon $v$. (2) It violates equation~\ref{main:eq:algo:deg:big}. In this event, note that in the beginning of the phase we had $v \in B_i$, $\text{deg}_v(X_i) = (1/2) \cdot \text{deg}_v(E_i) = (1/\lambda_i) \cdot \text{deg}_v(E_i)$ and  $\text{deg}_v(E_i) \geq \epsilon d_i/L$.  The former guarantee holds  since we set $X_i \leftarrow \text{SPLIT}(E_{B_i})$ at the end of the previous phase, whereas the latter guarantee follows from equation~\ref{main:eq:algo:big}. On the other hand, when the node $v$ violates equation~\ref{main:eq:algo:deg:big}, we find that $\text{deg}_v(X_i)$ differs from $(1/\lambda_i) \cdot \text{deg}_v(E_{i})$ by at least $(\epsilon/\lambda_i) \cdot \text{deg}_v(E_i) = (\epsilon/2) \cdot \text{deg}_v(E_i)$. Accordingly, during this time-interval (that starts at the beginning of the phase and ends when equation~\ref{main:eq:algo:deg:big} is violated), at least $\Omega(\epsilon^2 d_i/L)$ edge insertions/deletions incident upon $v$ must have taken place in $E_i$. To summarize, for each unit increment in $C$, we must have $\Omega(\epsilon^2 d_i/L)$ edge insertions/deletions in $E_i$. Thus, we have $C = O(t/(\epsilon^2 d_i/L)) = O(t L/(\epsilon^2d_i))$.

\subsection{Approximation guarantee from approximate skeletons: Proof of Theorem~\ref{th:approx:split}}
\label{sec:new:101}

We devote this section to the complete proof of Theorem~\ref{th:approx:split}. At a high level, the main idea behind the proof remains the same as in Section~\ref{sub:sec:ideal}. We will have to overcome several intricate obstacles, however, because now we are dealing with the relaxed notion of an approximate skeleton as defined in Section~\ref{main:sec:runtime}.

We start by focussing on the second part of  Theorem~\ref{th:approx:split}, which states that the degree of every node in $X \cup Y$ is at most $O(\text{poly} \log n)$. This is stated and proved in Lemma~\ref{lm:new:901}. 

\begin{lemma}
\label{lm:new:901}
Consider the subsets of edges $X \subseteq E$ and $Y \subseteq E$ as per Theorem~\ref{th:approx:split}. Then we have $\text{deg}_v(X \cup Y) = O(\text{poly} \log n)$ for every node $v \in V$. 
\end{lemma}

\begin{proof}
We first bound the degree of a node $v \in V$ in $X$. Towards this end, consider any level $i > L'$.  By equation~\ref{main:eq:algo:deg:spurious}, we have that $\text{deg}_v(X_i) \leq (1/\lambda_i) \cdot d_i = d_{L'}$ for all spurious nodes $v \in S_i$. By equation~\ref{main:eq:algo:deg:tiny}, we have $\text{deg}_v(X_i) \leq (1/\lambda_i) \cdot (2 \epsilon d_i/L) = 2 \epsilon d_{L'}/L \leq d_{L'}$ for all non-spurious tiny nodes $v \in T_i \setminus S_i$. Finally, by equation~\ref{main:eq:algo:deg:big} and Corollary~\ref{cor:lm:deg:1}, we have that $\text{deg}_v(X_i) \leq ((1+\epsilon)/\lambda_i) \cdot \text{deg}_v(E_i) \leq (1+\epsilon) \cdot (d_i/\lambda_i) = (1+\epsilon) \cdot d_{L'}$ for all non-spurious big nodes $v \in B_i \setminus S_i$. By Definition~\ref{def:approx:split}, a node belongs to exactly one of the three subsets -- $S_i$, $T_i \setminus S_i$ and  $B_i \setminus S_i$. Combining all these  observations, we get: $\text{deg}_v(X_i) \leq (1+\epsilon) \cdot d_{L'} = O(\text{poly} \log n)$ for all nodes $v \in V$. Now, summing over all  $i > L'$, we get: $\text{deg}_v(X) = \sum_{i > L'} \text{deg}_v(X_i) \leq (L - L') \cdot O(\text{poly} \log n) = O(\text{poly} \log n)$ for all the nodes $v \in V$. 

Next, we bound the degree of a node $v \in V$ in $Y$. Note that the degree thresholds in the levels $[0, L']$ are all at most $d_{L'}$. Specifically, for all $i \leq L'$ and  $v \in V$, Corollary~\ref{cor:lm:deg:1} implies that $\text{deg}_v(E_i) \leq d_i \leq d_{L'} = O(\text{poly} \log n)$. Hence, for every node $v \in V$, we have $\text{deg}_v(Y) = \sum_{i \leq L'} \text{deg}_v(E_i) \leq (L'+1) \cdot O(\text{poly} \log n) = O(\text{poly} \log n)$.  

To summarize, the maximum degree of a node in the edge-sets $X$ and $Y$ is $O(\text{poly} \log n)$. Hence,   for every node $v \in V$, we have: $\text{deg}_v(X \cup Y) = \text{deg}_v(X) + \text{deg}_v(Y) =  O(\text{poly} \log n)$. This concludes the proof of the lemma. 
\end{proof}

We now focus on the first part of Theorem~\ref{th:approx:split}, which guarantees the existence of a large fractional matching with support $X \cup Y$. This is stated in the lemma below. Note that Lemma~\ref{lm:new:901} and Lemma~\ref{lm:new:801} together imply Theorem~\ref{th:approx:split}.

\begin{lemma}
\label{lm:new:801}
Consider the subsets of edges $X \subseteq E$ and $Y \subseteq E$ as per Theorem~\ref{th:approx:split}. Then there exists a fractional matching $w'$ on support $X \cup Y$ such that $w(E) \leq O(1) \cdot w'(E)$. Here, $w$ is the fractional matching given by Theorem~\ref{th:runtime}.
\end{lemma}

We devote the rest of this section to the proof of Lemma~\ref{lm:new:801}.  As in Section~\ref{sub:sec:ideal}, we start by defining two fractional assignments $w^+ = \sum_{i > L'} w_i$ and $w^- = \sum_{i \leq L'} w_i$. In other words, $w^+$ captures the fractional weights assigned to the edges in levels $[L'+1, L]$ by the hierarchical partition, whereas $w^-$ captures the fractional weights assigned to the edges in the remaining levels $[0, L']$. The fractional assignment $w^+$ has support $\cup_{i > L'} E_i$, whereas the fractional assignment $w^-$ has support $\cup_{i \leq L'} E_i = Y$.  We have $w = w^+ + w^-$ and  $w(E) = w^+(E) + w^-(E)$. If at least half of the value of $w(E)$ is coming from the levels $[0, L']$, then there is nothing to prove. Specifically, suppose that $w^-(E) \geq (1/2) \cdot w(E)$. Then we can set $w' = w^-$ and obtain $w(E) \leq 2 \cdot w^-(E) = 2 \cdot w^-(Y) = 2 \cdot w'(Y) = 2 \cdot w'(X \cup Y)$.  This concludes the proof of the lemma. Accordingly, from this point onward, we will assume that at least half of the value of $w(E)$ is coming from the levels $i > L'$. Specifically, we have:
\begin{equation}
\label{eq:assume:101}
w(E) \leq 2 \cdot w^+(E)
\end{equation}
Given equation~\ref{eq:assume:101}, we will construct a fractional matching with support $X$ whose size is within a constant factor of $w(E)$.\footnote{Recall that $w$ is the fractional matching given by the hierarchical partition. See Section~\ref{sec:sodapaper}.}
We want to follow the argument applied to ideal skeletons in Section~\ref{sub:sec:ideal} (see Definition~\ref{def:ideal:split}). Accordingly, for every level $i > L'$ we now define a fractional assignment $\hat{w}_i$ with support $X_i$. 
\begin{eqnarray}
\hat{w}_i(e) & = & 1/d_{L'} \text{ for all edges } e \in X_i \nonumber \\
& = & 0  \ \ \ \ \ \  \text{ for all edges } e \in E \setminus X_i. \label{eq:w1:101}
\end{eqnarray}

We next define the fractional assignment $\hat{w}$.
\begin{eqnarray}
\hat{w} = \sum_{i > L'} \hat{w}_i \label{eq:w1:102}
\end{eqnarray}

In Section~\ref{sub:sec:ideal} (see Lemma~\ref{main:lm:structure}), we observed that $\hat{w}$ is a fractional matching with support $X$ whose size is exactly the same as $w^+(E)$. This observation, along with equation~\ref{eq:assume:101}, would have sufficed to prove Lemma~\ref{lm:new:801}. The intuition was that at every level $i > L'$, the degree of a node $v \in V$ in $X_i$ is exactly $(1/\lambda_i)$ times its degree in $E_i$. In contrast, the weight of an edge $e \in X_i$ under $\hat{w}_i$ is exactly $\lambda_i$ times its weight under $w_i$. This ensured that the weight of node  remained unchanged as we transitioned from $w_i$ to $\hat{w}_i$, that is, $W_v(w_i) = W_v(\hat{w}_i)$ for all nodes $v \in V$. Unfortunately, this guarantee  no longer holds for approximate-skeletons. It still seems natural, however, to compare the weights a node receives under these two fractional assignments $w_i$ and $\hat{w}_i$. This  depends on the status of the node under consideration, depending on whether the node belongs to the set $B_i \setminus S_i$, $T_i \setminus S_i$ or $S_i$ (see Definition~\ref{def:approx:split}). Towards this end, we derive Claims~\ref{cl:big:101},~\ref{cl:tiny:101},~\ref{cl:spurious:101}.  The first claim states that the weight of a non-spurious big node under $\hat{w}_i$ is {\em  close} to its weight under $w_i$. The second claim states that the weight of a non-spurious tiny node under $\hat{w}_i$ is {\em very small} (less than $2\epsilon/L$). The third claim states that the weight of a spurious node under $\hat{w}_i$ is at most one.

\begin{claim}
\label{cl:big:101}
For all $i > L'$ and $v \in B_i \setminus S_i$,  we have: 
$$(1-\epsilon) \cdot W_v(w_i) \leq  W_v(\hat{w}_i) \leq (1+\epsilon) \cdot W_v(w_i).$$
\end{claim}

\begin{proof}
Fix any level $i > L'$ and any node $v \in B_i \setminus S_i$. The claim follows from equation~\ref{main:eq:algo:deg:big} and the facts  below: 

\medskip
\noindent (1) $\lambda_i = d_i/d_{L'}$.  

\medskip
\noindent (2) $W_v(\hat{w}_i) = (1/d_{L'}) \cdot \text{deg}_v(X_i)$. See equation~\ref{eq:w1:101}. 

\medskip
\noindent (3) $W_v(w_i) = (1/d_i) \cdot \text{deg}_v(E_i)$. See Section~\ref{sec:sodapaper}.
\end{proof}

\begin{claim}
\label{cl:tiny:101}
For all levels $i > L'$ and non-spurious tiny nodes $v \in T_i \setminus S_i$, we have $W_v(\hat{w}_i) \leq 2\epsilon/L$.
\end{claim}

\begin{proof}
Fix any level $i > L'$ and any node $v \in T_i \setminus S_i$. The claim follows from equation~\ref{main:eq:algo:deg:tiny} and the facts below: 

\medskip
\noindent (1) $\lambda_i = d_i/d_{L'}$. 

\medskip
\noindent (2) $W_v(\hat{w}_i) = (1/d_{L'}) \cdot \text{deg}_v(X_i)$. See equation~\ref{eq:w1:101}.
\end{proof}

\begin{claim}
\label{cl:spurious:101}
For all $i > L'$ and $v \in S_i$, we have $W_v(\hat{w}_i) \leq 1$. 
\end{claim}

\begin{proof}
Fix any level $i > L'$ and any node $v \in S_i$. The claim follows from equation~\ref{main:eq:algo:deg:spurious} and the facts below: 

\medskip
\noindent (1) $\lambda_i = d_i/d_{L'}$. 

\medskip
\noindent (2) $W_v(\hat{w}_i) = (1/d_{L'}) \cdot \text{deg}_v(X_i)$. See equation~\ref{eq:w1:101}. 
\end{proof}

Unfortunately, the fractional assignment $\hat{w}$ need not necessarily be a fractional matching, the main reason being that at a level $i > L'$ the new weight $W_v(\hat{w}_i)$ of a spurious node $v \in S_i$ can be much larger than its original weight $W_v(w_i)$. Specifically, Claim~\ref{cl:spurious:101} permits that $W_v(\hat{w}_i) = 1$ for such a node $v \in S_i$. If there exists  a node $v \in V$ that belongs to $S_i$ at every level $i > L'$, then we might  have $W_v(\hat{w}) = \sum_{i > L'} W_v(\hat{w}_i) = \sum_{i > L'} 1 = (L-L') >> 1 \geq W_v(w)$. 

To address this concern regarding the weights of the spurious nodes, we switch from $\hat{w}$ to a new fractional assignment $w''$, which is defined as follows. For every level $i > L'$, we construct a  fractional assignment $w''_i$ that sets to zero the weight of every edge in $X_i$ that is incident upon a spurious node $v \in S_i$. For every other edge $e$, the weight $w''_i(e)$ remains the same as  $\hat{w}_i(e)$. Then we set $w'' = \sum_{i > L'} w''_i$. 
\begin{eqnarray}
w''_i(u, v) & = & \hat{w}_i(u, v) \text{ if } (u, v) \in X_i \text{ and } \{u, v\} \cap S_i = \emptyset \nonumber \\
& = & 0 \text{ if } (u, v) \in X_i \text{ and } \{u, v \} \cap S_i \neq \emptyset. \nonumber \\
& = & 0 \text{ else if } (u, v) \in E \setminus X_i  \label{eq:bound:new:001}
\end{eqnarray}
\begin{eqnarray}
w'' = \sum_{i > L'} w''_i \label{eq:bound:clean:1}
\end{eqnarray}

The above transformation guarantees that $W_v(w''_i) = 0$ for every spurious node $v \in S_i$ at level $i > L'$. Thus, the objection raised above regarding the weights of spurious nodes is no longer valid for the fractional assignment $w''_i$.  We now make three claims on the fractional assignments $\hat{w}$ and $w''$.

Claim~\ref{cl:new:w:101} bounds the maximum weight of a node under $w''$. Its proof appears in Section~\ref{sec:cl:new:w:101}.

\begin{claim}
\label{cl:new:w:101}
We have $W_v(w'') \leq 1+3\epsilon$ for all $v \in V$. 
\end{claim}

Claim~\ref{cl:new:w:102} states that the size of $w''$ is close to the size of $\hat{w}$. Its proof appears in Section~\ref{sec:cl:new:w:102}. 

\begin{claim}
\label{cl:new:w:102}
We have $w''(E) \geq \hat{w}(E) - 4 \epsilon \cdot  w^+(E)$. 
\end{claim}

Claim~\ref{cl:new:w:103} states that the size of $\hat{w}$ is within a constant factor of the size of $w^+$. Its proof appears in Section~\ref{sec:cl:new:w:103}. 

\begin{claim}
\label{cl:new:w:103}
We have $\hat{w}(E) \geq (1/8)  \cdot w^+(E)$. 
\end{claim}

\begin{corollary}
\label{cor:cl:new:w:103}
We have $w''(E) \geq (1/8 - 4\epsilon) \cdot w^+(E)$. 
\end{corollary}

\begin{proof}
Follows from Claims~\ref{cl:new:w:102} and~\ref{cl:new:w:103}. 
\end{proof}

To complete the proof of Lemma~\ref{lm:new:801}, we scale down the weights of the edges in $w''$ by a factor of $(1+3\epsilon)$. Specifically, we define a fractional assignment $w'$ such that:
\begin{eqnarray}
w'(e) & = & \frac{w''(e)}{(1+3\epsilon)} \text{ for all edges } e \in E. \nonumber
\end{eqnarray}

Since $w''$ has support $X$, the fractional assignment $w'$ also has support $X$, that is, $w'(e) = 0$ for all edges $e \in E \setminus X$. Claim~\ref{cl:new:w:101} implies that $W_v(w') = W_v(w'')/(1+3\epsilon) \leq 1$ for all nodes $v \in V$. Thus, $w'$ is fractional matching on support $X$. Since the edge-weights are scaled down by a factor of $(1+3\epsilon)$, Corollary~\ref{cor:cl:new:w:103} implies that:
\begin{equation}
\label{eq:last:stand:1}
w'(E) = \frac{w''(E)}{(1+3\epsilon)} \geq \frac{(1/8-4\epsilon)}{(1+3\epsilon)} \cdot w^+(E).
\end{equation}
Equations~\ref{eq:assume:101} and~\ref{eq:last:stand:1} imply that $w(E) \leq O(1) \cdot w'(E)$. This concludes the proof of Lemma~\ref{lm:new:801}.

\subsubsection{Proof of Claim~\ref{cl:new:w:101}}
\label{sec:cl:new:w:101}
Throughout the proof, we fix any given node $v \in V$. We will show that $W_v(w'') \leq 1+3\epsilon$. We start by making a simple observation:
\begin{equation}
\label{eq:less:1}
W_v(w''_i) \leq W_v(\hat{w}_i)  \text{ for all levels } i > L'. 
\end{equation}
Equation~\ref{eq:less:1} holds since we get the fractional assignment $w''_i$ from $\hat{w}_i$ by setting some edge-weights to zero and keeping the remaining edge-weights unchanged (see equation~\ref{eq:bound:new:001}).

By Definition~\ref{def:approx:split}, at every level $i > L'$ the node $v$ is part of exactly one of the three subsets -- $T_i \setminus S_i$, $B_i \setminus S_i$ and $S_i$. Accordingly, we can classify the levels into three types  depending on which of these subsets  $v$ belongs to  at that level. Further, recall that $W_v(w'') = \sum_{i > L'} W_v(w''_i)$.  We will separately bound the contributions from each type of levels towards the node-weight $W_v(w'')$.

We first bound the contribution towards $W_v(w'')$ from all the levels $i > L'$ where  $v \in T_i \setminus S_i$.
\begin{claim}
\label{cl:new:w:101:1}
We have:
$$\sum_{i > L' : v \in T_i \setminus S_i} W_v(w''_i) \leq 2\epsilon.$$
\end{claim}

\begin{proof}
Claim~\ref{cl:tiny:101} implies that:
\begin{equation}
\label{eq:less:2}
\sum_{i > L' : v \in T_i \setminus S_i} W_v(\hat{w}_i) \leq \sum_{i > L' : v \in T_i \setminus S_i} (2\epsilon/ L) \leq 2\epsilon. 
\end{equation}
The claim follows from equations~\ref{eq:less:1} and~\ref{eq:less:2}. 
\end{proof}
We next bound the contribution towards $W_v(w'')$ from all the levels $i > L'$ where $v \in B_i \setminus S_i$. 

\begin{claim}
\label{cl:new:w:101:2}
We have:
$$\sum_{i > L' : v \in B_i \setminus S_i} W_v(w''_i) \leq 1+\epsilon.$$
\end{claim}

\begin{proof}
Let $\text{LHS} = \sum_{i > L' : v \in B_i \setminus S_i} W_v(w''_i)$. We have:
\begin{eqnarray}
\text{LHS} & \leq & \sum_{i > L' : v \in B_i \setminus S_i} W_v(\hat{w}_i) \label{eq:less:3} \\
& \leq & \sum_{i > L' : v \in B_i \setminus S_i} (1+\epsilon) \cdot W_v(w_i) \label{eq:less:4} \\
& \leq & (1+\epsilon) \cdot \sum_{i = 0}^L W_v(w_i) =  (1+\epsilon) \cdot W_v(w) \nonumber \\
& \leq & (1+\epsilon) \label{eq:less:5} 
\end{eqnarray}
Equation~\ref{eq:less:3} holds because of equation~\ref{eq:less:1}. Equation~\ref{eq:less:4} follows from Claim~\ref{cl:big:101}.  Finally, equation~\ref{eq:less:5} holds since $w$ is a fractional matching (see Section~\ref{sec:sodapaper}). 
\end{proof}

Finally, we  bound the contribution towards $W_v(w'')$ from all the levels $i > L'$ where $v \in S_i$. 

\begin{claim}
\label{cl:new:w:101:3}
We have:
$$\sum_{i > L' : v \in S_i} W_v(w''_i) = 0.$$
\end{claim}

\begin{proof}
Consider any level $i > L'$ where $v \in S_i$. By equation~\ref{eq:bound:new:001}, every edge in $X_i$ incident upon $v$ has zero weight under $w''_i$, and hence  $W_v(w''_i) = 0$. The claim follows.
\end{proof}

Adding up the bounds given by Claims~\ref{cl:new:w:101:1},~\ref{cl:new:w:101:2} and~\ref{cl:new:w:101:3}, we get:
\begin{eqnarray*}
 W_v(w'') & = & \sum_{i > L'} W_v(w''_i) \\
& = & \sum_{i > L' : v \in T_i \setminus S_i} W_v(w''_i) + \sum_{i > L' : v \in B_i \setminus S_i} W_v(w''_i)  \\
& & + \sum_{i > L' : v \in S_i} W_v(w''_i) \\
& \leq & 2\epsilon + (1+\epsilon) + 0 = 1+3\epsilon.
\end{eqnarray*}
This concludes the proof of Claim~\ref{cl:new:w:101}.

\subsubsection{Proof of Claim~\ref{cl:new:w:102}}
\label{sec:cl:new:w:102}

For any given fractional assignment, the some of the node-weights is  two times the sum of the edge-weights (since each edge has two endpoints). Keeping this in mind, instead of relating the sum of the edge-weights under the fractional assignments $w'', \hat{w}$ and $w^+$ as stated in Claim~\ref{cl:new:w:102}, we will attempt to relate the sum of the node-weights under $w''$, $\hat{w}$ and $w^+$.

As we switch from the fractional assignment $\hat{w}_i$ to the fractional assignment $w''_i$ at some level $i > L'$, all we do is to set to zero the weight of any edge incident upon a spurious node in $S_i$. Hence, intuitively, the difference between the sum of the node-weights under $w'' = \sum_{i > L'} w''_i$ and $\hat{w} = \sum_{i > L'} \hat{w}_i$ should be bounded by the sum of the weights of the spurious nodes across all the levels $i > L'$. This is formally stated in the claim below.

\begin{claim}
\label{cl:new:w:102:1}
We have:
$$\sum_{v \in V} W_v(w'') \geq \sum_{v \in V} W_v(\hat{w}) -  \sum_{i > L'} \sum_{v \in S_i} 2 \cdot W_v(\hat{w}_i).$$
\end{claim}

\begin{proof}
The left hand side (LHS) of the inequality is exactly equal to two times the sum of the edge-weights under $w''$. Similarly, the first sum in the right hand side (RHS) is exactly equal to two times the sum of the edge-weights under $\hat{w}$. Finally, we can also express the second sum in the RHS as the sum of certain edge-weights under $\hat{w}$.

 Consider any edge $(x, y) \in E$. We will show that the contribution of this edge towards the LHS is at least its contribution towards the RHS, thereby proving the claim.

\medskip
\noindent
{\em Case 1.} $(x, y) \notin X_i$ for all $i > L'$. Then the edge $(x, y)$ contributes zero to the left hand side (LHS) and zero to the right hand side (RHS). 

\medskip
\noindent 
{\em Case 2.} $(x, y) \in X_i$ at some level $i > L'$, but none of the endpoints of the edge is spurious, that is, $\{x, y \} \cap S_i = \emptyset$. In this case, by equation~\ref{eq:bound:new:001}, the edge $(x, y)$ contributes $2 \cdot w''_i(x, y)$ to the LHS, $2 \cdot \hat{w}_i(x, y)$ to the first sum in the  RHS, and zero to the second sum in the RHS. Further, we have $w''_i(x, y) = \hat{w}_i(x, y)$. Hence, the edge makes exactly the same contribution towards the LHS and the RHS.

\medskip
\noindent
{\em Case 3.} $(x, y) \in X_i$ at some level $i > L'$, and at least one endpoint of the edge is spurious, that is, $\{x, y\} \cap S_i \neq \emptyset$. In this case, by equation~\ref{eq:bound:new:001}, the edge $(x, y)$ contributes zero to the LHS,  $2 \cdot \hat{w}(x,y)$ to the first sum in the RHS, and at least $2 \cdot \hat{w}(x,y)$ to the second sum in the RHS. Hence, the net contribution towards the RHS is at most zero. In other words, the contribution towards the LHS is at least the contribution towards the RHS. 
\end{proof}

At every level $i > L'$, we will now bound the sum of the weights of the spurious nodes $v \in S_i$ under $\hat{w}$ by the sum of the node-weights under $w_i$. We will use the fact that each spurious node gets weight at most one (see Claim~\ref{cl:spurious:101}), which implies that $\sum_{v \in S_i} W_v(\hat{w}_i) \leq |S_i|$. By equation~\ref{main:eq:algo:spurious}, we will upper bound the number of spurious nodes by the number of non-spurious big nodes. Finally, by equation~\ref{main:eq:algo:deg:big}, we will infer that each non-spurious big node has sufficiently large degree in $E_i$, and hence its weight under $w_i$ is also sufficiently large.  

\begin{claim}
\label{cl:new:w:102:2}
For every level $i > L'$, we have:
$$\sum_{v \in S_i} W_v(\hat{w}_i) \leq (2\delta L / \epsilon) \cdot \sum_{v \in V} W_v(w_i).$$
\end{claim}

\begin{proof}
Fix any level $i > L'$. Claim~\ref{cl:spurious:101} states that $W_v(\hat{w}_i) \leq 1$ for all nodes $v \in S_i$. Hence, we get:
\begin{equation}
\label{eq:less:11}
\sum_{v \in S_i} W_v(\hat{w}_i) \leq |S_i|
\end{equation}
Equation~\ref{main:eq:algo:spurious} implies that $|S_i| \leq \delta \cdot |B_i| \leq \delta \cdot (|B_i\setminus S_i| + |S_i|)$. Rearranging the terms, we get: $|S_i| \leq \frac{\delta}{1-\delta} \cdot | B_i \setminus S_i |$. Since $\delta < 1/2$ (see equation~\ref{eq:parameter}), we have:
\begin{equation}
\label{eq:less:12}
|S_i| \leq 2 \delta \cdot |B_i \setminus S_i|
\end{equation}
From equations~\ref{eq:less:11} and~\ref{eq:less:12}, we get:
\begin{equation}
\label{eq:less:10}
\sum_{v \in S_i} W_v(\hat{w}_i) \leq 2 \delta \cdot |B_i \setminus S_i|
\end{equation}
Now, equation~\ref{main:eq:algo:big} states that $\text{deg}_v(E_i) \geq (\epsilon d_i/L)$ for all nodes $v \in B_i \setminus S_i$. Further, in the hierarchical partition we have $W_v(w_i) = (1/d_i) \cdot \text{deg}_v(E_i)$ for all nodes $v \in V$ (see Section~\ref{sec:sodapaper}). Combining these two observations, we get: $W_v(w_i) \geq \epsilon/L$ for all nodes $v \in B_i \setminus S_i$. Summing over all nodes $v \in V$, we get:
\begin{equation}
\label{eq:less:13}
\sum_{v \in V} W_v(w_i) \geq \sum_{v \in B_i \setminus S_i} W_v(w_i) \geq (\epsilon / L) \cdot |B_i \setminus S_i|
\end{equation}
The claim follows from equations~\ref{eq:less:10} and~\ref{eq:less:13}. 
\end{proof}

\begin{corollary}
\label{cor:cl:new:w:102:2}
We have:
$$\sum_{i > L'} \sum_{v \in S_i} W_v(\hat{w}_i) \leq (2 \delta L/\epsilon) \cdot \sum_{v \in V} W_v(w^+).$$
\end{corollary}

\begin{proof}
Follows from summing  Claim~\ref{cl:new:w:102:2} over all levels $i > L'$, and noting that since $w^+ = \sum_{i > L'} w_i$, we have $W_v(w^+) = \sum_{i > L'} W_v(w_i)$ for all nodes $v \in V$. 
\end{proof}

From Claim~\ref{cl:new:w:102:1} and Corollary~\ref{cor:cl:new:w:102:2}, we get:
\begin{equation}
\label{eq:less:16}
\sum_{v \in V} W_v(w'') \geq \sum_{v \in V} W_v(\hat{w}) - (4\delta L/\epsilon) \cdot \sum_{v \in V} W_v(w^+)
\end{equation}

Since $\delta = \epsilon^2/L$ (see equation~\ref{eq:parameter}) and since the sum of the node-weights in a fractional assignment is exactly two times the sum of the edge-weights, Claim~\ref{cl:new:w:102} follows from equation~\ref{eq:less:16}.

\subsubsection{Proof of Claim~\ref{cl:new:w:103}}
\label{sec:cl:new:w:103}

Every edge $(u, v) \in X = \cup_{i > L'} X_i$ has at least one endpoint at a level $i > L'$ (see Definition~\ref{def:approx:split}). In other words, every edge in $X$ has at least one endpoint in the set $V^*$ as defined below. 

\begin{definition}
\label{def:great:v*101}
Define $V^* = \{ v \in V : \ell(v) > L'\}$ to be the set of all nodes at  levels strictly greater than $L'$. 
\end{definition}

Thus, under any given fractional assignment, the sum of the node-weights in $V^*$ is within a factor of $2$ of the sum of  the edge-weights in $X$. Since both the fractional assignments $\hat{w}$ and $w^+$ have support $X$, we get the following claim.

\begin{claim}
\label{cl:bound:002}
We have:
$$2 \cdot w^+(E) \geq \sum_{v \in V^*} W_v(w^+) \geq w^+(E).$$
$$2 \cdot \hat{w}(E) \geq \sum_{v \in V^*} W_v(\hat{w}) \geq \hat{w}(E).$$
\end{claim}

Since we want to compare the sums of the edge-weights under  $\hat{w}$  and $w^+$, by Claim~\ref{cl:bound:002} it suffices to focus on the sum of the node-weights in $V^*$ instead. 
Accordingly, we first lower bound the sum $\sum_{v \in V^*} W_v(\hat{w})$ in Claim~\ref{cl:bound:000}. In the proof, we only use the fact that for each level $i > L'$, the weight of a  node $v \in B_i \setminus S_i$ remains roughly the same under the fractional assignments $\hat{w}_i$ and $w_i$ (see Claim~\ref{cl:big:101}). 

\begin{claim}
\label{cl:bound:000}
We have: 
$$\sum_{v \in V^*} W_v(\hat{w}) \geq (1-\epsilon) \cdot \sum_{v \in V^*} \sum_{i > L' : v \in B_i \setminus S_i} W_v(w_i).$$
\end{claim}

\begin{proof}
Fix any node $v \in V^*$. By Claim~\ref{cl:big:101}, we have: $W_v(\hat{w}_i) \geq (1-\epsilon) \cdot W_v(w_i)$ at each level $i > L'$ where $v \in B_i \setminus S_i$. Summing over all such levels, we get:
\begin{equation}
\label{eq:less:50}
\sum_{i > L' : v \in B_i \setminus S_i} W_v(\hat{w}_i) \geq (1-\epsilon) \cdot \sum_{i > L' : v \in B_i \setminus S_i} W_v(w_i)
\end{equation}
Since $\hat{w} = \sum_{i > L'} \hat{w}_i$, we have: $$W_v(\hat{w}) \geq \sum_{i > L' : v \in B_i \setminus S_i} W_v(\hat{w}_i).$$ 
Hence, equation~\ref{eq:less:50} implies that:
$$W_v(\hat{w}) \geq (1-\epsilon) \cdot \sum_{i > L' : v \in B_i \setminus S_i} W_v(w_i).$$ 
We now sum the above inequality over all nodes $v \in V^*$. 
\end{proof}


It  remains to lower bound the right hand side (RHS) in Claim~\ref{cl:bound:000} by  $\sum_{v \in V^*} W_v(w^+)$. Say that a level $i > L'$ is of Type I, II or III for a node $v \in V^*$ if $v$ belongs to $B_i \setminus S_i$, $S_i$ or $T_i \setminus S_i$ respectively.  By Definition~\ref{def:approx:split}, for every node $v \in V^*$, the set of levels $i > L'$ is partitioned into these three types. The sum in the RHS of Claim~\ref{cl:bound:000} gives the contribution of the type I levels towards $\sum_{v \in V^*} W_v(w^+)$. In Claims~\ref{cl:bound:spurious:001} and~\ref{cl:bound:tiny:001}, we  respectively show that the type II and type III levels make negligible contributions towards the sum $\sum_{v \in V^*} W_v(w^+)$.  Note that the sum of these contributions from the type I, type II and type III levels {\em exactly} equals $\sum_{v \in V^*} W_v(w^+)$. Specifically, we have:
\begin{eqnarray}
\label{eq:less:69}
\sum_{v \in V^*} \sum_{i > L' : v \in B_i \setminus S_i} W_v(w_i) + \sum_{v \in V^*} \sum_{i > L' : v \in S_i} W_v(w_i)  
+ \sum_{v \in V^*} \sum_{i > L' : v \in T_i \setminus S_i} W_v(w_i) = \sum_{v \in V^*} W_v(w^+)
\end{eqnarray}

Hence, equation~\ref{eq:less:69}, Claim~\ref{cl:bound:spurious:001} and Claim~\ref{cl:bound:tiny:001} lead to the following lower bound on the right hand side of Claim~\ref{cl:bound:000}.

\begin{corollary}
\label{cor:bound:005}
We have:
$$\sum_{v \in V^*} \sum_{i > L': v \in B_i \setminus S_i} W_v(w_i) \geq (1 - 40\epsilon) \cdot \sum_{v \in V^*} W_v(w^+).$$
\end{corollary}

From Claim~\ref{cl:bound:000}, Corollary~\ref{cor:bound:005} and equation~\ref{eq:parameter}, we get:
\begin{equation}
\label{eq:less:0101}
\sum_{v \in V^*} W_v(\hat{w}) \geq (1/4) \cdot \sum_{v \in V^*} W_v(w^+)
\end{equation}
Finally, from Claim~\ref{cl:bound:002} and equation~\ref{eq:less:0101}, we infer that:
\begin{eqnarray*}
\hat{w}(E) \geq \left(1/2\right) \cdot \sum_{v \in V^*} W_v(\hat{w})  \geq (1/8) \cdot \sum_{v \in V^*} W_v(w^+) 
\geq (1/8) \cdot w^+(E)
\end{eqnarray*}
This concludes the proof of Claim~\ref{cl:new:w:103}. Accordingly, we devote the rest of this section to the proofs of Claims~\ref{cl:bound:spurious:001} and~\ref{cl:bound:tiny:001}.

\begin{claim}
\label{cl:bound:spurious:001}
We have: $$\sum_{v \in V^*} \sum_{i > L' : v \in S_i} W_v(w_i) \leq 8\epsilon \cdot \sum_{v \in V^*} W_v(w^+).$$
\end{claim}

\begin{proof}
The proof of this claim is very similar to the proof of Claim~\ref{cl:new:w:102:2} and Corollary~\ref{cor:cl:new:w:102:2}. Going through that proof, one can verify the following upper bound on the number of spurious nodes across all levels $i > L'$.
\begin{equation}
\label{eq:less:60}
\sum_{i > L'} |S_i| \leq (2\delta L/\epsilon) \cdot \sum_{v \in V} W_v(w^+)
\end{equation}
Since each $w_i$ is a fractional matching (see Section~\ref{sec:sodapaper}), we have $W_v(w_i) \leq 1$ for all nodes $v \in V$ and all levels $i > L'$. Hence, we get:
\begin{equation}
\label{eq:less:61}
\sum_{v \in V^*} \sum_{i > L' : v \in S_i} W_v(w_i) \leq \sum_{i > L'} |S_i|
\end{equation}
From equations~\ref{eq:less:60} and~\ref{eq:less:61}, we infer that:
\begin{equation}
\label{eq:less:62}
\sum_{v \in V^*} \sum_{i > L' : v \in S_i} W_v(w_i) \leq (2 \delta L/\epsilon) \cdot \sum_{v \in V} W_v(w^+)
\end{equation}
Since the sum of the node-weights under any fractional assignment is equal to twice the sum of the edge-weights, Claim~\ref{cl:bound:002} implies that:
\begin{equation}
\label{eq:less:63}
\sum_{v \in V} W_v(w^+) = 2 \cdot w^+(E) \leq 4 \cdot \sum_{v \in V^*} W_v(w^+)
\end{equation}
Claim~\ref{cl:bound:spurious:001} follows from equations~\ref{eq:parameter},~\ref{eq:less:62} and~\ref{eq:less:63}. 
\end{proof}

\begin{claim}
\label{cl:bound:tiny:001}
We have: $$\sum_{v \in V^*} \sum_{i > L' : v \in T_i \setminus S_i} W_v(w_i) \leq 32\epsilon \cdot \sum_{v \in V^*} W_v(w^+).$$
\end{claim}

\begin{proof}
Fix any node $v \in V^*$. By equation~\ref{main:eq:algo:tiny}, we have $\text{deg}_v(E_i) \leq (2\epsilon d_i/L)$ at each level $i > L'$ where $v \in T_i \setminus S_i$. Further, the fractional matching $w_i$ assigns a weight $1/d_i$ to every edge in its support $E_i$ (see Section~\ref{sec:sodapaper}). Combining these two observations, we get: $W_v(w_i) = (1/d_i) \cdot \text{deg}_v(E_i) \leq 2\epsilon/L$ at each level $i > L'$ where $v \in T_i \setminus S_i$. Summing over all such levels, we get:
\begin{equation}
\label{eq:less:70}
\sum_{i > L' : v \in T_i \setminus S_i} W_v(w_i) \leq 2\epsilon
\end{equation} 
If we sum equation~\ref{eq:less:70} over all  $v \in V^*$, then we get:
\begin{equation}
\label{eq:less:71}
\sum_{v \in V^*} \sum_{i > L' : v \in T_i \setminus S_i} W_v(w_i) \leq 2\epsilon \cdot |V^*|
\end{equation}
A node $v \in V^*$ has level $\ell(v) > L'$. Hence, all the edges incident upon this node also have level at least $L'+1$. This implies that such a node $v$ receives zero weight from the fractional assignment $w^- = \sum_{i \leq L'} w_i$, for any edge in the support of $w^-$ is at level at most $L'$. Thus, we have: $W_v(w) = W_v(w^+) + W_v(w^-) = W_v(w^+)$ for such a node $v$. Now, applying Theorem~\ref{th:runtime}, we get:
\begin{equation}
\label{eq:less:80}
1/(1+\epsilon)^2 \leq W_v(w^+)  \text{ for all nodes } v \in V^*. 
\end{equation}
Summing equation~\ref{eq:less:80} over all nodes $v \in V^*$ and multiplying both sides by $(1+\epsilon)^2$, we get:
\begin{equation}
\label{eq:less:81}
|V^*| \leq (1+\epsilon)^2 \cdot \sum_{v \in V^*} W_v(w^+)
\end{equation}
Since $(1+\epsilon)^2 \leq 4$ and $V^* \subseteq V$, equations~\ref{eq:less:71},~\ref{eq:less:81} imply that:
\begin{equation}
\label{eq:less:82}
\sum_{v \in V^*} \sum_{i > L' : v \in T_i \setminus S_i} W_v(w_i) \leq 8\epsilon \cdot \sum_{v \in V} W_v(w^+)
\end{equation}
The claim follows from equations~\ref{eq:less:63} and~\ref{eq:less:82}.
\end{proof}



\section{Bipartite graphs}
\label{main:sec:bipartite}

Ideally, we would like  to present a dynamic algorithm on bipartite graphs that proves Theorem~\ref{th:new:approx}. Due to space constraints, however, we will only  prove  a weaker result stated in Theorem~\ref{th:abbrv} and defer the complete proof of Theorem~\ref{th:new:approx} to the full version. Throughout this section, we will use the notations and concepts introduced in Section~\ref{sec:prelim}. 

\begin{theorem}
\label{th:abbrv}
There is a randomized dynamic algorithm that maintains a $1.976$ approximation to the maximum matching in a bipartite graph in $O(\sqrt{n}\log n)$ expected update time.
\end{theorem}

 In Section~\ref{sec:ed-kernel}, we present a result from~\cite{Arxiv} which shows how to maintain a $(2+\epsilon)$-approximation to the maximum matching in bipartite graphs in $O(\sqrt{n}/\epsilon^2)$ update time. In Section~\ref{sec:better:than:2}, we build upon this result and prove Theorem~\ref{th:abbrv}. In Section~\ref{sec:extensions}, we allude to the extensions that lead us to the proof of Theorem~\ref{th:new:approx} in the full version of the paper.

\subsection{$(2+\epsilon)$-approximation in $O(\sqrt{n}/\epsilon^2)$ update time}
\label{sec:ed-kernel}
The first step  is to define the concept of a {\em kernel}. Setting $\epsilon = 0, d = 1$ in Definition~\ref{def:ed-kernel},  we note that the kernel edges in a $(0,1)$-kernel forms a maximal matching -- a matching where every unmatched edge has at least one matched endpoint. For general $d$, we note that the kernel edges in a $(0,d)$-kernel forms a maximal $d$-matching -- which is a maximal subset of edges where each node has degree at most $d$. In Lemma~\ref{lm:approx:ed-kernel} and Corollary~\ref{cor:approx:ed-kernel}, we show that the kernel edges in any $(\epsilon, d)$-kernel preserves the size of the maximum matching within a factor of $2/(1-\epsilon)$. Since $d$ is the maximum degree of a node in an $(\epsilon, d)$-kernel, a $(1+\epsilon)$-approximation to the maximum matching within a kernel can be maintained in $O(d/\epsilon^2)$ update time using Theorem~\ref{th:gupta:peng}. Lemma~\ref{lm:maintain:ed-kernel} shows that the set of kernel edges themselves can be maintained in $O(n/(\epsilon d))$ update time. Setting $d = \sqrt{n}$ and combining all these observations, we get our main result in Corollary~\ref{cor:ed-kernel:main}.

\begin{definition}
\label{def:ed-kernel}
Fix any $\epsilon \in (0, 1)$, $d \in [1, n]$. Consider any subset of  nodes  $T \subset V$ in the graph $G = (V, E)$, and any subset of edges $H \subseteq E$. The pair $(T, H)$ is called an $(\epsilon, d)$-kernel of $G$ iff: (1) $\text{deg}_v(H) \leq d$ for all nodes $v \in V$, (2) $\text{deg}_v(H) \geq (1-\epsilon) d$ for all nodes $v \in T$, and (3) every edge $(u, v) \in E$ with both endpoints $u, v \in V \setminus T$ is part of the subset $H$. We define the set of nodes $T^c = V \setminus T$, and say that the nodes in $T$ (resp. $T^c$) are ``tight'' (resp. ``non-tight''). The edges in $H$ are called ``kernel edges''.
\end{definition}

\begin{lemma}
\label{lm:approx:ed-kernel}
Consider any integral matching $M \subseteq E$ and let $(T, H)$ be any $(\epsilon, d)$-kernel of $G = (V, E)$ as per Definition~\ref{def:ed-kernel}. Then there is a fractional matching $w''$ in $G$ with support $H$ such that $\sum_{v \in V} W_v(w'') \geq (1-\epsilon) \cdot |M|$. 
\end{lemma}

The proof of Lemma~\ref{lm:approx:ed-kernel} appears in Section~\ref{sec:lm:approx:ed-kernel}. 

\begin{corollary}
\label{cor:approx:ed-kernel}
Consider any $(\epsilon, d)$-kernel as per Definition~\ref{def:ed-kernel}.  We have $\text{Opt}(H) \geq (1/2) \cdot (1-\epsilon) \cdot \text{Opt}(E)$. 
\end{corollary}

\begin{proof}
Let $M \subseteq E$ be a maximum cardinality matching in $G = (V, E)$. Let $w''$ be a fractional matching with support $H$ as per Lemma~\ref{lm:approx:ed-kernel}. Since in a bipartite graph the size of the maximum cardinality matching is the same as the size of the maximum fractional matching (see Theorem~\ref{main:main:th:structure}), we get: $\text{Opt}(H) = \text{Opt}_f(H) \geq w''(H) = (1/2) \cdot \sum_{v \in V} W_v(w'') \geq (1/2) \cdot (1-\epsilon) \cdot |M| = (1/2) \cdot (1-\epsilon) \cdot \text{Opt}(E)$. 
\end{proof}

\begin{lemma}
\label{lm:maintain:ed-kernel}
In the dynamic setting, an $(\epsilon, d)$-kernel can be maintained in $O(n/(\epsilon d))$ amortized update time. 
\end{lemma}

\begin{proof}(Sketch)
When an edge $(u, v)$ is inserted into the graph, we simply check if both its endpoints are non-tight. If yes, then we insert $(u, v)$ into $H$. Next, for each endpoint $x \in \{u, v\}$, we check if $\text{deg}_x(H)$ has now become equal to $d$ due to this edge insertion. If yes, then we delete the node $x$ from $T^c$ and insert it into $T$. All these operations can be performed in constant time. 

Now, consider the deletion of an edge $(u, v)$. If both $u, v$ are non-tight, then we have nothing else to do. Otherwise, for each tight endpoint $x \in \{u, v\}$, we check if $\text{deg}_x(H)$ has now fallen below the threshold $(1-\epsilon) d$ due to this edge deletion. If yes, then we might need to change the status of the node $x$ from tight to non-tight. Accordingly, we  scan through all the edges in $E$ that are incident upon $x$, and try to insert as many of them into $H$ as possible. This step takes $\Theta(n)$ time in the worst case since the degree of the node $x$ can be $\Theta(n)$. However, the algorithm ensures that this event occurs  only after $\epsilon d$ edges incident upon $x$ are deleted from $E$. This is true since we have a {\em slack} of $\epsilon d$ between the largest and smallest possible degrees of a tight node.  Thus, we get an amortized update time of $O(n/(\epsilon d))$.  
\end{proof}

\begin{corollary}
\label{cor:ed-kernel:main}
In a bipartite graph, one can maintain a $(2+6\epsilon)$-approximation to the size of the maximum matching in $O(\sqrt{n}/\epsilon^2)$ amortized update time. 
\end{corollary}

\begin{proof}(Sketch)
We set $d = \sqrt{n}$ and maintain an $(\epsilon, d)$-kernel $(T, H)$ as per Lemma~\ref{lm:maintain:ed-kernel}. This takes $O(\sqrt{n}/\epsilon)$ update time. Next, we note that the maximum degree of a node  in $H$ is $d = \sqrt{n}$ (see Definition~\ref{def:ed-kernel}). Accordingly, we can apply Theorem~\ref{th:gupta:peng} to maintain a $(1+\epsilon)$-approximate maximum matching  $M_H \subseteq H$ in $O(\sqrt{n}/\epsilon^2)$ update time. Hence, by Corollary~\ref{cor:approx:ed-kernel}, this matching $M_H$ is a $2 (1+\epsilon)/(1-\epsilon) \leq (2 +6\epsilon)$-approximation to the maximum matching in $G = (V, E)$. 
\end{proof}

\subsubsection{Proof of Lemma~\ref{lm:approx:ed-kernel}}
\label{sec:lm:approx:ed-kernel}
First, define a fractional assignment $w$ as follows. For every edge $(u, v) \in H$ incident on a tight node, we set $w(e) = 1/d$, and for every other edge  $(u, v) \in E$, set $w(u, v) = 0$. Since each node $v \in V$ has $\text{deg}_v(H) \leq d$, it is easy to check that $W_v(w) \leq 1$ for all nodes $v \in V$. In other words, $w$ forms a fractional matching in $G$.

Next, we define another fractional assignment $w'$. First, for every node $v \in T^c$, we define a ``capacity'' $b(v) = 1 - W_v(w) \in [0, 1]$. Next, for every edge $(u, v) \in H \cap M$ whose both endpoints are non-tight,  set $w'(u, v) = \min(b(u), b(v))$. For every other edge $(u, v) \in E$, set $w'(u, v) = 0$. 

We finally define $w'' = w+w'$. Clearly, the fractional assignment $w''$ has support $H$, since for every edge $(u, v) \in E \setminus H$, we have $w(u, v) = w'(u, v) = 0$. Hence, the lemma follows from Claims~\ref{cl:ed-kernel:1} and~\ref{cl:ed-kernel:2}. 

\begin{claim}
\label{cl:ed-kernel:1}
We have $W_v(w'') \leq 1$ for all nodes $v \in V$, that is, $w''$ is a fractional matching in $G$. 
\end{claim}

\begin{proof}
If a  node $v$ is tight, that is, $v \in T$, then we have $W_v(w'') = W_v(w) + W_v(w') = W_v(w) \leq 1$. Hence, for the rest of the proof, consider any node from the remaining subset $v \in T^c = V \setminus T$. There are two cases to consider here. 

\medskip
\noindent {\em Case 1.} If $v$ is not matched in $M$, then we have $W_v(w') = 0$, and hence $W_v(w'') = W_v(w) + W_v(w') = W_v(w) \leq  1$. 

\medskip
\noindent {\em Case 2.} If $v$ is matched in $M$, then let $u$ be its mate, i.e., $(u, v) \in M$. Here, we have $W_v(w') = w'(u, v) = \min(1 -W_u(w), 1 - W_v(w)) \leq 1 - W_v(w)$. This implies that $W_v(w'') = W_v(w) + W_v(w') \leq 1$. This concludes the proof.
\end{proof}

\begin{claim}
\label{cl:ed-kernel:2}
We have $\sum_{v \in V} W_v(w'') \geq (1-\epsilon) \cdot |M|$. 
\end{claim}

\begin{proof}
Throughout the proof, fix any edge $(u, v) \in M$. We will show that $W_u(w'') + W_v(w'') \geq (1-\epsilon)$. The claim will then follow if we sum over all the edges  in $M$. 

\medskip
\noindent {\em Case 1.} The edge $(u, v)$ has at least one tight endpoint. Let $u \in T$. In this case, we have $W_u(w'') + W_v(w'') \geq W_u(w'') = W_u(w) + W_u(w') \geq W_u(w) = (1/d) \cdot \text{deg}_u(H) \geq (1-\epsilon)$.  

\medskip
\noindent {\em Case 2.} Both the endpoints of $(u, v)$ are non-tight. Without any loss of generality, let $W_u(w) \geq W_v(w)$. In this case, we have $W_u(w'') + W_v(w'') \geq W_u(w'') = W_u(w) + W_u(w') = W_u(w) + w'(u, v) = W_u(w) + \min(1- W_u(w), 1-W_v(w)) = W_u(w) + (1-W_u(w)) = 1$. This concludes the proof.   
\end{proof}

\subsection{Better than $2$-approximation}
\label{sec:better:than:2}

The approximation guarantee derived in Section~\ref{sec:ed-kernel} follows from Claim~\ref{cl:ed-kernel:2}. Looking back at the proof of this claim, we observe that we actually proved a stronger statement: Any matching $M \subseteq E$ satisfies the property that $W_u(w'') + W_v(w'') \geq (1-\epsilon)$ for all matched edges $(u, v) \in M$, where $w''$ is a fractional matching with support $H$ that depends on $M$. In the right hand side of this inequality, if we replace the term $(1-\epsilon)$ by anything larger than $1$, then we will get a better than $2$ approximation (see the proof of Corollary~\ref{cor:approx:ed-kernel}). The reason it was not possible to do so in Section~\ref{sec:ed-kernel} is as follows. Consider a matched edge $(u, v) \in M$ with $u \in T$ and $v \in T^c$. Since $u$ is tight, we have  $1-\epsilon \leq W_u(w) = W_v(w'') \leq 1$. Suppose that $W_u(w'') = 1-\epsilon$. In contrast, it might well be the case that $W_v(w)$ is very close to being zero (which will happen if $\text{deg}_v(H)$ is very small).  Let $W_v(w) \leq \epsilon$. Also note that $W_v(w'') = W_v(w) + W_v(w') = W_v(w) \leq \epsilon$ since no edge that gets nonzero weight under $w'$ can be incident on $v$ (for $v$ is already incident upon an edge in $M$ whose other endpoint is tight). Hence, in this instance we will have $W_u(w'') + W_v(w'') \leq (1-\epsilon) + \epsilon = 1$, where  $(u, v) \in M$ is a matched edge with one tight and one non-tight endpoint. 

The above discussion suggests that we ought to ``throw in'' some additional edges into our kernel -- edges whose one endpoint is tight and the other endpoint is non-tight with a very small degree in $H$. Accordingly, we introduce the notion of {\em residual edges} in Section~\ref{sec:residual}. We show that the union of the kernel edges and the residual edges preserves the size of the maximum matching within a factor of strictly less than $2$. Throughout the rest of this section, we set the values of two parameters $\delta, \epsilon$ as follows. 
\begin{equation}
\label{eq:residual:parameter}
\delta = 1/20, \epsilon = 1/2000
\end{equation}

\subsubsection{The main framework: Residual edges}
\label{sec:residual}

We maintain an $(\epsilon, d)$-skeleton $(T, H)$ as in Section~\ref{sec:ed-kernel}. We further partition the set of non-tight nodes $T^c = V \setminus T$ into two subsets: $B \subseteq T^c$ and $S = T^c \setminus B$. The set of nodes in $B$ (resp. $S$) are called ``big'' (resp. ``small''). They satisfy the following degree-thresholds: (1) $\text{deg}_v(H) \leq 2 \delta d/(1-\delta)$ for all small nodes $v \in S$, and (2) $\text{deg}_v(H) \geq (2 \delta -\epsilon) d/ (1-\delta)$ for all big nodes $v \in B$. Let $E^r = \{ (u, v) \in E : u \in T, v \in S \}$ be the subset of edges joining the tight and the small nodes. We maintain a {\em maximal} subset of edges $M^r \subseteq E^r$ subject to the following constraints: (1) $\text{deg}_v(M^r) \leq 1$ for all tight nodes $v \in T$ and (2) $\text{deg}_v(M^r) \leq 2$ for all small nodes $v \in S$. The edges in $M^r$ are called the ``residual edges''. The degree of a node in $M^r$ is called its ``residual degree''. The corollary below follows from the maximality of the set $M^r \subseteq E^r$. 

\begin{corollary}
\label{cor:residual}
If an edge $(u, v) \in E^r$ with $u \in T$, $v \in S$ is not in $M^r$, then either $\text{deg}_v(M^r) = 1$ or $\text{deg}_u(M^r) = 2$.
\end{corollary}

\begin{lemma}
\label{lm:residual:update:time}
We can maintain the set of kernel edges $H$ and the residual edges $M^r$ in $O(n \log n/(\epsilon d))$ update time.
\end{lemma}

\begin{proof}(Sketch)
We maintain an $(\epsilon, d)$-kernel as per  the proof of Lemma~\ref{lm:maintain:ed-kernel}. We  maintain the node-sets $B, S \subseteq T^c$ and the edge-set $E^r$ in the same lazy manner:  A node changes its status only after $\Omega(\epsilon d)$ edges incident upon it are either inserted into or deleted from $G$ (since $\delta$ is a constant), and when that happens we might need to make $\Theta(n)$ edge insertions/deletions in $E^r$. This gives  the same amortized update time of $O(n/(\epsilon d))$ for maintaining the edge-set $E^r$.  

In order to maintain the set of residual edges $M^r \subseteq E^r$, we use a simple modification of the dynamic algorithm of Baswana et al.~\cite{BaswanaGS15} that maintains a maximal matching in $O(\log n)$ update time. This holds since $M^r$ is a {\em maximal $b$-matching} in $E^r$ where each small node can have at most two matched edges incident upon it, and each tight node can have at most one matched edge incident upon it. 
\end{proof}

\begin{lemma}
\label{lm:approx:residual}
Fix any $(\epsilon, d)$ kernel $(T, H)$ as in Section~\ref{sec:ed-kernel},  any set of residual edges $M^r$ as in Section~\ref{sec:residual}, and any matching $M \subseteq E$. Then we have a fractional matching $w''$ on support $H \cup M^r$ such that $\sum_{v \in V} W_v(w'') \geq (1+\delta/4) \cdot |M|$. 
\end{lemma}

\noindent {\em Roadmap for the rest of this section.} The statement of Lemma~\ref{lm:approx:residual} above is similar to that of Lemma~\ref{lm:approx:ed-kernel} in Section~\ref{sec:ed-kernel}. Hence, using a similar argument as in Corollary~\ref{cor:approx:ed-kernel}, we infer that the set of edges $M^r \cup H$ preserves the size of the maximum matching within a factor of $2/(1+\delta/4)$. Since $\text{deg}_v(M^r \cup H) = \text{deg}_v(H) + \text{deg}_v(M^r) \leq d+2$ for all nodes $v \in V$ (see Definition~\ref{def:ed-kernel}), we can maintain a $(1+\epsilon)$-approximate maximum matching in $H \cup M^r$ using Theorem~\ref{th:gupta:peng} in $O(d/\epsilon^2)$ update time. This matching will give a $2(1+\epsilon)/(1+\delta/4) = 1.976$-approximation to the size of maximum matching in $G$ (see equation~\ref{eq:residual:parameter}). The total update time is $O(d/\epsilon^2 + n\log n/(\epsilon d))$, which becomes $O(\sqrt{n} \log n)$ if we set $d = \sqrt{n}$ and plug in the value of $\epsilon$. This concludes the proof of Theorem~\ref{th:abbrv}. 

It remains to prove Lemma~\ref{lm:approx:residual}, which is done in Section~\ref{sec:lm:approx:residual}.

\subsubsection{Proof of Lemma~\ref{lm:approx:residual}}
\label{sec:lm:approx:residual}

We will define four fractional assignments $w, w^r, w', w''$. It might be instructive to contrast the definitions of the fractional assignments $w, w'$ and $w''$ here with  Section~\ref{sec:lm:approx:ed-kernel}.

\medskip
\noindent {\em The fractional assignment $w$:} Set $w(e) = (1-\delta)/d$ for every edge $e \in H$ incident upon a tight node. Set $w(e) = 0$ for every other edge $e \in E$. Hence, we have $W_v(w) = ((1-\delta)/d) \cdot \text{deg}_v(H)$ for all nodes $v \in V$. Accordingly, recalling the bounds on $\text{deg}_v(H)$ for various types of nodes (see Definition~\ref{def:ed-kernel}, Section~\ref{sec:residual}), we get:
\begin{eqnarray}
 W_v(w) \leq (1-\delta) \text{ for all nodes } v \in V. \label{eq:residual:1} \\
  W_v(w) \leq 2\delta \text{ for all small nodes } v \in S. \label{eq:residual:2} \\
W_v(w) \geq (1-\delta) (1 -\epsilon) \text{ for all tight nodes } v \in T. \label{eq:residual:3} \\
 W_v(w) \geq 2\delta - \epsilon \text{ for all big nodes } v \in B. \label{eq:residual:4}
 \end{eqnarray} 
 
 \medskip
 \noindent {\em The fractional assignment $w^r$:} Set $w^r(e) = \delta$ for every residual edge $e \in M^r$. Set $w^r(e) = 0$ for every other edge $e \in E$. 
 
 \medskip
 \noindent {\em The fractional assignment $w'$:} For every node $v \in T^c$, we define a ``capacity'' $b(v)$ as follows. If $v \in B \subseteq T^c$, then $b(v) = 1 - W_v(w)$. Else if $v \in S = T^c \setminus B$, then $b(v) = 1 - 2 \delta - W_v(w)$. Hence, equations~\ref{eq:residual:1},~\ref{eq:residual:2} imply that:  
 \begin{eqnarray}
 b(v) \geq \delta \text{ for all big nodes } v \in B. \label{eq:residual:big} \\
 b(v) \geq 1 - 4 \delta \text{ for all small nodes } v \in S. \label{eq:residual:5}
 \end{eqnarray}
For every edge $(u, v) \in M$ with $u, v \in T^c = V \setminus T$, we set $w'(u, v) = \min(b(u), b(v))$.  For every other edge $e \in E$, we set $w'(e) = 0$. By Definition~\ref{def:ed-kernel}, every edge whose both endpoints are non-tight is  a kernel edge. Hence, an edge gets nonzero weight under $w'$ only if it is part of the kernel. 
 
 \medskip
 \noindent {\em The fractional assignment $w''$:} Define $w'' = w+w^r+w'$. 
 
 \medskip
 \noindent {\em Roadmap for the rest of the proof.}
Each of the fractional assignments $w, w^r, w'$ assigns zero weight to every edge $e \in E \setminus (H \cup M^r)$. Hence, the fractional assignment $w'' = w+w^r+w'$ has support $H \cup M^r$. In Claim~\ref{cl:residual:1}, we show that $w''$ is a fractional matching in $G$. Moving on, in Definition~\ref{def:new:matching}, we partition the set of matched edges in $M$ into two parts. The subset $M_1 \subseteq M$ consists of those matched edges that have one tight and one small endpoints, and the subset $M_2 = M \setminus M_1$ consists of the remaining edges.  In Claims~\ref{cl:residual:3} and~\ref{cl:residual:4}, we relate the node-weights under $w, w', w^r$ with the sizes of the matchings $M_1$ and $M_2$. Adding up the bounds from Claims~\ref{cl:residual:3} and~\ref{cl:residual:4}, Corollary~\ref{cor:residual:10} lower bounds the sum of the node-weights under $w''$ by the size of the matching $M$. Finally,  Lemma~\ref{lm:approx:residual} follows from Claim~\ref{cl:residual:1} and Corollary~\ref{cor:residual:10}.
 
 \begin{claim}
 \label{cl:residual:1}
 We have $W_v(w'') \leq 1$ for all nodes $v \in V$. 
 \end{claim}

\begin{proof}
Consider any node $v \in V$. By equation~\ref{eq:residual:1}, we have: $W_v(w) \leq 1-\delta$. Also note that $W_v(w^r) \leq \delta$ for all tight nodes $v \in T$, $W_v(w^r) \leq 2 \delta$ for all small nodes $v \in S$, and $W_v(w^r) = 0$ for all big nodes $v \in B$. This holds since the degree (among the edges in $M^r$) of  a tight, small and big node is at most one, two and zero respectively. Next, note that for all nodes $v \in T^c$, we have $W_v(w') \leq b(v)$. This holds since  there is at most one edge in $M \cap H$ incident upon $v$ (since $M$ is a matching). So at most one edge incident upon $v$ gets a nonzero weight under $w'$, and the weight of this edge is at most $b(v)$. Finally, note that every edge with nonzero weight under $w'$ has both the endpoints in $T^c$. Hence, we have $W_v(w') = 0$ for all tight nodes $v \in T$. To complete the proof, we now consider three possible cases.

\medskip
\noindent {\em  Case 1.  $v \in T$.} Here, $W_v(w'') = W_v(w) + W_v(w^r) + W_v(w') = W_v(w) + W_v(w^r)  \leq W_v(w) + \delta \leq (1-\delta) + \delta = 1$.  

\medskip
\noindent {\em Case 2. $v \in S$.} Here, $W_v(w'') = W_v(w) + W_v(w^r) + W_v(w') \leq W_v(w) + 2\delta + b(v) = W_v(w) + 2 \delta + (1 - 2 \delta - W_v(w)) = 1$.  

\medskip
\noindent {\em Case 3. $v \in B$.} Here, $W_v(w'') = W_v(w) + W_v(w^r) + W_v(w') = W_v(w) + W_v(w') \leq W_v(w) + b(v) = W_v(w) + (1 - W_v(w)) = 1$. 
\end{proof}

\begin{definition}
\label{def:new:matching}
Partition the set of edges in $M$ into two parts: $M_1 = \{ (u, v) \in M : u \in T, v \in S \}$ and $M_2 = M \setminus M_1$. 
\end{definition}

\begin{claim}
\label{cl:residual:3}
Recall Definition~\ref{def:new:matching}. We have: 
$$\sum_{(u, v) \in M_2} W_u(w+w') + W_v(w+w') \geq (1+\delta/4) \cdot |M_2|.$$
\end{claim}

\begin{proof}
Fix any edge $(u, v) \in M_2$, and let $\text{LHS} = W_u(w+ w') + W_v(w+w')$. We will show that $\text{LHS} \geq (1+\delta/4)$. The claim will then follow if we sum over all such edges  $M_2$. We recall equation~\ref{eq:residual:parameter} and consider four possible cases.

\medskip
\noindent {\em Case 1. Both endpoints are tight, that is, $u, v \in T$.} Here, from equation~\ref{eq:residual:3} we get: $\text{LHS} \geq 2 \cdot (1-\delta -\epsilon +\delta \epsilon) \geq (1+\delta/4)$. 

\medskip
\noindent {\em Case 2. One endpoint is tight, and one endpoint is big, that is,  $u \in T$, $v \in B$.} Here, from equations~\ref{eq:residual:3},~\ref{eq:residual:4} we get: $\text{LHS} \geq (1- \delta - \epsilon + \delta \epsilon) + (2\delta - \epsilon) \geq (1+\delta - 2\epsilon) \geq 1+\delta/4$. 

\medskip
\noindent {\em Case 3. Both endpoints are non-tight, that is, $u , v \in B \cup S$.} Without any loss of generality, let $b(u) \geq b(v)$. Note that $(u, v) \in H$ since both $u, v \in T^c$, and hence $(u, v) \cap M_2 \cap H$. Thus, we have $W_u(w') = W_v(w') = w'(u, v) = b(v)$ since at most one matched edge can be incident upon a node. Now, either $v \in B$ or $v \in S$. In the former case, from equation~\ref{eq:residual:5} we get: $\text{LHS} \geq W_u(w') + W_v(w') = 2 \cdot b(v) \geq 2 (1-4\delta) \geq (1+\delta/4)$. In the latter case, from equation~\ref{eq:residual:big} we get: $\text{LHS} \geq (W_v(w) + W_v(w')) + W_u(w') = (b(v) + W_v(w)) + b(v) =  1  + b(v) \geq 1+\delta \geq 1+\delta/4$. 
\end{proof}

\begin{claim}
\label{cl:residual:4}
Recall Definition~\ref{def:new:matching}. We have:
$$\sum_{(u,v) \in M_1} \left(W_u(w) + W_v(w)\right) + \sum_{v \in V} W_v(w^r)  \geq  (1+\delta/4) \cdot |M_1|.$$
 \end{claim}
 
\begin{proof}
Every edge $(u, v) \in M_1$ has one endpoint $u \in T$. Thus, Applying equation~\ref{eq:residual:3}   we get: $W_u(w) + W_v(w) \geq W_u(w) \geq 1 - \delta - \epsilon$. Summing over all such edges, we get:
\begin{equation}
\label{main:eq:approx:20}
\sum_{(u, v) \in M_1} W_u(w) + W_v(w) \geq (1 - \delta - \epsilon) \cdot |M_1|
\end{equation}
Recall  that   $\text{deg}_u(M^r) \leq 1$ for every tight node $u \in T$. Accordingly, we classify each tight node as being either ``full'' (in which case $\text{deg}_u(M^r) = 1$) or ``deficient'' (in which case $\text{deg}_u(M^r) = 0$).  Further, recall that each edge $(u, v) \in M_1$ has one tight and one small endpoints. We say that an edge $(x, y) \in M_1$ is {\em deficient} if the  tight endpoint of the edge is deficient. Now, consider any deficient edge $(x, y) \in M_1$ where $x \in T$ and $y \in S$. Since $\text{deg}_x(M^r) = 0$, it follows that  $(x, y) \in E^r \setminus M^r$.  From the maximality of $M^r$, we infer that $\text{deg}_y(M^r) = 2$. Accordingly, there must be two edges  $(x', y), (x'', y) \in M^r$ with $x', x'' \in T$.   It follows that  both the nodes $x', x''$ are full. We say that the tight nodes $x', x''$ are {\em conjugates} of the deficient edge $(x, y) \in M_1$. In other words, we have shown that every deficient edge in $M_1$ has two conjugate tight nodes. Further,  the same tight node $x'$ cannot be a conjugate of two different deficient edges in $M_1$, for otherwise each of those deficient edges will contribute one towards $\text{deg}_{x'}(M^r)$, and we will get $\text{deg}_{x'}(M^r) \geq 2$, which is a contradiction. Thus, a simple counting argument implies that the number of conjugate tight nodes is exactly twice the number of deficient matched edges in $M_1$. Let $D(M_1), F, C$ respectively denote   the set of deficient matched edges in $M_1$,  the set of full tight nodes and the set of conjugate tight nodes. Thus, we get:
\begin{equation}
\label{main:eq:approx:30}
T \supseteq F \supseteq C, \ \ D(M_1) \subseteq M_1, \text{ and } |C| = 2 \cdot |D(M_1)| 
\end{equation}
Now, either $|D(M_1)| \leq (1/3) \cdot |M_1|$ or $|D(M_1)| > (1/3) \cdot |M_1|$. In the former case, at least a $(2/3)^{rd}$ fraction of the edges in $M_1$ are not deficient, and each such edge has one tight endpoint that is full. Thus, we get $|F| \geq (2/3) \cdot |M_1|$. In the latter case, from equation~\ref{main:eq:approx:30} we get $|F| \geq |C| = 2 \cdot |D(M_1)| > (2/3) \cdot |M_1|$. Thus, in either case we have $|F| \geq (2/3) \cdot |M_1|$. Since each node $v \in F \subseteq T$ has $\text{deg}_v(M^r) = 1$, and since each edge $e \in M^r$ has weight $w^r(e) = \delta$, it follows that $W_v(w^r) = \delta$ for all nodes $v \in F \subseteq T$. Hence, we get $\sum_{v \in T} W_v(w^r) \geq 
\delta \cdot |F| \geq (2\delta/3) \cdot |M_1|$. Next, we note that each edge in $M^r$ contributes the same amount $\delta$ towards the weights of both its endpoints -- one tight and the other small. Thus, we have:
$$\sum_{v \in S} W_v(w^r) = \sum_{v \in T} W_v(w^r) \geq (2\delta/3) \cdot |M_1|.$$ Since $B \cup S \subseteq V$ and $B \cap S = \emptyset$, we get: 
$$\sum_{v \in V} W_v(w^r) \geq \sum_{v \in B \cup S} W_v(w^r) \geq (4\delta/3) \cdot |M_1|.$$ This inequality, along with equation~\ref{main:eq:approx:20}, gives us:
\begin{eqnarray*}
\sum_{(u, v) \in M_1} \left(W_u(w) + W_v(w) \right) + \sum_{v \in V} W_v(w^r)  \\ 
\geq (1-\delta-\epsilon) \cdot |M_1| + (4\delta/3) \cdot |M_1|  = (1+\delta/3 - \epsilon) \cdot |M_1| \\
\geq (1+\delta/4) \cdot |M_1|.
\end{eqnarray*}
The last inequality follows from equation~\ref{eq:residual:parameter}. 
\end{proof}

\begin{corollary}
\label{cor:residual:10}
We have: 
$$\sum_{v \in V} W_v(w'') \geq (1+\delta/4) \cdot |M|.$$ 
\end{corollary}

\begin{proof}
Since $|M| = |M_1| + |M_2|$, the corollary follows from adding the inequalities stated in Claims~\ref{cl:residual:3} and~\ref{cl:residual:4}, and noting that no node-weight under $w''$ is counted twice in the left hand side.
\end{proof}

\subsection{Extensions}
\label{sec:extensions}
We gave a randomized algorithm for maximum bipartite matching that maintains a better than $2$ approximation in $O(\sqrt{n} \log n)$ update time. In the full version of the paper, we derandomize this scheme using the following idea. Instead of applying the randomized maximal matching algorithm from~\cite{BaswanaGS15} for maintaining the set of residual edges $M^r$,  we maintain a {\em residual fractional matching} using the deterministic algorithm from~\cite{Arxiv} (see Theorem~\ref{th:runtime}).  To carry out the approximation guarantee analysis, we have to change the proof of Lemma~\ref{lm:approx:residual} (specifically, the proof of Claim~\ref{cl:residual:4}). 

To get arbitrarily small polynomial update time, we maintain a partition of the node-set into multiple levels. The top level consists of all the tight nodes (see Definition~\ref{def:ed-kernel}). We next consider the subgraph induced by the non-tight nodes. Each edge in this subgraph is a kernel edge (see Definition~\ref{def:ed-kernel}). Intuitively, we split the node-set  of this  subgraph again into two parts by defining a kernel within this subgraph. The tight nodes we get in this iteration forms the next level in our partition of $V$. We keep doing this for $K$ levels, where $K$ is a sufficiently large integer. We show that (a) this structure can be maintained in $O(n^{2/K})$ update time, and (b) by combining the fractional matchings from all these levels, we can get an $\alpha_K$ approximate maximum fractional matching in $G$, where $1 \leq \alpha_K < 2$. By Theorem~\ref{main:main:th:structure}, this gives $\alpha_K$-approximation to the size of the maximum integral matching in $G$. See the full version of the paper for the details.

\section{Open Problems}
\label{sec:open}

In this paper, we presented two deterministic dynamic algorithms for maximum matching. Our first algorithm maintains a $(2+\epsilon)$-approximate maximum matching in a general graph in $O(\text{poly} (\log n, 1/\epsilon))$ update time. The exponent hidden in the polylogorithmic factor of the update time, however,  is rather huge. It will be interesting to bring down the update time of this algorithm to $O(\log n/\epsilon^2)$ without increasing the approximation factor. This will match the update time  in~\cite{Arxiv} for maintaining a  {\em fractional} matching. 

We also  showed how to maintain a better than $2$ approximation to the size of the maximum matching on bipartite graphs in $O(n^{2/K})$ update time, for every sufficiently large integer $K$. The approximation ratio  approaches $2$ as $K$ becomes large. The main open problem here is to design a dynamic algorithm that gives better than $2$ approximation in polylogarithmic update time. This remains open even on bipartite graphs and even if one allows randomization. 


\bibliographystyle{plain}
\bibliography{citations}

\newpage
\part{DYNAMIC ALGORITHM FOR GENERAL GRAPHS: FULL DETAILS}
\newpage

\section{Preliminaries}
\label{label:general:sec:prelim}

In this part of the writeup, we will show how to maintain a $(2+\epsilon)$-approximate maximum matching in a general graph with $O(\text{poly} (\log n, 1/\epsilon))$ update time. \\

\noindent {\bf Notations.} 

\medskip
\noindent A ``matching'' $M \subseteq E$ in a graph $G = (V, E)$ is a subset of edges that do not share any common endpoints. The ``size'' of a matching is the number of edges contained in it. We will also be interested in the concept of a ``fractional matching''. Towards this end, consider any weight function $w : E \rightarrow [0,1]$ that assigns fractional weights to the edges in the graph. Given any such weight function, we  let $W(v, E') = \sum_{(u,v) \in E'} w(u,v)$ denote the total weight a node $v \in V$ receives from its incident edges in $E' \subseteq E$. We also define $w(E') = \sum_{e \in E'} w(e)$  to be the sum of the weights assigned to the edges in $E' \subseteq E$. Now, the weight function $w$ is a ``fractional matching'' iff we have $W(v, E) \leq 1$ for all $v \in V$. The ``value'' of this  fractional matching $w$ is given by $w(E)$.  Given a graph $G = (V, E)$, we will let $\text{{\sc Opt}}_G$ (resp. $\text{{\sc Opt}}^f_G$) denote the maximum ``size of a matching''  (resp. ``value of a fractional matching'') in $G$. Then it is well known that $\text{{\sc Opt}}_G \leq \text{{\sc Opt}}^f_G \leq (3/2) \cdot \text{{\sc Opt}}_G$. To see an example where the second inequality is tight, consider a triangle-graph and assign a weight $1/2$ to each edge in this triangle. This gives a fractional matching of value $3/2$, whereas every matching in this graph is of size one.  We say that a fractional matching $w$ is ``$\lambda$-maximal'', for $\lambda \geq 1$, iff for every edge $(u,v) \in E$, we have $W(u, E) + W(v, E) \geq 1/\lambda$.  The value of a $\lambda$-maximal fractional matching in $G$ is a $2\lambda$-approximation to $\text{{\sc Opt}}_G^f$. Throughout the paper, we will use the symbol $\text{deg}(v, E')$ to denote the number of edges in a subset $E' \subseteq E$ that are incident upon a node $v \in V$. 


\medskip
\noindent {\bf  Maintaining a large fractional matching:}  

\medskip
\noindent
The paper~\cite{Arxiv} maintains a $(2+\epsilon)$-approximate maximum fractional matching in $O(\log n/\epsilon^2)$ update time. We will now summarize their result.  Consider an input graph $G = (V, E)$ with $|V| = n$ nodes and $|E| = m$ edges. Fix any two constants $\alpha, \beta \geq 1$ and define $L = \lceil \log_{\beta} (n/\alpha) \rceil$. An ``$(\alpha, \beta)$-partition'' of the input graph $G = (V, E)$ partitions the node-set $V$ into $(L+2)$ subsets $\{V_{i}\}, i \in \{-1, 0, 1, \ldots, L\}$. We say that the nodes belonging to the subset $V_i$ are in ``level $i$''. We denote the level of a node  $v$ by $\ell(v)$, i.e.,  $v \in V_i$ iff $\ell(v) = i$. We next define the ``level of an edge'' $(u,v)$ to be $\ell(u,v) = \max(\ell(u), \ell(v))$. In other words, the level of an edge is the maximum level of its endpoints. We let $E_i = \{ e \in E : \ell(e) = i\}$ denote the set of edges at level $i$, and define the subgraph $G_i = (V, E_i)$. This partitions the edge-set $E$  into $(L+2)$ subsets $\{E_i\}, i \in \{-1, 0, 1, \ldots, L\}$. Finally,  for each level $i$, we define the value $d_i = \beta^i \cdot (\alpha \beta)$, and assign a weight $w(e) = 1/d_i$ to  all the edges $e \in E_i$. We  require that an $(\alpha, \beta)$-partition satisfy the invariants below.

\begin{invariant}
\label{label:general:inv:partition:edge}
For every edge $e \in E$, we have $\ell(e) \geq 0$. In other words, every edge has at least one endpoint $v$ at level $\ell(v) \geq 0$, or, equivalently $E_{-1} = \emptyset$. This means that  $V \setminus V_{-1}$ forms a vertex cover in $G$. 
\end{invariant}

\begin{invariant}
\label{label:general:inv:partition}
We have $W(v, E) \in [1/(\alpha \beta), 1]$ for all nodes $v \in V \setminus V_{-1}$, and $W(v, E) \leq 1$ for all nodes $v \in V_{-1}$. This, along with Invariant~\ref{label:general:inv:partition:edge}, means that $w$ is an $(\alpha \beta)$-maximal fractional matching in $G$. 
\end{invariant}

\begin{corollary}
\label{label:general:cor:lm:deg:1}
In an $(\alpha, \beta)$-partition, every node $v \in V$ has  $\text{deg}(v, E_i) \leq d_i$ for all levels $i \in [0,  L]$. Thus, we refer to $d_i$ as being the ``degree-threshold'' of level $i$.
\end{corollary}

\begin{proof}
Invariant~\ref{label:general:inv:partition:edge} implies that  $W(v, E_i) =  \text{deg}(v, E_i) \cdot (1/d_i) \leq 1$
\end{proof}

\begin{theorem}\cite{Arxiv}
\label{label:general:th:runtime}
 Set $\alpha = 1+3\epsilon$ and $\beta = 1+\epsilon$, which means  $L = O(\log n/\epsilon)$. We can maintain the node-sets $\{V_i\}$, the edge-sets $\{E_i\}$ and the edge-weights $\{w(e)\}$ in an $(\alpha, \beta)$-partition of $G$  in $O(\log n/\epsilon^2)$  update time. Hence, the edge-weights   form a   $(2\alpha \beta)$-approximate maximum fractional matching in $G$. 
\end{theorem}

\medskip
\noindent {\bf Maintaining a near-optimal matching in a bounded degree graph:} 

\medskip

\noindent Gupta and Peng~\cite{GuptaP13}  maintains a $(1+\epsilon)$-approximate maximum matching  in $O(\sqrt{m}/\epsilon^2)$  time. The next theorem follows as an immediate consequence of their algorithm.

\begin{theorem}
\label{label:general:th:gupta:peng}
If we are guaranteed that the maximum degree in a graph $G = (V, E)$ never exceeds some threshold $d$, then we can maintain a $(1+\epsilon)$-approximate maximum matching in $G$ in $O(d/\epsilon^2)$ update time. 
\end{theorem}

\noindent {\bf Problem definition.}

\medskip
\noindent
 The input graph $G = (V, E)$ is changing via a sequence of ``edge-updates'', where the term ``edge-update'' refers to the insertion/deletion of an edge in $E$. The node-set $V$ remains fixed. In the beginning, we have an empty graph, i.e., $E= \emptyset$. We will show how to maintain a $(2+\epsilon)$-approximate maximum matching in $G$ in $O(\text{poly} (\log n, 1/\epsilon))$ amortized update time, where $n = |V|$ is the number of nodes in the input graph.

\subsection{Setting some parameter values}
\label{label:general:sec:parameter}
Throughout this paper, we  will be dealing with the parameters defined below. We let $n = |V|$ denote the number of nodes in the input graph. Basically, we treat $\epsilon$ and $\gamma$ as small positive constants, and $\delta = \Theta(1/\text{poly} \log n)$. It should be clear to the reader that as long as the number of nodes $n$ is sufficiently large, all the equations stated below will be satisfied. We will use these equations throughout the rest of the paper. 
\begin{eqnarray}
\label{label:general:eq:mu}
0 < \epsilon < 1 \\
\label{label:general:eq:alpha} 
\alpha = 1+3\epsilon \\
\label{label:general:eq:beta}
\beta = 1+\epsilon \\
L = \lceil \log_{\beta} (n/\alpha) \rceil \label{label:general:eq:first:first:L} \\
\label{label:general:eq:L} \label{label:general:eq:L:floor}
L \text{ is a positive integer and } L = O(\log n/\epsilon) \text{ and } \lfloor L^4/2 \rfloor \geq L^4/4 \geq (1/\epsilon). \\
\label{label:general:eq:first}\label{label:general:eq:delta:100}
0 < \delta < 1 \\
\label{label:general:eq:d}
d \text{ is a positive integer and } L^4 \leq d \leq n. \\
\label{label:general:eq:L'}
L'  = \left\lceil \log_{\beta} \left(\frac{2 L^4}{\alpha \beta}\right) \right\rceil \\
\label{label:general:eq:Li} 
L_d = \left\lceil \log_2 \left(\frac{d}{L^4}\right) \right \rceil  \\
\label{label:general:eq:Li:lambda}
2^{L_d} = \frac{d}{\lambda_d L^4}  \text{ where } 1/2 \leq \lambda_d \leq 1. \\
L = O(\log n/\epsilon) \text{ and } L_d \leq L \label{label:general:eq:lastt} \\
\label{label:general:eq:imp:1}
\delta < 1/2 
\end{eqnarray}
\begin{eqnarray}
\label{label:general:eq:imp:4} 
L \geq L_{d_i} \text{ for all } i \in [L', L] \\
\label{label:general:eq:imp:67}
L^2 \geq 4/\epsilon \\
\label{label:general:eq:imp:68}
L \geq 3 \alpha \beta \\
L \geq 3 \label{label:general:eq:L:3} \label{label:general:eq:imp:72} \label{label:general:eq:imp:71} \label{label:general:eq:main:9}  \\
L_d \geq \gamma \label{label:general:eq:101} \\
L \geq 4L_d/\gamma \label{label:general:eq:102}\\
L^2 \geq (8 e^{\gamma} L_d)/(\epsilon \gamma \lambda_d) \label{label:general:eq:Li:epoch} 
\end{eqnarray}

We next define the parameters $h_0, h_1, h_2$.  Note that for all $k \in [0,2]$, we have $h_k \leq 1$, and, furthermore,  $h_k$ can be made very close to one if  the parameters $\delta, \epsilon, \gamma$ are sufficiently small and $K$ is sufficiently large. 

\begin{eqnarray}
\label{label:general:eq:h0}
h_0 & = &  e^{-\gamma} \cdot (1+4\epsilon)^{-1} \\
\label{label:general:eq:h0:less}
1/2 \leq h_0 & \leq & (1+4\epsilon)^{-1} \\
\label{label:general:eq:h1}
h_1 & = & h_0 \cdot e^{-\gamma}  \cdot ((\alpha \beta)^{-1}-3\epsilon)  \\
\label{label:general:eq:h2}
h_2 & = & (1-32 \delta L^4/\epsilon^2) \cdot (h_1 - \epsilon - 1/K) \text{ where } K \text{ is a large positive integer.} 
\end{eqnarray}

\subsection{Skeleton of a graph}
\label{label:general:sec:skeleton}

Set the parameter values as in Section~\ref{label:general:sec:parameter}. Throughout the paper, we consider an $(\alpha, \beta)$-partition (see Section~\ref{label:general:sec:prelim}) of the graph $G = (V, E)$.  Below, we introduce the concept of a skeleton of a graph.


\begin{definition}
\label{label:general:def:skeleton}
Consider a  graph $G = (V, E)$ with $n$ nodes and  maximum degree $d$, that is, $\text{deg}(v, E) \leq d$ for all $v \in V$.   A tuple $(B, T, S, X)$ with   $B, T, S \subseteq V$ and $X \subseteq E$, is  a ``skeleton'' of $G$ if and only if:
\begin{enumerate}
\item We have $B \cup T = V$ and $B \cap T = \emptyset$. In other words, the node-set $V$ is partitioned into subsets $B$ and $T$. The nodes in $B$ are called ``big'', whereas the nodes in $T$ are called ``tiny''.
\item For every big node $v \in B$, we have $\text{deg}(v, E) > \epsilon d/L^2$, where $L$ is the same as in the $(\alpha, \beta)$ partition.
\item For every tiny node  $v \in T$, we have $\text{deg}(v, E) < 3\epsilon d/L^2$. 
\item We have $|S| \leq 4\delta \cdot |B|$. The nodes in $S$ are called ``spurious''.
\item Recall equations~\ref{label:general:eq:d},~\ref{label:general:eq:Li} and~\ref{label:general:eq:Li:lambda}. For every big, non-spurious node $v \in B \setminus S$, we have: $$e^{-\gamma} \cdot (\lambda_d L^4/d) \cdot \text{deg}(v, E) \leq \text{deg}(v, X) \leq e^\gamma \cdot (\lambda_d L^4/d) \cdot \text{deg}(v, E).$$
\item For every node $v \in \mathcal{V}$, we have $\text{deg}(v, X) \leq \lambda_d L^4 +2$.
\item For every tiny, non-spurious node $v \in T \setminus S$, we have $\text{deg}(v, X) \leq 3\epsilon \lambda_d L^2 + 2$.
\end{enumerate}
\end{definition}

\noindent Thus, in a skeleton $(B, T, S, X)$ of a graph $G = (V, E)$, each node is classified as being either ``big'' or ``tiny'', but not both. Barring a few exceptions (which are called ``spurious'' nodes), the big nodes have large degrees and the tiny nodes have small degrees. The number of spurious nodes, however, is negligible in comparison with the number of big nodes. In a skeleton, the degrees of the big, non-spurious nodes are scaled down by a factor very close to $(\lambda_d L^4/d)$, where $d$ is the maximum degree of a node in $G$. Further, the maximum degree of any node in a skeleton is $\lambda_d L^4 + 2$. Finally, the maximum degree of any tiny, non-spurious node in the skeleton is $3\epsilon \lambda_d L^2 + 2$. 
The next theorem shows that we can efficiently maintain the edge-set of a skeleton of a graph.

\begin{theorem}
\label{label:general:th:maintain:skeleton}
Given  a graph $G = (V, E)$ with $n$ nodes, we can maintain the edges $e \in H$ of a skeleton $(B, T, S, H)$ of $G$ under a sequence of edge insertions/deletions in $E$. We assume that  $E = \emptyset$ in the beginning. Our algorithm is deterministic and has an amortized update time of $O(L^2 L_d^3/(\epsilon \gamma \delta))$. See equation~\ref{label:general:eq:Li}.
\end{theorem}

\noindent The next theorem will be helpful in deriving the approximation guarantee of our algorithm.

\begin{theorem}
\label{label:general:th:main}
Consider an $(\alpha, \beta)$-partition of the graph $G = (V, E)$. For every level $i \in [L',  L]$, let $(B_i , T_i, S_i, X_i)$ be a skeleton of $G_i = (V, E_i)$. Define the edge-sets $X = \bigcup_{i = L'}^L X_i$, and $Y =  \bigcup_{i=0}^{L'-1} E_i$. Then the following conditions are satisfied.
\begin{enumerate}
\item For every node $v \in V$, we have $\text{deg}(v, X \cup Y) = O(L^5)$.
\item The size of the maximum cardinality matching in  $G(X \cup Y) = (V, X \cup Y)$ is a $(2/h_2) \cdot (1+\epsilon)$-approximation to  the size of the maximum cardinality matching in $G = (V, E)$. See equation~\ref{label:general:eq:h2}.
\end{enumerate}
\end{theorem}

Our main result is summarized in Theorem~\ref{label:general:cor:th:main}, which follows from Theorems~\ref{label:general:th:maintain:skeleton} and~\ref{label:general:th:main}.

\begin{theorem}
\label{label:general:cor:th:main}
In a dynamic setting, we can deterministically  maintain a $(2+\epsilon)$-approximate maximum matching in a graph $G = (V, E)$ in $O(\text{poly } (\log n, \epsilon^{-1}))$ amortized update time.
\end{theorem}

\begin{proof}
We set the parameters $\gamma$ and $\delta$ as follows.
$$\gamma = \epsilon, \delta = \epsilon^3/L^4.$$
The input graph $G = (V, E)$ changes via a sequence of edge-updates in $E$. In this setting, we maintain the edge-sets $\{E_i\}$ of an $(\alpha, \beta)$-partition of $G$ as per Theorem~\ref{label:general:th:runtime}. This requires $O(\log n/\epsilon^2)$ update time. Next, for each level $i \in [L', L]$, we maintain the edge-set $X_i \subseteq E_i$ as per Theorem~\ref{label:general:th:maintain:skeleton}. By equation~\ref{label:general:eq:lastt}, this requires   $O(\text{poly } (\log n, (\epsilon \gamma \delta )^{-1})) = O(\text{poly } (\log n, \epsilon^{-1}))$ update time. Finally, by Theorem~\ref{label:general:th:main},  the maximum degree of a node in the subgraph $G(X \cup Y)$ is $O(L^4)$. Hence, we maintain a $(1+\epsilon)$-approximate maximum matching $M' \subseteq (X \cup Y)$ in the subgraph $G(X \cup Y)$ as per Theorem~\ref{label:general:th:gupta:peng}. This also requires $O(\text{poly} (\log n, \epsilon^{-1}))$ time.  Accordingly, the total update time of the algorithm is  $O(\text{poly} (\log n, \epsilon^{-1}))$. 

By Theorem~\ref{label:general:th:main}, the maximum matching in $G(X \cup Y)$ is a  $(2/h_2) \cdot (1+\epsilon)$-approximation to the maximum matching in $G$. Recall that our algorithm maintains $M'$, which is a $(1+\epsilon)$-approximation to the maximum matching in $G(X \cup Y)$. Thus, $M'$ is a $(2/h_2) \cdot (1+\epsilon)^2$-approximation to the maximum matching in $G$. Since $\delta = \epsilon^3 / L^4$ and $\gamma = \epsilon$, equations~\ref{label:general:eq:h0},~\ref{label:general:eq:h1} and~\ref{label:general:eq:h2} imply that $(1/h_2) \cdot (1+\epsilon)^2 = 1+ O(\epsilon)$ for sufficiently small $\epsilon$.

To summarize, we maintain a matching $M' \subseteq X \cup Y \subseteq E$. The size of this matching $M'$ is a $(2+O(\epsilon))$-approximation to the size of the maximum cardinality matching in $G$. Further, our algorithm for maintaining $M'$ requires $O(\text{poly } (\log n, \epsilon^{-1}))$ update time. This concludes the proof of Theorem~\ref{label:general:cor:th:main}.
\end{proof}

In Section~\ref{label:general:sec:th:main} we prove Theorem~\ref{label:general:sec:th:main}, and in Section~\ref{label:general:sec:algo} we prove Theorem~\ref{label:general:th:maintain:skeleton}.

\section{Deriving the approximation guarantee: Proof of Theorem~\ref{label:general:th:main}}
\label{label:general:sec:th:main}

\renewcommand{\L}{\mathcal{L}}
\newcommand{\I}{\mathcal{I}}

The first part of the theorem, which upper bounds the maximum degree a node can have in $X \cup Y$, is relatively easy to prove.  The basic idea is that for any given level $i \in [0, L]$, the degree of a node is $O(L^4)$ among the level-$i$ edges in $X \cup Y$. This is true for the levels $i \geq L'$ because of the condition (6) in Definition~\ref{label:general:def:skeleton}. On the other hand, for $i < L'$, every edge $e \in E_i$ has weight $w(e) = \Omega(1/L^4)$. Thus, the degree of a node cannot be too high among those edges, for otherwise the total weight received by the node would exceed one. This is summarized in the lemma below.

\begin{lemma}
\label{label:general:lm:max:deg:100}
We have $\text{deg}(v, X \cup Y) = O(L^5)$ for all nodes $v \in V$.
\end{lemma}

\begin{proof}
Consider any node $v \in V$. By Definition~\ref{label:general:def:skeleton}, we have $\text{deg}(v, X_i) = O(L^4)$ for all levels $i \in [L', L]$. Next, note that at every level $i < L'$, we have $d_i \leq d_{L'} = \beta^{L'} \cdot (\alpha \beta) = O(L^4)$. The last equality follows from equation~\ref{label:general:eq:L'}. Hence, from Corollary~\ref{label:general:cor:lm:deg:1}, we infer that $\text{deg}(v, E_i) \leq d_i = O(L^4)$ for all levels $i <  L'$ as well. Since $Y = \bigcup_{i=0}^{L'-1} E_i$, summing over all the levels $i \in [0,L]$, we infer that $\text{deg}(v, X \cup Y) = O(L^5)$.
\end{proof}

Accordingly, we devote the rest of this section towards proving the second part of Theorem~\ref{label:general:th:main}. But, before proceeding any further, we need to introduce the following notations.
\begin{itemize}
\item Let $V_X = \{ v \in V : \ell(v) \geq L'\}$ denote the set of nodes lying at levels $L'$ or above in the $(\alpha, \beta)$-partition. Similarly, we let $V_Y = \{ v \in V : \ell(v) < L' \}$ denote the set of nodes lying at levels $(L'-1)$ or below. Clearly, the node-set $V$ is partitioned by these two subsets $V_X \subseteq V$ and $V_Y = V \setminus V_X$. Furthermore, note that the level of any edge $(u,v)$ in the $(\alpha, \beta)$-partition is defined as the maximum level among its endpoints. Hence, any edge $e \in Y = \bigcup_{i = 0}^{L'-1} E_i$ has both of its endpoints in $V_Y$. On the other hand, any edge $e \in E_i$ with $i \geq L'$ has at least one of its endpoints in $V_X$. As a corollary,  every edge $e \in X$ has at least one endpoint in $V_X$. We will use these observations throughout the rest of this section.
\end{itemize}
We will now try to explain why we  consider these node-sets $V_X$ and $V_Y$. For this we need to look back at the reason behind  constructing the sets $\{X_i\}$ in the first place.  Consider any node $v \in V_X$. Such a node belongs to a level $\ell(v) \geq L'$ in the $(\alpha, \beta)$-partition. In an ideal scenario, for each $i \geq L'$, we want  the degree of  $v$ to drop by a factor of $(L^4/d_i)$ as we switch from the edge-set $E_i$ to the edge-set $X_i \subseteq E_i$ (see the condition (5) in Definition~\ref{label:general:def:skeleton}). If we  succeed in this goal, then we can  increase the weight of every edge in $X_i$ by the inverse factor $(d_i/L^4)$. As a consequence, the total weight received by the node $v$ from the level-$i$ edges will remain the same as before. To continue with the discussion, recall that every edge $e \in E_i$ had a weight $w(e) = 1/d_i$ in the $(\alpha, \beta)$-partition. Let $w'(e) = w(e) \cdot (d_i/L^4) = 1/L^4$ denote the scaled up weight of the surviving edges $e \in X_i \subseteq E_i$. Thus, we expect to see that $W(v, E_i) = W'(v, X_i)$ at every level $i \in [L', L]$. Since $v \in V_X$, every edge incident upon $v$ belongs to some $E_i$ with $i \geq L'$. Accordingly, summing over all the levels $i \in [L', L]$, we expect to get the following guarantee.
$$W(v, E) = \sum_{i=L'}^L W(v, E_i) = \sum_{i = L'}^L W'(v, X_i) = W'(v, X) \geq (\alpha \beta)^{-1}.$$
The last inequality is supposed to hold due to Invariant~\ref{label:general:inv:partition} (recall that $\ell(v) \geq L'$). Thus, in an ideal scenario, if we assign a weight $w'(e) = L^{-4}$ to every edge $e \in X$, then we should get a fractional matching where every node in $V_X$ gets a weight very close to one. In other words, such a matching will be ``nearly maximal'' in the graph induced by the edges $e \in E$ with $\ell(e) \geq L'$ (this holds since every such edge will have at least one endpoint in $V_X$). Thus, there exists a near-maximal fractional matching $w'$  in the subgraph $(V, X)$ that is also near-maximal in $(V, \bigcup_{i\geq L'} E_i)$. Since maximal matchings are $2$-approximations to maximum matchings, intuitively, as we switch from the edge-set $\bigcup_{i \geq L'} E_i$ to the edge-set $X$, we loose roughly a factor of $2$ in the approximation ratio (albeit in the fractional world). All the remaining edges $\bigcup_{i < L'} E_i$ go to the set $Y$, and we ought to loose nothing there in terms of the approximation guarantee. Thus, intuitively, there should be a matching defined on the edge-set $X \cup Y$ that is a good approximation to the maximum matching in the original graph $G = (V, E)$.

\begin{itemize}
\item It is now high  time to emphasize  that none of the above observations are exactly true to begin with. For example, as the Definition~\ref{label:general:def:skeleton} states, there will be spurious nodes whose degrees might not drop by much as we go from from the edge-set $E_i$ to the edge-set $X_i$, $i \in [L', L]$.  There might be tiny nodes whose degrees might drop by more than we bargained for during the same transition. And even  the degrees of the big, non-spurious nodes might not drop exactly by the factor $(L^4/d_i)$. Nevertheless, we will try to carry out the plan outlined above, hoping to avoid all the roadblocks that we might come across while writing the formal proof. 
\end{itemize}

\noindent Along these lines, we now formally define the weight function $w' : X \rightarrow [0, 1]$ in equation~\ref{label:general:n:eq:w'}.

\begin{equation}
\label{label:general:n:eq:w'}
w'(u, v)  = h_0 \cdot (\lambda_{d_i} L^4)^{-1} \text{ for every edge } (u,v) \in X_i, i \in [L', L].
\end{equation}

\noindent Note that $\lambda_{d_i}$ lies between $1/2$ and $1$, and $h_0$ is very close to one (see equations~\ref{label:general:eq:Li:lambda},~\ref{label:general:eq:h0}). Thus, although $w'$ might assign different weights to different edges, none of these weights are too far away from $1/L^4$ (this is consistent with our preliminary discussion above). Not surprisingly, the next lemma shows that under $w'$, every node $v \in V_X$ gets a weight that is lower bounded by a value very close to one. The proof of Lemma~\ref{label:general:n:lm:w'} appears in Section~\ref{label:general:sec:n:lm:w'}.

\begin{lemma}
\label{label:general:n:lm:w'}
We have $W'(v, X) \geq h_1$ for all nodes $v \in V_X$. In other words, for every node $v$ at level $\ell(v) \in [L', L]$ in the $(\alpha, \beta)$-partition, the weight received by $v$ under $w'$ is very close to one.
\end{lemma}

Unfortunately, however, the weights $\{w'(e)\}, e\in X,$ do not give us a fractional matching. For example, consider any node $v \in V_X$ that is spurious at every level $i \geq L'$, i.e., $v \in S_i$ for all $i \in [L', L]$. At each level $i \geq L'$, the node $v$ can have $\text{deg}(v, X_i) = \Theta(L^4)$ (see the condition (6) in Definition~\ref{label:general:def:skeleton}), and its weight from this level  $W'(v, X_i) = \text{deg}(v, X_i) \cdot  (h_0 \lambda_{d_i}^{-1}) \cdot L^{-4}$ can be as large as $\Theta(1)$. So the total weight it receives from all the levels can be as large as $(L-L') \cdot \Theta(1) = \Theta(L)$. The next lemma shows that this actually is the worst possible scenario. The proof of Lemma~\ref{label:general:n:lm:w':bound} appears in Section~\ref{label:general:sec:n:lm:w':bound}.

\begin{lemma}
\label{label:general:n:lm:w':bound}
We have $W'(v, X) \leq 2L$ for every node $v \in V$.
\end{lemma}

Since the spurious nodes are the ones that are preventing $w'$ from being a fractional matching, we  assign zero weight to any edge that is incident upon a spurious node, and leave the weights of the remaining edges as in $w'$. Accordingly, we define the weight function $w'' : X \rightarrow [0,1]$ in equation~\ref{label:general:n:eq:w*}.

\begin{equation}
\label{label:general:n:eq:w*}
w''(u,v)   =   \begin{cases} 
w'(u,v) & \text{ if }  \{u, v\} \cap S = \emptyset, \text{ where } S = \bigcup_{i \geq L'} S_i; \\
0 & \text{ otherwise.}
\end{cases} 
\end{equation}

We highlight an important fact here. Consider an edge $(u, v) \in X_i, i \geq L'$. It might very well be the case that one of its endpoints, say $v$, is non-spurious at level $i$ but spurious at some other level $i' \in \{L', \ldots, L\} \setminus \{i\}$ (i.e.,  $v \notin S_i$ and $v \in S_{i'}$). Even under this situation, $w''$ assigns  zero weight to the edge $(u,v)$. Although this seems to be a stringent condition, it turns out not to affect our proof in any  way.

As expected, the new weights $\{w''(e)\}, e \in X,$ actually defines a fractional matching. This is shown in Lemma~\ref{label:general:n:lm:w:star}, and the proof of this lemma appears in Section~\ref{label:general:n:lm:w:star}.
\begin{lemma}
\label{label:general:n:lm:w:star}
We have $W''(v, X) \leq 1$ for all nodes $v \in V$. In other words, the weight function $w''$ defines a fractional matching in the subgraph induced by the edges in $X$. 
\end{lemma}

To take stock of the situation, we have constructed two weight functions $w' : X \rightarrow [0, 1]$ and $w'' : X \rightarrow [0, 1]$. The former nearly saturates all the nodes in $V_X$, in the sense that each such node receives a total weight that is at least $\eta$, for some $\eta$ very close to one (see Lemma~\ref{label:general:n:lm:w'}).  But $w'$ is not a fractional matching. The weight function $w''$, on the other hand, is a fractional matching (see Lemma~\ref{label:general:n:lm:w:star}) but might not saturate all the nodes in $V_X$. Our goal is to get the best of the both these worlds.  As a first step towards this goal, in Lemma~\ref{label:general:n:lm:spurious:count} we  show that the nodes that are spurious at some level (a.k.a, the ``troublemakers'')  are negligibly small in numbers compared to the size of the set $V_X$. The proof of Lemma~\ref{label:general:n:lm:spurious:count} appears in Section~\ref{label:general:sec:n:lm:spurious:count}.

\begin{lemma}
\label{label:general:n:lm:spurious:count}
The number of nodes that are spurious at some level is negligibly small compared to the number of nodes $v$ at levels $\ell(v) \geq L'$ in the $(\alpha, \beta)$-partition. Specifically, we have:
$$|S| \leq (8 \delta L^3/\epsilon) \cdot |V_X|.$$
\end{lemma}

If we switch from the weight function $w'$ to the weight function $w''$, the only difference we see is that the edges incident upon the nodes in $S$ have been ``turned off''. Let these specific edges be called ``transient''. Under $w'$, each node in $S$ receives a total weight of at most $2L$ (see Lemma~\ref{label:general:n:lm:w':bound}). Thus, we get an upper bound of $2L \cdot |S|$  on  the total weight  assigned by $w'$ to all the transient edges. Since each edge is incident upon two nodes, as we switch from $w'$ to $w''$ the total weight of the nodes in $V_X$ drops by at most $4 L \cdot |S|$. Consequently, by a simple counting argument,  at most $(4L/\epsilon) \cdot |S|$ many nodes in $V_X$ can experience a larger than $\epsilon$ drop in their weights due to this transition. Since every node in $V_X$ has weight at least $h_1$ under $w'$ (see Lemma~\ref{label:general:n:lm:w'}), we get the following lemma. The proof of Lemma~\ref{label:general:n:lm:spurious:good} appears in Section~\ref{label:general:sec:n:lm:spurious:good}.

\begin{lemma}
\label{label:general:n:lm:spurious:good}
Under $w''$, only a negligible fraction of the nodes in $V_X$ do not receive large weights.  Specifically, define the set of nodes $Q_X = \{ v \in V_X : W''(v, X) < h_1 - \epsilon\}$. Then we have:
$$|Q_X| \leq (32 \delta L^4/\epsilon^2) \cdot |V_X|.$$
\end{lemma}

\begin{corollary}
\label{label:general:n:cor:lm:spurious:good}
We have:  $|V_X| \leq (1- 32 \delta L^4/\epsilon^2)^{-1} \cdot |V_X \setminus Q_X|.$
\end{corollary}

\begin{proof}
Follows from Lemma~\ref{label:general:n:lm:spurious:good}.
\end{proof}

Recall that every edge in $\bigcup_{i \geq L'} E_i$ has at least one endpoint in $V_X$. In lieu of our high level discussions at the start of this section, we see that $w''$ is a ``nearly-maximal'' fractional matching  in $\bigcup_{i \geq L'} E_i$ in the following sense: It ``nearly saturates'' ``almost all'' the nodes in $V_X$ (see Lemmas~\ref{label:general:n:lm:w:star} and~\ref{label:general:n:lm:spurious:good}). Remaining faithful to our initial plan, in Lemma~\ref{label:general:n:lm:main} we  now show the existence of a large matching in $M^* \subseteq X \cup Y$. The construction of this matching $M^*$ will crucially rely on the properties of the weight function  $w''$. The proof of Lemma~\ref{label:general:n:lm:main} appears in Section~\ref{label:general:sec:n:lm:main}.

\begin{lemma}
\label{label:general:n:lm:main}
Let $X^* =  \{ (u,v) \in X : w''(u,v) > 0 \}$ be the set of edges that receive nonzero weights under $w''$, and let $E^* = X^* \cup Y$.  Define $G^* = (V, E^*)$ to be the subgraph of $G = (V, E)$ induced by the edges in $E^*$. Then for every matching $M \subseteq E$ in $G$, there is a matching $M^* \subseteq E^*$ in $G^*$ such that $|M| \leq (2/h_2) \cdot (1+\epsilon) \cdot |M^*|$.
\end{lemma}

Theorem~\ref{label:general:th:main} now follows from Lemmas~\ref{label:general:lm:max:deg:100} and~\ref{label:general:n:lm:main}.

\subsection{Some basic notations}
\label{label:general:n:sec:notations}

Here, we quickly recall some basic notations introduced earlier. We also introduce a few notations. All these notations will be used throughout the rest of this section.

\begin{itemize}
\item Define $B = \bigcup_{i = L'}^L B_i$ to be the set of nodes that are big at some level $i \geq L'$.
\item Define $S = \bigcup_{i=L'}^L S_i$ to be the set of nodes that are spurious at some level $i \geq L'$.
\item Define $T = \bigcup_{i=L'}^L T_i$ to be the set of nodes that are tiny at some level $i \geq L'$.
\item Given any node $v \in V$, let $\mathcal{L}_v(B)$ be the levels in $[L', L]$ where it is big. Thus, a level $i \in [L', L]$ belongs to the set $\mathcal{L}_v(B)$ iff $v \in B_i$.
\item Given any node $v \in V$, let $\mathcal{L}_v(T)$ be the levels in $[L', L]$ where it is tiny. Thus, a level $i \in [L', L]$ belongs to the set $\mathcal{L}_v(T)$ iff $v \in T_i$.
\item Given any node $v \in V$, let $\mathcal{L}_v(S)$ be the levels in $[L', L]$ where it is spurious. Thus, a level $i \in [L', L]$ belongs to the set $\mathcal{L}_v(S)$ iff $v \in S_i$.
\item Since a node is either big or tiny (but not both) at each level, it follows that for every node $v \in V$ the set $\{L', \ldots, L\}$ is partitioned by the three subsets $\mathcal{L}_v(S)$, $\mathcal{L}_v(B) \setminus \mathcal{L}_v(S)$ and $\mathcal{L}_v(T) \setminus \mathcal{L}_v(S)$.
\item Recall the notations $V_X$ and $V_Y$ introduced right after the proof of Lemma~\ref{label:general:lm:max:deg:100}.
\end{itemize}

\subsection{Proof of Lemma~\ref{label:general:n:lm:w'}}
\label{label:general:sec:n:lm:w'}

\begin{itemize}
\item Throughout this proof, we will use the notations defined in Section~\ref{label:general:n:sec:notations}.
\end{itemize}
\noindent
Throughout this section, fix any node $v \in V_X$. Thus, the node $v$ is at level $\ell(v) \geq L'$ in the $(\alpha, \beta)$-partition.

\begin{claim}
\label{label:general:n:cl:n:lm:w':1}
Under $w$, the node $v$ gets negligible weight  from all the levels $i \in [L', L]$ where it is tiny. Specifically, we have: 
$$\sum_{i \in \L_v(T)} W(v, X_i)   \leq 3 \epsilon.$$ 
\end{claim}

\begin{proof}
 Consider any level $i \in \L_v(T)$ where the node $v$ is tiny. By the condition (2) in Definition~\ref{label:general:def:skeleton}, we have $\text{deg}(v, E_i) < 3 \epsilon d_i/L^2$. Since each edge in $E_i$ received a weight $1/d_i$ under $w$, we get:
\begin{equation}
\label{label:general:n:eq:w':1}
W(v, E_i) = \text{deg}(v, E_i) \cdot (1/d_i) \leq 3\epsilon/L^2 
\end{equation}
Since there are at most $L$ levels in the range $[L', L]$, summing equation~\ref{label:general:n:eq:w':1} over all $i \in \L_v(T)$, we get:
$$\sum_{i \in \L_v(T)} W(v, E_i) \leq \left|\L_v(T)\right| \cdot (3\epsilon/L^2) \leq 3\epsilon/L \leq 3\epsilon.$$
The last inequality holds since $L \geq 1$ (see equation~\ref{label:general:eq:imp:72}).
\end{proof}

\begin{corollary}
\label{label:general:n:cor:1} 
Under $w$,  the node $v$ receives close to one weight from the levels $i \in [L', L]$ where it is big.
Specifically, we have: $$\sum_{i \in \L_v(B)} W(v, E_i) \geq (\alpha \beta)^{-1} - 3 \epsilon.$$
\end{corollary}

\begin{proof}
Since $v \in V_X$, the node $v$ belongs to a level $\ell(v) \geq L'$ in the $(\alpha, \beta)$-partition. Hence, every edge $(u,v) \in E$ incident upon $v$  has level $\ell(u,v) = \max(\ell(u), \ell(v)) \geq L'$. Thus, from Invariant~\ref{label:general:inv:partition} we get:
\begin{equation}
\label{label:general:n:eq:cor:1}
\sum_{i \in L'}^L W(v, E_i) = W(v, E) \geq (\alpha \beta)^{-1}
\end{equation}
At any level in the $(\alpha, \beta)$-partition the node $v$ is either big or tiny, but it cannot be both at the same time. In other words,  the set of levels $\{L', \ldots, L\}$ is partitioned into two subsets: $\L_v(B)$ and $\L_v(T)$. Hence, the corollary follows from equation~\ref{label:general:n:eq:cor:1} and Claim~\ref{label:general:n:cl:n:lm:w':1}.
\end{proof}

\begin{claim}
\label{label:general:n:cl:n:lm:w':2}
Consider the weight received by the node $v$ under $w'$ from  the levels in $[L', L]$ where it is big. This weight is  nearly equal to the weight it receives under $w$ from the same levels. Specifically, we have:
$$\sum_{i \in \mathcal{L}_v(B)} W'(v, X_i) \geq (h_0  e^{-\gamma}) \cdot \sum_{i \in \mathcal{L}_v(B)}  W(v, E_i).$$
\end{claim}

\begin{proof}
Consider any  level $i \in \L_v(B)$ where the node $v$ is big.   By the condition (5) of Definition~\ref{label:general:def:skeleton}, the degree of $v$ drops roughly by a multiplicative factor of $(\lambda_{d_i} L^4/d_i)$ as we go from $E_i$ to $X_i$. Thus, we have:
\begin{equation}
\label{label:general:n:eq:w':4}
\text{deg}(v, X_i) \geq e^{-\gamma} \cdot (\lambda_{d_i} L^4/d_i) \cdot \text{deg}(v, E_i)
\end{equation}
Each edge in $X_i$ receives exactly $h_0 \cdot (\lambda_{d_i} L^4)^{-1}$ weight under $w'$. Hence, equation~\ref{label:general:n:eq:w':4} implies that:
\begin{equation}
\label{label:general:n:eq:w':5}
W'(v, X_i) = h_0  \cdot (\lambda_{d_i} L^4)^{-1} \cdot \text{deg}(v, X_i) \geq (h_0 e^{-\gamma})  \cdot (1/d_i) \cdot \text{deg}(v, E_i)
\end{equation}
Since each edge in $E_i$ receives $1/d_i$ weight under $w$, equation~\ref{label:general:n:eq:w':5} implies that:
\begin{equation}
\label{label:general:n:eq:w':6}
W'(v, X_i) \geq (h_0 e^{-\gamma}) \cdot W(v, E_i) 
\end{equation}
The claim follows if we sum equation~\ref{label:general:n:eq:w':6} over all levels $i \in \L_v(B)$.
\end{proof}

Under $w'$, the total weight received by the node $v$  is at least the weight it receives from the levels $i \in [L', L]$ where it is big. Thus, from Corollary~\ref{label:general:n:cor:1} and Claim~\ref{label:general:n:cl:n:lm:w':2}, we infer that:
$$W'(v, X) \geq \sum_{i \in \L_v(B)} W'(v, X_i) \geq (h_0 e^{-\gamma}) \cdot ((\alpha \beta)^{-1} - 3\epsilon) = h_1.$$
The last equality follows from equation~\ref{label:general:eq:h1}. This concludes the proof of the lemma.

\subsection{Proof of Lemma~\ref{label:general:n:lm:w':bound}}
\label{label:general:sec:n:lm:w':bound}

Fix any node $v \in V$. We will first bound the weight received by the node $v$ under $w'$ from any given level $i \in [L', L]$. Towards this end, note that by the condition (6) in Definition~\ref{label:general:def:skeleton}, we have $\text{deg}(v, X_i) \leq \lambda_{d_i} \cdot L^4 + 2 \leq 2 \lambda_{d_i} L^4$.  The last inequality holds since $1/2 \leq \lambda_{d_i} \leq 1$ (see equation~\ref{label:general:eq:Li:lambda}) and $L^4 \geq 4$ (see equation~\ref{label:general:eq:imp:71}). Since every edge in $X_i$ receives a weight $h_0 \cdot (\lambda_{d_i} L^4)^{-1}$ under $w'$, equation~\ref{label:general:eq:h0:less} gives us:
$$W'(v, X_i) = \text{deg}(v, X_i) \cdot  h_0 \cdot (\lambda_{d_i} L^4)^{-1} \leq 2 h_0 \leq 2$$
Summing the above inequality over all levels $i \in [L', L]$, we get:
$$W'(v, X) = \sum_{i=L'}^L W'(v, X_i) \leq 2 (L - L'+1) \leq 2L.$$
This concludes the proof of the lemma.

\subsection{Proof of Lemma~\ref{label:general:n:lm:w:star}}
\label{label:general:sec:n:lm:w*}

\begin{itemize}
\item Throughout this proof, we will use the notations defined in Section~\ref{label:general:n:sec:notations}.
\end{itemize}
\noindent
Throughout this section, we fix any node $v \in V$. Recall that the set of levels $\{L', \ldots, L\}$ is partitioned into three subsets: $\L_v(S)$, $\L_v(T) \setminus \L_v(S)$ and $\L_v(B) \setminus \L_v(S)$. Thus, we have:
\begin{equation}
\label{label:general:eq:obvious}
W''(v, X) = \sum_{i \in \L_v(S)} W''(v, X_i) + \sum_{i \in \L_v(T) \setminus \L_v(S)} W''(v, X_i) + \sum_{i \in \L_v(B) \setminus \L_v(S)} W''(v, X_i)
\end{equation}

We separately bound the weights received by the node $v$ under $w''$ from these three different types of levels. The lemma  follows by adding up the bounds from Claims~\ref{label:general:n:cl:n:lm:w*:1},~\ref{label:general:n:cl:n:lm:w*:2} and~\ref{label:general:n:cl:n:lm:w*:3}.

\begin{claim}
\label{label:general:n:cl:n:lm:w*:1}
The node $v$ gets zero weight under $w''$ from all the levels $i \in [L', L]$ where it is spurious. Specifically, we have $\sum_{i \in \L_v(S)} W''(v, E_i) = 0$. 
\end{claim}

\begin{proof}
If the node $v$ is spurious at some level $i$, then each of its incident edges in $X_i$ gets zero weight under $w''$ (see equation~\ref{label:general:n:eq:w*}).  The claim follows.
\end{proof}

\begin{claim}
\label{label:general:n:cl:n:lm:w*:2}
The node $v$ gets negligible weight under $w''$ from all the levels $i \in [L', L]$ where it is tiny but non-spurious. Specifically, we have: 
$$\sum_{i \in \L_v(T) \setminus \L_v(S)} W''(v, X_i)   \leq (1+4\epsilon)^{-1} \cdot (4 \epsilon).$$ 
\end{claim}

\begin{proof}
 Consider any level $i \in \L_v(T) \setminus \L_v(S)$ where the node $v$ is tiny but non-spurious. Note that $\lambda_{d_i} \geq 1/2$ (see equation~\ref{label:general:eq:Li:lambda}). Hence, Definition~\ref{label:general:def:skeleton} and equation~\ref{label:general:eq:imp:67} imply that: 
\begin{equation}
\label{label:general:n:eq:w:star:1}
\text{deg}(v, X_i) \leq 3 \epsilon \lambda_{d_i} L^2 + 2 \leq 4 \epsilon \lambda_{d_i} L^2
\end{equation} 
Each edge in $X_i$ receives at most  $h_0 \cdot (\lambda_{d_i} L^{4})^{-1}$ weight under $w''$. Thus, equations~\ref{label:general:n:eq:w:star:1} and~\ref{label:general:eq:h0:less} imply that:
\begin{equation}
\label{label:general:n:eq:w:star:2}
W''(v, X_i) \leq h_0 \cdot (\lambda_{d_i} L^{4})^{-1} \cdot \text{deg}(v, X_i) \leq 4 \epsilon h_0 L^{-2} \leq  (1+4\epsilon)^{-1} \cdot (4 \epsilon/L^{2})
\end{equation}
Since there are at most $L$ levels in the range $[L', L]$, summing equation~\ref{label:general:n:eq:w:star:2} over all  $i \in \L_v(T) \setminus \L_v(S)$ gives:
 \begin{eqnarray}
 \sum_{i \in \L_v(T) \setminus \L_v(S)} W''(v, X_i)   \leq (1+4\epsilon)^{-1} \cdot (4 \epsilon/L) \label{label:general:n:eq:w:star:3}
 \end{eqnarray}
 This  claim follows from equation~\ref{label:general:n:eq:w:star:3} and the fact that $L \geq 1$ (see equation~\ref{label:general:eq:imp:72}).
 \end{proof}

 \begin{claim}
 \label{label:general:n:cl:n:lm:w*:3}
The node $v$ receives at most $(1+4\epsilon)^{-1}$ weight under $w''$ from all the levels $i \in [L', L]$ where it is big and non-spurious. Specifically, we have:
 $$\sum_{i \in \mathcal{L}_v(B) \setminus \L_v(S)} W''(v, X_i) \leq (1+4 \epsilon)^{-1}.$$ 
  \end{claim}
 
 \begin{proof}
Consider any  level $i \in \L_v(B) \setminus \L_v(S)$ where the node $v$ is big but non-spurious.   By the condition (5) of Definition~\ref{label:general:def:skeleton}, the degree of $v$ drops roughly by a multiplicative factor of $(\lambda_{d_i} L^4/d_i)$ as we go from $E_i$ to $X_i$. Specifically, we have:
\begin{equation}
\label{label:general:n:eq:w:star:4}
\text{deg}(v, X_i) \leq e^{\gamma} \cdot (\lambda_{d_i} L^4/d_i) \cdot \text{deg}(v, E_i)
\end{equation}
Each edge in $X_i$ receives at most $h_0 \cdot (\lambda_{d_i} L^4)^{-1}$ weight under $w''$. Hence, equations~\ref{label:general:eq:h0},~\ref{label:general:n:eq:w:star:4} imply that:
\begin{equation}
\label{label:general:n:eq:w:star:5}
W''(v, X_i) \leq h_0 \cdot (\lambda_{d_i} L^4)^{-1} \cdot \text{deg}(v, X_i) \leq  (1+4\epsilon)^{-1} \cdot (1/d_i) \cdot \text{deg}(v, E_i)
\end{equation}
Since each edge in $E_i$ receives $1/d_i$ weight under $w$, equation~\ref{label:general:n:eq:w:star:5} implies that:
\begin{equation}
\label{label:general:n:eq:w:star:6}
W''(v, X_i) \leq (1+4\epsilon)^{-1} \cdot W(v, E_i) 
\end{equation}
The claim follows if we sum equation~\ref{label:general:n:eq:w:star:6} over all levels $i \in \L_v(B) \setminus \L_v(S)$, and recall that the sum $\sum_{i \L_v(B) \setminus \L_v(S)} W(v, E_i)$ itself is at most one (for $w$ is a fractional matching in $G$).
 \end{proof}

 The lemma follows from equation~\ref{label:general:eq:obvious} and Claims~\ref{label:general:n:cl:n:lm:w*:1},~\ref{label:general:n:cl:n:lm:w*:2},~\ref{label:general:n:cl:n:lm:w*:3}.

\subsection{Proof of Lemma~\ref{label:general:n:lm:spurious:count}}
\label{label:general:sec:n:lm:spurious:count}

\begin{itemize}
\item Throughout this proof, we will use the notations defined in Section~\ref{label:general:n:sec:notations}.
\end{itemize}
\noindent
We first show that the number of nodes that are spurious at some level is negligibly small compared to the number of nodes that are big at some level. Towards this end, recall the condition (4) in Definition~\ref{label:general:def:skeleton}. For each level $i \in [L', L]$, this  implies that $|S_i| \leq 4 \delta \cdot |B_i| \leq 4 \delta \cdot |B|$. The last inequality holds since $B_i \subseteq B$. Hence, summing over all the levels $i \in [L', L]$, we get:
\begin{equation}
\label{label:general:n:eq:spurious:count:1}
|S| = \left| \bigcup_{i=L'}^{L} S_i \right| \leq \sum_{i=L'}^L |S_i| \leq (L-L'+1) \cdot 4\delta \cdot |B| \leq (4\delta L)  \cdot |B|
\end{equation}
It remains to upper bound the size of the set $B$ in terms of the size of the set $V_X$. Towards this end, we first show that every node in $B$ receives sufficiently large weight under $w$. Specifically, fix any node $v \in B$. By definition, we have $v \in B_i$ at some level $i \in [L', L]$. Hence, the condition (1) in Definition~\ref{label:general:def:skeleton} implies that $\text{deg}(v, E_i) > \epsilon d_i/L^2$. Since each edge in $E_i$ receives a weight $1/d_i$ under $w$, we get: $W(v, E_i) = (1/d_i) \cdot \text{deg}(v, E_i) > \epsilon/L^2$. To summarize, we have the following guarantee.
\begin{equation}
\label{label:general:n:eq:spurious:count:2}
\sum_{i = L'}^L W(v, E_i) \geq \epsilon/L^2 \ \ \text{ for all nodes } v \in B.
\end{equation}
Recall that  the level of an edge $(u, v)$ in the $(\alpha, \beta)$-partition is given by $\ell(u, v) = \max(\ell(u), \ell(v))$. Further, the set $E_i$ is  precisely  the set of those edges $e \in E$ with $\ell(e) = i$.  Hence, each edge $(u,v) \in \bigcup_{i = L'}^L E_i$ has at least one endpoint in 
$V_X = \{ x \in V : \ell(v) \in [L', L]\}$. Thus, a simple counting argument gives us:
\begin{equation}
\label{label:general:n:eq:spurious:count:3}
\sum_{v \in B} \sum_{i=L'}^L W(v, E_i) \leq 2 \cdot \sum_{v \in V_X} \sum_{i=L'}^L W(v, E_i)
\end{equation}
The above inequality holds for the following reason. Consider any edge $(u,v) \in \bigcup_{i=L'}^L E_i$. Either both the endpoints $u, v$ belong to $V_X$, or exactly one of the endpoints $u, v$ belong to $V_X$. In the former case, the edge $(u,v)$ contributes at most $2 \cdot w(u,v)$ to the left hand side, and exactly $4 \cdot w(u,v)$ to the right hand side. In the latter case, the edge $(u,v)$ contributes at most $2 \cdot w(u,v)$ to the left hand side, and exactly $2 \cdot w(u,v)$ to the right hand side. Thus, the contribution of every relevant edge to the left hand side is upper bounded by its contribution to the right hand side.  

Next, recall that  $w$ defines a fractional matching in $G$, and so we have $W(v, E) \leq 1$ for every node $v \in V$. Accordingly, from equations~\ref{label:general:n:eq:spurious:count:2} and~\ref{label:general:n:eq:spurious:count:3}, we infer that:
$$|B| \cdot (\epsilon/L^2) \leq \sum_{v \in B} \sum_{i=L'}^L W(v, E_i) \leq 2 \cdot \sum_{v \in V_X} \sum_{i=L'}^L W(v, E_i) \leq 2 \cdot |V_X|.$$
Rearranging the terms of the above inequality, we get our desired upper bound on $|B|$.
\begin{equation}
\label{label:general:n:eq:spurious:count:4}
|B| \leq (2L^2/\epsilon) \cdot |V_X|
\end{equation}
Finally, from equations~\ref{label:general:n:eq:spurious:count:1},~\ref{label:general:n:eq:spurious:count:4}, we infer that $|S| \leq (8 \delta L^3/\epsilon) \cdot |V_X|$. This concludes the proof of the lemma.

\subsection{Proof of Lemma~\ref{label:general:n:lm:spurious:good}}
\label{label:general:sec:n:lm:spurious:good}

\begin{itemize}
\item Throughout this proof, we will use the notations $V_X$ and $V_Y$ defined right after the proof of Lemma~\ref{label:general:lm:max:deg:100}.
\end{itemize}
\noindent The weights $w''$ are constructed from $w'$ by switching off the nodes in $S$ and the edges incident upon them (see equations~\ref{label:general:n:eq:w'},~\ref{label:general:n:eq:w*}). Since each node in $S$ has at most $2L$ weight under $w'$ (see Lemma~\ref{label:general:n:lm:w':bound}), and since each edge is incident upon two nodes, the transition from $w'$ to $w''$ decreases the sum  of the node-weights in $V_X$ by at most $4L \cdot |S|$. This insight is formally stated in the claim below.

\begin{claim}
\label{label:general:cl:n:lm:spurious:good:1}
We have:
$$\sum_{v \in V_X} W''(v, X) \geq \sum_{v \in V_X} W'(v, X) -  4L \cdot |S|.$$
\end{claim}

\begin{proof}
To prove the claim, we first show  that:
\begin{equation}
\label{label:general:n:eq:new:100}
\sum_{v \in V_X} W''(v, X) + 2 \cdot \sum_{v \in S} W'(v, X) \geq \sum_{v \in V_X} W'(v, X) 
\end{equation}
Equation~\ref{label:general:n:eq:new:100} follows from a simple counting argument. Consider any edge $(u,v) \in X$ that has at least one endpoint in $V_X$. These are the edges that contribute towards the right hand side of the above inequality.   Now, there are two possible cases to consider. 
\begin{itemize}
\item {\em Case 1.} At least one of the endpoints $u, v$ belong to $S$. In this case, we have $w''(u,v) = 0$. However, due to  the term $2 \cdot \sum_{v \in S} W'(v, S)$, the edge $(u,v)$ contributes at least $2 w'(u,v)$ towards the left hand side.  And clearly, the edge $(u,v)$ can contribute at most $2 w'(u,v)$ towards the right hand side. 
\item {\em Case 2.} None of the endpoints $u, v$ belong to $S$. In this case, we have $w'(u, v) = w''(u,v)$, and so the contribution of the edge $(u,v)$ towards the left hand side is at least as much as its contribution towards the right hand side. 
\end{itemize}
Next, we recall that $W'(v, X) \leq 2L$ for every node $v \in V$ (see Lemma~\ref{label:general:n:lm:w':bound}).  Thus, we have:
\begin{equation}
\label{label:general:n:eq:new:101}
4L \cdot |S| \geq 2 \cdot \sum_{v \in S} W'(v, X)
\end{equation}
From equations~\ref{label:general:n:eq:new:100} and~\ref{label:general:n:eq:new:101}, we infer that:
\begin{equation}
\label{label:general:n:eq:new:102}
\sum_{v \in V_X} W''(v, X) + 4 L \cdot |S| \geq \sum_{v \in V_X} W'(v, X) 
\end{equation}
The claim follows from equation~\ref{label:general:n:eq:new:102}.
\end{proof}

As we make a transition from $w'$ to $w''$, the total weight of the nodes in $V_X$ drops by at most $4L \cdot |S|$ (see Claim~\ref{label:general:cl:n:lm:spurious:good:1}). Hence, by a simple counting argument, due to this transition at most $(4L/\epsilon) \cdot |S|$ nodes  can experience their weights dropping by more than $\epsilon$. Next, recall that under $w'$, every node $v \in V_X$ has weight $W'(v, X) \geq h_1$ (see Lemma~\ref{label:general:n:lm:w'}).  Thus, every node in $Q_X$ has experienced a drop of $\epsilon$ due to the transition from $w'$ to $w''$. Accordingly, the size of the set $Q_X$ cannot be larger than $(4L/\epsilon) \cdot |S|$. Since $|S| \leq  (8 \delta L^3/\epsilon) \cdot |V_X|$, we infer that $|Q_X| \leq (32 \delta L^4/\epsilon^2) \cdot |V_X|$. This concludes the proof of the lemma.

\subsection{Proof of Lemma~\ref{label:general:n:lm:main}}
\label{label:general:sec:n:lm:main}

\begin{itemize}
\item Throughout this proof, we will use the notations $V_X$ and $V_Y$ defined right after the proof of Lemma~\ref{label:general:lm:max:deg:100}.
\end{itemize}
\noindent 
We begin by noting that under $w''$, the weights of the edges in $X^*$ are very close to one another.

\begin{observation}
\label{label:general:n:ob:close}
Let $c = \lfloor (L^4/2) \rfloor$. Then we have $1/(8c) \leq w''(e) \leq 1/c$ for every edge $e \in X^*$.
\end{observation}

\begin{proof}
Consider any edge $(u,v) \in X^*$. This edge belongs to some level $i \in [L', L]$, and since $(u,v) \in X^*$, by definition the edge has nonzero weight under $w''$. Hence, equations~\ref{label:general:n:eq:w'} and~\ref{label:general:n:eq:w*} implies that $w''(u,v) = (h_0/\lambda_{d_i}) \cdot L^{-4}$. Since $1/2 \leq h_0, \lambda_{d_i} \leq 1$ (see equations~\ref{label:general:eq:Li:lambda},~\ref{label:general:eq:h0:less}), we get: $1/(2L^4) \leq w''(u, v) \leq 2/L^4$. The observation now follows from the fact that $c = \lfloor (L^4/2) \rfloor \geq (L^4/4)$. See equation~\ref{label:general:eq:L:floor}.
\end{proof}

As an important corollary, we get an upper bound on the maximum degree of a node in $G_{X^*} = (V, X^*)$.

\begin{corollary}
\label{label:general:n:cor:ob:close}
We have $\text{deg}(v, X^*) \leq 8 c$ for every node $v \in V$.
\end{corollary}

\begin{proof}
The corollary holds since $1 \geq W''(v, X) = W''(v, X^*) \geq \text{deg}(v, X^*) \cdot (1/c)$ for all nodes $v \in V$. The first inequality follows from Lemma~\ref{label:general:n:lm:w:star}. The last inequality follows from Observation~\ref{label:general:n:ob:close}.
\end{proof}

\subsubsection{Outline of the proof}
\label{label:general:n:sec:outline}
Before proceeding any further, we give a high level overview of our approach. Suppose that we make two simplifying assumptions (which will be relaxed in Section~\ref{label:general:n:sec:complete}).
\begin{assume}
\label{label:general:n:ob:1}
Each edge $e \in X^*$ receives exactly the same weight $1/c''$ under $w''$, for some integer $c''$. Thus, we have $w''(e) = 1/c''$ for all $e \in X^*$. By Observation~\ref{label:general:n:ob:close},  the weights $\{w''(e)\}, e \in X^*,$ are already very close to one another. Here, we  take this one step further by assuming that they are exactly equal. We also assume that these edge-weights are inverse of some integer value. 
\end{assume}

\begin{assume}
\label{label:general:n:ob:2}
There is no edge $(u,v) \in E$ in the $(\alpha, \beta)$-partition at level $\ell(u,v) < L'$. In other words, we are assuming that the edge-set $Y$ is empty. 
\end{assume}
 
By Lemma~\ref{label:general:n:lm:w:star},  $w''$ defines a fractional matching in the graph $G = (V, E)$. We will now see that under the above two assumptions, we can say something more about $w''$. Towards this end, first note that as $Y = \emptyset$ (see Observation~\ref{label:general:n:ob:2}), every edge $(u,v) \in E$ has level $\ell(u,v) = \max(\ell(u), \ell(v)) \geq L'$ in the $(\alpha, \beta)$-partition. Since $V_X$ is precisely the set of nodes $v \in V$ with levels $\ell(v) \geq L'$,  we infer the following fact.
\begin{fact}
\label{label:general:n:fact:1}
Every edge $(u,v) \in E$ has at least one endpoint in $V_X$. 
\end{fact}
On the other hand, by Lemma~\ref{label:general:n:lm:spurious:good}, an overwhelming fraction of the nodes in $V_X$ receive weights that are close to one under $w''$. This, along with Fact~\ref{label:general:n:fact:1}, implies  that $w''$ is a ``near-maximal'' matching in $G$, in the sense that almost all the edges in $G$ have at least one nearly tight endpoint under $w''$. Hence, the value of $w''$, given by $w''(E) = \sum_{e \in E} w''(e)$, is very close to being a $2$-approximation to the size of the maximum cardinality matching in $G$. We will now show that there exists a matching $M^* \subseteq X^*$ whose size is very close to $w''(E)$. This will imply the desired guarantee we are trying to prove.

We now bound the value of the fractional matching $w''$. Since each edge in $X^*$ receives $1/c''$ weight under $w''$ (see Assumption~\ref{label:general:n:ob:1}), and since every other edge gets zero weight, we have $w''(E) = |X^*|/c''$. This is summarized as a  fact below.

\begin{fact}
\label{label:general:n:fact:2}
We have $w''(E) = |X^*|/c''$.
\end{fact}

It remains to show the existence of a large matching in $G_{X^*}$ of size very close to $w''(E)$. Towards this end, we first note that $c''$ is an upper  bound on the maximum degree of a node in $G_{X^*} = (V, X^*)$. This holds since $W''(v, X^*) = (1/c'') \cdot \text{deg}(v, X^*) \leq 1$ for all $v \in V$ (see Lemma~\ref{label:general:n:lm:w:star}). Thus, by Vizing's theorem~\cite{vizing1,vizing2}, there is a proper edge-coloring of the edges in $G_{X^*}$ using only $c'' + 1$ colors. Recall that a proper edge-coloring assigns one color to every edge in the graph, while ensuring that the edges incident upon the same node get different colors. We take any $c''+1$ edge-coloring in $G_{X^*}$ as per Vizing's theorem. By a simple counting argument, there exists a color  that is assigned to at least $|X^*|/(c''+1)$ edges in $X^*$. The edges that receive this color constitute a matching in $G_{X^*}$ (for they cannot share a common endpoint). In other words, there exists a matching $M^* \subseteq X^*$ of size $|X^*|/ (c''+1)$.  From Fact~\ref{label:general:n:fact:2}, we infer that  $|M^*| \geq w''(E) \cdot (c''/(c''+1))$. Since $c'' = \Theta(\text{poly } L) = \Theta(\text{poly} \log n)$, the size of the matching $M^*$ is indeed very close to the value $w''(E)$. This gives us the desired bound.

\subsubsection{The complete proof}
\label{label:general:n:sec:complete}

Unfortunately, while giving a complete proof, we cannot rely on the simplifying assumptions from Section~\ref{label:general:n:sec:outline}. Since Assumption~\ref{label:general:n:ob:1} no longer holds, a natural course of action is to discretize the edge-weights under $w''$ into constantly many buckets, without loosing too much of the value of the fractional matching defined by $w''$. This will at least ensure that the number of different weights is some constant, and we  hope to extend our proof technique from Section~\ref{label:general:n:sec:outline} to this situation. 

We pick some sufficiently large constant integer $K$, and round down  the weights under $w''$ to the nearest multiples of $1/(8Kc)$. Let $w^*$ be these rounded weights.  Specifically, for every edge $e \in X^*$, we have:
\begin{eqnarray}
\label{label:general:n:eq:w*:new}
w^*(e) & = & \frac{\tau^*(e)}{8Kc}, \ \text{ where } \tau^*(e) \text{ is a positive integer such that } \frac{\tau^*(e)}{8Kc} \leq w''(e) < \frac{\tau^*(e)+1}{8Kc}.
\end{eqnarray}

As usual, given any subset of edges $E' \subseteq E$ and  node $v \in V$, we define $W^*(v, E')$ to be the total weight received by $v$ from its incident edges in $E'$, and $w^*(E')$ to be the total weight of all the edges in $E'$. We now list down some basic properties of the weights $w^*$, comparing them against the weights $w''$.

\begin{observation}
\label{label:general:n:ob:discrete:1}
For all $e \in X^*$, we have  $(8c)^{-1} \leq w^*(e) \leq w''(e) \leq  (c)^{-1}$ and $K \leq \tau^*(e) \leq 8K$.
\end{observation}

\begin{proof}
Follows from Observation~\ref{label:general:n:ob:close} and equation~\ref{label:general:n:eq:w*:new}.
\end{proof}

\begin{observation} 
\label{label:general:n:ob:discrete:2}
For every node $v \in V$, we have  $W''(v, X^*) - 1/K \leq W^*(v, X^*) \leq W''(v, X^*)$. 
\end{observation}

\begin{proof}
Fix any node $v \in V$. It has at most $8c$ edges incident upon it in $G_{X^*} = (V, X^*)$ (see Corollary~\ref{label:general:n:cor:ob:close}). As we make a transition from $w''$ to $w^*$, each of these edges incurs a weight-loss of at most $1/(8Kc)$ (see equation~\ref{label:general:n:eq:w*:new}). Hence, the total loss in the weight received by the node $v$ from its incident edges is at most $(8c) \cdot 1/(8Kc) = 1/K$. Accordingly, we have $W^*(v, X^*) \geq W''(v, X^*) - 1/K$.
\end{proof}

\paragraph{Roadmap.} The rest of the section is organized as follows.
\begin{enumerate}
\item Recall that $M \subseteq E$ is a matching in the graph $G = (V, E)$ we want to compete against. We first construct two multi-graphs $\G_{X^*} = (V, \E_{X^*})$ and $\G_{M \cap Y} = (V, \E_{M \cap Y})$. The multi-edges in $\E_{X^*}$ are constructed from the edge-set $X^*$. Similarly, as the notation suggests, the multi-edges in $\E_{M \cap Y}$ are constructed from the edge-set $M \cap Y$ (which consists of the matched edges in $M$ with both endpoints at a level less than $L'$ in the $(\alpha, \beta)$-partition). We also define $\G = (V, \E)$ to be the union of the two multi-graphs $\G_{X^*}$ and $\G_{M \cap Y}$, so that $\E = \E_{X^*} \cup \E_{M \cap Y}$.
\item  Next, we construct a fractional matching $w_{\E} : \E \rightarrow [0,1]$ in the multi-graph $\G$. This fractional matching is ``uniform'', in the sense that it assigns exactly the same weight to every multi-edge. This is a significant property, since the fractional matching we had to deal with in the simplified proof of Section~\ref{label:general:n:sec:outline} was also uniform, and this fact was used while comparing the value of the fractional matching against the size of the integral matching constructed out of Vizing's theorem~\cite{vizing1,vizing2}. As in Section~\ref{label:general:n:sec:outline}, we will show that the value of the fractional matching $w_{\E}(\E) = \sum_{e \in \E} w_{\E}(e)$ is very close to being a $2$-approximation to the size of $M \subseteq E$ (the matching in $G$ we are competing against).
\item Finally, we construct a matching $M^* \subseteq X^* \cup Y$ whose size is very close to $w_{\E}(\E)$. To construct this matching $M^*$, we use a generalized version of Vizing's theorem for multi-graphs.
\item Steps 2 and 3 imply that the size of the matching $M^*$ is very nearly within a factor of $2$ of $|M|$. This concludes the proof of Lemma~\ref{label:general:n:lm:main}.
\end{enumerate}

\bigskip
\paragraph{Step I: Constructing the multi-graphs $\G_{X^*} = (V, \E_{X^*})$,  $\G_{M \cap Y} = (V, \E_{M \cap Y})$ and $\G = (V, \E)$.}
\ 

\bigskip
\noindent  We create $\tau^*(e)$ copies of each edge $e \in X^*$, and let $\E_{X^*}$ denote the collection of these multi-edges. Next,  we create $\max\left(0, 8Kc - \text{deg}(u, \E_{X^*}) - \text{deg}(v, \E_{X^*})\right)$ copies of every edge $(u,v) \in M \cap Y$, and let  $\E_{M \cap Y}$ denote the collection of these multi-edges. We also define $\E = \E_{X^*} \cup \E_{M \cap Y}$. We will be interested in the multi graphs $\G_{X^*} = (V, \E_{X^*})$, $\G_{M \cap Y} = (V, \E_{M \cap Y})$ and $\G = (V, \E)$. We now state three simple observations that will be useful later on.

\begin{observation}
\label{label:general:n:ob:multigraph:1}
In the multi-graph $\G$, the degree of every node $v \in V$ is at most $8Kc$.
\end{observation}

\begin{proof}
Throughout the proof, fix any node $v \in V$. We first bound the degree of this node among the multi-edges in $\E_{X^*}$. Accordingly, consider an   edge $(u,v) \in X^*$ that contributes to $\text{deg}(v, \E_{X^*})$. Under $w^*$, this edge has weight $w^*(u,v) = \tau^*(u,v)/(8Kc)$. While constructing the multi-graph $\G$, we simply created $\tau^*(u,v)$ copies of this edge. Thus, summing over all such edges $(u,v) \in X^*$, we get:
\begin{eqnarray}
\text{deg}(v, \E_{X^*}) & = & \sum_{(u,v) \in X^*} \tau^*(u,v) \nonumber  \\
& = & \sum_{(u,v) \in X^*} w^*(u,v) \cdot (8Kc) \nonumber \\
& = & W^*(v, X^*) \cdot (8Kc) \nonumber \\
& \leq & 8Kc \label{label:general:eq:lastmoment}
\end{eqnarray}
The last inequality holds since $X^* \subseteq X$ is the set of edges in $X$ that receive nonzero weights under $w''$, and hence, we have $W^*(v, X^*) \leq W''(v, X^*) = W''(v, X) \leq 1$ (see Observation~\ref{label:general:n:ob:discrete:2} and Lemma~\ref{label:general:n:lm:w:star}).

\medskip
Now, we have two cases to consider. 

\medskip
\noindent {\em Case 1. $v \in V_X$.} In this case, all the multi-edges in $\G$ that are incident upon it originate from the edge-set $X^*$. Thus, we have $\text{deg}(v, \E) = \text{deg}(v, \E_{X^*})$. The observation now follows from equation~\ref{label:general:eq:lastmoment}.

\medskip
\noindent {\em Case 2. $v \in V_Y$.} In this case, note that the node $v$ can be incident upon at most one edge in $M \cap Y$ (for $M$ is a matching). Thus, by our construction of the multigraph, at most $(8Kc - \text{deg}(v, \E_{X^*}))$  multi-edges incident upon $v$ are included in $\E_{M \cap Y}$. We therefore conclude that:
$$\text{deg}(v, \E) = \text{deg}(v, \E_{X^*}) + \text{deg}(v, \E_{M \cap Y}) \leq \text{deg}(v, \E_{X^*}) + 8Kc - \text{deg}(v, \E_{X^*}) = 8Kc.$$
\end{proof}

\begin{observation}
\label{label:general:n:ob:multigraph:2}
In the multi-graph $\G_{X^*}$, there are at most $8K$ multi-edges joining any two given nodes. 
\end{observation}

\begin{proof}
Consider any two nodes $u, v \in V$. If $(u,v) \notin X^*$, then by our construction, there cannot be any multi-edge between $u$ and $v$ in $\G_{X^*}$. Hence, we assume that $(u,v) \in X^*$. Recall that $w^*(u,v) = \tau^*(u,v)/(8Kc)$ and that $\tau^*(u,v)$ copies of the edge $(u,v)$ are added to the multi-graph $\G_{X^*}$. Thus, it remains to show that $\tau^*(u,v) \leq 8K$, which clearly holds since $w^*(u,v) \leq 1/c$ (see Observation~\ref{label:general:n:ob:discrete:1}).
\end{proof}

\begin{observation}
\label{label:general:n:ob:multigraph:3}
In  multi-graph $\G$, we have $\text{deg}(u, \E) + \text{deg}(v, \E) \geq 8Kc$ for each edge $(u,v) \in M \cap Y$.
\end{observation}

\begin{proof}
If $\text{deg}(u, \E_{X^*}) + \text{deg}(v, \E_{X^*}) \geq 8Kc$, then there is nothing more to prove. So let $\text{deg}(u, \E_{X^*}) + \text{deg}(v, \E_{X^*}) = 8Kc - \mu$ for some integer $\mu \geq 1$. In this case, by our construction, $\mu$ copies of the edge $(u,v)$ get added to $\E_{M \cap Y}$, and hence to $\E = \E_{X^*} \cup \E_{M \cap Y}$. Thus, we again get: $\text{deg}(u, \E) + \text{deg}(v, \E) \geq 8Kc$. 
\end{proof}

\bigskip
\paragraph{Step II: Constructing the fractional matching $w_{\E} : \E \rightarrow [0,1]$ in $\G$.} 
\

\bigskip
\noindent We assign a weight $w_{\E}(e) = (8Kc)^{-1}$ to every multi-edge $e \in \E$. Since every node in $\G$ has degree at most $8Kc$, we infer that:
$$W_{\E}(v, \E) = \text{deg}(v, \E) \cdot 1/(8Kc) \leq 1.$$
This shows that the weights $\{ w_{\E}(e) \}, e \in \E,$ constitute a fractional matching in the multi-graph $\G$. We now state an easy bound on the value of this fractional matching.

\begin{observation}
\label{label:general:n:ob:frac:1}
We have $\sum_{e \in \E} w_{\E}(e) = |\E|/(8Kc)$. 
\end{observation}

As an aside, we intimate the reader that in Step III  we will  construct a matching $M^* \subseteq X^* \cup Y$ of size at least $|\E|/ (8K(c+1))$. Thus, Observation~\ref{label:general:n:ob:frac:1}  will imply that the size of $M^*$ is a $(1+1/c)$-approximation to the value of $w_{\E}$. At the present moment, however, our goal is to compare  the value of $w_{\E}$ against the size of the matching $M$.  We now make two simple observations that will be  very helpful in this regard.

\begin{observation}
\label{label:general:n:ob:g:w}
For every edge $(u,v) \in M \cap Y$, we have $W_{\E}(u) + W_{\E}(v) \geq 1$. 
\end{observation}

\begin{proof}
Fix any edge $(u,v) \in M \cap Y$. Since $(u, v) \in Y$, both $u$ and $v$ are at  levels less than $L'$ in the $(\alpha, \beta)$-partition. This means that $u, v \notin V_X$. So there can be no multi-edge in $\E_{X^*}$ between $u$ and $v$.  We consider two possible cases.
\begin{itemize}
\item {\em Case 1.}  $\text{deg}(u, \E_{X^*}) + \text{deg}(v, \E_{X^*}) \geq 8Kc$. Note that every multi-edge in $\E_{X^*}$ gets a weight $1/(8Kc)$ under $w_{\E}$. Since no such multi-edge joins $u$ and $v$, we have: 
\begin{eqnarray*}
W^*(u, \E) + W^*(v, \E) & \geq & W^*(u, \E_{X^*}) + W^*(v, \E_{X^*}) \\
& = & \left(\text{deg}(u, \E_{X^*}) + \text{deg}(v, \E_{X^*})\right) \cdot (8Kc)^{-1} \\
& \geq & 1
\end{eqnarray*}
\item {\em Case 2.}   $\text{deg}(u, \E_{X^*}) + \text{deg}(v, \E_{X^*}) = 8Kc - \mu$, where $\mu > 0$ (say).  In this case, we also make $\mu$ copies of the edge $(u,v)$ and include them in $\E_{M \cap Y}$. Note that every multi-edge in $\E = \E_{X^*} \cup \E_{M \cap Y}$ gets a weight $1/(8Kc)$ under $w_{\E}$. Since there is no multi-edge between $u$ and $v$ in $\E_{X^*}$, we get: 
\begin{eqnarray*}
W^*(u, \E) + W^*(v, \E) & \geq & W^*(u, \E_{X^*}) + W^*(v, \E_{X^*}) + W^*(u, \E_{M \cap Y})\\
& = & \left(\text{deg}(u, \E_{X^*}) + \text{deg}(v, \E_{X^*}) + \mu \right) \cdot (8Kc)^{-1} \\
& \geq & 1
\end{eqnarray*}
\end{itemize}
\end{proof}

\begin{observation}
\label{label:general:n:ob:g:w:1}
We have $W_{\E}(v) \geq h_1 - \epsilon - (1/K)$ for every node $v \in V_X \setminus Q_X$ (see Lemma~\ref{label:general:n:lm:spurious:good}).
\end{observation}

\begin{proof}
Consider any node $v \in V_X \setminus Q_X$. Consider any edge $(u,v) \in X^*$ incident upon $v$. This edge has weight $w^*(u,v) = \tau^*(u,v)/(8Kc)$ under $w^*$ (see equation~\ref{label:general:n:eq:w*:new}). Under $w_{\E}$, on the other hand, we create $\tau^*(u,v)$ copies of this edge and assign a weight $1/(8Kc)$ to each of these copies. The total weight assigned to all these copies, therefore, remains exactly equal to $\tau^*(u,v)/(8Kc)$.  Next, note that since $v \in V_X$, there cannot be any edge in $Y$ that is incident upon $v$. Thus, we have:
$$W_{\E}(v, \E) = W_{\E}(v, \E_{X^*}) = W^*(v, X^*) \geq W''(v, X^*) - (1/K) \geq h_1 - \epsilon - (1/K).$$
The second-last inequality follows from Observation~\ref{label:general:n:ob:discrete:2}. The last inequality follows from Lemma~\ref{label:general:n:lm:spurious:good}.
\end{proof}

We will now bound the sum of the node-weights under $w_{\E}$ by the size of the matching $M$.

\begin{claim}
\label{label:general:n:cl:almost}
We have $|M| \leq (1/h_2) \cdot \sum_{v \in V} W_{\E}(v, \E)$.
\end{claim}

\begin{proof}
Every edge $(u,v) \in M \setminus Y$ has at least one endpoint at level $L'$ or higher in the $(\alpha, \beta)$-partition. In other words, every edge $(u,v) \in M \setminus Y$ has at least one endpoint in $V_X$. Since each node in $V_X$ can get matched by at most one edge in $M$, we have:
\begin{equation}
\label{label:general:n:eq:almost}
|M \setminus Y| \leq |V_X|  
\end{equation}
Each edge $(u, v) \in M \cap Y$, on the other hand, has $W_{\E}(u, \E) + W_{\E}(v, \E) \geq 1$ (see Observation~\ref{label:general:n:ob:g:w}). Since both  these endpoints $u, v$ belong to the node-set $V_Y = V \setminus V_X$, summing  over all edges in $M \cap Y$, we get:
\begin{equation}
\label{label:general:n:eq:almost:1}
|M \cap Y| \leq \sum_{v \in V_Y} W_{\E}(v, \E)
\end{equation}
Since $|M| = |M \cap Y| + |M \setminus Y|$, equations~\ref{label:general:n:eq:almost} and~\ref{label:general:n:eq:almost:2} give us the following upper bound on  $|M|$.
\begin{equation}
\label{label:general:n:eq:almost:2}
|M| \leq \sum_{v \in V_Y} W_{\E}(v, \E) + |V_X|
\end{equation}
To complete the proof, we will now upper bound the size of the set $V_X$ by the sum $\sum_{v \in V_X} W_{\E}(v, \E)$. The main idea is simple. Lemma~\ref{label:general:n:lm:spurious:good} shows that only a negligible  fraction of the nodes in $V_X$ belong to the set $Q_X$, and Observation~\ref{label:general:n:ob:g:w:1} shows that for every other node  $v \in V_X \setminus Q_X$, the weight  $W_{\E}(v, \E)$ is very close to one. To be more specific, we infer that:
\begin{eqnarray}
|V_X| & \leq & (1- 32 \delta L^4/\epsilon^2)^{-1} \cdot |V_X \setminus Q_X| \label{label:general:n:eq:almost:10} \\
& \leq & (1-32 \delta L^4/\epsilon^2)^{-1} \cdot (h_1 - \epsilon - 1/K)^{-1} \cdot \sum_{v \in V_X \setminus Q_X} W_{\E}(v, \E) \label{label:general:n:eq:almost:11} \\
& = & h_2^{-1}  \cdot \sum_{v \in V_X \setminus Q_X} W_{\E}(v, \E) \label{label:general:n:eq:almost:12} \\
& \leq & h_2^{-1} \cdot \sum_{v \in V_X} W_{\E}(v, \E) \label{label:general:n:eq:almost:13}
\end{eqnarray}
Equation~\ref{label:general:n:eq:almost:10} follows from Corollary~\ref{label:general:n:cor:lm:spurious:good}. Equation~\ref{label:general:n:eq:almost:11} follows from Observation~\ref{label:general:n:ob:g:w:1}. Equation~\ref{label:general:n:eq:almost:12} follows from equation~\ref{label:general:eq:h2}. 

Since the node-set $V$ is partitioned by the subsets $V_X$ and $V_Y$, from equations~\ref{label:general:n:eq:almost:13} and~\ref{label:general:n:eq:almost:2} we infer that:
\begin{eqnarray}
|M| & \leq & \sum_{v \in V_Y} W_{\E}(v, \E) + h_2^{-1} \cdot \sum_{v \in V_X} W_{\E}(v, \E) \nonumber \\
& \leq & h_2^{-1} \cdot \left(\sum_{v \in V_Y} W_{\E}(v, \E) +  \sum_{v \in V_X} W_{\E}(v, \E)\right) \nonumber \\
& = & h_2^{-1} \cdot \sum_{v \in V} W_{\E}(v, \E) \nonumber
\end{eqnarray}
This concludes the proof of the claim.
\end{proof}

As an immediate corollary, we get a bound on  $|M|$ in terms of the value of the fractional matching $w_{\E}$.

\begin{corollary}
\label{label:general:n:cor:cl:almost}
We have $|M| \leq (2/h_2) \cdot \sum_{e \in \E} w_{\E}(e)$.
\end{corollary}

\begin{proof}
Follows from Claim~\ref{label:general:n:cl:almost} and the fact that each multi-edge is incident upon two nodes.
\end{proof}

\bigskip
\paragraph{Step III: Constructing the  matching $M^* \subseteq X^* \cup Y$.}
\

\bigskip
\noindent Every node in the multigraph $\G_{X^*}$ has degree at most $8Kc$ (see Observation~\ref{label:general:n:ob:multigraph:1}), and any two given nodes  in $\G_{X^*}$ have at most $8K$ multi-edges between them (see Observation~\ref{label:general:n:ob:multigraph:2}). Hence, by Vizing's theorem~\cite{vizing1,vizing2}, there is a proper edge-coloring in $\G_{X^*}$ that uses only $8Kc + 8K = 8K(c+1)$ colors. Let $\lambda \rightarrow \{1, \ldots,  8K(c+1)\}$ be such a coloring, where $\lambda(e)$ denotes the color assigned to the multi-edge $e \in \E_{X^*}$. By definition, two multi-edges incident upon the same node get different colors under $\lambda$.

We  use $\lambda$ to get a proper edge-coloring $\lambda_{\E} : \E \rightarrow \{1, \ldots, 8K(c+1)\}$ of the multi-graph $\G$ with the same number of colors. This is done as follows.
\begin{itemize}
\item First, we set $\lambda_{\E}(e) = \lambda(e)$ for every multi-edge $e \in \E_{X^*}$. This ensure that two multi-edges in $\E_{X^*}$ incident upon the same node  get different colors under $\lambda_{\E}$, for $\lambda$ is already a proper coloring of $\G_{X^*}$. 
\item It remains to assign a color to every multi-edge in $\G_{M \cap Y}$. Towards this end, recall that each multi-edge in $\G_{M \cap Y}$ originates from some edge in $M \cap Y$. Consider any edge $(u,v) \in M \cap Y$. There are two possible cases.
\begin{itemize}
\item {\em Case 1. $\text{deg}(u, \E_{X^*})  + \text{deg}(v, \E_{M \cap Y}) \geq 8Kc$}. 

 In this case, we do not have any copy of the edge $(u,v)$ in $\G_{M \cap Y}$. There is nothing to be done as far as the coloring $\lambda_{\E}$ is concerned.
\item {\em Case 2. $\text{deg}(u, \E_{X^*}) + \text{deg}(v, \E_{M \cap Y}) = 8Kc - \mu$, for some  integer $\mu \geq 1$.}

In this case,  there are $\mu$ copies of the edge $(u,v)$ in  $\G_{M \cap Y}$. While assigning colors to  these copies, we only need to ensure that they do not come into conflict with the colors that have already been  assigned to the multi-edges in $\E_{X^*}$  incident upon $u$ and $v$. But, we are guaranteed that there are exactly $(8Kc - \mu)$ of these potentially troubling edges. Since we have a palette of $(8Kc + 8K)$  colors to begin with, even after assigning the colors to the edges in $\E_{X^*}$, we are left with at least $(8Kc+ 8K) - (8Kc - \mu) = 8K + \mu$ colors that have not been used on any multi-edge incident upon $u$ or $v$. Using these leftovers, we can easily color all the $\mu$ copies of the edge $(u,v)$ without creating any conflict with the colors assigned to the multi-edges in $\E_{X^*}$.
\end{itemize}
\end{itemize}
Thus, there exists a proper coloring of $\G = (V, \E)$ using $8K(c+1)$ colors. Since each multi-edge is assigned one color,  one of these $8K(c+1)$ colors  will hit at least $|\E|/(8K(c+1))$ multi-edges, and  those multi-edges will surely constitute a matching (for they cannot share a common endpoint). This shows the existence of a matching $M^* \subseteq X^* \cup Y$ of size at least $|\E|/(8K(c+1))$. Thus, we get the following claim.

\begin{claim}
\label{label:general:n:cl:int}
There exists a matching $M^* \subseteq X^* \cup Y$ of size at least $\left| \E \right|/(8K(c+1))$. 
\end{claim}

\bigskip

\paragraph{Step IV: Wrapping things up (The approximation guarantee)}
\

\bigskip
By Claim~\ref{label:general:n:cl:int}, there exists a matching $M^* \subseteq X^* \cup Y$ of size at least $\left| \E \right|/(8K(c+1))$. By Claim~\ref{label:general:n:cor:cl:almost}, the size of the matching $M$  is at most $(2/h_2)$ times the value of the fractional matching $w_{\E}$ defined on $\G = (V, \E)$. Finally, by Observation~\ref{label:general:n:ob:frac:1}, the value of the fractional matching $w_{\E}$ is exactly $\left| \E \right|/(8Kc)$. Thus, we conclude that:
$$ |M| \leq (2/h_2) \cdot \sum_{e \in \E} w_e(\E) = (2/h_2) \cdot \frac{\left| \E \right|}{(8Kc)} \leq (2/h_2) \cdot (1+1/c) \cdot |M^*| \leq (2/h_2) \cdot (1+\epsilon) \cdot |M^*|.$$
The last inequality holds since $c = \lfloor (L^4/2) \rfloor \geq (1/\epsilon)$ (see Observation~\ref{label:general:n:ob:close} and equation~\ref{label:general:eq:L:floor}). In other words, there exists a matching $M^* \subseteq X^* \cup Y$ whose size is a $(2/h_2)  \cdot (1+\epsilon)$-approximation to the size of the matching $M$. This concludes the proof of Lemma~\ref{label:general:n:lm:main}.

\section{Maintaining the edge-set of a skeleton: Proof of Theorem~\ref{label:general:th:maintain:skeleton}}
\label{label:general:sec:algo}

In this section, we consider the following dynamic setting. We are given an input graph $G = (V, E)$ with $|V| = n$ nodes. In the beginning, the edge-set $E$ is empty. Subsequently, at each time-step, an edge is inserted into (or deleted from) the graph $G$. However, at each time-step, we know for sure that the maximum degree of any node in $G$ is at most $d$. Thus, we have the following guarantee.

\begin{lemma}
\label{label:general:lm:deg:1}
We have $\text{deg}(v, E) \leq d$ for all nodes $v \in V$.
\end{lemma}

We will present  an algorithm for maintaining the edge-set $X$ of a skeleton  of $G$ (see Definition~\ref{label:general:def:skeleton}).

\paragraph{Roadmap.} The rest of this section is organized as follows.
\begin{itemize}
\item In Section~\ref{label:general:sec:algo:roadmap}, we present a high level overview of our approach.
\item In Section~\ref{label:general:sec:maintain:prelim}, we define the concepts of a ``critical structure'' and a ``laminar structure'' that will be used in later sections. We also motivate the connection between these two concepts and the notion of a skeleton of the graph.
\item In Section~\ref{label:general:sec:building:block}, we describe two subroutines that will be crucially used by our algorithm.
\item In Section~\ref{label:general:sec:main:algo}, we  present our algorithm for dynamically maintaining  critical  and  laminar structures and analyze some of its basic properties.
\item In Section~\ref{label:general:sec:maintain:skeleton}, we  show that our algorithm in Section~\ref{label:general:sec:maintain:critical}  gives us the edge-set  of some skeleton of $G$. See Theorem~\ref{label:general:th:maintain:skeleton:minor}.
\item In Section~\ref{label:general:sec:analyze:time}, we analyze the amortized update time of our algorithm. See Theorem~\ref{label:general:th:maintain:runtime:main}.
\item Finally, Theorem~\ref{label:general:th:maintain:skeleton} follows from Theorem~\ref{label:general:th:maintain:skeleton:minor} and Theorem~\ref{label:general:th:maintain:runtime:main}. 
\end{itemize}

\subsection{A high level overview of our approach}
\label{label:general:sec:algo:roadmap}

We begin by introducing the  concepts of a ``critical structure'' (see Definition~\ref{label:general:def:critical:structure}) and a ``laminar structure'' (see Definition~\ref{label:general:def:laminar:structure}). Broadly speaking, in a critical structure we classify each node $v \in V$  as being either ``active'' or ``passive''. Barring a few outliers that are called ``$c$-dirty'' nodes, every active (resp. passive) node has degree larger (resp. smaller) than $\epsilon d/L^2$ (resp. $3\epsilon d/L^2$). While the  $c$-dirty nodes are too few in numbers to have any impact on the approximation ratio of our algorithm, the flexibility of having  such nodes turns out to be very helpful if we want to maintain a critical structure in the dynamic setting. 

The edge-set $H \subseteq E$ of a critical structure consists of all the edges incident upon active nodes, and a laminar structure consists of a family of $(L_d+1)$ subsets of these edges $H = H_0 \supseteq H_1 \supseteq \cdots \supseteq H_{L_d}$, where $L_d = \lceil \log_2 (d/L^4) \rceil$ (see equation~\ref{label:general:eq:Li}).  The set $H_j$ is identified as the ``layer $j$'' of the laminar structure. 

Ignoring some low level details, our goal is to reduce the degree of each active node by a factor of half across successive layers. Thus, the degree of an active node $v$ in the last layer $H_{L_d}$ will be very close to $2^{-L_d}$ times its degree among the edges in $H$. Since every edge incident upon an active node is part of $H$, we get $\text{deg}(v, H_{L_d}) \simeq 2^{-L_d} \cdot \text{deg}(v, H) = 2^{-L_d} \cdot \text{deg}(v, E) = (\lambda_d L^4/d) \cdot \text{deg}(v, E)$ for every active node $v \in V$. The last equality holds since $L_d$ is chosen in such a way that $2^{L_d} = \lambda_d d/L^4$ (see equation~\ref{label:general:eq:Li:lambda}).  

We will also try to ensure that the degree of every node (not necessarily active) reduces by at least a factor of half (it is allowed to reduce by more) across successive layers. This will imply that $\text{deg}(v, H_{L_d})$ is at most $2^{-L_d} \cdot \text{deg}(v, H) = (\lambda_d L^4/d) \cdot \text{deg}(v, H) \leq (\lambda_d L^4/d) \cdot \text{deg}(v, E)$ for every node $v \in V$.  Finally, recall that $\text{deg}(v, E) \leq d$ for every node $v \in V$ (see Lemma~\ref{label:general:lm:deg:1}), and that $\text{deg}(v, E) < 3\epsilon d/L^2$ for the passive nodes $v$. Thus, roughly speaking, we will have $\text{deg}(v, H_{L_d}) \leq \lambda_d L^4$ for all nodes $v \in V$, and $\text{deg}(v, H_{L_d}) \leq 3\epsilon \lambda_d L^2$ for all passive nodes $v$.  

The main purpose of the preceding discussion  was to show the link between the concepts of critical and laminar structures on the one hand, and the notion of a skeleton of a graph on the other. Basically, ignoring the  small number of  $c$-dirty nodes, (a) the set of active (resp. passive) nodes correspond to the set of big  (resp. tiny) nodes in Definition~\ref{label:general:def:skeleton}, and (b) the edges in the last layer $H_{L_d}$ of the laminar structure correspond to the edge-set $X$ in Definition~\ref{label:general:def:skeleton}. This shows that in order to maintain the edge-set $X$ of a skeleton of $G$, it suffices to maintain a critical structure  $(A, P, D_c, H)$ and a corresponding laminar structure, where the symbols $A, P, D_c \subseteq V$ respectively denote the sets of active, passive and $c$-dirty nodes. 


\subsection{Critical and laminar structures} 
\label{label:general:sec:maintain:prelim}

We first define the concept of a ``critical structure''. 

\begin{definition}
\label{label:general:def:critical:structure}
A tuple $(A, P, D_c, H)$, where $A, P, D_c \subseteq V$ and $H \subseteq E$, is called a ``critical structure'' of $G = (V, E)$ iff the following five conditions are satisfied:
\begin{enumerate} 
\item We have $\text{deg}(v, E) > \epsilon d/L^2$ for all nodes $v \in A \setminus D_c$. 
\item We have $\text{deg}(v, E) < 3 \epsilon d/L^2$ for all nodes $v \in P \setminus D_c$.
\item We have:
\begin{equation}
\label{label:general:eq:critical:count} 
|D_c| \leq (\delta/(L_d+1)) \cdot |A|.
\end{equation}
\item We have $H = \{ (u,v) \in E : \text{ either } u \in A \text{ or } v \in A\}$.
\item We have $P = V \setminus A$. Thus, the node-set $V$ is partitioned by the subsets $A$ and $P$. 
\end{enumerate}
The nodes in $A$ (resp. $P$) are called ``active'' (resp. ``passive''). The nodes in $D_c$ (resp. $V \setminus D_c$) are called ``$c$-dirty'' (resp. ``$c$-clean''), where the symbol ``$c$'' stands for the term ``critical''. 
\end{definition}
Intuitively, the above definition says  that (a) the active nodes have large degrees in $G$, (b) the passive nodes have small degrees in $G$, (c) some nodes (which are called $c$-dirty) can violate the previous two conditions, but their  number is negligibly small compared to the number of active nodes, and (d) the edge-set $H$ consists of all the edges in $E$ with at least one active endpoint.  Roughly speaking, the nodes in $A \setminus D_c$ will belong to the set of big nodes $B$ in the skeleton, and the nodes in $P \setminus D_c$ will belong to the set of tiny nodes $T$ (see Definition~\ref{label:general:def:skeleton}).  In Section~\ref{label:general:sec:maintain:critical}, we present an algorithm for maintaining  a critical structure $(A, P,  D_c, H)$. Below, we highlight one important property of a critical structure that follows from Definition~\ref{label:general:def:critical:structure}, namely, all the edges in $E$ that are incident upon an active node are part of the set $H$.
\begin{observation}
\label{label:general:ob:edge}
In a critical structure $(A, P,  D_c, H)$, we have $\text{deg}(v, H) = \text{deg}(v, E)$ for all $v \in A$.
\end{observation}

We next define the concept of a ``laminar structure''.

\begin{definition}
\label{label:general:def:laminar:structure}
Consider a critical structure $(A, P, D_c, H)$ as per Definition~\ref{label:general:def:critical:structure}. A laminar structure w.r.t. $(A, P, D_c, H)$ consists of $(L_d + 1)$ ``layers'' $0, \ldots, L_d$. Each layer $j \in [0, L_d]$ is associated with a tuple $(C_{lj}, D_{lj}, H_j)$ where $C_{lj}, D_{lj} \subseteq A$ and $H_j \subseteq H$. The nodes in $C_{lj}$ (resp. $D_{lj}$) are called $l$-clean (resp. $l$-dirty) at layer $j$, where ``$l$'' stands for the term ``laminar''.
\end{definition}

We will be interested in a laminar structure that satisfies the four invariants described below.

\medskip
The first invariant states that the edge-sets corresponding to successive layers are contained within one another. In other words, the edge-sets $H_0, \ldots, H_{L_d}$ form a laminar family. Furthermore, we have $H = H_0$. 

\begin{invariant}
\label{label:general:inv:laminar:edge}
$H = H_0 \supseteq H_1 \supseteq \cdots \supseteq H_{L_d}$.
\end{invariant}

The second invariant states that at each layer $j \in [0, L_d]$, each active node is classified as either  $l$-dirty or $l$-clean. In other words, the set of active nodes is partitioned by the subsets of $l$-clean and $l$-dirty nodes at each layer $j$. Next, just like the edge-sets $H_j$, the sets of $l$-clean nodes are also contained within one another across successive layers. The sets $C_0, \ldots, C_{l,L_d}$ also form a laminar family. However, unlike the edge-set $H_0$, which is always equal to $H$, the set $C_0$ can be properly contained within the set of active nodes $A$. To be more precise, we define $D_{l0} = A \cap D_c$ to be the set of active nodes that are also $c$-dirty, and $C_{l0}
 = A \setminus D_{l0}$ to be the set of active nodes that are $c$-clean. 
\begin{invariant}
\label{label:general:inv:laminar:contain}
The following conditions hold.
\begin{enumerate}
\item $D_{l0} = A \cap D_c$.
\item $D_{lj} \cap C_{lj} = \emptyset$ and $D_{lj} \cup C_{lj} = A$ for every layer $j \in [0, L_d]$.
\item $D_{l0} \subseteq D_{l1} \subseteq \cdots \subseteq D_{l, L_d}$, and accordingly, $A \supseteq C_0 \supseteq C_1 \supseteq \cdots \supseteq C_{l,L_d}$.
\end{enumerate}
\end{invariant}

Recall the discussion in Section~\ref{label:general:sec:algo:roadmap}. In an ideal scenario, we would like  to reduce the degree of every active node by a factor of $1/2$ across successive layers. Specifically, we would like to have $\text{deg}(v, H_j) = (1/2) \cdot \text{deg}(v, H_{j-1})$ for every active node $v \in A$ and layer $j \in [1, L_d]$. Such a structure, however,  is  very difficult to maintain in a dynamic setting. Accordingly, we introduce some slack, and we are happy as long as  $\text{deg}(v, H_j)$, instead of being exactly equal to $(1/2) \cdot \text{deg}(v, H_{j-1})$, lies in the range $\left[1/(2\eta) \cdot \text{deg}(v, H_{j-1}), (\eta/2) \cdot \text{deg}(v, H_{j-1})\right]$ for $\eta = (1+\gamma/L_d)$. We want every node in the set $C_{lj} \subseteq A$ to satisfy this approximate degree-splitting condition at all the layers $j' \leq j$.  In other words, if an active node $v$ is $l$-clean at layer $j$, then its degree ought to have been split roughly by half across successive layers in the interval $[0,j]$. This motivates our third invariant.

\begin{invariant}
\label{label:general:inv:laminar:clean}
For every layer $j \in [1, L_d]$ and every node $v \in C_{lj}$, we have:
$$(1+\gamma/L_d)^{-1} \cdot \frac{\text{deg}(v, H_{j-1})}{2} \leq \text{deg}(v, H_j) \leq (1+\gamma/L_d) \cdot \frac{\text{deg}(v, H_{j-1})}{2}.$$
\end{invariant}

At this point, the reader might object that the above invariant only specifies the approximate degree-splitting condition at layer $j$ for the nodes in $C_{lj}$, whereas our declared  goal was to enforce this condition at all the layers in the interval $[0, j]$. Fortunately for us, Invariant~\ref{label:general:inv:laminar:contain} comes to our rescue, by requiring that a $l$-clean node at layer $j$ must also $l$-clean at every layer $j' \leq j$. Hence, Invariants~\ref{label:general:inv:laminar:contain} and~\ref{label:general:inv:laminar:clean} together imply the desired guarantee, namely, that the degree of a node $v \in C_{lj}$ is split approximately by half across successive layers in the entire range $[0, j]$.  The next lemma states this simple observation in a formal language.

\begin{lemma}
\label{label:general:lm:clean}
A laminar structure  has for each layer $j \in [1, L_d]$ and each node $v \in C_{lj}$:
$$\frac{deg(v, E)}{2^j \cdot (1+\gamma/L_d)^j} \leq \text{deg}(v, H_j) \leq (1+\gamma/L_d)^j \cdot \frac{\text{deg}(v, E)}{2^j}.$$
\end{lemma}

\begin{proof}
Consider any layer $j \in [1, L_d]$ and any node $v \in C_{lj}$. Since $C_{lj} \subseteq C_{lk}$ for all $k \in [1, j]$ (see Invariant~\ref{label:general:inv:laminar:contain}), we have $v \in C_{lk}$ for every layer $k \in [1, j]$. Hence, by Invariant~\ref{label:general:inv:laminar:clean} we have:
\begin{equation}
\label{label:general:eq:lm:clean}
\frac{deg(v, H_{k-1})}{2 \cdot (1+\gamma/L_d)} \leq \text{deg}(v, H_{k}) \leq (1+\gamma/L_d) \cdot \frac{\text{deg}(v, H_{k-1})}{2} \ \ \text{ for all } k \in [1,j].
\end{equation}
From equation~\ref{label:general:eq:lm:clean} we infer that:
\begin{equation}
\label{label:general:eq:lm:clean:100}
\frac{deg(v, H_0)}{2^j \cdot (1+\gamma/L_d)^j} \leq \text{deg}(v, H_j) \leq (1+\gamma/L_d)^j \cdot \frac{\text{deg}(v, H_0)}{2^j}.
\end{equation}
Since $C_{lj} \subseteq A$ (see Invariant~\ref{label:general:inv:laminar:contain}) and $v \in C_{lj}$, we have $v \in A$. Hence, Observation~\ref{label:general:ob:edge} and Invariant~\ref{label:general:inv:laminar:edge} imply that $\text{deg}(v, H_{0}) = \text{deg}(v, H) = \text{deg}(v, E)$. The  lemma now follows from  equation~\ref{label:general:eq:lm:clean:100}. 
\end{proof}

We now bound the degree of a $l$-clean node in the  last layer $L_d$ in terms of its  degree  in the graph $G$.

\begin{corollary}
\label{label:general:cor:lm:clean}
For each node $v \in C_{l,L_d}$ in a laminar structure, we have:
$$e^{-\gamma} \cdot (\lambda_d L^4/d) \cdot deg(v, E) \leq \text{deg}(v, H_{L_d}) \leq e^{\gamma} \cdot (\lambda_d L^4/d) \cdot deg(v, E).$$
\end{corollary}

\begin{proof}
Fix any node $v \in C_{l,L_d}$. Setting $j = L_d$ in Lemma~\ref{label:general:lm:clean}, we get:
$$\frac{deg(v, E)}{2^j \cdot (1+\gamma/L_d)^{L_d}} \leq \text{deg}(v, H_j) \leq (1+\gamma/L_d)^{L_d} \cdot \frac{\text{deg}(v, E)}{2^j}.$$
The corollary now follows from equation~\ref{label:general:eq:Li:lambda} and the fact that $(1+\gamma/L_d)^{L_d} \leq e^{\gamma}$. 
\end{proof}

We want to ensure that the edge-set $H_{L_d}$ corresponds to the edge-set $X$ of the skeleton in Definition~\ref{label:general:def:skeleton}. Comparing the guarantee stated in Corollary~\ref{label:general:cor:lm:clean} with the condition (5) in Definition~\ref{label:general:def:skeleton}, it  becomes obvious that each $l$-clean node in the last layer $L_d$ will  belong to the set of big, non-spurious nodes $B \setminus S$ in a skeleton. 
Furthermore, pointing towards the similarities between the conditions (1), (2) in Definition~\ref{label:general:def:critical:structure} and the conditions (2), (3) in Definition~\ref{label:general:def:skeleton}, we hope to convince the reader that every active, $c$-clean node will belong to the set $B$, and that  every  passive, $c$-clean node will belong to  the set $T$.  This, however, is not the end of the story, for the condition (1) in Definition~\ref{label:general:def:skeleton} requires that the sets $B$ and $T$ actually partition the set of all nodes $V$, whereas we have not yet assigned the  $c$-dirty nodes  to either $B$ or $T$. There is an easy fix for this. We will simply assign all the $c$-dirty nodes with $\text{deg}(v, E) > \epsilon d/L^2$ to  $B$, and the remaining $c$-dirty nodes to  $T$. This will take care of the conditions (1), (2), (3) in Definition~\ref{label:general:def:skeleton}.

By Invariant~\ref{label:general:inv:laminar:contain}, we have $A \setminus D_c = A \setminus D_{l0} = C_{l0}$. Hence, at this point the reader might raise the following objection: So far we have argued that the nodes in $A \setminus D_c = C_{l0}$ will belong to the set $B$ and that the nodes in $C_{l,L_d}$ will belong to the set $B \setminus S$. Intuitively, this alludes to the fact  that the nodes in $C_{l0} \setminus C_{l, L_d} = D_{l, L_d} \setminus D_{l0}$ will belong to the set $S$ of spurious nodes. As we will see in Section~\ref{label:general:sec:maintain:skeleton}, this is indeed going to be the case. But the condition (4) in Definition~\ref{label:general:def:skeleton} places an upper bound on the maximum number of spurious nodes we can have.\footnote{The reader might point out that the condition (3) in Definition~\ref{label:general:def:critical:structure} places an upper bound on the number of $c$-dirty nodes. However, it is easy to check that the size of the set $D_{l,L_d}$ can potentially be much larger than that of $D_c$.} Till now, we have not stated any such analogous upper bound on the maximum number of nodes in $D_{l, L_d}$. To alleviate this concern, we introduce our fourth and final invariant.

\begin{invariant}
\label{label:general:inv:laminar}
We have: $$ \left| D_{lj} \right| \leq \frac{(j+1)\delta}{(L_d+1)} \cdot |A| \ \ \text{ for every layer } j \in [0, L_d].$$
\end{invariant}

Note that the invariant places an upper bound on the number of $l$-dirty nodes {\em at every layer $j \in [0,L_d]$}. At first glance, this might seem to be too stringent a requirement, for all that we initially asked for was an upper bound on the number of $l$-dirty nodes in layer $L_d$. However, as we shall see in later sections, this invariant will  be very helpful in bounding the amortized update time of our algorithm for maintaining critical and  laminar structures. 

\medskip
The remaining two conditions (6) and (7) in Definition~\ref{label:general:def:skeleton} {\em cannot} be derived from the Definitions~\ref{label:general:def:critical:structure},~\ref{label:general:def:laminar:structure} and the four invariants stated above. Instead, they will follow from some specific properties of our algorithm for  maintaining  the critical and laminar structures. Ignoring some low level details, these properties are: 
\begin{enumerate}
\item Most of the time during the course of our algorithm, the layers  in the range $[1, L_d]$ keep losing edges (i.e., no edge is inserted into $H_j$ for $j \geq 1$) and hence the degrees of the nodes in these layers keep decreasing. We emphasize that this property is not true for $j = 0$. 
\item Consider a time-instant $t$ where we see an exception to the rule (1), i.e., an edge is inserted into some layer $j \in [1, L_d]$ by our algorithm. Then there exists a layer $j' \in [1, j-1]$ with the following property: Just after the time-instant $t$, the degree of every node drops roughly by a factor of $2$ across successive layers in the interval $[j', L_d]$. 
\item Using (1) and (2), we  upper bound  the maximum degree a node can have in any layer $j \in [1, L_d]$.  As a corollary, we can show that $\text{deg}(v, H_{L_d}) \leq \lambda_d L^4 + 2$ for all nodes $v \in V$, and that $\text{deg}(v, H_{L_d}) \leq 3 \epsilon \lambda_{d} L^2 +2$ for all nodes $v \in P \setminus D_c$. As we will see later on, these two bounds will respectively correspond to the conditions (6) and (7) in Definition~\ref{label:general:def:skeleton}. 
\end{enumerate}

\noindent To conclude this section, we note that these crucial   properties of our algorithm are formally derived and stated in Section~\ref{label:general:sec:algo:properties} (see Lemmas~\ref{label:general:lm:monotone:laminar},~\ref{label:general:lm:correct:1} and Corollary~\ref{label:general:cor:lemma:correct:1}).



\subsection{Two basic subroutines}
\label{label:general:sec:building:block}

In this section, we describe two subroutines that will be heavily used by  our algorithm in Section~\ref{label:general:sec:main:algo}.

\paragraph{SPLIT$(\mathcal{E})$.}
The first  subroutine is called SPLIT($\mathcal{E}$) (see Figure~\ref{label:general:fig:split}). This  takes as input an edge-set $\mathcal{E}$ defined over the nodes in $V$. Roughly speaking, its goal is to reduce the degree of every node  by a factor of half, and it succeeds in its goal for all the nodes with degree at least three.  To be more precise, the output of the subroutine is a subset of edges $\mathcal{E}' \subseteq \mathcal{E}$ that satisfies the two properties stated in Figure~\ref{label:general:fig:split}. The subroutine runs in $O(|\mathcal{E}|)$ time, and can be implemented as follows: 
\begin{itemize}
\item Create a new node $v^*$. For every  node $v \in V$ with odd degree $\text{deg}(v, \mathcal{E})$,  create a new edge $(v^*, v)$. Let $\mathcal{E}_{new}$ denote the set of newly created edges, and let $V^* = V \cup \{v^*\}$ and $\mathcal{E}^* = \mathcal{E} \cup \mathcal{E}_{new}$. It is easy to check that every node in the graph $G^* = (V, \mathcal{E}^*)$ has an even degree. Hence, in $O(|\mathcal{E}^*|)$ time, we can compute an Euler tour in $G^*$. We construct the edge-set $\mathcal{E}'$ by first selecting the alternate edges in this Euler tour, and then discarding those selected edges that are part of $\mathcal{E}_{new}$. The subroutine returns the edge-set $\mathcal{E}'$ and runs in  $O(|\mathcal{E}^*|) = O(|\mathcal{E}|)$ time. See the paper~\cite{ipl} for  details.
\end{itemize}

\begin{figure}[htbp]
\centerline{\framebox{
\begin{minipage}{5.5in}
\begin{tabbing}
1. \ \ \ \ \= The input is a set of edges $\mathcal{E}$ defined over the nodes in $V$. \\ \\ 
2. \> The output is a subset  $\mathcal{E}' \subseteq \mathcal{E}$, with the following property: \\ 
\> \ \ \ \ \ \= For every node $v \in \mathcal{V}$, we have  $\frac{\text{deg}(v, \mathcal{E})}{2} -1 \leq \text{deg}(v, \mathcal{E}') \leq \frac{\text{deg}(v, \mathcal{E})}{2} + 1$.\\  \\
3. \> The subroutine runs in $O(|\mathcal{E}|)$ time.
\end{tabbing}
\end{minipage}
}}
\caption{\label{label:general:fig:split} SPLIT$(\mathcal{E})$.}
\end{figure}

\paragraph{REBUILD$(j)$.}
The second subroutine is  called  REBUILD($j$), where $j \in [1, L_d]$. We will ensure that during the course of our algorithm, the conditions stated in Figure~\ref{label:general:fig:rebuild:init} are satisfied just before every call to REBUILD($j$). As the name suggests, the subroutine REBUILD($j$) will then rebuild the layers $k \in [j, L_d]$ from scratch. See Figure~\ref{label:general:fig:rebuild} for details. We will show that at the end of the call to REBUILD($j$), the Invariants~\ref{label:general:inv:laminar:edge},~\ref{label:general:inv:laminar:contain},~\ref{label:general:inv:laminar:clean} will continue to hold, and Invariant~\ref{label:general:inv:laminar} will hold for all layers $j' \in [0, L_d]$. We will also derive some nice properties of this subroutine that will be useful later on.

\begin{figure}[htbp]
\centerline{\framebox{
\begin{minipage}{5.5in}
\begin{tabbing}
1. \ \ \=  The tuple $(A, P, D_c, H)$ is a critical structure as per Definition~\ref{label:general:def:critical:structure}. \\ \\
2. \> The corresponding laminar structure (see Definition~\ref{label:general:def:laminar:structure}) satisfies Invariants~\ref{label:general:inv:laminar:edge},~\ref{label:general:inv:laminar:contain} and~\ref{label:general:inv:laminar:clean}. \\ \\
3. \>  Invariant~\ref{label:general:inv:laminar} holds at every layer $j' \in [0,j-1]$. \\ \\
4. \> Invariant~\ref{label:general:inv:laminar} is violated at layer $j$.
\end{tabbing}
\end{minipage}
}}
\caption{\label{label:general:fig:rebuild:init} Initial conditions just before a call to REBUILD$(j)$, $j \in [1, L_d]$.}
\end{figure}

\begin{figure}[htbp]
\centerline{\framebox{
\begin{minipage}{5.5in}
\begin{tabbing}
1. \ \ \ \= {\sc For} $k = j$ to $L_d$ \\
2. \> \ \ \ \ \ \= $H_{k} \leftarrow \text{SPLIT}(H_{k-1})$. \\
3. \> \> $D_{l,k} \leftarrow D_{l, k-1}$. \\
4. \> \> $C_{l,k} \leftarrow C_{l,k-1}$.
\end{tabbing}
\end{minipage}
}}
\caption{\label{label:general:fig:rebuild} REBUILD$(j)$, $j \in [1, L_d]$. Just before a call to this subroutine, the conditions in Figure~\ref{label:general:fig:rebuild:init} hold.}
\end{figure}

The subroutine SPLIT($\mathcal{E}$) outputs a subset of edges $\mathcal{E}' \subseteq \mathcal{E}$ where the degree of each node is  half times its original degree, plus-minus one (see Figure~\ref{label:general:fig:split}). Since the subroutine REBUILD($j$) iteratively sets $H_{k} \leftarrow \text{SPLIT}(H_{k-1})$ for $k = j$ to $L_d$ (see Figure~\ref{label:general:fig:rebuild}), we can get a nearly tight bound on the degree of a node in a layer $k \in [j, L_d]$ when the subroutine  REBUILD($j$) finishes execution.

\begin{lemma}
\label{label:general:lm:fig:reduce}
Fix any $j \in [1, L_d]$. At the end of a call to the subroutine REBUILD($j$), we have:
$$\frac{\text{deg}(u,H_{k-1})}{2} - 1 \leq \text{deg}(u, H_{k}) \leq \frac{\text{deg}(u, H_{k-1})}{2} + 1 \ \ \text{ for all  } u \in V, k \in [j, L_d].$$ 
\end{lemma}

The next lemma follows directly from the descriptions in Figures~\ref{label:general:fig:rebuild:init} and~\ref{label:general:fig:rebuild}.
\begin{lemma}
\label{label:general:lm:rebuild:minor:1}
\label{label:general:lm:rebuild:minor:2}
 Consider any  call to the subroutine REBUILD($j$) during the course of our algorithm.
 \begin{itemize}
 \item The call does not alter the sets $A, P, D_c$ and $H$. By Figure~\ref{label:general:fig:rebuild:init}, the tuple $(A, P, D_c, H)$ is a critical structure (see Definition~\ref{label:general:def:critical:structure}) just before the call, and it continues to remain so at the end of the call. 
\item The call does not alter the layers $k < j$, i.e., the sets $\{D_{lk}, C_{lk}, H_k\}, k \in [0,j-1],$ do not change. 
\item Finally, at the end of the call we have:  $D_{lk} = D_{l,j-1}$ and $C_{lk} = C_{l,j-1}$ for all layers $k \in [j,L_d]$.
\end{itemize}
\end{lemma}

\medskip
The proof of the next lemma appears in Section~\ref{label:general:sec:property:rebuild}.

\begin{lemma}
\label{label:general:lm:rebuild:main}
At the end of a call to  REBUILD($j$), Invariants~\ref{label:general:inv:laminar:edge},~\ref{label:general:inv:laminar:contain},~\ref{label:general:inv:laminar:clean} and~\ref{label:general:inv:laminar} are satisfied.
\end{lemma}

\subsubsection{Proof of Lemma~\ref{label:general:lm:rebuild:main}}
\label{label:general:sec:property:rebuild}

\newcommand{\ts}{t_{\text{start}}}
\newcommand{\te}{t_{\text{end}}}

Throughout this section, we fix a layer $j \in [1, L_d]$ and a given call to the subroutine REBUILD($j$). We  let $t_{\text{start}}$  denote the point in time just before the call to REBUILD($j$), whereas we let $t_{\text{end}}$ denote the point in time just after the call to REBUILD($j$). We consider all the four invariants one after another.

\bigskip
\paragraph{Proof for Invariant~\ref{label:general:inv:laminar:edge}.} 
\ 

\medskip
\noindent By Figure~\ref{label:general:fig:rebuild:init}, we have $H = H_0 \supseteq H_1 \supseteq \cdots \supseteq H_{j-1}$ at time $\ts$. The call to REBUILD($j$) does not change any of the sets $H, H_0, \ldots, H_{j-1}$, and it iteratively sets $H_{k} \leftarrow \text{SPLIT}(H_{k-1})$ for $k = j$ to $L_d$. Hence, by Figure~\ref{label:general:fig:split}, we have $H_{j-1} \supseteq H_j \supseteq \cdots \supseteq H_{L_d}$ at time $\te$. Putting all these observations together, we have $H = H_0 \supseteq H_1 \supseteq \cdots \supseteq H_{L_d}$ at time $\te$. This shows that Invariant~\ref{label:general:inv:laminar:edge} is satisfied at time $\te$. 

\bigskip
\paragraph{Proof for Invariant~\ref{label:general:inv:laminar:contain}.}
\ 

\medskip
\noindent
By Figure~\ref{label:general:fig:rebuild:init}, at time $\ts$ we have:
\begin{itemize}
\item $D_{l0} = D_c \cap A$. 
\item $D_{lk} \cap C_{lk} = \emptyset$ and $D_{lk} \cup C_{lk} = A$ for all layers $k \in [1,j-1]$.
\item $D_{l0} \subseteq D_{l1} \subseteq \cdots \subseteq D_{l,j-1}$ and $A \supseteq C_{l0} \supseteq C_{l1} \supseteq \cdots \supseteq C_{l,j-1}$. 
\end{itemize}
The call to REBUILD($j$) does not change the sets $D_c$ and $A$.  The sets $\{D_{lk}, C_{lk}\}, k < j,$ are also left untouched. Finally, it is guaranteed that at time $\te$ we have:
\begin{itemize}
\item $D_{lk} = D_{l,j-1}$ and $C_{lk} = C_{l,k-1}$ for all layers $k \in [j, L_d]$.
\end{itemize}
From all these observations, we infer that Invariant~\ref{label:general:inv:laminar:contain} holds at time $\te$.

\bigskip
\paragraph{Proof for Invariant~\ref{label:general:inv:laminar}.}
\

\medskip
\noindent
By Figure~\ref{label:general:fig:rebuild:init} (item 3), at time $\ts$ we have:
$$|D_{lk}| \leq \frac{\delta (k+1)}{(L_d+1)} \cdot |A| \ \ \text{ at every layer } k \in [0, j-1].$$ 
The call to REBUILD($j$) does not change the set $A$. Neither does it change the sets $\{D_{lk}\}, k < j$. It is also guaranteed that at time $\te$ we have $D_{lk} = D_{l,j-1}$ for all layers $k \in [j, L_d]$. Thus, at time $\te$ we have:
$$|D_{lk}| = |D_{l,j-1}| \leq \frac{\delta j}{(L_d+1)} \cdot |A| < \frac{\delta (k+1)}{(L_d+1)} \cdot |A| \ \ \text{ at every layer } k \in [j, L_d].$$
From all these observations, we infer that Invariant~\ref{label:general:inv:laminar} holds at time $\te$.

 \bigskip
 \paragraph{Proof for Invariant~\ref{label:general:inv:laminar:clean}.}
 \
 
 \medskip
\noindent  Invariant~\ref{label:general:inv:laminar:clean} holds at time $\ts$, and a call to REBUILD($j$) does not alter the layers $k  < j$. Thus, we have:
\begin{equation}
\label{label:general:eq:clean:proof:1}
\text{Invariant~\ref{label:general:inv:laminar:clean} holds at time } \te \text{ for every layer } k \in [1, j-1] \text{ and node } u \in C_{l,k}.
\end{equation}
Accordingly, we only need to focus on the layers $k \in [j, L_d]$. Since the set $C_{l,j-1}$ does not change during the call to REBUILD($j$), we will refer to $C_{l,j-1}$ without any ambiguity. Further, the call to REBUILD($j$) ensures that  $C_{l,k} = C_{l,j-1}$ for all $k \in [j, L_d]$ at time $\te$. Thus,  in order to prove that Invariant~\ref{label:general:inv:laminar:clean} holds at time $\te$ in the remaining layers $k \in [j, L_d]$, it suffices to show Claim~\ref{label:general:cl:tend:main}.
\begin{claim}
\label{label:general:cl:tend:main}
Consider any node $v \in C_{l,j-1}$. At time $\te$, we have:
$$(1+\gamma/L)^{-1} \cdot \frac{\text{deg}(v, H_{k-1})}{2} \leq \text{deg}(v, H_k) \leq (1+\gamma/L) \cdot \frac{\text{deg}(v, H_{k-1})}{2} \ \ \text{ at every layer } k \in [j, L_d].$$
\end{claim}

Throughout the rest of this section, we fix a node $v \in C_{l,j-1}$, and focus on proving Claim~\ref{label:general:cl:tend:main} for the node $v$. The main idea is very simple: The call to REBUILD($j$) ensures that upon its return,  $\text{deg}(v, H_k)$ is very close to $(1/2) \cdot \text{deg}(v, H_{k-1})$ for all layers $k \in [j, L_d]$ (see Lemma~\ref{label:general:lm:fig:reduce}). Barring some technical details, this is sufficient to show that $\text{deg}(v, H_k)$ is within the prescribed range at each layer $k \in [j, L_d]$.

\begin{claim}
\label{label:general:cl:tend}
At time $\te$, we have:
$$\frac{\text{deg}(v, H_{k-1})}{2} - 1 \leq \text{deg}(v, H_k) \leq \frac{\text{deg}(v, H_{k-1})}{2} + 1 \ \ \text{ at every layer } k \in [j, L_d].$$
\end{claim}

\begin{proof}
Follows from Lemma~\ref{label:general:lm:fig:reduce}.
\end{proof}

\begin{claim}
\label{label:general:cl:tend:1}
At time $\te$, we have $\text{deg}(v, H_k) \geq L$ at every layer $k \in [j-1, L_d]$.
\end{claim}

\begin{proof}
\end{proof}

\begin{claim}
\label{label:general:cl:quick:1}
Consider any number $x \geq L$. We have:
\begin{itemize}
\item  $x/2 - 1 \geq (x/2) \cdot (1+\gamma/L_d)^{-1}$, and
\item $x/2 + 1 \leq (x/2) \cdot (1+\gamma/L_d)$.
\end{itemize}
\end{claim}

\begin{proof}
To prove the first part of the claim, we infer that:
\begin{eqnarray}
(x/2) - (x/2) \cdot (1+\gamma/L_d)^{-1} & = & \left(\frac{x}{2}\right) \cdot \left(\frac{\gamma/L_d}{1+\gamma/L_d}\right) \nonumber \\
& \geq & \left(\frac{L \gamma}{2 L_d}\right) \cdot \left(\frac{1}{1+\gamma/L_d}\right) \nonumber \label{label:general:eq:quick:1}  \\
& \geq & \frac{L \gamma}{4L_d} \label{label:general:eq:quick:2} \\
& \geq & 1 \label{label:general:eq:quick:3}
\end{eqnarray}
Equation~\ref{label:general:eq:quick:2} follows from equation~\ref{label:general:eq:101}. Equation~\ref{label:general:eq:quick:3} follows from equation~\ref{label:general:eq:102}.

\medskip
To prove the second part of the claim, we infer that:
\begin{eqnarray}
(x/2) \cdot (1+\gamma/L_d) - (x/2)  & = & (x/2) \cdot (\gamma/L_d) \nonumber \\
& \geq & \frac{L\gamma}{2L_d} \nonumber \label{label:general:eq:quick:4}  \\
& \geq & 1 \label{label:general:eq:quick:5} 
\end{eqnarray}
Equation~\ref{label:general:eq:quick:5} follows from equation~\ref{label:general:eq:102}.
\end{proof}

\begin{claim}
\label{label:general:cl:tend:2}
At time $\te$, we have
$$(1+\gamma/L)^{-1} \cdot \frac{\text{deg}(v, H_{k-1})}{2} \leq \text{deg}(v, H_k) \leq (1+\gamma/L) \cdot \frac{\text{deg}(v, H_{k-1})}{2} \ \ \text{ at every layer } k \in [j, L_d].$$
\end{claim}

\begin{proof}
Consider any layer $k \in [j, L_d]$. Let $x$ be the value of $\text{deg}(v, H_{k-1})$ at time $\te$. Hence, at time $\te$, the value of $\text{deg}(v, H_k)$ lies in the range $[x/2 - 1, x/2+1]$ (see Claim~\ref{label:general:cl:tend}). Since $x \geq L$ (see Claim~\ref{label:general:cl:tend:1}), we infer that the range $[x/2-1, x/2+1]$ is completely contained within the range $\left[ (1+\gamma/L)^{-1} \cdot (x/2), (1+\gamma/L) \cdot (x/2)\right]$ (see Claim~\ref{label:general:cl:quick:1}). Thus, at time $\te$, the value of $\text{deg}(v, H_k)$ also falls within the range $\left[ (1+\gamma/L)^{-1} \cdot (x/2), (1+\gamma/L) \cdot (x/2)\right]$. This concludes the proof of the claim.
\end{proof}

Claim~\ref{label:general:cl:tend:main} follows from Claim~\ref{label:general:cl:tend:2}.

\subsection{Our algorithm for maintaining  critical and  laminar structures}
\label{label:general:sec:main:algo}
\label{label:general:sec:maintain:critical}
\label{label:general:sec:maintain:laminar}

We use the term ``edge-update'' to refer to the insertion/deletion of an edge in the graph $G = (V, E)$. Thus, in a dynamic setting the graph $G = (V, E)$ changes due to a sequence of edge-updates. In this section, we  present an algorithm that maintains a critical structure  and a laminar structure in such a dynamic setting, and ensures that all the invariants from Section~\ref{label:general:sec:maintain:prelim} are satisfied.  Before delving into technical details, we first present a high level overview of our approach.

\paragraph{A brief outline of our algorithm.}
 Our algorithm works in ``phases'', where the term ``phase'' refers to a contiguous block of edge-updates.  In the beginning of a phase, there are no dirty nodes, i.e., we have $D_c = \emptyset$ and $D_{lj} = \emptyset$ for all $j \in [0, L_d]$, and all the four invariants from Section~\ref{label:general:sec:maintain:prelim} are satisfied.

\medskip
\noindent {\em The algorithm in the middle of a phase.} We next describe how to modify the critical structure after  an edge-update in the middle of a phase. Towards this end, consider an edge-update that corresponds to the insertion/deletion of the edge $(u,v)$ in $G = (V, E)$. This edge-update changes the degrees of the endpoints $u, v$, and this  might lead to some node $x \in \{u,v\}$ violating   one of the first two conditions in Definition~\ref{label:general:def:critical:structure}. But this can happen only if the node $x$ was $c$-clean before the edge-update. Hence,   the easiest way to fix the problem, if there is one, is to change the status of the node to $c$-dirty. And  we do exactly the same. Next, we  ensure that $H$ remains the set of edges incident upon the active nodes. Towards this end, we check if either of the endpoints $u, v$ is active. If the answer is yes, then the edge $(u,v)$ is inserted into (resp. deleted from) $H$ along with the insertion (resp. deletion) of $(u,v)$ in $G = (V, E)$.

From the above discussion, it is obvious that the sets of active and passive nodes do not change in the middle of a phase. And we ensure that $H$ remains the set of edges incident upon the active nodes. In contrast, as the phase goes on, we  see more and more  $c$-clean nodes becoming $c$-dirty. The set of $c$-dirty nodes, accordingly, keeps getting larger along with the passage of time. The phase terminates when the size of the set $D_c$ exceeds the threshold $(\delta/(L_d+1)) \cdot |A|$, thereby violating the condition (3) in Definition~\ref{label:general:def:critical:structure}. 

We next show how to maintain the laminar structure  in the middle of a phase.  
\begin{enumerate}
\item Whenever an edge $(u,v)$ is inserted into $H$, we set $H_0 \leftarrow H_0 \cup \{(u,v)\}$, and whenever an edge $(u,v)$ is deleted from the graph, we set $H_j \leftarrow H_j \setminus \{(u,v)\}$ for all layers $j \in [0, L_d]$. This ensures that we always have $H = H_0 \supseteq H_1 \supseteq \cdots \supseteq H_{L_d}$. Furthermore, whenever an active node $v$ becomes $c$-dirty, we set $D_{lj} \leftarrow D_{lj} \cup \{v\}$ for all layers $j \in [0, L_d]$. This ensures that Invariant~\ref{label:general:inv:laminar:contain} is satisfied.
\item Whenever we see a node $v \in C_{l,j}$ violating Invariant~\ref{label:general:inv:laminar:clean} at layer $j \in [1, L_d]$, we set $D_{l,k} \leftarrow D_{l,k} \cup \{v\}$ and $C_{l,k} \leftarrow C_{l,k} \setminus \{v\}$ for all layers $k \in [j, L_d]$. This restores the validity of Invariant~\ref{label:general:inv:laminar:clean} without tampering with Invariant~\ref{label:general:inv:laminar:contain}.
\item Whenever Invariant~\ref{label:general:inv:laminar} gets violated, we find the smallest index $j$ at which $|D_{lj}| > (\delta (j+1)/(L_d+1)) \cdot |A|$, and call the subroutine REBUILD($j$) (see Figure~\ref{label:general:fig:rebuild}).
\end{enumerate}
\noindent Fix any layer $j \in [1, L_d]$. A call to REBUILD($j'$) with $j' > j$ does not affect the layers in the range $[0, j]$. Consequently,  the set $D_{lj}$ (resp. $C_{lj}$) keeps getting bigger (resp. smaller) till the point arrives where  REBUILD($j'$) is called for some $j' \leq j$.\footnote{For $j = 0$, we have $D_{l0} = D_c \cap A$. Since the set $A$ remains unchanged and the set $D_c$ keeps getting bigger along with the passage of time in the middle of a phase, we conclude that the same thing happens with the set $D_{l0}$. } To be more specific, in the middle of a phase, a node $v$ can switch from $D_{lj}$  to $C_{lj}$ only if REBUILD($j'$) is called for some $j' \leq j$.  In a similar vein, an edge  gets inserted into $H_j$ only if REBUILD($j'$) is called for some $j' \leq j$; and at other times in the middle of a phase the set  $H_j$ can only keep shrinking (see item (1) above). To summarize, the reader should note that the algorithm satisfies some nice monotonicity properties. These will be very helpful in our analysis in later sections.

\medskip
\noindent {\em Dealing with the termination of a phase.}
When a phase terminates, we need to do some cleanup work before starting the next phase. Specifically, recall that a phase terminates when the number of $c$-dirty nodes goes beyond the acceptable threshold. At this stage,  we  shift  some active $c$-dirty nodes $v \in A \cap D_c$ from the set $A$ to the set $P$, shift some passive $c$-dirty nodes $v \in P \cap D_c$ from the set $P$ to the set $A$, and finally set $D_c = \emptyset$. At the end of these operations, we have $D_c = \emptyset$, and we perform these operations in such a way that Definition~\ref{label:general:def:critical:structure} is satisfied. Next, we   construct the entire laminar structure from scratch by calling the subroutine REBUILD($1$) (see Figure~\ref{label:general:fig:rebuild}). Subsequently, we start the next phase. 

\paragraph{Roadmap.} We will now present the algorithm in details. The rest of this section is organized as follows.
\begin{itemize}
\item In Section~\ref{label:general:sec:phase:begin}, we state the initial conditions that hold in the beginning of a phase.
\item In Section~\ref{label:general:sec:phase:middle}, we describe our algorithm in the middle of a phase.
\item In Section~\ref{label:general:sec:phase:end}, we describe the cleanup that needs to be performed at the end of a phase.
\item In Section~\ref{label:general:sec:algo:properties}, we derive some useful properties of our algorithm.
\end{itemize}

\subsubsection{Initial conditions in the beginning of a phase}
\label{label:general:sec:phase:begin}

Just before  the first edge insertion/deletion of a phase, the following conditions are satisfied.
\begin{enumerate}
\item There are no $c$-dirty nodes, i.e., $D_c = \emptyset$. 
\item There are no $l$-dirty nodes at any layer $j \in [0, L_d]$, i.e., $D_{lj} = \emptyset$ for all $j \in [0, L_d]$.
\end{enumerate}

\noindent In the beginning of the very first phase, the graph $G = (V, E)$ has an empty edge-set. At that instant, every node is passive and the conditions (1) and (2) are satisfied.

\medskip
By induction hypothesis, suppose that the conditions (1) and (2) are satisfied when  the $k^{th}$ phase is about to begin, for some integer $k \geq 1$. We will process the edge insertions/deletions in $G$ during the $k^{th}$ phase using the algorithm described  in Sections~\ref{label:general:sec:phase:middle} and~\ref{label:general:sec:phase:end}. This algorithm will ensure that the conditions (1) and (2) are satisfied at the start of the $(k+1)^{th}$ phase.

\subsubsection{Handling edge insertion/deletions in the middle of a phase}
\label{label:general:sec:phase:middle}
Suppose that an edge $(u,v)$ is inserted into (resp. deleted from) the graph $G = (V, E)$ in the middle of a phase. In this section, we will show how to handle this edge-update. 

\bigskip
\paragraph{Step I: Updating the critical structure.}
\
\medskip
\noindent We update the critical structure $(A, P, D_c, H)$ as follows.
\begin{itemize}
\item If the edge $(u,v)$ has at least one active endpoint, i.e., if $\{u,v \} \cap A \neq \emptyset$, then:
\begin{itemize}
\item If we are dealing with the insertion of the edge $(u,v)$ into $G$, then  set $H \leftarrow H \cup \{(u,v)\}$. 

Else if we are dealing with the deletion of the edge $(u,v)$ from $G$, then  set $H \leftarrow H \setminus \{(u,v)\}$.  
\end{itemize} 
This  ensures that the condition (4) in Definition~\ref{label:general:def:critical:structure} remain satisfied. 
\end{itemize}
Next, note that the edge-update  changes the degrees of the endpoints $u, v$. Hence,  to satisfy  the conditions (1) and (2) in Definition~\ref{label:general:def:critical:structure}, we  perform the following operations on each node $x \in \{u, v\}$. 
\begin{itemize}
\item If $x \in P \setminus D_c$ and $\text{deg}(x, E) \geq 3\epsilon d/L^2$, then set $D_c \leftarrow D_c \cup \{x\}$. 

Else if $x \in A \setminus D_c$ and $\text{deg}(x, E) \leq \epsilon d/L^2$, then set $D_c \leftarrow D_c \cup \{x\}$. 
\end{itemize}
At this point, the conditions (1), (2), (4) and (5) in Definition~\ref{label:general:def:critical:structure} are satisfied. But, we cannot yet be sure about the remaining condition (3). The next step in our algorithm will resolve this issue.

\bigskip
\paragraph{Step II: Deciding if we have to terminate the phase.}
\

\medskip
\noindent
Step I of the algorithm  might have made changed the status of one or both the endpoints $u, v$ from $c$-clean to $c$-dirty. If this is the case, then this increases the number of $c$-dirty nodes. To find out if this violates the condition (3) in Definition~\ref{label:general:def:critical:structure}, we check if $|D_c| > (\delta/(L_d+1)) \cdot |A|$.  
\begin{itemize}
\item If $|D_c| > (\delta/(d+1)) \cdot |A|$, then the number of $c$-dirty nodes have increased beyond the acceptable threshold, and so we terminate the current phase and move on to Section~\ref{label:general:sec:phase:end}. Else if $|D_c| \leq (\delta/(L_d+1)) \cdot |A|$, then all the conditions in Definition~\ref{label:general:def:critical:structure} are satisfied, and we move on to Step III.
\end{itemize}

\bigskip
\paragraph{Step III: Updating the laminar structure.}
If we have reached this stage, then we know for sure that the critical structure $(A, P, D_c, H)$ now satisfies all the conditions in Definition~\ref{label:general:def:critical:structure}. It only remains to update the laminar structure, which is done as follows.
\begin{enumerate}
\item For each $x \in \{u, v\}$:
\begin{itemize}
\item If $x \in A$  and  Step I has converted the node $x$ from $c$-clean to $c$-dirty, then:

 For all $k \in [0, L_d]$, set $D_{lk} \leftarrow D_{lk} \cup \{x\}$ and $C_{lk} \leftarrow  C_{lk} \setminus \{x\}$. 
\end{itemize}
This ensures that $D_{l0}$ remains equal to $D_c \cap A$ and that  Invariant~\ref{label:general:inv:laminar:contain} continues to hold.
\item If Step I inserts the edge $(u,v)$ into $H$, then set $H_0 \leftarrow H_0 \cup \{(u,v)\}$.

Else if Step I deletes the edge $(u,v)$ from $H$, then set $H_j \leftarrow H_j \setminus \{(u,v)\}$ for all layers $j \in [0, L_d]$.

This  ensures that Invariant~\ref{label:general:inv:laminar:edge} is satisfied.  Also note that if the edge-update under consideration is an insertion, this does not affect the edges in the  layers $j \geq 1$.
\item The previous operations might have changed the degree of some  endpoint $x \in \{u, v\}$ in the laminar structure, and hence, we have ensure that no node $x \in \{u,v\}$ violates Invariant~\ref{label:general:inv:laminar:clean}. Accordingly, we call the subroutine CLEANUP($x$) for each node $x \in \{u,v\}$. See Figure~\ref{label:general:fig:ensure:clean}. The purpose of these calls is to restore Invariant~\ref{label:general:inv:laminar:clean} without affecting Invariant~\ref{label:general:inv:laminar:contain}.

Due to the calls to CLEANUP($u$) and CLEANUP($v$), the sets $D_{lj}$ get bigger and the set $C_{lj}$ gets smaller. In other words, a node that is $l$-dirty at some layer $j$ never becomes $l$-clean at layer $j$ due to these calls. This monotonicity property will be very helpful in the analysis of our algorithm. 
\item The calls to CLEANUP($u$) and CLEANUP($v$)  might have increased the number of $l$-dirty nodes at some layers, and hence, Invariant~\ref{label:general:inv:laminar} might get violated at some layer $j \in [0, L_d]$.  

Next, note that since we gone past Step II,  we must have $|D_c| \leq \delta/(L_d+1) \cdot |A|$. Since $D_{l0} = D_c \cap A$ (see item (1) above), we have $|D_{l0}| \leq |D_c| \leq \delta /(L_d+1) \cdot |A|$. In other words, we only need to be concerned about  Invariant~\ref{label:general:inv:laminar} for layers $j > 0$. 

To address this concern, we now call the subroutine VERIFY(). See Figure~\ref{label:general:fig:insert}. 

Lemma~\ref{label:general:lm:insert:laminar} shows that Invariant~\ref{label:general:inv:laminar} holds for all layers $j \in [0, L_d]$ at the end of the call to VERIFY(). At this stage  all the invariants hold, and we are ready to handle the next edge-update in $G$.
\end{enumerate}

\begin{figure}[htbp]
\centerline{\framebox{
\begin{minipage}{5.5in}
\begin{tabbing}
1. \ \ \ \= {\sc For} $j = 1$ to $L_d$: \\
2. \> \ \ \ \ \ \ \= {\sc If} $x \in C_{lj}$ and the node $x$ violates Invariant~\ref{label:general:inv:laminar:clean} at layer $j$, {\sc Then} \\
3. \> \> \ \ \ \ \ \ \ \ \ \ \ \ \ \= Set $D_{lj'} \leftarrow D_{lj'} \cup \{x\}$ and $C_{lj'} \leftarrow C_{lj'} \setminus \{x\}$ for all $j' \in [j, L_d]$. \\
4. \> \> \> RETURN.
\end{tabbing}
\end{minipage}
}}
\caption{\label{label:general:fig:ensure:clean} CLEANUP($x$).}
\end{figure}

\begin{figure}[htbp]
\centerline{\framebox{
\begin{minipage}{5.5in}
\begin{tabbing}
1. \ \ \ \= {\sc For} $j = 1$ to $L_d$: \\
2. \> \ \ \ \ \ \ \ \= {\sc If} $|D_{lj}| > (\delta(j+1)/(L_d+1)) \cdot |A|$, {\sc Then} \\
3. \> \> \ \ \ \ \ \ \ \ \ \= Call the subroutine REBUILD($j$). See Figure~\ref{label:general:fig:rebuild}. \\
4. \> \> \> RETURN. 
\end{tabbing}
\end{minipage}
}}
\caption{\label{label:general:fig:insert} VERIFY().}
\end{figure}

\begin{lemma}
\label{label:general:lm:insert:laminar}
At the end of a call to VERIFY(), Invariant~\ref{label:general:inv:laminar} holds for all layers $j \in [0, L_d]$.
\end{lemma}

\begin{proof}
Since we have gone past Step II,  we must have $|D_c| \leq \delta/(L_d+1) \cdot |A|$. Since $D_{l0} = D_c \cap A$ (see item (1) in Step III), we have $|D_{l0}| \leq |D_c| \leq \delta /(L_d+1) \cdot |A|$. In other word, Invariant~\ref{label:general:inv:laminar} holds for $j = 0$ at the end of the call to VERIFY(). Henceforth, we focus on the layers $j > 0$.

If the subroutine VERIFY() does not make any call to REBUILD($j'$) with $j' \in [1, L_d]$ during its execution, then clearly Invariant~\ref{label:general:inv:laminar} holds for all layers $j \in [1, L_d]$. Otherwise, let $k \in [1, L_d]$ be the index such that the subroutine REBUILD($k$) is called during the execution of VERIFY(). 

Since the subroutine REBUILD($j$) was not called for any $j \in [1, k-1]$, it means that Invariant~\ref{label:general:inv:laminar} was already satisfied for layers $j \in [1, k-1]$. Next, note that the subroutine REVAMP() is terminated after  REBUILD($k$) finishes execution. Hence, Lemma~\ref{label:general:lm:rebuild:minor:1} ensures that $C_{lj} = C_{k-1}$ for all $j \in [k, L_d]$ at the end of the call to REVAMP(). So at that instant we have $|D_{lj}| = |D_{l, k-1}|$ for all $j \in [k, L_d]$. Since Invariant~\ref{label:general:inv:laminar} holds for $j = k-1$, we infer that Invariant~\ref{label:general:inv:laminar} also holds for all $j \in [k, L_d]$ at that instant. 
\end{proof}

\subsubsection{Terminating a phase}
\label{label:general:sec:phase:end}
We terminate the current phase when the number of $c$-dirty nodes becomes larger than $(\delta/(L_d +1))$ times the number of active nodes. To address this concern, we call the subroutine REVAMP() as described below.

\paragraph{REVAMP().} 
\begin{enumerate}
\item We will first  modify the critical structure $(A, P, D_c, H)$. Let $V' \leftarrow D_c$ be the set of $c$-dirty nodes at the time REVAMP()  is called.  We run the {\sc For} loop described below.

 {\sc For all} nodes $v \in V'$:
\begin{itemize}
\item Set $D_c \leftarrow D_c \setminus \{v\}$. So the node $v$ is no longer $c$-dirty, and it might  violate the conditions (1) and (2) in Definition~\ref{label:general:def:critical:structure}. To address this concern, we perform the following operations.
\item If $v \in A$ and $\text{deg}(v, E) < 3 \epsilon d/L^2$, then set $A \leftarrow A \setminus \{v\}$, $P \leftarrow P \cup \{v\}$, and $H \leftarrow H \setminus \{ (u,v) \in E : u \in P \}$. In other words, we change the status of the node from active to passive, which satisfies the conditions (1) and (2) in Definition~\ref{label:general:def:critical:structure}. Next,  to satisfy the condition (4) in Definition~\ref{label:general:def:critical:structure}, we delete from the set $H$ those edges whose one endpoint is $v$ and other endpoint is some passive node. 
\item Else if $v \in P$ and $\text{deg}(v, E) >  \epsilon d/L^2$, then set $A \leftarrow A \cup \{v\}$, $P \leftarrow P \setminus \{v\}$, and $H \leftarrow H \cup \{ (u,v) \in E\}$. In other words, we change the status of the node from active to passive, which satisfies the first and second conditions in Definition~\ref{label:general:def:critical:structure}. Next,  to satisfy the fourth condition in Definition~\ref{label:general:def:critical:structure}, we  ensure that all the edges incident upon $v$ belong to  $H$.
\end{itemize}
At the end of the For loop, there are no $c$-dirty nodes, i.e., $D_c = \emptyset$ and all the conditions in Definition~\ref{label:general:def:critical:structure} are satisfied. Thus, at this stage  $(A, P, D_c, H)$ is a  critical structure with no $c$-dirty nodes. 
\item Next, we update our laminar structure as follows.
\begin{itemize}
\item Set $H_0 \leftarrow H$, and $H_j \leftarrow \emptyset$ for all $j \in [1, L_d]$.
\item Set $D_{l0} \leftarrow \emptyset$ and $D_{lj} \leftarrow A$ for all $j \in [1, L_d]$.
\item Set $C_{l0} \leftarrow A$ and $C_{lj} \leftarrow \emptyset$ for all $j \in [1, L_d]$. 
\end{itemize}
At this stage, all the conditions stated in Figure~\ref{label:general:fig:rebuild:init} are satisfied for $j = 1$. Accordingly, we call the subroutine REBUILD($1$). This constructs the layers $j \in [1, L_d]$ of the laminar structure from scratch.  
\end{enumerate}

\noindent By Lemma~\ref{label:general:lm:rebuild:main}, at the end of the call to REBUILD($1$), all the four invariants from Section~\ref{label:general:sec:maintain:prelim} are satisfied. Furthermore,  Lemma~\ref{label:general:lm:rebuild:minor:1} implies that $D_{lj} = D_{l0} = D_c = \emptyset$ for all $j \in [1, L_d]$. Hence, both the conditions (1) and (2) stated in Section~\ref{label:general:sec:phase:begin} hold at this time.  So we  start a new phase from the next edge-update.

\subsubsection{Some useful properties of our algorithm}
\label{label:general:sec:algo:properties}

In this section we derive some nice properties of our algorithm that will be used in later sections.
In Lemma~\ref{label:general:lm:property:inv}, we  note that our algorithm satisfies  all the four invariants from Section~\ref{label:general:sec:maintain:prelim}. Next, in Lemma~\ref{label:general:lm:init} (resp. Lemma~\ref{label:general:lm:monotone:laminar}), we summarize the way the critical (resp. laminar) structure changes with the passage of time within a given phase.  The proofs of these three lemmas follow directly from the description of our algorithm, and are omitted. 

\begin{lemma}
\label{label:general:lm:property:inv}
Suppose that we maintain a critical structure  and a laminar structure as per our algorithm. Then Invariants~\ref{label:general:inv:laminar:edge},~\ref{label:general:inv:laminar:contain},~\ref{label:general:inv:laminar:clean} and~\ref{label:general:inv:laminar} are satisfied after each edge-update in the graph $G = (V, E)$. 
\end{lemma}

\begin{lemma}
\label{label:general:lm:init}
\label{label:general:lm:monotone:dirty}
\label{label:general:lm:middle}
\label{label:general:lm:end}
Consider the maintenance of the critical structure in any given phase as per our algorithm. 
\begin{itemize}
\item In the beginning of the phase, there are no $c$-dirty nodes (i.e., $D_c = \emptyset$).
\item In the middle of a phase, a node can change from being $c$-clean to $c$-dirty, but not the other way round.
\item The sets of active and passive nodes do not change in the middle of the phase.
\item The phase ends when the number of $c$-dirty nodes  exceeds the threshold $(\delta/(L_d+1)) \cdot |A|$, and at this point the subroutine REVAMP() is called. At the end of the call to REVAMP(), the next phase begins. 
\end{itemize}
\end{lemma}

\begin{lemma}
\label{label:general:lm:monotone:laminar}
Consider the maintenance of the laminar structure in any given phase as per our algorithm.
\begin{itemize}
\item In the beginning of the phase, there are no $l$-dirty nodes, i.e., $D_{lj} = \emptyset$ for all $j \in [0, L_d]$.
\item Consider any layer $j \in [1, L_d]$, and focus on any time interval  in the middle of the phase where no call is made to REBUILD($k$) with $k \in [1, j]$. During such a time interval: 
\begin{itemize}
\item  No edge gets inserted into the set $H_j$. In other words, the edge-set $H_j$ keeps shrinking. Furthermore, an edge $e \in H_j$ gets deleted from $H_j$ only if it gets deleted from the graph $G = (V, E)$. 
\item No node $v \in A$ gets moved from $D_{lj}$ to $C_{lj}$. In other words, the node-set $D_{lj}$ (resp. $C_{lj}$) keeps growing (resp. shrinking).
\end{itemize} 
\item At layer $j = 0$, we have $H_0 = H$ and $D_{l0} = A \cap D_c$. Thus, by Lemma~\ref{label:general:lm:monotone:dirty}, the node-set $D_{l0}$ (resp. $C_{l0}$) keeping growing (resp. shrinking) throughout the duration of the phase. 
\end{itemize}
\end{lemma}

\noindent Next,   we upper bound  the maximum degree a node can have in a given layer $j \in [1, L_d]$.  We emphasize that the bounds in Lemma~\ref{label:general:lm:correct:1} are artifacts of the specific algorithm we have designed. In other words, these bounds  {\em cannot} be derived  only by looking at the  definitions and the  invariants stated in Section~\ref{label:general:sec:maintain:prelim}.  Specifically, we exploit two important properties of our algorithm.

\begin{itemize}
\item  Consider  any layer $j \in [1, L_d]$.  In a given phase, the degree of a node $v \in V$ in this layer can increase only during a call to  REBUILD($k$) for some $k \in [1, j]$. See Lemma~\ref{label:general:lm:monotone:laminar}. 
\item  Roughly speaking, just after  the end of a call  to  REBUILD($k$),  the degree of each node $v \in V$ drops by at least a factor of $2$ across successive layers in the range $[k, L_d]$. See Lemma~\ref{label:general:lm:fig:reduce}.
\end{itemize}
\noindent The proof of Lemma~\ref{label:general:lm:correct:1} appears in Section~\ref{label:general:sec:degree:half}.

\begin{lemma}
\label{label:general:lm:correct:1}
\label{label:general:lm:correct:2}
Suppose that we maintain a critical structure  and a laminar structure as per our algorithm, and let  $j \in [0, L_d]$ be any layer in the laminar structure. Then we always have: 
\begin{eqnarray}
\label{label:general:eq:lm:correct:1}
\text{deg}(v, H_{j}) \leq d \cdot 2^{-j} + \sum_{j'=0}^{j-1} 2^{-j'} \ \ \text{ for all nodes } v \in V. \\
\label{label:general:eq:lm:correct:2}
\text{deg}(v, H_{j}) \leq (3\epsilon d/L^2) \cdot 2^{-j} + \sum_{j'=0}^{j-1} 2^{-j'} \ \ \text{ for all nodes } v \in P \setminus D_c.
\end{eqnarray} 
\end{lemma}

\begin{corollary}
\label{label:general:cor:lemma:correct:1}
\label{label:general:cor:lemma:correct:2}
Suppose that we maintain a critical structure  and a laminar structure as per our algorithm. Then we always have: 
\begin{enumerate}
\item $\text{deg}(v, H_{L_d}) \leq  \lambda_d L^4 + 2 \ \ \text{ for all nodes } v \in V.$
\item $\text{deg}(v, H_{L_d}) \leq 3\epsilon \lambda_d L^2 + 2 \ \ \text{ for all nodes } v \in P \setminus D_c.$
\end{enumerate} 
\end{corollary}

\begin{proof}
Follows from Lemma~\ref{label:general:lm:correct:2} and  equation~\ref{label:general:eq:Li:lambda}.
\end{proof}

\subsubsection{Proof of Lemma~\ref{label:general:lm:correct:1}}
\label{label:general:sec:degree:half}
\paragraph{Proof of equation~\ref{label:general:eq:lm:correct:1}.}
\

\medskip
\noindent 
 We  prove equation~\ref{label:general:eq:lm:correct:1} for a given node $v \in V$, using  induction on  the number of edge-updates seen so far.

\begin{itemize}
\item {\em Base step.} Equation~\ref{label:general:eq:lm:correct:1} holds for  node $v$ after the $t^{th}$ edge-update in $G$, for $t = 0$.

The base step is true since initially the graph $G$ is empty and $\text{deg}(v, H_j) = 0$ for all $j \in [0,L_d]$.
\item {\em Induction step.} Suppose that equation~\ref{label:general:eq:lm:correct:1} holds for node $v$ after the $t^{th}$ edge-update in $G$, for some integer $t \geq 0$. Given this induction hypothesis, we will show that equation~\ref{label:general:eq:lm:correct:1} continues to hold for node $v$ after the $(t+1)^{th}$ edge-update in $G$.
\end{itemize}

\noindent From this point onwards, we focus on proving the induction step. There are two possible cases to consider.

\bigskip
\noindent {\em Case 1. The $(t+1)^{th}$ edge-update in $G$ resulted in the termination of a phase.}

\medskip
\noindent In this case, the following conditions hold after our algorithm handles the $(t+1)^{th}$ update in $G$.
\begin{eqnarray}
\label{label:general:eq:induction:great:1}
\text{deg}(v, H_0) & \leq & d \\
\text{deg}(v, H_k) & \leq & \frac{\text{deg}(v, H_{k-1})}{2} + 1 \ \ \text{ at every layer } k \in [1, L_d]. \label{label:general:eq:induction:great:2}
\end{eqnarray}
Equation~\ref{label:general:eq:induction:great:1} holds since $d$ is the maximum degree a node can have in the graph $G = (V, E)$ (see Lemma~\ref{label:general:lm:deg:1}) and  $H_0 \subseteq E$. Equation~\ref{label:general:eq:induction:great:2} holds since the $(t+1)^{th}$ edge-update marks the termination of a phase. Hence, our algorithm calls the subroutine REBUILD($1$) while handling the $(t+1)^{th}$ edge-update (see item (2) in Section~\ref{label:general:sec:phase:end}). At the end of this call to REBUILD($1$), the degree of any node in a layer $k \in [1, L_d]$ is  $(1/2)$ times its degree in the previous layer, plus-minus one (see Lemma~\ref{label:general:lm:fig:reduce}).

By equations~\ref{label:general:eq:induction:great:1} and~\ref{label:general:eq:induction:great:2}, we get the following guarantee  after  the $(t+1)^{th}$ edge-update in $G$.
\begin{eqnarray}
\label{label:general:eq:induction:great:3}
\text{deg}(v, H_j) \leq d \cdot 2^{-j} + \sum_{j'=0}^{j-1} 2^{-j'} \ \ \text{ at every layer } j \in [0, L_d].
\end{eqnarray}
\noindent This concludes the proof of the induction step.

\bigskip
\noindent {\em Case 2. The $(t+1)^{th}$ edge-update in $G$ does not result in the termination of a phase.}

\medskip
\noindent
In this case, the $(t+1)^{th}$ edge-update falls in the middle of a phase, and is handled by the algorithm in Section~\ref{label:general:sec:phase:middle}. Specifically, the edge-sets $H_0, \ldots, H_{L_d}$ are modified by the procedure described in Step III (see Section~\ref{label:general:sec:phase:middle}). There are two operations performed by this procedure that concern us, for they are the only ones that tamper with the edge-sets $H_0, \ldots, H_{L_d}$.
\begin{itemize}
\item (a) In item (2) of Step III (Section~\ref{label:general:sec:phase:middle}), we might change some of the edge-sets $H_0, \ldots, H_{L_d}$. 
\item (b) In item (4) of Step III (Section~\ref{label:general:sec:phase:middle}) we call VERIFY(), which in turn might call  REBUILD($j$) for some layer $j \in [1, L_d]$. And  a call to REBUILD($j$) reconstructs the edge-sets $H_j, \ldots, H_{L_d}$.
\end{itemize}
\noindent Operation (a)  never inserts an edge into a layer $j > 0$. Thus, due to this operation  $\text{deg}(v, H_j)$ can only decrease, provided  $j > 0$. This ensures that for each layer $j > 0$, $\text{deg}(v, H_j)$ continues to satisfy the desired upper bound of equation~\ref{label:general:eq:lm:correct:1} after operation (a). The value of $\text{deg}(v, H_0)$, however, can increase due to operation (a). But this does not concern us, for we always have $\text{deg}(v, H_0) \leq \text{deg}(v, E) \leq d$ (see Lemma~\ref{label:general:lm:deg:1}). So the upper bound prescribed by equation~\ref{label:general:eq:lm:correct:1} for layer $j = 0$ is trivially satisfied all the time.

Operation (b) tampers with  the edge-sets $H_0, \ldots, H_{L_d}$ only if the subroutine REBUILD($j$), for some $j \in [1, L_d]$, is called during the execution of VERIFY(). We focus on this call to REBUILD($j$). Just before the call begins, equation~\ref{label:general:eq:lm:correct:1} holds for node $v$. The call to REBUILD($j$) does not alter the layers $j' \in [0, j-1]$ (see Lemma~\ref{label:general:lm:rebuild:minor:1}). Accordingly, at the end of the call to REBUILD($j$), we have:
\begin{equation}
\label{label:general:eq:induction:great:10}
\text{deg}(v, H_k) \leq d \cdot 2^{-k} + \sum_{j' = 0}^{k-1} 2^{-j'} \ \ \text{ at every layer } k \in [0, j-1].
\end{equation}
Furthermore, the call to REBUILD($j$) ensures that the degree of a node at any layer $k \in [j, L_d]$ is half of its degree in the previous layer, plus-minus one (see Lemma~\ref{label:general:lm:fig:reduce}). Accordingly, at the end of the call to REBUILD($j$), we have:
\begin{equation}
\label{label:general:eq:induction:great:11}
\text{deg}(v, H_k) \leq \frac{\text{deg}(v, H_{k-1})}{2} + 1 \ \ \text{ at every layer } k \in [j, L_d].
\end{equation}
By equations~\ref{label:general:eq:induction:great:10} and~\ref{label:general:eq:induction:great:11}, at the end of the call to REBUILD($j$) we have:
\begin{equation}
\label{label:general:eq:induction:great:12}
\text{deg}(v, H_k) \leq d \cdot 2^{-k} + \sum_{j' = 0}^{k-1} 2^{-j'} \ \ \text{ at every layer } k \in [j, L_d].
\end{equation}
Equations~\ref{label:general:eq:induction:great:10} and~\ref{label:general:eq:induction:great:12} conclude the proof of the induction step.

\bigskip

\paragraph{Proof of equation~\ref{label:general:eq:lm:correct:2}.}
\

\medskip
\noindent
Throughout the proof, fix any node $v \in V$. We will show that the node satisfies equation~\ref{label:general:eq:lm:correct:2} in every ``phase'' (see Section~\ref{label:general:sec:maintain:critical}). Accordingly, fix any phase throughout the rest of the proof. Also recall Lemma~\ref{label:general:lm:init}, which summarizes the way the critical structure changes along with the passage of time within a given phase. Since the node-set $V$ is partitioned into the subsets $A$ and $P$, there are two cases to consider.
\begin{itemize}
\item {\em Case 1.} The node $v$ is active (i.e., part of the set $A$) in the beginning of the phase. Here, the node continues to remain active throughout the duration of the phase. So the lemma is trivially true for  $v$.
\item {\em Case 2.} The node $v$ is passive (i.e., part of the set $P$) in the beginning of the phase. Here,  the node continues to remain passive till the end of the phase. We will consider two possible sub-cases.
\begin{itemize}
\item {\em Case 2a. The node never becomes $c$-dirty during the phase. } In this case, we have $v \in P \setminus D_c$ and  $\text{deg}(v, E) \leq (3\epsilon d/L^2)$ throughout the duration of the phase (see Definition~\ref{label:general:def:critical:structure}). Hence, we can  adapt the proof of equation~\ref{label:general:eq:lm:correct:1} by replacing $d$ with $(3\epsilon d/L^2)$ as an upper bound on the maximum degree of the node, and get the following guarantee: Throughout the duration of the phase, we have $\text{deg}(v, H_{j}) \leq (3\epsilon d L^{-2})/2^j + \sum_{j'=0}^{j-1} 1/2^{j'}$ for every layer $j \in [0, L_d]$.
\item {\em Case 2b. The node becomes $c$-dirty at some time $t$ in the middle of the phase, and from that point onwards, the node remains $c$-dirty till the end of the phase.} In this case, the lemma is trivially true for the node from time $t$ till the end of the phase. Hence, we consider the remaining time interval from the beginning of the phase till time $t$. During this interval, we have $v \in P \setminus D_c$, and we can apply the same argument as in Case 2a. We thus conclude that the lemma holds for the node $v$ throughout the duration of the phase.
\end{itemize}
\end{itemize}

\subsection{Maintaining the edge-set of an skeleton}
\label{label:general:sec:maintain:skeleton}

In Section~\ref{label:general:sec:maintain:prelim}, we  defined the concepts of a critical structure  and a laminar structure   in the graph $G = (V, E)$. In Section~\ref{label:general:sec:main:algo}, we  described our algorithm for maintaing these two structures in a dynamic setting. In this section, we will show that the edges $e \in H_{L_d}$ (i.e., the edges from the last layer of the laminar structure) constitute the edge-set $X$ of a skeleton of the graph $G$ (see Definition~\ref{label:general:def:skeleton}). Towards this end, we define the node-sets $B, T, S \subseteq V$ and the edge-set $X \subseteq E$ as follows.
\begin{eqnarray}
\label{label:general:eq:mapping:s}
S & = & D_c \cup D_{l, L_d}  \label{label:general:eq:mapping:s} \\
B & = & (A \setminus S) \cup \left\{ v \in S : \text{deg}(v, E) > \epsilon d/L^2\right\} \label{label:general:eq:mapping:b} \\
T & = & (P \setminus S) \cup \left\{ v \in S : \text{deg}(v, E \leq \epsilon d/L^2\right\} \label{label:general:eq:mapping:t} \\
 X & = & H_{L_d} \label{label:general:eq:mapping:x} 
\end{eqnarray}

The next lemma shows that the node-sets $B, T, S \subseteq V$ and the edge-set $X \subseteq E$ as defined above satisfies all the seven properties stated in Definition~\ref{label:general:def:skeleton}. This lemma implies the main result of this section, which is summarized in Theorem~\ref{label:general:th:maintain:skeleton:minor}. As far as the proof of Lemma~\ref{label:general:lm:mapping:1} is concerned, the reader should note that  the conditions (1) -- (5)  follow directly from Definitions~\ref{label:general:def:critical:structure},~\ref{label:general:def:laminar:structure} and the four invariants stated in Section~\ref{label:general:sec:maintain:prelim}.  In other words, all the critical and laminar structures that satisfy these invariants will also satisfy the conditions (1) -- (5). In sharp contrast, we need to use a couple of crucial properties of our algorithm (see Corollaries~\ref{label:general:cor:lemma:correct:1} and~\ref{label:general:cor:lemma:correct:2}) to prove the remaining two conditions (6) and (7). And it is not difficult to see that there are instances of critical and laminar structures that satisfy the four invariants from Section~\ref{label:general:sec:maintain:prelim} but do not satisfy the conditions (6) and (7). 

\begin{lemma}
\label{label:general:lm:mapping:1}
Suppose that  a critical structure $(A, P, D_c, H)$ and a laminar structure $(H_1, \ldots, H_{L_d})$ are maintained as per the algorithm in Section~\ref{label:general:sec:maintain:critical}. Also suppose that the node-sets $B, T, S \subseteq V$ and the edge-set $X \subseteq E$ are defined as in equations~\ref{label:general:eq:mapping:s},~\ref{label:general:eq:mapping:b},~\ref{label:general:eq:mapping:t} and~\ref{label:general:eq:mapping:x}. Then the following conditions hold.
\begin{enumerate}
\item We have $B \cup T = V$ and $B \cap T = \emptyset$.
\item For every  node $v \in B$, we have $\text{deg}(v, E) > \epsilon d/L^2$. 
\item For every node  $v \in T$, we have $\text{deg}(v, E) < 3\epsilon d/L^2$. 
\item We have $|S| \leq 4\delta \cdot |B|$. 
\item For every  node $v \in B \setminus S$, we have: $$e^{-\gamma} \cdot (\lambda_d L^4/d) \cdot \text{deg}(v, E) \leq \text{deg}(v, X) \leq e^\gamma \cdot (\lambda_d L^4/d) \cdot \text{deg}(v, E).$$
\item For every node $v \in V$, we have $\text{deg}(v, X) \leq \lambda_d L^4 +2$.
\item For every  node $v \in T \setminus S$, we have $\text{deg}(v, X) \leq 3\epsilon \lambda_d L^2 + 2$.
\end{enumerate}
\end{lemma}

\begin{proof} \

\begin{enumerate}
\item  Recall that $A \subseteq V$ and $P = V \setminus A$ (see Definition~\ref{label:general:def:critical:structure}). Since $S \subseteq V$ (see equation~\ref{label:general:eq:mapping:s} and Definitions~\ref{label:general:def:critical:structure},~\ref{label:general:def:laminar:structure}), we infer that the node-set $V$ is partitioned into three subsets: $A \setminus S$, $P \setminus S$ and $S$. Hence, from equations~\ref{label:general:eq:mapping:b} and~\ref{label:general:eq:mapping:t}, we get: $B \cup T = V$ and $B \cap T = \emptyset$.
\item Definition~\ref{label:general:def:critical:structure} implies that $\text{deg}(v, E) > \epsilon d/L^2$ for all nodes $v \in A \setminus D_c$. Since $D_c \subseteq S$ (see equation~\ref{label:general:eq:mapping:s}), we get: $\text{deg}(v, E) > \epsilon d/L^2$ for all nodes $v \in A \setminus S$. Hence, equation~\ref{label:general:eq:mapping:b} implies that $\text{deg}(v, E) > \epsilon d/L^2$ for all nodes $v \in B$.
\item Definition~\ref{label:general:def:critical:structure} implies that $\text{deg}(v, E) < 3\epsilon d/L^2$ for all nodes $v \in P \setminus D_c$. Since $D_c \subseteq S$ (see equation~\ref{label:general:eq:mapping:s}), we get: $\text{deg}(v, E) < 3\epsilon d/L^2$ for all nodes $v \in P \setminus S$. Hence, equation~\ref{label:general:eq:mapping:t} implies that $\text{deg}(v, E) < 3\epsilon d/L^2$ for all nodes $v \in T$.
\item  The basic idea behind the proof of condition (4) is as follows. Equation~\ref{label:general:eq:critical:count} implies that $|D_c| \leq \delta \cdot |A|$. Since $S = D_c \cup D_{l, L_d}$, Invariant~\ref{label:general:inv:laminar} implies that $|S| \leq 2 \delta \cdot |A|$. So the size of the set $S$ is negligible in comparison with   the size of the set $A$. Hence, the size of the set $S$ is negligible even in comparison with the (smaller) set  $A \setminus S$. Since $A \setminus S \subseteq B$, it follows that  $|S|$ is negligible in comparison with $|B|$ as well. Specifically, we can show that $|S| \leq 4 \delta \cdot |B|$. Below, we present the proof in full details.

Recall that $D_{l, L_d} \subseteq A$ (see Invariant~\ref{label:general:inv:laminar:contain}). Hence, we have: 
\begin{equation}
\label{label:general:eq:funky:0}
|A| = \left| D_{l, L_d} \right| + \left| A \setminus  D_{l, L_d}\right|
\end{equation}
Invariant~\ref{label:general:inv:laminar} states that:
\begin{eqnarray}
\left| D_{l, L_d} \right|  \leq  \delta \cdot \left| A  \right|
 \label{label:general:eq:funky:1} 
 \end{eqnarray}
 Recall that $\delta < 1/2$ (see equation~\ref{label:general:eq:imp:1}). Hence, equations~\ref{label:general:eq:funky:0} and~\ref{label:general:eq:funky:1} imply that:
 \begin{eqnarray}
 \label{label:general:eq:funky:2}
 |A| \leq \frac{1}{(1-\delta)} \cdot \left|A \setminus  D_{l, L_d} \right| \leq 2 \cdot \left| A \setminus D_{l, L_d} \right| 
 \end{eqnarray}
 Equations~\ref{label:general:eq:funky:1} and~\ref{label:general:eq:funky:2} imply that:
 \begin{eqnarray}
 \left| D_{l, L_d} \right| \leq 2 \delta \cdot  \left| A \setminus  D_{l, L_d} \right| 
 \label{label:general:eq:funky:3}
 \end{eqnarray}
 Next,  Definition~\ref{label:general:def:critical:structure} implies that $|D_c| \leq \delta \cdot |A|$. Accordingly, from equation~\ref{label:general:eq:funky:2} we get:
 \begin{eqnarray}
 \label{label:general:eq:funky:4}
 |D_c| \leq 2 \delta \cdot \left| A \setminus  D_{l, L_d} \right|
 \end{eqnarray}
 Equations~\ref{label:general:eq:funky:3} and~\ref{label:general:eq:funky:4} imply that:
 \begin{equation}
 \label{label:general:eq:funky:5}
 \left| D_c \cup D_{l, L_d} \right| \leq 4\delta \cdot  \left| A \setminus  D_{l, L_d} \right|
  \end{equation}
Since $S = D_c \cup D_{l, L_d}$ (see equation~\ref{label:general:eq:mapping:s}) and  $D_{l0} = D_c \cap A \subseteq D_{l,L_d}$ (see Invariant~\ref{label:general:inv:laminar:contain}), we have:
\begin{equation}
\label{label:general:eq:funky:6}
A \setminus S = A \setminus D_{l, L_d}
\end{equation}
From equations~\ref{label:general:eq:mapping:s},~\ref{label:general:eq:mapping:b},~\ref{label:general:eq:funky:5} and~\ref{label:general:eq:funky:6}, we conclude that:
$$|S| \leq 4 \delta \cdot |B|.$$
\item Equations~\ref{label:general:eq:mapping:s},~\ref{label:general:eq:mapping:b} and Invariant~\ref{label:general:inv:laminar:contain} imply that:
\begin{eqnarray}
B \setminus S    =   A \setminus S  \subseteq  A \setminus D_{l,L_d} = C_{l,L_d}. \label{label:general:eq:funky:10} 
\end{eqnarray}
Corollary~\ref{label:general:cor:lm:clean} and equations~\ref{label:general:eq:mapping:x}, \ref{label:general:eq:funky:10} imply that for all nodes $v \in B \setminus S$, we have:
\begin{eqnarray*}
e^{-\gamma} \cdot (\lambda_d L^4/d) \cdot \text{deg}(v, E) \leq \text{deg}(v, X) \leq e^{\gamma} \cdot (\lambda_d L^4/d) \cdot \text{deg}(v, E).
\end{eqnarray*}
\item Since $X = H_{L_d}$ (see equation~\ref{label:general:eq:mapping:x}), Corollary~\ref{label:general:cor:lemma:correct:1} implies that $\text{deg}(v, X) \leq \lambda_d L^4 +2$ for all  $v \in V$. 
\item Equations~\ref{label:general:eq:mapping:s} and~\ref{label:general:eq:mapping:t}  imply that:
\begin{eqnarray}
\label{label:general:eq:funky:20}
T \setminus S = P \setminus S \subseteq P \setminus D_c
\end{eqnarray}
By Corollary~\ref{label:general:cor:lemma:correct:2} and equations~\ref{label:general:eq:mapping:x},~\ref{label:general:eq:funky:20}, for all nodes $v \in T \setminus S$, we have: $\text{deg}(v, X) \leq 3\epsilon \lambda_d L^2 + 2$.
\end{enumerate}
\end{proof}

Below, we summarize the main result of this section.

\begin{theorem}
\label{label:general:th:maintain:skeleton:minor}
Suppose that  a critical structure $(A, P, D_c, H)$ and a laminar structure $(H_1, \ldots, H_{L_d})$ are maintained as per the algorithm in Section~\ref{label:general:sec:maintain:critical}. Then the edges in the last layer $H_{L_d}$ form the edge-set of a skeleton of $G = (V, E)$.
\end{theorem}

\begin{proof}
Follows from Definition~\ref{label:general:def:skeleton} and Lemma~\ref{label:general:lm:mapping:1}.
\end{proof}

\subsection{Bounding the amortized update time of our algorithm}
\label{label:general:sec:analyze:time}

We implement our algorithm using the most obvious data structures. To be more specific, we maintain a doubly linked list for each of the node-sets $A, P, D_c, \{D_{lj}\}$ and $\{ C_{lj} \}$. Further, we maintain counters that store the sizes of each of these sets. We also maintain the edge-sets $E, H$ and $\{ H_j\}$ using standard adjacency list data structures, and we maintain the degree of each node in each of these edge-sets. Whenever an ``element'' (an  edge or a node) appears in some linked list, we store a pointer from that element to its position in the linked list. Using these pointers, an element can be added to (or deleted from) a list in constant time.

\subsubsection{A few notations and terminologies}
\label{label:general:sec:runtime:notations}

Recall that initially the graph $G = (V, E)$ has zero edges. We will use the term ``edge-update'' to denote the insertion/deletion of an edge in $G = (V, E)$. Thus, in the dynamic setting, the graph $G = (V, E)$ changes via a sequence of edge-updates. For any integer $t \geq 0$, we will use the term ``edge-update $t$'' to refer to the $t^{th}$ edge-update in this sequence. Each phase  corresponds to a contiguous block of edge-updates (see Section~\ref{label:general:sec:main:algo}). We can, therefore, identify a phase by an interval $[t, t']$, where $t < t'$ are positive integers. Such a phase begins with the edge-update $t$ and ends just after the edge-update $t'$. The next phase starts with the edge-update $t'+1$.

\subsubsection{Roadmap}

The rest of this section is organized as follows.
\begin{itemize}
\item In Section~\ref{label:general:sec:lm:runtime:exclude},  we show that excluding the calls to the subroutines REVAMP() and REBUILD($j$), our algorithm handles each edge-update  in $O(L_d)$ time in the worst case. 
\item In Section~\ref{label:general:sec:time:rebuild}, we bound the time taken by a single call to REBUILD($j$). 
\item In Section~\ref{label:general:sec:runtime:revamp}, we show that the subroutine REVAMP() runs in $O(L^2 L_d/(\epsilon \delta))$ amortized time (see Lemma~\ref{label:general:lm:amortized:time:revamp}). This includes the time taken by the calls to  REBUILD($1$) at the end of a phase. 
\item Finally, in Section~\ref{label:general:sec:runtime:rebuild}, we show that for a given layer $j \in [1, L_d]$, the amortized time taken by a call to REBUILD($j$) in the middle of a phase is $O(L^2 L_d^2/(\epsilon \gamma \delta))$ (see Lemma~\ref{label:general:lm:bound:laminar:main}). Since there are $L_d$ possible values of $j$, the total amortized time of all the calls to REBUILD(.) in the middle of a phase is given by $O(L^2 L_d^3/(\epsilon \gamma \delta))$. 
\end{itemize}
\noindent Summing over all these running times, we get our main result.

\begin{theorem}
\label{label:general:th:maintain:runtime:main}
Our algorithm for maintaining a critical and a laminar structure handles an edge insertion/deletion in $G = (V, E)$ in $O(L^2 L_d^3/(\epsilon \gamma \delta))$ amortized time.
\end{theorem} 

\subsubsection{A simple bound}
\label{label:general:sec:lm:runtime:exclude}

We devote this section to the proof of the following lemma. 

\begin{lemma}
\label{label:general:lm:runtime:exclude}
Excluding the calls made to the REVAMP() and REBUILD($j$) subroutine, our algorithm handles an edge-update in $G = (V, E)$ in $O(L_d)$ time.
\end{lemma}

\begin{proof}
Recall the description of our algorithm in Section~\ref{label:general:sec:main:algo}. For the purpose of this lemma, we only need to concern ourselves with Section~\ref{label:general:sec:phase:middle}, for Section~\ref{label:general:sec:phase:begin} only states the initial conditions in the beginning of a phase and Section~\ref{label:general:sec:phase:end} describes the REVAMP() subroutine. Accordingly, we analyze the time taken by the Steps I, II and III from Section~\ref{label:general:sec:phase:middle}, one after the other. 
\begin{itemize}
\item {\em Running time of Step I.}

Here, at most one edge is inserted into (or deleted from) the set $H$, and at most two nodes become $c$-dirty. It is easy to check that these operations can be implemented in $O(1)$ time.
\item {\em Running time of Step II.} 

This step only asks us to compare the number of $c$-dirty nodes against the number of active nodes. Hence, this can be implemented in $O(1)$ time.  
\item {\em Running time of Step III.}

Here, we bound the time taken by each of the items (1), (2), (3) and (4) in Step III.
\begin{enumerate}
\item Item (1) moves at most two nodes from $D_{lk}$ to $C_{lk}$ for all $k \in [0, L_d]$, and  requires $O(L_d)$ time.
\item Item (2) inserts/deletes an edge from at most $(L_d+1)$ edge-sets $\{H_j\}$, and requires $O(L_d)$ time.
\item Item (3) makes two calls to CLEANUP($x$). Each of these calls runs in $O(L_d)$ time. 
\item Item (4)  calls  VERIFY(), which requires $O(L_d)$ time excluding the call to REBUILD($j$).
\end{enumerate}
Hence, the total time taken by Step III is $O(L_d)$.
\end{itemize}
\noindent Hence, excluding the calls to REBUILD($j$) and REVAMP(), handling an edge-update  requires $O(L_d)$ time. 
\end{proof}

\subsubsection{Analyzing the running time of a single call to REBUILD($j$)}
\label{label:general:sec:time:rebuild}

\noindent We will bound  the running time of a call to the subroutine REBUILD($j$) as described in Figure~\ref{label:general:fig:rebuild}. But first, recall the initial conditions stated in Figure~\ref{label:general:fig:rebuild:init} and the properties of the subroutine stated in Lemma~\ref{label:general:lm:rebuild:minor:1}.

\begin{lemma}
\label{label:general:lm:data:structure}
Consider any layer $j \in [1, L_d]$, and note that the subroutine REBUILD($j$) does not change the set of active nodes. A call to the subroutine  REBUILD($j$) runs in   $O(|A| \cdot d \cdot 2^{-j})$ time.
\end{lemma}

\begin{proof}
The call to  REBUILD($j$) does not affect the layers $j' < j$. Instead, the subroutine iteratively sets $H_{j'} \leftarrow \text{SPLIT}(H_{j'-1})$ for $j' = j$ to $L_d$. See Figure~\ref{label:general:fig:rebuild}.

Consider the first iteration of the For loop in Figure~\ref{label:general:fig:rebuild}, where we set $H_{j} \leftarrow \text{SPLIT}(H_{j-1})$. Since $H_{j-1} \subseteq H$ (see Definition~\ref{label:general:def:laminar:structure}), and $H$ is the set of edges incident upon active nodes (see Definition~\ref{label:general:def:critical:structure}), each edge  in $H_{j-1}$ has at least one  endpoint in the set $A$. Also by Lemma~\ref{label:general:lm:correct:1}, the maximum degree of a node in a layer $j'$ is $O(d \cdot 2^{-j'})$. Hence, there are at most $O(|A| \cdot d \cdot 2^{-(j-1)})$ edges in $H_{j-1}$. Thus, we have:
\begin{equation}
\label{label:general:eq:data:structure}
|H_{j-1}| = O(|A| \cdot d \cdot 2^{-(j-1)}).
\end{equation}

By Figure~\ref{label:general:fig:split}, we can implement  the subroutine SPLIT($H_{j-1}$) in $O(|H_{j-1}|)$ time. After setting $H_j \leftarrow \text{SPLIT}(H_{j-1})$, we have to perform the following operations: (a) $D_{l,j} \leftarrow D_{l,j-1}$ and (b) $C_{l,j} \leftarrow C_{l,j-1}$. Since both the sets $D_{l,j-1}$ and $C_{l,j-1}$ are contained in $A$ (see Invariant~\ref{label:general:inv:laminar:contain}), these operations can be performed in $O(|A|)$ time. Since $j \leq L_d$, we have   $d \cdot 2^{-(j-1)} \geq 1$  (see equation~\ref{label:general:eq:Li}). Thus,  the time taken by the first iteration of the For loop in Figure~\ref{label:general:fig:rebuild} is $O(|H_{j-1}| + |A|) = O(|A| \cdot d \cdot 2^{-(j-1)})$. 

Applying the same argument by which we arrived at equation~\ref{label:general:eq:data:structure},  immediately after setting   $H_j \leftarrow \text{SPLIT}(H_{j-1})$, we have the following guarantee:
\begin{equation}
\label{label:general:eq:data:structure:1}
|H_{j}| = O(|A| \cdot d \cdot 2^{-j})
\end{equation}
Accordingly, the second iteration of the For loop in Figure~\ref{label:general:fig:rebuild}, where we set $H_{j+1} \leftarrow \text{SPLIT}(H_j)$, can be implemented in $O(|A| \cdot d \cdot 2^{-j})$ time. In general, it can be shown that setting $H_{k} \leftarrow \text{SPLIT}(H_{k-1})$ will take $O(|A| \cdot d \cdot 2^{- (k-1)})$ time for $k \in [j, L_d]$. So the  runtime of the subroutine REBUILD($j$) is given by:
$$\sum_{k= j}^{L_d} O\left(|A| \cdot d \cdot 2^{-(k-1)} \right) = O\left(|A| \cdot d \cdot 2^{-j} \right).$$
This concludes the proof of Lemma~\ref{label:general:lm:data:structure}. 
\end{proof}

\subsubsection{Bounding the amortized update time of the subroutine REVAMP()}
\label{label:general:sec:runtime:revamp}

\newcommand{\U}{\mathcal{U}}

Recall the notations and terminologies introduced in Section~\ref{label:general:sec:runtime:notations}. Our algorithm works in ``phases'' (see the discussion in the beginning of Section~\ref{label:general:sec:main:algo}). For every integer $k \geq 1$, suppose that the phase $k$ starts with the edge-update $t_k$. Thus, we have $t_1 = 1$ and $t_k < t_{k+1}$ for every  integer $k \geq 1$, and each phase $k$ corresponds to the interval $[t_k, t_{k+1}-1]$. Next, note that a call is made to the subroutine REVAMP() at the end of each phase (sec Section~\ref{label:general:sec:phase:end}).  So the $k^{th}$ call to the subroutine REVAMP() occurs  while processing the edge-update $t_{k+1}-1$.


\begin{itemize}
\item Throughout the rest of this section, we will use the properties of our algorithm outlined in Lemma~\ref{label:general:lm:middle}.
\end{itemize}

\noindent We  introduce the following notations.
\begin{itemize}
 \item In the beginning of a phase $k \geq 1$, we have $D_c = \emptyset$. As the phase progresses, the set $D_c$ keeps getting bigger and bigger.  Finally, due to the edge-update $t_{k+1}-1$, the size of the set $D_c$ exceeds the threshold $(\delta/(L_d+1)) \cdot |A|$.  Let $D_c^+{(k)}$ denote the status of the set  $D_c$ precisely at this  instant. Note that immediately afterwords, a call is made to the subroutine REVAMP(), and when the subroutine finishes execution we again find that the set $D_c$ is empty. Then we initiate the next phase. 
\item The indicator variable $\text{{\sc Dirty}}_c^+(v, k) \in \{0, 1\}$  is set to one  iff $v \in D_c^+(k)$. Thus, we have:
\begin{equation}
\label{label:general:eq:indicator:dirty}
\left| D_c^+(k) \right| = \sum_{v \in V} \text{{\sc Dirty}}_c^+(v, k)
\end{equation}
 \item The sets $A$ and $P$ do not change in the middle of a phase. Hence, without any ambiguity, we let $A(k)$ and $P(k)$ respectively denote the status of the set $A$ (resp. $P$) during phase $k \geq 1$. 
 \item Consider a counter $\U_v$. Initially,  the graph $G = (V, E)$ is empty and $\U_v = 0$. Subsequently, whenever an edge-update occurs in the graph $G = (V,E)$, if the edge being inserted/deleted is incident upon the node $v$, then we set $\U_v \leftarrow \U_v + 1$. 
 \item Let $\U_v(t)$ denote the status of  $\U_v$ after the edge-update $t \geq 0$. Note that $\U_v(0) = 0$, and in general, $\U_v(t)$ gives the number of edge-updates that are incident upon $v$ among the first $t$ edge-updates.
 \end{itemize}

\paragraph{Roadmap.} In Lemma~\ref{label:general:lm:runtime:revamp} and Corollary~\ref{label:general:cor:lm:runtime:revamp}, we bound the time taken by the call to REVAMP() at the end of a given phase $k \geq 1$. Specifically, we show that such a call runs in $O(\left|D_c^+(k) \right| \cdot d \cdot L_d \cdot \delta^{-1})$ time. In Lemma~\ref{label:general:lm:update:count:1}, we show that on average, a node $v$ can become $c$-dirty only after $\Omega(\epsilon d L^{-2})$ edge-updates  incident upon $v$ have taken place in the graph $G = (V, E)$. By Corollary~\ref{label:general:cor:lm:runtime:revamp}, processing a $c$-dirty node at the end of the phase requires  $O(d L_d \delta^{-1})$ time. Hence, we get an amortized update time of $O((dL_d \delta^{-1})/(\epsilon d L^{-2})) = O(L^2 L_d /(\epsilon \delta))$ for the subroutine REVAMP(). This is proved in Lemma~\ref{label:general:lm:revamp:final}.

\begin{lemma}
\label{label:general:lm:runtime:revamp}
The  call to REVAMP() at the end of phase $k$ runs in $O(|D_c^+(k) \cup A(k)| \cdot d)$ time. 
\end{lemma}

\begin{proof}
Recall the description of the subroutine REVAMP() in Section~\ref{label:general:sec:phase:end}. The node-set $V'$ in the For loop is given by the set $D_c^+(k)$. We first bound the running time of this For loop (see item (1) in Section~\ref{label:general:sec:phase:end}).

 During  a single iteration of the For loop, we pick one $c$-dirty node $v$ and make it $c$-clean. Next, we might have to change the status of the node from active to passive (or vice versa) and remove (resp. add) the edges incident upon $v$ from (resp. to) the set $H$. Since the node $v$ has at most $d$ edges incident upon it (see Lemma~\ref{label:general:lm:deg:1}), one iteration of the For loop can be implemented in $O(d)$ time. Since the For loop runs for $D_c^+(k)$  iterations, the total time taken by the For loop is $O(|D_c^+(k)| \cdot d)$.
 
Next, we bound the time taken to implement item (2) in Section~\ref{label:general:sec:phase:end}. In the beginning of the For loop, we had $|A(k)|$ many active nodes. In the worst case, the For loop can potentially change the status of every  passive node in $D_c^+(k)$ to active. Thus, there are at most $|D_c^+(k) \cup A(k)|$ many active nodes at the end of the For loop. So the first part of item (2), where we update the edge-sets $\{H_j\}$ and the node-sets $\{D_{lj}, C_{lj}\}$, can be implemented in $O(|D_c^+(k) \cup A(k)| \cdot d)$ time. Next, by Lemma~\ref{label:general:lm:data:structure}, the call to REBUILD($1$) also takes   $O(|D_c^+(k) \cup A_k| \cdot d)$ time.

Thus,  the total time taken by call to REVAMP() at the end of phase $k$ is $O(|D_c^+(k) \cup A(k)| \cdot d)$. 
\end{proof}

\begin{corollary}
\label{label:general:cor:lm:runtime:revamp}
The  call to  REVAMP() at the end of phase $k$ runs in $O(|D_c^+(k)| \cdot d \cdot (L_d/\delta))$ time.
\end{corollary}

\begin{proof}
Follows from Lemma~\ref{label:general:lm:runtime:revamp} and the fact that $|A(k)| = O((L_d/\delta) \cdot |D_c^+(k)|)$ (see Lemma~\ref{label:general:lm:end}).
\end{proof}

We now explain the statement of Lemma~\ref{label:general:lm:update:count:1}.  First, note that a node can change its status from $c$-clean to $c$-dirty at most once during the course of a single phase (see Lemma~\ref{label:general:lm:init}). Accordingly, consider any node $v \in V$, and let $\psi(k) = \sum_{k'=1}^k \text{{\sc Dirty}}_c^+(v, k')$ be the number  of times  the node $v$ becomes $c$-dirty during the first $k$ phases, where $k \geq 1$ is an integer. Then at least $(2 \epsilon d/L^2) \cdot \psi(k)$ edge-updates in the first $k$ phases were incident upon $v$. To summarize, this lemma shows that on average one needs  to have $(2\epsilon d/L^2)$ edge-updates incident upon $v$ before the node $v$  changes its status from $c$-clean to $c$-dirty.

\begin{lemma}
\label{label:general:lm:update:count:1}
Fix any node $v \in V$ and any positive integer $k$. We have:
\begin{eqnarray}
\label{label:general:eq:updates:new:1}
\U_v(t_{k+1}-1) \geq 
\left(\frac{2\epsilon d}{L^2}\right) \cdot \sum_{k'=1}^k \text{{\sc Dirty}}_c^+(v, k')
\end{eqnarray}
\end{lemma}

\begin{proof}

We will use induction on the value of the sum $\psi(k) = \sum_{k'=1}^k \text{{\sc Dirty}}_c^+(v, k')$. Specifically, throughout the proof we fix a phase $k \geq 1$. Next, we   define $\psi(k) = \sum_{k'=1}^k \text{{\sc Dirty}}_c^+(v, k')$, and consider the last edge-update $t$ in the interval $[t_1, t_{k+1}-1]$ where the node $v$ becomes  $c$-dirty from $c$-clean. We will show that   $\U_v(t) \geq (2\epsilon d/L^2) \cdot \psi(k)$. Since $t_{k+1}-1 \geq t$, this will imply that $\U_v(t_{k+1}-1) \geq \U_v(t) \geq (2\epsilon d /L^2) \cdot \psi(k)$. 

\bigskip
\noindent {\em The base step. $\psi(k) = 1$.} 

\medskip
\noindent In the beginning of phase one, the graph $G = (V, E)$ is empty,  $\text{deg}(v, E) = 0$,  and the node $v$ is passive and $c$-clean. Since $\psi(k) = 1$, there is exactly one phase $k' \in [1, k]$ such that $v \in D_c^+(k)$. Recall the description of our algorithm in Section~\ref{label:general:sec:phase:middle}. It follows that $\text{deg}(v, E)$ was less than the threshold $(3\epsilon d/L^2)$ from the start of phase one till the end of phase $(k'-1)$. It was only after the  edge-update $t$ in the middle of phase $k'$ that we saw $\text{deg}(v, E)$ reaching the threshold $(3\epsilon d/L^2)$, which compelled us to classify the node $v$ as $c$-dirty at that instant. Since $\text{deg}(v, E)$ was initially zero,  at least $(3\epsilon d/L^2)$ edges incident upon $v$ must have been inserted into $G = (V, E)$ during the first  $t$ edge-updates. Hence, we must have $\U_v(t) \geq (3 \epsilon d/L^2) = (3\epsilon d /L^2) \cdot \psi(k) > (2\epsilon d/L^2) \cdot \psi(k)$.

\bigskip
\noindent {\em The inductive step. $\psi(k) \geq 2$.} 

\medskip
\noindent Consider the phase $k^* < k$ where the node $v$ becomes $c$-dirty for the $(\psi(k)-1)^{th}$ time. Specifically, let $k^* = \min \{ k' \in [1, k-1] : \sum_{j = 1}^{k'} \text{{\sc Dirty}}_c^+(v, j) = \psi(k) -1 \}$.  Let $t^*$ be  the edge-update in phase $k^*$ due to which the node $v$ becomes $c$-dirty from $c$-clean. By induction hypothesis, we have:
\begin{equation}
\label{label:general:eq:induction:1}
\U_v(t^*) \geq \left(\frac{2\epsilon d}{L^2}\right) \cdot (\psi(k)-1)
\end{equation}
Consider the interval $[t_{k^*+1}, t_{k+1}-1]$  from the start of phase $(k^*+1)$ till the end of phase $k$. Let $t$ be the unique edge-update in this interval   where the node $v$  becomes $c$-dirty from $c$-clean. 
We will show that at least $(2\epsilon d/L^2)$ edge-updates incident upon $v$ took place during the interval $[t^*+1, t]$. Specifically, we will prove the following inequality. 
\begin{equation}
\label{label:general:eq:induction:2}
\U_v(t) - \U_v(t^*) \geq \left(\frac{2\epsilon d}{L^2}\right)
\end{equation}
Equations~\ref{label:general:eq:induction:1} and~\ref{label:general:eq:induction:2} will imply that $\U_v(t) \geq (2\epsilon d/L^2) \cdot \psi(k)$. To proceed with the proof of equation~\ref{label:general:eq:induction:2},  note that there are two possible ways by which the node $v$ could have become $c$-dirty at edge-update $t^*$, and accordingly, we consider two possible cases.
\begin{itemize}
\item {\em Case 1. 
In the beginning of phase $k^*$, the node $v$ was passive and $c$-clean, with $\text{deg}(v, E) < (3\epsilon d/L^2)$. Then with the edge-update $t^*$ in the middle of phase $k^*$, we saw $\text{deg}(v, E)$ reaching the threshold $(3\epsilon d/L^2)$, which forced us to classify the node $v$ as $c$-dirty at that instant.}

In this instance, we need to consider two possible sub-cases.
\begin{itemize}
\item {\em Case 1a.  At the end of phase $k^*$, i.e., immediately after edge-update $(t_{k^*+1}-1)$, the degree of the node $v$ is at most $(\epsilon d/L^2)$.}

In this case,  during the interval $[t^*, t_{k^*+1} -1]$ itself the degree of the node $v$ has changed from $(3 \epsilon d/L^2)$ to $(\epsilon d /L^2)$. This can happen only if $(3\epsilon d/L^2) - (\epsilon d/L^2) = (2 \epsilon d/L^2)$ edges incident upon $v$ was deleted from the graph during this interval. Since $t_{k^*+1}-1 < t$, we get:
$$\U_v(t) - \U_v(t^*) \geq \U_v(t_{k^*+1} -1) - \U_v(t^*) \geq \left(\frac{2\epsilon d}{L^2}\right).$$
\item {\em Case 1b.  At the end of phase $k^*$, i.e., immediately after edge-update $(t_{k^*+1}-1)$, the degree of the node $v$ is greater than $(\epsilon d/L^2)$.}

In this case, the node $v$ is classified as active  in  phase $k^*+1$ (see Section~\ref{label:general:sec:phase:end}).   Further,  the node $v$ is $c$-clean at the start of phase $k^*+1$ (see Lemma~\ref{label:general:lm:init}) and remains $c$-clean throughout the interval $[t_{k^*+1}, t-1]$ (this follows from our choice of $t$). Accordingly, the node  remains active throughout the interval $[t_{k^*+1}, t-1]$ (see Section~\ref{label:general:sec:phase:end}). To summarize, the node remains active and $c$-clean from the start of phase $k^*+1$ till the edge-update $t-1$.

By definition, after the  edge-update $t$ the node $v$ becomes $c$-dirty again. An active node becomes $c$-dirty only if its degree reaches the threshold $(\epsilon d/L^2)$. Thus, we must have $\text{deg}(v, E) = \epsilon d/L^2$ after edge-update $t$. Since $\text{deg}(v, E) = (3\epsilon d/L^2)$ after edge-update $t^*$, we infer that $\text{deg}(v, E)$ has dropped by $(3\epsilon d/L^2) - (\epsilon d/L^2) = (2 \epsilon d/L^2)$ during the interval $[t^*, t]$. This can happen only if at least $(2 \epsilon d/L^2)$ edges incident upon $v$ were deleted from the graph $G = (V, E)$ during the interval $[t^*, t]$. Thus, we get:
$$\U_v(t) - \U_v(t^*) \geq (2\epsilon d/L^2).$$
\end{itemize}
\item {\em Case 2.
In the beginning of phase $k^*$, the node $v$ was active and $c$-clean, with $\text{deg}(v, E) > (\epsilon d/L^2)$. Then with the edge-update $t^*$ in the middle of phase $k^*$, we saw $\text{deg}(v, E)$ reaching the threshold $(\epsilon d/L^2)$, which forced us to classify the node $v$ as $c$-dirty at that instant.} 

We can easily modify the proof of Case 1  for this case. Basically, the roles of the active and passive nodes are interchanged, and so are the roles of the  thresholds  $(3 \epsilon d/L^2)$ and $(\epsilon d/L^2)$.
\end{itemize}
\end{proof}

We are now ready to bound the amortized update time of the subroutine REVAMP().

\begin{lemma}
\label{label:general:lm:revamp:final}
\label{label:general:lm:amortized:time:revamp}
The subroutine REVAMP() has an amortized update time of $O(L^2 L_d/(\epsilon \delta))$.
\end{lemma}

\begin{proof}
Recall that $\U_v(t)$ denotes the number of edge-updates incident upon $v$ among the first $t$ edge-updates in $G = (V, E)$. Since each edge is incident upon two nodes,  we have $t = (1/2) \cdot \sum_{v \in V} \U_v(t)$ for every integer $t \geq 1$. Hence, we have the following guarantee at the end of every phase $k \geq 1$.
\begin{eqnarray}
t_{k+1} - 1 & = &  (1/2) \cdot \sum_{v \in V} \U_v(t_{k+1} -1)  \nonumber \\
& \geq & (\epsilon d/L^2) \cdot \sum_{v \in V} \sum_{k'=1}^k \text{{\sc Dirty}}_c^+(v, k') \label{label:general:eq:revamp:final:1} \\
& = & (\epsilon d/L^2) \cdot \sum_{k'=1}^k \sum_{v \in V} \text{{\sc Dirty}}_c^+(v, k') \nonumber \\
& = & (\epsilon d/L^2) \cdot \sum_{k'=1}^k \left|D_c^+(v,k')\right| \label{label:general:eq:revamp:final:2}
\end{eqnarray}  
Equation~\ref{label:general:eq:revamp:final:1} follows from Lemma~\ref{label:general:lm:update:count:1}. Equation~\ref{label:general:eq:revamp:final:2} follows from equation~\ref{label:general:eq:indicator:dirty}. From equation~\ref{label:general:eq:revamp:final:2} and Corollary~\ref{label:general:cor:lm:runtime:revamp}, we get the following guarantee.
\begin{itemize}
\item The total time spent on the calls to the subroutine REVAMP() during the first $k$ phases is given by: 
$$\sum_{k'=1}^k O(|D_c^+(k)| \cdot d \cdot (L_d/\delta)) = (t_{k+1}-1) \cdot O(L^2 L_d/(\epsilon \delta)).$$
\end{itemize}
Note that by definition, there are $(t_{k+1}-1)$ edge-updates  in the graph $G = (V, E)$ during the first $k$ phases. Accordingly, we infer that the subroutine REVAMP() has an amortized update time of $O(L^2 L_d/(\epsilon \delta))$.
\end{proof}

\subsubsection{Bounding the amortized update time of REBUILD($j$) in the middle of a phase}
\label{label:general:sec:analyze:time:middle}
\label{label:general:sec:runtime:rebuild}

Throughout this section, we fix any layer $j \in [1, L_d]$. We will analyze the amortized running time of a call made to the subroutine REBUILD($j$). Note that such a call  can be made under two possible circumstances.
\begin{itemize}
\item  While processing an edge-update in the middle of a phase,  a call to REBUILD($j$) can be  made by the subroutine VERIFY() in Figure~\ref{label:general:fig:insert}. This happens only if there are too many $l$-dirty nodes at layer $j$.
\item  While processing the last edge-update  of a phase,  the subroutine REVAMP() calls  REBUILD($j$) with $j = 1$.   The total running time of these calls to REBUILD($1$) is subsumed by the total running time of all the calls to REVAMP(), which in turn has already been analyzed in Section~\ref{label:general:sec:runtime:revamp}.
\end{itemize}

\noindent  Accordingly, in this section we  focus on the calls made to REBUILD($j$) from a given phase.

\begin{itemize}
\item Throughout the rest of this section, we will use  the properties  outlined in Lemmas~\ref{label:general:lm:init} and~\ref{label:general:lm:monotone:laminar}. Further, since we are considering a fixed layer $j \in [1, L_d]$,  we will sometimes omit the symbol ``$j$''  from the notations introduced in this section. 
\end{itemize}

\medskip
Before proceeding any further, we need to introduce the concept of an ``epoch''.


\begin{definition}
\label{label:general:def:epoch}
 An epoch $\left[\tau^{0}, \tau^{1}\right]$ consists of a  contiguous block of edge-updates in the middle of a phase, where $\tau^{0}$ (resp. $\tau^{1}$) denotes the first (resp. last) edge-update in the epoch.  Specifically, we have:
\begin{itemize}
\item While processing the edge-update $(\tau^0-1)$, a call is made to REBUILD($j'$) with $j' \in [1, j]$.
\item  While processing the edge-updates $t \in \left[\tau^0, \tau^1-1\right]$, no call is made to REBUILD($j'$) with $j' \in [1, j]$. 
\item While processing the edge-update $\tau^{1}$, a call is made to REBUILD($j$).
\end{itemize}
\noindent We say that the  epoch ``begins''  at the time instant just before the edge-update $\tau^0$, and ``ends''  at the time instant just before the call to REBUILD($j$) while processing the edge-update $\tau^1$.
\end{definition}
For the rest of this section, we fix a given epoch $\left[\tau^0, \tau^1\right]$ in the middle of the phase under consideration. By Lemma~\ref{label:general:lm:data:structure},  the call to REBUILD($j$) after edge-update $\tau^1$ takes $O(|A| \cdot d \cdot 2^{-j})$ time. We will show that the epoch lasts for $\Omega(\epsilon \gamma \delta \cdot (L_d L)^{-2} \cdot |A| \cdot (d/2^{j}))$ edge-updates (see Corollary~\ref{label:general:cor:lm:bound:laminar:2}).  Note that by definition, no call is made to REBUILD($j$) during the epoch. Hence, dividing the running time of the call that ends  the epoch by the number of edge-updates in the epoch gives us an amortized bound of $O(L_d^2 L^2/\epsilon \gamma \delta)$.

Thus, the main challenge is to lower bound the number of edge-updates in the epoch. Towards this end, we take a closer look at how the sets of $l$-dirty nodes at layers $j' \in [0, j]$ evolve with the passage of time. First, note the properties of the laminar structure specified by Invariant~\ref{label:general:inv:laminar:contain}. Since a call was made to REBUILD($j'$) with $j' \in [1, j]$ after the edge-update $(\tau^0-1)$, Lemma~\ref{label:general:lm:rebuild:minor:1}  implies that $D_{l0} \subseteq D_{l1} \subseteq \cdots \subseteq D_{l,j-1} = D_{l,j} \subseteq A$ when the epoch begins. In the middle of the epoch, the sets $\{D_{l,k}\}, k \in [0, j],$ can only get bigger with the passage of time (see Lemma~\ref{label:general:lm:monotone:laminar}), but they always remain contained within one another as a laminar family. To be more specific, suppose that while processing some edge-update in the epoch, we have to change the sets $D_{l0}, \ldots, D_{lj}$. This change has to be one of the following two types.
\begin{itemize}
\item (a) Some node $v \in A$ becomes $c$-dirty, and to satisfy Invariant~\ref{label:general:inv:laminar:contain} we move the node $v$ from $C_{lk}$ to $D_{lk}$ at all layers $k \in [0, j]$. See item (1) in Step III of Section~\ref{label:general:sec:phase:middle}.
\item (b) Some node $v \in C_{lj'}$ becomes $l$-dirty at layer $j' \in [1, j]$, and  we move the node $v$ from $C_{lk}$ to $D_{lk}$ at all layers $k \in [j', j]$. See the call to CLEANUP($x$) in Step III of Section~\ref{label:general:sec:phase:middle}.
\end{itemize}
To summarize, during the epoch the $l$-dirty sets $D_{l0} \subseteq D_{l1} \subseteq \cdots \subseteq D_{lj}$ keep getting bigger and bigger with the passage of time, without violating the laminar property. In other words, no node is ever deleted from one of these sets while the epoch is still in progress. 

The epoch ends because a call is made to REBUILD($j$) after edge-update $\tau^1$. This can happen only if at the end of the epoch, Invariant~\ref{label:general:inv:laminar} is violated at layer $j$ but satisfied at all the layers $k \in [0, j-1]$. See the call to the subroutine VERIFY() in Step III (Figure~\ref{label:general:fig:insert}) of Section~\ref{label:general:sec:phase:middle}. Thus, at the end of the epoch, we have $|D_{lj}| > (\delta (j+1)/(L_d+1)) \cdot |A|$ and $|D_{l,j-1}| \leq (\delta j/(L_d+1)) \cdot |A|$. Since $D_{l,j-1} \subseteq D_{l,j}$, we get:

\begin{lemma}
\label{label:general:lm:bound:laminar:0}
At the end of the epoch $\left[\tau^0, \tau^1\right]$, we have $|D_{lj} \setminus D_{l,j-1}| \geq (\delta /(L_d+1)) \cdot |A|$. 
\end{lemma}

Next, note that if a node $v$ belongs to $D_{lj} \setminus D_{l,j-1}$ at the end of the epoch, then the node $v$ must have been part of  the set $C_{lj}$ in the beginning of the epoch. This follows from the three observations stated below, which, in turn, follow from our discussion preceding Lemma~\ref{label:general:lm:bound:laminar:0}.
\begin{enumerate}
\item We always have $D_{l,j-1} \subseteq D_{lj} \subseteq A$ and $C_{lj} = A \setminus D_{lj}$. The set $A$ does not change during the epoch.
\item In the beginning of the epoch, we have $D_{l,j-1} = D_{lj}$. 
\item During the epoch, a node is never deleted from either of the sets $D_{l,j-1}$ and $D_{l,j}$.
\end{enumerate}

\noindent We formally state our observation in Lemma~\ref{label:general:lm:bound:laminar:-1}.

\begin{lemma}
\label{label:general:lm:bound:laminar:-1}
If a node $v$ belongs to $D_{lj} \setminus D_{l,j-1}$ at the end of the epoch $\left[\tau^0, \tau^1\right]$, then the node $v$ must have belonged to $C_{lj}$ in the beginning of the epoch.
\end{lemma}

The next lemma follows directly from Lemmas~\ref{label:general:lm:bound:laminar:0} and~\ref{label:general:lm:bound:laminar:-1}.

\begin{lemma}
\label{label:general:lm:bound:laminar:1}
There are at least $\Omega(\delta \cdot L_d^{-1} \cdot |A|)$ nodes  that belong to $C_{lj}$ when the epoch $\left[ \tau^0, \tau^1\right]$ begins and belong to $D_{lj} \setminus D_{l,j-1}$ when the epoch ends. 
\end{lemma}

Next, we will show that  a node $v$ moves from $C_{lj}$ to $D_{lj} \setminus D_{l,j-1}$  only after a large number of  edge-updates incident upon $v$. The complete proof of Lemma~\ref{label:general:lm:bound:laminar:2} appears in Section~\ref{label:general:sec:proof:lm:bound:laminar:2}. The main idea behind the proof, however, is simple. Consider a node $v$ that belongs to $C_{lj}$ when the epoch begins. Since a call was made to REBUILD($j'$) with  $j' \in [1, j]$ just before the start of the epoch, Lemma~\ref{label:general:lm:fig:reduce} implies that $\text{deg}(v, H_j)$ is very close to $(1/2) \cdot \text{deg}(v, H_{j-1})$ at that moment. On the other hand, the node $v$ moves from $C_{lj}$ to $D_{lj} \setminus D_{l,j-1}$ only if somewhere in the middle of the epoch it violates Invariant~\ref{label:general:inv:laminar:clean} in layer $j$, and this means that at that moment $\text{deg}(v, H_{j})$ is very far away from $(1/2) \cdot \text{deg}(v, H_{j-1})$. So $\text{deg}(v, H_j)$ moves from being close to $(1/2) \cdot \text{deg}(v, H_{j-1})$ to being far away from $(1/2) \cdot \text{deg}(v, H_{j-1})$ within the given epoch. This can happen only if either $\text{deg}(v, H_j)$ or $\text{deg}(v, H_{j-1})$ changes by a large amount during the epoch. In either case, we can show that a large number of edge-updates incident upon $v$ takes place during the epoch. 

\begin{lemma}
\label{label:general:lm:bound:laminar:2}
Take any node $v$ that is part of $C_{lj}$ when the epoch $\left[\tau^0, \tau^1\right]$ begins, and is part of $D_{lj} \setminus D_{l,j-1}$ when the epoch ends. At least $\Omega( \epsilon \gamma \cdot (L_d L^2)^{-1} \cdot (d/2^j))$ edge-updates in the epoch are incident upon   $v$.
\end{lemma}

\begin{corollary}
\label{label:general:cor:lm:bound:laminar:2}
The epoch $\left[\tau^0, \tau^1\right]$ lasts for at least $\Omega( \epsilon \gamma \delta \cdot (L_d L)^{-2} \cdot (d/2^j) \cdot |A|)$ edge-updates.
\end{corollary}

\begin{proof}
Let $C^*$ be the set of nodes that are part of $C_{lj}$ when the epoch begins and part of $D_{lj} \setminus D_{l,j-1}$ when the epoch ends. By Lemmas~\ref{label:general:lm:bound:laminar:1},~\ref{label:general:lm:bound:laminar:2}, at least $\Omega( \epsilon \gamma \delta \cdot L_d^{-2} \cdot L^{-2} \cdot |A| \cdot d \cdot 2^{-j})$ edge-updates in the epoch are incident upon the nodes in $C^*$. This lower bounds  the total number of edge-updates  during the epoch.  
\end{proof}

We are now ready to derive the main result of this section.

\begin{lemma}
\label{label:general:lm:bound:laminar:main}
\label{label:general:lm:amortized:runtime:rebuild}
The amortized update time of REBUILD($j$) in the middle of a phase is $O(L^2 L_d^2/\epsilon \gamma \delta)$. 
\end{lemma}

\begin{proof}
Note that no call is made to REBUILD($j$) during the epoch, and  the call to REBUILD($j$) after edge-update $\tau^1$ requires $O(|A| \cdot d \cdot 2^{-j})$ time (see Lemma~\ref{label:general:lm:data:structure}). The lemma now follows from Corollary~\ref{label:general:cor:lm:bound:laminar:2}. 
\end{proof}

\subsubsection{Proof of Lemma~\ref{label:general:lm:bound:laminar:2}} 
\label{label:general:sec:proof:lm:bound:laminar:2}

\paragraph{Notations.}  We introduce some notations that will be used throughout the proof. 
\begin{itemize}
\item Let $x_k$ be the value of $\text{deg}(v, H_k)$, for $k \in [0, L_d]$, when the epoch begins. Similarly,  let $x$ be the value of $\text{deg}(v, E)$ when the epoch begins. 
\item Consider the unique time-instant in the epoch when the node $v$ is moved from $C_{l,j}$ to  $D_{lj} \setminus D_{l,j-1}$. For $k \in [0, L_d]$, let $x_k + \Delta_k$ be the value of $\text{deg}(v, H_k)$ at this time-instant, where $\Delta_k$ is some integer. 
\end{itemize}

\paragraph{Two simple observations.}
Just before the epoch begins, a call is made to REBUILD($j'$) for some $j' \in [1,j]$. At the end of the call, $\text{deg}(v, H_j)$ equals $1/2$ times $\text{deg}(v, H_{j-1})$,  plus-minus one (see Lemma~\ref{label:general:lm:fig:reduce}). Thus, we have:
\begin{equation}
\label{label:general:eq:lm:bound:laminar:1}
\left(\frac{x_{j-1}}{2}\right) - 1 \leq x_j \leq  \left(\frac{x_{j-1}}{2}\right) + 1
\end{equation}

Consider the unique time-instant in the epoch when the node $v$ is moved from $C_{lj}$ to $D_{lj} \setminus D_{l,j-1}$. This event can take place only if the node $v$ was violating Invariant~\ref{label:general:inv:laminar:clean} at layer $j$ at that instant. Thus, we have:

\begin{equation}
\label{label:general:eq:lm:bound:laminar:2}
x_j  + \Delta_j \notin \left[ (1+\gamma/L_d)^{-1} \left( \frac{x_{j-1} + \Delta_{j-1}}{2} \right), (1+\gamma/L_d) \left( \frac{x_{j-1} + \Delta_{j-1}}{2} \right) \right]
\end{equation}

\paragraph{The main idea.}  When the epoch begins, equation~\ref{label:general:eq:lm:bound:laminar:1} implies that $\text{deg}(v, H_j)$ is very close to $(1/2) \cdot \text{deg}(v, H_{j-1})$. In contrast, equation~\ref{label:general:eq:lm:bound:laminar:2} implies that somewhere in the middle of the epoch, $\text{deg}(v, H_j)$ is quite far away from $(1/2) \cdot \text{deg}(v, H_{j-1})$. Intuitively, this can happen only if either $\text{deg}(v, H_j)$ or $\text{deg}(v, H_{j-1})$ has changed by a large amount during this interval, i.e., either $|\Delta_j|$ or $|\Delta_{j-1}|$ is large. This is shown in Claim~\ref{label:general:cl:y} and Corollary~\ref{label:general:cor:cl:y}. Finally, in Claim~\ref{label:general:cl:final} and Corollary~\ref{label:general:cor:cl:final}, we show that this can happen only if the degree of $v$ in the graph $G = (V, E)$ itself has changed by a large amount, which means that a large number of edge-updates incident upon $v$ have taken place during the interval under consideration. Lemma~\ref{label:general:lm:bound:laminar:2} follows from Corollary~\ref{label:general:cor:cl:final}.


\medskip
To prove Claim~\ref{label:general:cl:y} and Corollary~\ref{label:general:cor:cl:y}, we  first need to show that  $x_{j-1}$ is not much smaller than $d/2^{j}$.

\begin{claim}
\label{label:general:cl:epoch:simple}
We have $x_{j-1} \geq (2 \epsilon \cdot e^{-\gamma} L^{-2}) \cdot (d/2^{j})$.
\end{claim}

\begin{proof}
When the epoch begins,  the node $v$ belongs to the set $C_{lj}$, and we know that $ C_{lj} \subseteq C_{l,j-1}$ (see Invariant~\ref{label:general:inv:laminar:contain}). Hence, when the epoch begins, the node $v$ is part of $C_{l,j-1}$. Accordingly, by Lemma~\ref{label:general:lm:clean}:
\begin{equation}
\label{label:general:eq:epoch:100}
x_{j-1} \geq \frac{x}{2^{j-1} \cdot (1+\gamma /L_d)^{j-1}}
\end{equation}

We now lower bound $x$. When the epoch begins, the node $v$ belongs to the set $C_{lj}$, and we know that $C_{lj} \subseteq A \setminus D_{l0} = A \setminus D_c$ (see Invariant~\ref{label:general:inv:laminar:contain}). Hence, by the condition (1) in Definition~\ref{label:general:def:critical:structure}:\begin{equation}
\label{label:general:eq:epoch:101}
x > \epsilon d/L^2
\end{equation} 
From equations~\ref{label:general:eq:epoch:100} and~\ref{label:general:eq:epoch:101} we infer that:
\begin{equation}
\label{label:general:eq:epoch:102}
x_{j-1} \geq \frac{\epsilon \cdot d}{2^{j-1} \cdot (1+\gamma /L_d)^{j-1} \cdot L^2}
\end{equation}
Since $(1+\gamma /L_d)^{j-1} < (1+\gamma/L_d)^{L_d} \leq e^{\gamma}$, equation~\ref{label:general:eq:epoch:102} implies that:
\begin{eqnarray*}
x_{j-1} \geq \left(\frac{\epsilon}{e^{\gamma} L^2}\right) \cdot \left(\frac{d}{2^{j-1}}\right) = \left(\frac{2 \epsilon}{e^{\gamma} L^2}\right) \cdot \left(\frac{d}{2^{j}}\right)\end{eqnarray*}
This concludes the proof of the claim.
\end{proof}

\begin{corollary}
\label{label:general:cor:cl:epoch:simple}
We have $x_{j-1} \geq 8 L_d/\gamma$.
\end{corollary}

\begin{proof}
Form Claim~\ref{label:general:cl:epoch:simple}, we infer that:
\begin{eqnarray}
x_{j-1} & \geq & \left(\frac{\epsilon}{e^{\gamma} L^2}\right) \cdot \left(\frac{d}{2^{j-1}}\right) \nonumber \\
& \geq & \left(\frac{\epsilon}{e^{\gamma} L^2}\right) \cdot \left(\frac{d}{2^{L_d}}\right) \label{label:general:eq:cor:cl:epoch:simple:1} \\
& = & \frac{\epsilon \cdot \lambda_d \cdot L^2}{e^{\gamma}} \label{label:general:eq:cor:cl:epoch:simple:2} \\
& \geq & \frac{8 L_d}{\gamma} \label{label:general:eq:cor:cl:epoch:simple:3}
\end{eqnarray}
Equation~\ref{label:general:eq:cor:cl:epoch:simple:1} holds since $j \leq L_d$. Equation~\ref{label:general:eq:cor:cl:epoch:simple:2} follows from equation~\ref{label:general:eq:Li:lambda}. Equation~\ref{label:general:eq:cor:cl:epoch:simple:3} follows from equation~\ref{label:general:eq:Li:epoch}.
\end{proof}

\begin{claim}
\label{label:general:cl:y}
Either $|\Delta_j| \geq \gamma \cdot (16 L_d)^{-1} \cdot x_{j-1}$ or $|\Delta_{j-1}| \geq \gamma \cdot (16 L_d)^{-1} \cdot x_{j-1}$. 
\end{claim}

\begin{proof}
Throughout the proof, we set $\mu = (1+ \gamma/L_d)$. Looking at equation~\ref{label:general:eq:lm:bound:laminar:2}, we consider two  cases.

\medskip
\noindent {\em Case 1. $x_j + \Delta_j > (1/2) \cdot \mu \cdot \left(x_{j-1}+\Delta_{j-1}\right).$} 

\medskip
\noindent In this case, since $x_j \leq x_{j-1}/2+1$ (see equation~\ref{label:general:eq:lm:bound:laminar:1}), we get:
\begin{eqnarray}
\Delta_j \geq \mu \cdot \left(\frac{x_{j-1}+\Delta_{j-1}}{2}\right) - \left(\frac{x_{j-1}}{2}\right) - 1 \geq \frac{\mu \cdot \Delta_{j-1}}{2} + (\mu -1) \cdot \left(\frac{x_{j-1}}{2}\right) -1 \nonumber
\end{eqnarray}
Rearranging the terms in the above inequality, we get:
\begin{eqnarray}
\Delta_j - \frac{\mu \cdot \Delta_{j-1}}{2} & \geq & (\mu -1) \cdot \left(\frac{x_{j-1}}{2}\right) - 1 \nonumber \\
& \geq & (\mu -1) \cdot \left(\frac{x_{j-1}}{2}\right) - (\mu -1) \cdot \left(\frac{x_{j-1}}{4}\right) \label{label:general:eq:funny:3} \\
& = & (\mu -1) \cdot \left(\frac{x_{j-1}}{4}\right) \label{label:general:eq:funny:4}
\end{eqnarray}
Equation~\ref{label:general:eq:funny:3} holds since $x_{j-1} > 8 L_d/\gamma$ (see Corollary~\ref{label:general:cor:cl:epoch:simple}) and $(\mu - 1) = \gamma/L_d$. Thus, from equation~\ref{label:general:eq:funny:4}, we infer that:
\begin{equation}
\label{label:general:eq:tedious:1}
|\Delta_j| + \left| \frac{\mu \cdot \Delta_{j-1}}{2} \right| \geq  (\gamma/L_d) \cdot \left( \frac{x_{j-1}}{4} \right)
\end{equation}
Hence, either $|\Delta_j|$ is at least $1/2$ times the right hand side of equation~\ref{label:general:eq:tedious:1}, or else $|\Delta_{j-1}|$ is at least $1/\mu$ times the right hand side of equation~\ref{label:general:eq:tedious:1}. The claim follows since $\mu = 1+\gamma/L_d \leq 2$.

\bigskip
\noindent {\em Case 2. $x_j + \Delta_j < (1/2) \cdot (1/\mu) \cdot \left(x_{j-1}+\Delta_{j-1}\right)$.} 

\medskip
\noindent In this case, since $x_j \geq x_{j-1}/2-1$ (see equation~\ref{label:general:eq:lm:bound:laminar:1}), we get:
\begin{eqnarray}
\Delta_j \leq   \left(\frac{x_{j-1}+\Delta_{j-1}}{2\mu}\right) - \left(\frac{x_{j-1}}{2}\right) + 1 \leq \left(\frac{\Delta_{j-1}}{2\mu}\right) + (1/\mu -1) \cdot \left(\frac{x_{j-1}}{2}\right) +1 \nonumber
\end{eqnarray}
Rearranging the terms in the above inequality, we get:
\begin{eqnarray}
\left(\frac{\Delta_{j-1}}{2\mu}\right) - \Delta_j & \geq & (1-1/\mu) \cdot \left(\frac{x_{j-1}}{2}\right) - 1 \nonumber \\
& \geq & \left(\frac{\mu -1}{\mu}\right) \cdot \left(\frac{x_{j-1}}{2}\right) - (\mu -1) \cdot \left(\frac{x_{j-1}}{4}\right) \label{label:general:eq:funny:10} \\
& \geq & \left(\frac{\mu -1}{\mu}\right) \cdot \left(\frac{x_{j-1}}{4}\right) \label{label:general:eq:funny:11}
\end{eqnarray}
Equations~\ref{label:general:eq:funny:10} and~\ref{label:general:eq:funny:11} hold since $\mu = 1+\gamma/L_d$ and $x_{j-1} > 8L_d/\gamma = 8/(\mu -1)$ (see Corollary~\ref{label:general:cor:cl:epoch:simple}). From equation~\ref{label:general:eq:funny:11}, we infer that:
\begin{equation}
\label{label:general:eq:one}
|\Delta_j| + \left| \frac{\Delta_{j-1}}{2\mu} \right| \geq \left(\frac{\mu -1}{\mu}\right) \cdot \left(\frac{x_{j-1}}{4}\right) = (\gamma/L_d) \cdot \left(\frac{x_{j-1}}{4 \mu}\right)
\end{equation}
Since $\mu = 1+\gamma/L_d \leq 2$, from equation~\ref{label:general:eq:one} we have:
\begin{equation}
\label{label:general:eq:two}
|\Delta_j| + \left| \frac{\Delta_{j-1}}{2\mu} \right| \geq (\gamma/L_d) \cdot \left(\frac{x_{j-1}}{8}\right)
\end{equation}
Hence, either $|\Delta_j|$ is at least $1/2$ times the right hand side of equation~\ref{label:general:eq:tedious:1}, or else $|\Delta_{j-1}|$ is at least $\mu$ times the right hand side of equation~\ref{label:general:eq:tedious:1}. The claim follows since $\mu = 1+\gamma/L_d \geq 1$.   
\end{proof}

\begin{corollary}
\label{label:general:cor:cl:y}
Either $|\Delta_j| \geq \epsilon \gamma \cdot (8 e^{\gamma} L_d L^2)^{-1} \cdot   (d/2^{j})$ or $|\Delta_{j-1}| \geq \epsilon \gamma \cdot (8 e^{\gamma} L_d L^2)^{-1} \cdot (d/2^{j})$. 
\end{corollary}

\begin{proof}
Follows from Claims~\ref{label:general:cl:epoch:simple} and~\ref{label:general:cl:y}.
\end{proof}

\begin{claim}
\label{label:general:cl:final}
At least $\max(|\Delta_j|, |\Delta_{j-1}|)$ edge-updates in the epoch $\left[\tau^1, \tau^0\right]$ are incident upon the node $v$.
\end{claim}

\begin{proof}
The proof consists of two steps.
\begin{itemize}
\item {\em Step 1. We show that at least $|\Delta_{j}|$ edge-updates in the epoch are incident upon the node $v$.}

To prove this step, note that $j \in [1, L_d]$. By definition, no call is made to REBUILD($j'$) with $j' \in [1, j]$ during the epoch. Hence, no edge is inserted into $H_j$ during the epoch. Further, during the epoch, an edge $e \in H_j$ can get deleted from $H_j$ only if the edge gets deleted from the graph $G = (V, E)$ itself (see Lemma~\ref{label:general:lm:monotone:laminar}). Thus, during any interval within the epoch   $\text{deg}(v, H_j)$ can only decrease, and furthermore, the absolute value of this change in $\text{deg}(v, H_j)$ is at most  the number edge-updates incident upon $v$. 
\item {\em Step 2. We show that at least $|\Delta_{j-1}|$ edge-updates in the epoch are incident upon the node $v$.}

Here, we consider two possible cases.
\begin{itemize}
\item {\em Case 1. $j-1 \in [1, L_d]$.}

In this case, the argument is exactly similar to Case 1.
\item {\em Case 2. $j-1 = 0$.}

In this case, since the node $v$ is active, we have $\text{deg}(v, H_0) = \text{deg}(v, H) = \text{deg}(v, E)$ throughout the duration of the epoch. Hence, during any interval within the epoch, the absolute value of the change in $\text{deg}(v, H_0)$ is at most the number edge-updates incident upon $v$. 
\end{itemize}
\end{itemize}
\end{proof}

\begin{corollary}
\label{label:general:cor:cl:final}
At least $\Omega(\epsilon \gamma \cdot (L_d L^2)^{-1} \cdot (d/2^j))$ edge-updates in the epoch $\left[\tau^0, \tau^1\right]$ are incident upon $v$.
\end{corollary}

\begin{proof}
Note that $e^{\gamma} \in (1, e)$ since $\gamma \in (0,1)$. The corollary now follows from Corollary~\ref{label:general:cor:cl:y} and Claim~\ref{label:general:cl:final}.
\end{proof}

\newpage

\part{DYNAMIC ALGORITHM FOR BIPARTITE GRAPHS: FULL DETAILS}
\newpage

\section{Notations and Preliminaries}
\label{label:bipartite:sec:notation}

\newcommand{\W}{\mathcal{W}}

This part of the writeup  presents our dynamic algorithm for bipartite graphs for maintaining  a better than $2$ approximation to the size of the maximum matching. The main result is summarized in Theorem~\ref{label:bipartite:th:new:approx}.

\paragraph{Notations.} We now define some notations that will be used throughout this half of the paper.  Let $G = (V, E)$ denote the input graph with $n = |V|$ nodes and $m = |E|$ edges. Given any subset of edges $E' \subseteq E$ and any node $v \in V$, we let $N_v(E') = \{ u \in V : (u, v) \in E' \}$ denote the set of neighbors of $v$ that are connected to $v$ via an edge in $E'$. Furthermore, we let  $\text{deg}_v(E') = |N_v(E')|$ denote the number of edges in $E'$ that are incident upon $v$.  Finally, for any subset of edges $E' \subseteq E$, we let $V(E') = \{ v \in V : \text{deg}_v(E') > 0 \}$ denote the set of endpoints of the edges in $E'$. 

\paragraph{Fractional assignments.} A ``fractional assignment''  is a function $w : E \rightarrow \mathbf{R}^+$. It assigns a nonnegative   ``weight''  $w(e)$ to every edge $e \in E$ in the input graph $G = (V, E)$. The set of edges  $\text{Support}(w) = \{ e \in E : w(e) > 0 \}$ with positive weights is called the ``support'' of  the fractional assignment $w$. The ``weight'' of a node $v \in V$ under $w$ is given by $W_v(w) = \sum_{(u,v) \in E} w(u,v)$. In other words, $W_v(w)$ denotes the total weight received by $v$ from its incident edges, under the fractional assignment $w$. We now define the ``addition'' of two fractional assignments. Given any two fractional assignments $w, w'$, we say that $w + w'$ is a new fractional assignment such that $(w+w') (e) = w(e) + w'(e)$ for every edge $e \in E$. It is easy to check that this addition operation is commutative and associative, i.e., we have $w + w' = w' + w$ and $w + (w' + w'') = (w+w') + w''$ for any three fractional assignments $w, w', w''$ defined on the same graph. Given any subset of edges $E' \subseteq E$ and any fractional assignment $w$, we let $w(E') = \sum_{e \in E'} w(e)$ denote the sum of the weights of the edges in $E'$ under $w$.  The ``size'' of a fractional assignment $w$ is defined as $w(E) = \sum_{e \in E} w(e)$.

\paragraph{Fractional $b$-matchings.} Suppose that we assign a ``capacity'' $b_v \geq 0$ to each node $v \in V$ in the input graph $G = (V, E)$. A fractional assignment $w$ is called a ``fractional-$b$-matching'' with respect to these capacities iff $W_v(w) \leq b_v$ for every node $v \in V$. The size of this fractional $b$-matching is given by $w(E) = \sum_{e \in E} w(e)$, and its support is defined as $\text{Support}(w) = \{ e \in E : w(e) > 0 \}$. We say that a fractional $b$-matching $w$ is ``maximal'' iff for every edge $(u, v) \in E$, either $W_v(w) = b_v$ or $W_u(w) = b_u$.  

\paragraph{``Extending'' a fractional $b$-matching.} During the course of our algorithm, we will often consider some subgraph $G' = (V', E')$ of the input graph $G = (V, E)$, with $V' \subseteq V$ and $E' \subseteq E$, and define a  fractional $b$-matching $w' : E' \rightarrow \mathbf{R}^+$ on $G'$ with respect to the node-capacities $\{ b'(v) \}, v \in V'$.  In such cases, to ease notation, we will often pretend that $w'$ is a fractional $b$-matching on $G$ itself. We will do this by setting $w'(e) = 0$ for all edges $e \in E \setminus E'$ and $b'(v) = 0$ for all nodes $v \in V \setminus V'$. With this notational convention, we will be able to ``add'' two fractional assignments $w', w''$  defined on two different subgraphs $G', G''$ of $G$.

\paragraph{Fractional matchings.} A fractional assignment $w$ in the input graph $G = (V, E)$ is called a ``fractional matching'' iff we have $W_v(w) \in [0, 1]$ for all nodes $v \in V$. In other words, this is a fractional $b$-matching where every node has capacity one. The ``size'' of this fractional matching is given by $w(E) = \sum_{e \in E} w(e)$, and its support is defined as $\text{Support}(w) = \{ e \in E : w(e) > 0 \}$. We say that a fractional matching $w$ is ``maximal'' iff for every edge $(u, v) \in E$, either $W_v(w) = 1$ or $W_u(w) = 1$.

\paragraph{(Integral) matchings.}  In the input graph $G = (V, E)$, an (integral) matching $M \subseteq E$ is a subset of edges that have no common endpoints. This can be thought of as a fractional matching $w_M : E \rightarrow [0,1]$ with $\text{Support}(w_M) = M$ such that $w_M(e) = 1$ for all edges $e \in M$. The size of this matching is given by $|M| = w_M(E)$. We say that the matching $M$ is ``maximal'' iff for every edge $(u, v) \in E \setminus M$, at least one of its endpoints $\{u, v \}$ is matched under $M$. This is equivalent to the condition that the corresponding fractional matching $w_M$ is maximal.  We will need the following theorem, which shows that the maximum size of a fractional matching is no more than the maximum size of an integral matching in a bipartite graph. 
\begin{theorem}
\label{label:bipartite:th:main:structure}
In an input graph $G = (V, E)$, consider any  fractional matching $w$ with support $E' \subseteq E$. If the graph $G$ is bipartite, then there is a matching $M \subseteq E'$  of size at least $w(E)$. 
\end{theorem}

\paragraph{Our result.}  In Section~\ref{label:bipartite:sec:invariants}, we first fix three parameters $\epsilon, \delta, K$ as per equations~\ref{label:bipartite:eq:parameter:1}~---~\ref{label:bipartite:eq:parameter:7} (we assume that the number of nodes $n$ in the input graph is sufficiently large).  Then we describe some invariants that define three fractional assignments $w, w^r, w_1^*$ in $G = (V, E)$. In Theorem~\ref{label:bipartite:th:feasible:main}, we show that $(w+w^r+w_1^*)$ forms a fractional matching in $G = (V, E)$. In Theorem~\ref{label:bipartite:th:approx:main}, we show that the size of $(w+w^r+w_1^*)$ is a $(1/f)$-approximation to the maximum possible size of a fractional matching in $G = (V, E)$, where:
$$f = (1/2) \cdot \left( \frac{(1+\delta/3)}{(1+\epsilon)} - 4 K \epsilon\right).$$
In Section~\ref{label:bipartite:sec:algo}, we show how to maintain a $(1+\epsilon)^2$-approximation to the size of $(w+w^r+w^*_1)$ with $\kappa(n) = O((10/\epsilon)^{K+8} \cdot n^{2/K})$ amortized update time (see Theorem~\ref{label:bipartite:th:runtime:main:main}).  This implies that we can maintain a $(1+\epsilon)^2/f$-approximation to the size of the maximum fractional matching in $G$ with $\kappa(n)$ update time. By Theorem~\ref{label:bipartite:th:main:structure}, the size of the maximum fractional matching in $G$ equals the size of the maximum cardinality (integral) matching in $G$. Thus, we reach the following conclusion.
\begin{itemize}
\item In $O((10/\epsilon)^{K+8} \cdot n^{2/K})$ amortized update time, we can maintain a $(1+\epsilon)^2/f$-approximation to the size of the maximum matching in $G$. 
\end{itemize}
We now set the values of $\epsilon, \delta$, and define two new parameters $\alpha_K, \beta_K$ as follows.
\begin{eqnarray}
\label{label:bipartite:eq:upload:1}
\alpha_K  = \frac{2 \cdot (1+\epsilon)^3}{(1+\delta/3) - 4 K\epsilon(1+\epsilon)} \text{ and }  \beta_K = (10/\epsilon)^{K+8}, \text{ where } \epsilon = \frac{1}{36 K \cdot 10^{K+4}} \text{ and }  \delta = 10^4 \cdot K \epsilon.
\end{eqnarray}
It is easy to check that this setting of values satisfies equations~\ref{label:bipartite:eq:parameter:1} --~\ref{label:bipartite:eq:parameter:7}. Further, note that $\alpha_K = (1+\epsilon)^2/f$. Thus, we get the following theorem:
\begin{theorem}
\label{label:bipartite:th:new:approx}
Fix any  positive integer  $K$ and define $\alpha_K, \beta_K$ as  per equation~\ref{label:bipartite:eq:upload:1}. In a dynamic setting, we can  maintain a $\alpha_K$-approximation to the value of the maximum matching in a bipartite graph $G = (V, E)$ with $O(\beta_K \cdot n^{2/K})$ amortized update time. Note that $1 \leq \alpha_K < 2$ for every sufficiently large integer $K$.
\end{theorem}

\subsection{An important technical theorem}
\label{label:bipartite:sec:thm:critical}

We devote this section to the proof of the following theorem. This result will be crucially used later on in the design and analysis of our dynamic algorithm. 

\begin{theorem}
\label{label:bipartite:thm:critical}
Consider a bipartite graph $G = (V, E)$, where the node-set $V$ is partitioned into two subsets $A \subseteq V$ and $B = V \setminus A$ such that every edge $e \in E$ has one endpoint in $A$ and another endpoint in $B$.  Fix any number $\lambda \in [0, 1/2]$, and suppose that each node $v \in B$ has a capacity $b(v) = 2 \lambda$. Furthermore, suppose that each node $u \in A$ has a capacity $b(u) \in [0, \lambda]$.  Let $w$ be a maximal fractional $b$-matching in the graph $G$ with respect to these capacities. Thus, for every edge $(u, v) \in E$, we have either $W_u(w) = b(u)$ or $W_v(w) = b(v)$. Further, for every node $v \in V$, we have $0 \leq W_v(w) \leq b(v)$.

  Let $M \subseteq E$ be a matching in $G$, i.e., no two edges in $M$ share a common endpoint.   Finally, let  $A(M) = \{ u \in A : \text{deg}_M(u) = 1 \}$ and $B(M) = \{ v \in B : \text{deg}_M(v) = 1\}$ respectively denote the set of nodes from $A$ and $B$ that are matched under $M$.  Then we have:
$$\sum_{v \in V} W_v(w) \geq (4/3)  \cdot \sum_{u \in A(M)} b(u).$$
\end{theorem}

\medskip
\noindent {\bf Proof of Theorem~\ref{label:bipartite:thm:critical}.} \ 

\medskip
\noindent For each node $u \in A(M)$, define $\Delta_u(w) = b(u) - W_u(w)$ to be the ``slack'' at node $u$ with respect to $w$. Since for all nodes $u \in A(M)$, we have $0 \leq W_u(w) \leq b(u) \leq \lambda$, we infer that $\Delta_u(w) \in [0, \lambda]$ for all nodes $u \in A(M)$. We will show that:
\begin{equation}
\label{label:bipartite:eq:slack}
\sum_{u \in A(M)} \Delta_u(w) \leq (1/2) \cdot \sum_{u \in A} W_u(w)
\end{equation} 
We claim that equation~\ref{label:bipartite:eq:slack} implies Theorem~\ref{label:bipartite:thm:critical}. To see why this is true,  consider two possible cases.
\begin{itemize}
\item {\em Case 1.} $\sum_{u \in A(M)} \Delta_u(w) \leq (1/3) \cdot \sum_{u \in A(M)} b(u)$. 

In this case, we have:
\begin{eqnarray}
\sum_{u \in A(M)} W_u(w) & = & \sum_{u \in A(M)} \left( b(u) - \Delta_u(w) \right) \nonumber \\
& = & \sum_{u \in A(M)} b(u) - \sum_{u \in A(M)} \Delta_u(w) \nonumber \\
& \geq & \sum_{u \in A(M)} b(u) - (1/3) \cdot \sum_{u \in A(M)} b(u) \nonumber \\
& = & (2/3) \cdot \sum_{u \in A(M)} b(u) \label{label:bipartite:eq:lm:slack:1}
\end{eqnarray}
Now, since an edge $e \in E$ contributes the same amount $w(e)$ to each of the sums $\sum_{u \in A} W_u(w)$ and $\sum_{v \in B} W_v(w)$, we have $\sum_{u \in A} W_u(w) = \sum_{v \in B} W_v(w)$. Thus, we get:
\begin{eqnarray}
\sum_{v \in V} W_v(w) & = & 2 \cdot \sum_{u \in A} W_u(w) \nonumber \\
& \geq & 2 \cdot \sum_{u \in A(M)} W_u(w) \label{label:bipartite:eq:lm:slack:2} \\
& \geq & (4/3) \cdot \sum_{u \in A(M)} b(u) \label{label:bipartite:eq:lm:slack:3}
\end{eqnarray}
Equation~\ref{label:bipartite:eq:lm:slack:2} holds since $A(M) \subseteq A$. Equation~\ref{label:bipartite:eq:lm:slack:3} follows from equation~\ref{label:bipartite:eq:lm:slack:1}. The theorem follows from equation~\ref{label:bipartite:eq:lm:slack:3}. 

\item {\em Case 2.} $\sum_{u \in A(M)} \Delta_u(w) > (1/3) \cdot \sum_{u \in A(M)} b(u)$. 

In this case, we have:
\begin{eqnarray}
\sum_{u \in A} W_u(w) & \geq & 2 \cdot \sum_{u \in A(M)} \Delta_u(w) \label{label:bipartite:eq:lm:slack:4} \\
& > & (2/3) \cdot \sum_{u \in A(M)} b(u) \label{label:bipartite:eq:lm:slack:5}
\end{eqnarray}
Equation~\ref{label:bipartite:eq:lm:slack:4} follows from equation~\ref{label:bipartite:eq:slack}. Next, just as in Case 1, we argue that an edge $e \in E$ contributes the same amount $w(e)$ to each of the sums $\sum_{u \in A} W_u(w)$ and $\sum_{v \in B} W_v(w)$. So we have: $\sum_{u \in A} W_u(w) = \sum_{v \in B} W_v(w)$. Thus, we get:
\begin{eqnarray}
\sum_{v \in V} W_v(w) & = & 2 \cdot \sum_{u \in A} W_u(w) \nonumber \\
& > & (4/3) \cdot \sum_{u \in A(M)} b(u) \label{label:bipartite:eq:lm:slack:6}
\end{eqnarray}
Equation~\ref{label:bipartite:eq:lm:slack:6} follows from equation~\ref{label:bipartite:eq:lm:slack:5}. The theorem follows from equation~\ref{label:bipartite:eq:lm:slack:6}.
\end{itemize}

\noindent Thus, in order to prove the theorem, it suffices to prove equation~\ref{label:bipartite:eq:slack}.  This is shown below.
\begin{itemize}
\item {\em Proof of equation~\ref{label:bipartite:eq:slack}.}

Let $A^*(M) = \{ u \in A(M) : \Delta_u(w) > 0 \}$ denote the subset of nodes in $A(M)$ that have nonzero slack under $w$. For each node $u \in A(M)$, let $u(M) \in B(M)$ denote the node $u$ is matched to under $M$. Since $w$ is a maximal fractional $b$-matching with respect to the capacities $\{ b(v) \}, v \in V$, we infer that $W_{u(M)}(w) = b(u(M)) = 2\lambda$ for all nodes $u \in A^*(M)$. Further, we note that $\Delta_u(w) = b(u) - W_u(w)\leq b(u) \leq \lambda$ for all nodes $u \in A^*(M) \subseteq A$. Thus, we get:
\begin{equation}
\label{label:bipartite:eq:lm:slack:7}
\Delta_u(w) \leq (1/2) \cdot W_{u(M)}(w) \text{ for all nodes } u \in A^*(M). 
\end{equation}
From equation~\ref{label:bipartite:eq:lm:slack:7}, we infer that:
\begin{eqnarray*}
\sum_{u \in A(M)} \Delta_u(w) & = & \sum_{u \in A(M) \setminus A^*(M)} \Delta_u(w) + \sum_{u \in A^*(M)} \Delta_u(w) \\
& = & \sum_{u \in A^*(M)} \Delta_u(w) \\
& \leq & (1/2) \cdot \sum_{u \in A^*(M)} W_{u(M)}(w) \\
& \leq & (1/2) \cdot \sum_{v \in B} W_{v}(w) \\
& = & (1/2) \cdot \sum_{u \in A} W_u(w)
\end{eqnarray*}
The last equality holds since each edge $e \in E$ contributes the same amount $w(e)$ towards the sums $\sum_{v \in B} W_v(w)$ and $\sum_{u \in A} W_u(w)$, and so these two sums are equal. This concludes the proof of equation~\ref{label:bipartite:eq:slack}. 
\end{itemize}

\section{Invariants maintained by our algorithm}
\label{label:bipartite:sec:invariants}

Throughout the rest of this half of the paper, we fix a sufficiently large integral constant $K \geq 10$, and sufficiently small constants $\epsilon, \delta \in (0,1)$. We will assume that $K, \epsilon, \delta$ satisfy the following guarantees.
\begin{eqnarray}
\label{label:bipartite:eq:parameter:1}
\epsilon = 1/N \text{ for some positive integer } N \gg K,  \text{ and } \delta < \frac{1}{36 \cdot 10^{K-2}}. \\
\epsilon \ll \delta, \text{ and } \delta \text{ is an integral multiple of } \epsilon. \label{label:bipartite:eq:parameter:2} \\
\delta < \frac{1}{13 \cdot 10^K} \label{label:bipartite:eq:parameter:3} \\
\epsilon \gg 1/n^{1/K}, \text{ where } n = |V| \text{ is the number of nodes in the input graph.} \label{label:bipartite:eq:parameter:4} \\
K < n^{1/K} \label{label:bipartite:eq:parameter:5} \\
\sum_{j = 0}^{i} n^{j/K} \leq  2 \cdot n^{i/K} \text{ for all } i \in \{1, \ldots, K\}.  \label{label:bipartite:eq:parameter:6} \\
\log n \leq n^{1/K} \label{label:bipartite:eq:parameter:7}.
\end{eqnarray}
We  maintain a family of $K$ subgraphs $G_1, \ldots, G_K$ of the input graph $G = (V, E)$. For $1 \leq i \leq K$, let $Z_i$ and $E_i$ respectively denote the node-set and the edge-set of the $i^{th}$ subgraph, i.e.,  $G_i = (Z_i, E_i)$. We ensure that:
\begin{eqnarray*}
& (a) & V = Z_K \supseteq Z_{K-1} \supseteq \cdots \supseteq Z_1.  \\
& (b)  & \text{For  } 1\leq  i \leq  K, \text{ the set } E_i = \{ (u, v) \in E : u, v\in Z_i\} \text{ consists of  the edges with both endpoints in } Z_i. 
\end{eqnarray*}


\noindent We  define the ``level'' of a node $v \in V$ to be the minimum index $i \in \{1, \ldots, K\}$ such that $v \in Z_i$. We denote the level of a node $v$ by $\ell(v)$. Thus, we have:
\begin{equation}
\label{label:bipartite:eq:label}
\ell(v) = \min_{i \in [1, K]} \{ v \in Z_i \} \ \ \text{ for all nodes } v \in V. 
\end{equation}
The corollary below follows from our definition of the level of a node. 

\begin{corollary}
\label{label:bipartite:cor:level}
Consider any node $v \in V$ with level $\ell(v) = i \in \{1, \ldots, K\}$. We have:  
$$v \in Z_j \text{ for all } j \in \{i, \ldots, K\}, \text{ and }  v \notin Z_j \text{ for all } j \in \{1, \ldots, i-1\}.$$ 
\end{corollary}

\subsection{An overview of the structures maintained by our algorithm}
\label{label:bipartite:sec:overview}

For each index $i \in \{2, \ldots, K\}$, our algorithm will maintain the following structures. \\
\begin{enumerate}
\item The subgraph $G_i = (Z_i, E_i)$.  \\
\item For each node $v \in Z_i$, two capacities $b_i(v), b_i^r(v) \geq 0$. The former is called the ``normal'' capacity of the node at level $i$, whereas the latter is called the ``residual'' capacity of the node at level $i$.  \\
\item A fractional $b$-matching $w_i$ in the subgraph $G_i = (Z_i, E_i)$ with respect to the normal capacities $\{b_i(v) \}$, $v \in Z_i$. This is called the ``normal'' fractional assignment at level $i$. The support of $w_i$ will be denoted by $H_i = \{ e \in E_i : w_i(e) > 0\}$. The fractional assignment $w_i$ will be ``uniform'', in the sense that every edge $e \in H_i$ gets the same weight $w_i(e) = 1/d_i$, where $d_i = n^{(i-1)/K}$. 

\medskip
 To ease notation, we will often extend $w_i$ to the input graph $G = (V, E)$ as described in Section~\ref{label:bipartite:sec:notation}. Thus,  $b_i(v) = 0$ for all $v \in V \setminus Z_i$ and $w_i(e) = 0$ for all $e \in E \setminus E_i$. Accordingly, if a node $v$ does not belong to $Z_i$, then  $W_v(w_i) = 0$.  \\
\item For each node $v \in Z_i$, a ``discretized node-weight'' $\mathcal{W}_v(w_i) \leq b_i(v)$ that is an integral multiple of $\epsilon$ and gives a close estimate of the normal node-weight $W_v(w_i) = \sum_{(u,v) \in E_i} w_i(u,v)$.  Intuitively, one can think of $\W_v(w_i)$ as follows: It ``rounds up'' the value of $W_v(w_i)$ to a multiple of $\epsilon$, while ensuring that  the rounded value (1) does not exceed $b_i(v)$, and (2) does not  differ too much from the actual value.  The second property will be used in the analysis of our algorithm.

\medskip
We will often extend this notation to the entire node-set $V$ by  assuming that $\mathcal{W}_v(w_i) = 0$ for all nodes $v \in V \setminus Z_i$.  We will also set $\mathcal{W}_v(w) = \sum_{j=2}^K \mathcal{W}_v(w_j)$ for all nodes $v \in V$ (see equation~\ref{label:bipartite:eq:new:w} below).  \\
\item A fractional $b$-matching $w^r_i$ in the subgraph $G_i = (Z_i, E_i)$ with respect to the residual capacities $\{b_i^r(v) \}, v \in Z_i$. This is called the ``residual'' fractional assignment at level $i$. \\

\medskip
To ease notation, we will often extend $w_i^r$ to the input graph $G = (V, E)$ as described in Section~\ref{label:bipartite:sec:notation}. Thus,   $b_i^r(v) = 0$ for all $v \in V \setminus Z_i$ and $w_i^r(e) = 0$ for all $e \in E \setminus E_i$. Accordingly, we also set $W_v(w^r_i) = 0$ for all nodes $v \in V \setminus Z_i$.  \\
\item A partition of the node-set $Z_i$ into three subsets: $T_i, B_i, S_i$. The nodes in $T_i, B_i$ and $S_i$ will be respectively called ``tight'', ``big'' and ``small'' at level $i$. \\
\end{enumerate}

\noindent We will often use the phrase ``structures for level $i$'' while describing our algorithm. This is  defined below. 

\begin{definition}
\label{label:bipartite:def:structure}
For every $i \in \{2, \ldots, K\}$, the phrase ``structures for level $i$'' refers to the subgraph $G_i = (Z_i, E_i)$, the node-capacities $\{b_i(v), b_i^r(v)\}$, $v \in Z_i$,   the fractional assignments $w_i, w_i^r$, the discretized node-weights $\{ \mathcal{W}_v(w_i) \}, v \in Z_i$, and the partition of the node-set $Z_i$ into subsets $T_i, B_i, S_i$.
\end{definition}

\noindent Next, our algorithm will maintain the following structures for level one. \\
\begin{enumerate}
\item The subgraph $G_1 = (Z_1, E_1)$. \\
\item For each node $v \in Z_1$, a capacity $b_1^*(v) \geq 0$. \\
\item A fractional $b$-matching $w_1^*$ in  $G_1 = (Z_1, E_1)$ with respect to the node-capacities $\{b_1^*(v)\}, v \in Z_1$. 
To ease notation, we will often extend $w_1^*$ to the input graph $G = (V, E)$ as described in Section~\ref{label:bipartite:sec:notation}. So we will  set $b_1^*(v) = 0$ for all $v \in V \setminus Z_1$ and $w_1^*(e) = 0$ for all $e \in E \setminus E_1$. 
\end{enumerate}

\begin{definition}
\label{label:bipartite:def:structure:one}
The phrase ``structures for level $1$''  will refer to the subgraph $G_1 = (Z_1, E_1)$, the node-capacities $\{b_1^*(v) \}$, $v \in Z_1$, and the fractional assignment $w_1^*$. 
\end{definition}
\noindent Our algorithm will ensure the following property. \\
\begin{itemize}
\item Fix any  $1 \leq i \leq K$. The structures for level $i$  depend only on the structures for  levels $i+1 \leq j  \leq K$. Equivalently, the structures for a  level $j < i$ do not influence the structures for level $i$. 
\end{itemize}

\medskip
\noindent We will define two fractional assignments -- $w$ and $w^r$  -- as follows.
\begin{equation}
\label{label:bipartite:eq:new:w}
w = \sum_{j=2}^K w_j \text{ and } w^r = \sum_{j=2}^K w^r_j
\end{equation}
We will show that the fractional assignment $(w+w_1^*+w^r)$ forms a fractional matching in the input graph $G = (V, E)$. Further, the size of this fractional matching is {\em strictly} within a factor  of $2$ of the size of the maximum cardinality matching in $G$. Our main results are thus summarized in the two theorems below. Note that equations~\ref{label:bipartite:eq:parameter:1} and~\ref{label:bipartite:eq:parameter:2} guarantee that $f > 1/2$ (see Theorem~\ref{label:bipartite:th:approx:main}).

\begin{theorem}
\label{label:bipartite:th:feasible:main}
The fractional assignment $(w+w_1^*+w^r)$ is a fractional matching in the  graph $G = (V, E)$.
\end{theorem}

\begin{theorem}
\label{label:bipartite:th:approx:main}
The size of the fractional assignment $(w+w_1^*+w^r)$ is at least $f$ times the size of the  maximum cardinality matching in the input graph $G = (V, E)$, where
$$f = (1/2) \cdot \left(\frac{(1+\delta/3)}{(1+\epsilon)} - 4 K  \epsilon\right).$$
\end{theorem}

\paragraph{Organization for the rest of this section.}  \ \\

\begin{itemize}
\item In Section~\ref{label:bipartite:sec:inv:2}, we describe the invariants that are used to define the structures for levels $j \in \{2, \ldots, K\}$.
\item  In Section~\ref{label:bipartite:sec:feasible:inv:2}, we show that the above invariants are ``consistent'' in the following sense: There is a way to construct the structures for all levels $j > 1$ that satisfy these invariants. 
\item In Section~\ref{label:bipartite:sec:inv:1}, we describe the invariants that are used to define the structures for level one.
\item In Section~\ref{label:bipartite:sec:feasible:inv:1}, we again prove the ``consistency'' of these invariants. Specifically, we show that there is a way to construct the structures for level one that satisfy the invariants.
\item   In Section~\ref{label:bipartite:sec:property}, we derive some properties from the invariants that will be useful later while analyzing the update time of our algorithm.  
\item In Section~\ref{label:bipartite:sec:feasible}, we prove Theorem~\ref{label:bipartite:th:feasible:main} by showing that the fractional assignment $(w+w_1^*+w^r)$ forms a fractional matching in the input graph $G = (V, E)$. 
\item Finally, in Section~\ref{label:bipartite:sec:approx}, we bound the size of the fractional assignment $(w+w_1^*+w^r)$ and prove the required approximation guarantee as stated in Theorem~\ref{label:bipartite:th:approx:main}. \\
\end{itemize}

\noindent Before moving on to Section~\ref{label:bipartite:sec:inv:2}, we introduce a notation that will be used multiple times. For every  $i \in \{1, \ldots,  K-1\}$, we  define the ``small-index'' $s_i(v)$ of a node $v \in Z_i$ to be the minimum level  $j \in \{i+1, \ldots, K\}$ where the node is small (i.e., $v \in S_j$).  If $i = K$ or if $v \notin S_j$ for all $j \in \{ i +1, \ldots, K\}$, then we define $s_i(v) = \infty$. 
\begin{eqnarray}
\label{label:bipartite:eq:level-small}
 s_i(v)  = \begin{cases} \infty  & \text{ if } i = K \text{ or } v \notin S_j \ \ \forall j \in \{i+1, \ldots, K\}; \\ 
\min \left\{ j \in \{ i+1, \ldots, K \} : v \in S_j \right \} & \text{ otherwise.}
\end{cases}
\end{eqnarray}
Intuitively, $s_i(v)$ captures the notion of the ``most recent'' level larger than $i$ where $v$ is small.

\subsection{Invariants for levels $i \in \{2, \ldots, K\}$}
\label{label:bipartite:sec:inv:2}

Fix any level $i \in \{2, \ldots, K\}$, and recall Definition~\ref{label:bipartite:def:structure}. Suppose we are given the structures for all the levels $j \in \{i+1, \ldots, K\}$. Given this input,  we present the invariants that  determine the structures for level $i$.
 
First, in Definition~\ref{label:bipartite:inv:node-partition:1}, we specify how to derive the subgraph $G_i = (Z_i, E_i)$. It states that if $i < K$, then the node-set $Z_i$ consists of all the ``non-tight'' nodes (see Section~\ref{label:bipartite:sec:overview}) at level $i+1$. Else if $i = K$, then the set $Z_i$ consists of all the nodes in the input graph $G = (V, E)$. Further, the set $E_i$ consists of the edges in $G$ with both endpoints in $Z_i$. In other words, $G_i$ is the subgraph of $G$ induced by the subset of nodes $Z_i \subseteq V$.

\begin{definition}
\label{label:bipartite:inv:node-partition:1}
For all $i \in \{2, \ldots, K\}$, the node-set $Z_i$ is defined as follows.
\begin{equation} \nonumber
Z_i = \begin{cases} V & \text{ for } i = K; \\
Z_{i+1} \setminus T_{i+1} & \text{ for } i \in \{1, \ldots, K-1\}. 
\end{cases}
\end{equation}
Further, we define $E_{i} = \{ (u, v) \in E : u, v \in Z_i \}$ to be the set of edges in $E$ with both endpoints in $Z_i$. 
\end{definition}

In Definition~\ref{label:bipartite:inv:normal:capacity}, we specify how to derive the normal  node-capacity of a node   $v \in Z_i$. This is  determined by two quantities: (1) the total {\em discretized weight} received by the node from the levels larger than $i$; and (2) the value of the index $s_i(v)$. From equation~\ref{label:bipartite:eq:level-small}, it follows that  the value of $s_i(v)$ is determined by the  partitions of the node-sets $\{ Z_j \}$ into subsets $\{ T_j, S_j, B_j \}$ for all  $i < j \leq K$. Thus, the node-capacities $\{ b_i(v) \}, v \in Z_i$, are uniquely determined by the structures for levels $j \in \{i+1, \ldots, K\}$.

\begin{definition}
\label{label:bipartite:inv:normal:capacity}
Fix any level $i \in \{2, \ldots, K\}$ and any node $v \in Z_i$. The capacity $b_i(v)$ is defined as follows. 
\begin{eqnarray}
\label{label:bipartite:eq:kernel-capacity}
b_i(v) = \begin{cases} 1 - \delta - \sum_{j = i+1}^K \mathcal{W}_v(w_{j}) & \text{ if } s_i(v) = \infty; \\ \\
1 -  5 \cdot 10^{(K-s_i(v) + 1)} \cdot \delta - \sum_{j = i+1}^K \mathcal{W}_v(w_{j}) & \text{ else if } s_i(v) \in \{i+1, \ldots, K\}. 
\end{cases}
\end{eqnarray}
\end{definition}

Note that by equation~\ref{label:bipartite:eq:parameter:2}, $\delta$ is an integral multiple of $\epsilon$. Furthermore, we have $1 = N \cdot \epsilon$ for some integer $N$ (see equation~\ref{label:bipartite:eq:parameter:1}). In Section~\ref{label:bipartite:sec:overview}, while describing the structures for a level, we stated that the discretized weight of a node is also an integral multiple of $\epsilon$. Thus, for every node $v \in Z_i$, we infer that $\W_v(w_j)$ is an integral multiple of $\epsilon$ for all levels $j > i$. Hence, Definition~\ref{label:bipartite:inv:normal:capacity} implies that $b_i(v)$ is an integral multiple of $\epsilon$ for all nodes $v \in Z_i$.  In Section~\ref{label:bipartite:sec:feasible:inv:2}, we show that $b_i(v) > 0$ for  every node $v \in Z_i$ (see Lemma~\ref{label:bipartite:lm:positive:capacity}). Accordingly, we get: $b_i(v) \geq \epsilon$ for every node $v \in Z_i$. 

In Definition~\ref{label:bipartite:def:inv:normal:assignment} and Invariant~\ref{label:bipartite:inv:normal:assignment}, we describe how to construct the fractional $b$-matching $w_i$ in $G_i = (Z_i, E_i)$ with respect to the capacities $\{ b_i(v) \}$. Every edge in the support of $w_i$ gets a weight $1/d_i$. Further, this weight $1/d_i$  is negligibly small compared to the smallest node-capacity $\epsilon$ (see equation~\ref{label:bipartite:eq:parameter:4}). Hence, it is indeed possible to construct a nontrivial fractional matching $w_i$ with nonempty support $H_i$ in this manner.

\begin{definition}
\label{label:bipartite:def:inv:normal:assignment}
We define $d_i = n^{(i-1)/K}$ for all $i \in \{1, \ldots, K\}$.
\end{definition}

\begin{invariant}
\label{label:bipartite:inv:normal:assignment}
Consider any level $2 \leq i \leq K$.  We maintain  a fractional $b$-matching in $G_i$ with respect to the node-capacities $\{b_i(v) \}, v \in Z_i$, and denote this fractional $b$-matching by $w_i$.  Thus, we have: $W_v(w_i) \leq b_i(v)$ for all nodes $v \in Z_i$.
 We define $H_i = \{ e \in E_i : w_i(e) > 0\}$ to be the support of  $w_i$.   We ensure that $w_i(e) = 1/d_i$ for all $e \in H_i$. In other words, the fractional assignment $w_i$ is ``uniform'', in the sense that it assigns the same weight to every edge in its support. 
\end{invariant}

After constructing the fractional assignment $w_i$, as per Invariant~\ref{label:bipartite:inv:discretized} we give each node $v \in Z_i$ a discretized weight $\mathcal{W}_v(w_i)$ that is an integral multiple of $\epsilon$. These discretized weights serve as estimates of actual node-weights $\{W_v(w_i)\}$.  In Section~\ref{label:bipartite:sec:feasible:inv:2}, we show that it is indeed possible to satisfy Invariant~\ref{label:bipartite:inv:discretized}. 

\begin{invariant}
\label{label:bipartite:inv:discretized}
Consider any level $i \in \{2, \ldots, K\}$. For every node $v \in Z_i$, we maintain a  value $0 \leq \mathcal{W}_v(w_i) \leq b_i(v)$ so that $\mathcal{W}_v(w_i) - 2 \epsilon < W_v(w_i) \leq \mathcal{W}_v(w_i)$. Further,  $\mathcal{W}_v(w_i)$ is an integral multiple of $\epsilon$. 
\end{invariant}

 Invariant~\ref{label:bipartite:inv:node-partition}  implies that  $w_i$ is  ``maximal''  with respect to the discretized node-weights $\{ \W_v(w_i) \}$. To be more specific, the node-set $Z_i$ is partitioned into three subsets $T_i$, $B_i$ and $S_i$. The nodes in $T_i$, $B_i$ and $S_i$ are respectively called ``tight'', ``big'' and ``small'' at level $i$ (see Definition~\ref{label:bipartite:def:inv:node-partition}).  A node $v$ is tight iff  $\mathcal{W}_v(w_i) = b_i(v)$. For a tight node $v$, if we take an edge $(u,v)$ from $E_i \setminus H_i$ and insert it into $H_i$, then the weights $W_v(w_i), \mathcal{W}_v(w_i)$ might exceed the capacity $b_i(v)$ -- thereby violating Invariant~\ref{label:bipartite:inv:discretized}. But this issue does not arise if the node $v$ is non-tight. Accordingly, the invariant requires every edge in $E_i \setminus H_i$   to have at least one non-tight endpoint. 
 
 In Section~\ref{label:bipartite:sec:feasible:inv:2}, we show that the subsets $S_i, T_i \subseteq Z_i$  are mutually disjoint (see Lemma~\ref{label:bipartite:lm:partition}). This implies that the node-set $Z_i$ is indeed partitioned into three subsets -- $T_i, S_i$ and $B_i$ -- as claimed in the invariant below. Intuitively, this claim holds since $\delta$ is sufficiently small. To be more specific, consider a node $v \in T_i$. By Definition~\ref{label:bipartite:def:inv:node-partition}, we have $\W_v(w_i) = b_i(v)$ for such a node $v$. Hence, from Definition~\ref{label:bipartite:inv:normal:capacity}, we infer that either $\W_v(w_i) = b_i(v) =  1 - \delta - \sum_{j>i} \W_v(w_j)$, or $\W_v(w_i) = b_i(v) = 1 - 5 \cdot 10^{K-s_i(v)+1} \cdot \delta - \sum_{j > i} \W_v(w_j)$.  Rearranging the terms, we get: either $\sum_{j \geq i} \W_v(w_j) = 1 - \delta$ or $\sum_{j \geq i} \W_v(w_j) = 1 - 5 \cdot 10^{K-s_i(v)+1} \cdot \delta$. In either case, if $\delta$ is sufficiently small, then we can infer that $\sum_{j \geq i} \W_v(w_j) > 8 \cdot 10^{K-i} \cdot \delta$, which implies that $v \notin S_i$ by Definition~\ref{label:bipartite:def:inv:node-partition}. In other words, a tight node can never be small, which guarantees that $S_i \cap T_i = \emptyset$.

\begin{definition}
\label{label:bipartite:def:inv:node-partition}
At each level $2 \leq i \leq K$, the node-sets $T_i, B_i, S_i \subseteq Z_i$ are defined as follows.
\begin{eqnarray}
T_i & = & \{ v \in Z_i : \mathcal{W}_v(w_i) =  b_i(v)  \} \\
S_i & = & \left\{ v \in Z_i :  \sum_{j= i}^K  \mathcal{W}_v(w_j)  \leq 8 \cdot 10^{K-i} \cdot \delta \right\} \\
B_i & = & Z_i \setminus (T_i \cup S_i)
\end{eqnarray}
\end{definition} 
 
\begin{invariant}
\label{label:bipartite:inv:node-partition}
At each level $i \in \{2, \ldots, K\}$, the node-set $Z_i$ is partitioned by the   subsets $T_i, B_i, S_i$. 
Further, for every edge $(u,v) \in E_i$ with $u , v \in Z_i \setminus T_i$, we have $(u, v) \in H_i$. In other words, every edge between two non-tight nodes belongs to the support of $w_i$.   
\end{invariant}

The next definition specifies how to derive the residual capacity of a node $v \in Z_i$ based on (a) the partition of the node-set $Z_i$ into subsets $T_i, S_i, B_i$ and  (b) the value of $s_i(v)$. From equation~\ref{label:bipartite:eq:level-small}, it follows that the value of $s_i(v)$, in turn, depends on the partitions of $\{Z_{j} \}$ into subsets $\{T_j, S_j, B_j\}$ for all $j > i$. 
\begin{definition}
\label{label:bipartite:inv:residual:capacity}
Fix any level $i \in \{2, \ldots, K\}$ and any node $v \in Z_i$. The capacity $b_i^r(v)$ is defined as follows. 
\begin{eqnarray}
\label{label:bipartite:eq:residual-capacity}
b_i^r(v) = \begin{cases}
0 & \text{ if } v \notin (S_i \cup T_i); \\
 8 \cdot 10^{(K-i)} \cdot \delta & \text{ else if } v \in S_i; \\
 \delta & \text{ else if } v \in T_i \text{ and } s_i(v) = \infty; \\
 4 \cdot 10^{(K-s_i(v)+1)} \cdot  \delta & \text{ else if } v \in T_i \text{ and } s_i(v) \in \{i+1, \ldots, K\}. 
\end{cases}
\end{eqnarray}
\end{definition}

Note that $b_i^r(v) \geq 0$ for all nodes $v \in Z_i$. Further, note that $b_i^r(v)$ is nonzero only if $v \in T_i \cup S_i$. We define a subgraph $G_i^r = (T_i \cup S_i, E_i^r)$ of $G_i$ where the edge-set $E_i^r = \{ (u, v) \in E_i : u \in T_i, v \in S_i \}$ consists of those edges in $E_i$ that have one tight and one small endpoints. The edges in $E_i^r$ are called ``residual edges in level $i$''. Similarly, the subgraph $G_i^r$ is called the ``residual subgraph'' in level $i$. The next invariant states how to construct the residual fractional matching $w_i^r$ on the residual subgraph $G_i^r$.

\begin{invariant}
\label{label:bipartite:inv:residual:matching}
For all $i \in [2, K]$, let $E_i^r = \{ (u, v) \in E_i : u \in T_i, v \in S_i\}$ be the set of edges in $E_i$ with one tight and one small endpoints. Let $G^r_i = (T_i \cup S_i, E_i^r)$ be the subgraph of $G_i = (Z_i, E_i)$ induced by these edges. The fractional assignment $w^r_i$ is a maximal fractional  $b$-matching in  $G_i^r$ with respect to the residual capacities $\{ b_i^r(v) \}, v \in T_i \cup S_i$. Thus, for every edge $(u, v) \in E^r$, either $W_v(w^r_i) = b_i^r(v)$ or $W_u(w^r_i) = b_i^r(u)$. We refer to $E_i^r$, $G_i^r$ and $w_i^r$ as residual edges, residual subgraph and residual fractional matching at level $i$ respectively. We implicitly assume that $w_i^r(e) = 0$ for all edges $e \in E_i \setminus E_i^r$. This ensures that $w_i^r$ is also a valid fractional $b$-matching in $G_i$ with respect to the residual capacities described in Definition~\ref{label:bipartite:inv:residual:capacity}.
\end{invariant}
We now explain the link between Invariant~\ref{label:bipartite:inv:residual:matching} and Theorem~\ref{label:bipartite:thm:critical}. Suppose that $\lambda = 4 \cdot 10^{(K-i)} \cdot \delta$. Recall Definition~\ref{label:bipartite:inv:residual:capacity}. Consider any node $v \in T_i$. If $s_i(v) = \infty$, then we have $b_i^r(v) = \delta \leq \lambda$. On the other hand, if $s_i(v) \in \{i+1, \ldots, K\}$, then we also have $b_i^r(v) = 4 \cdot 10^{(K-s_i(v)+1)} \cdot \delta \leq 4 \cdot 10^{(K-i)} \cdot \delta \leq \lambda$. Thus, we have $b_i^r(v) \in [0, \lambda]$ for all tight nodes $v \in T_i$. In contrast, we have $b_i^r(v) = 8 \cdot 10^{(K-i)} \cdot \delta = 2\lambda$ for all small nodes $v \in S_i$. We therefore reach the following conclusion: If we set $\lambda = 4 \cdot 10^{(K-i)} \cdot \delta$, $A = T_i$, $B = S_i$, $E = E_i^r$ and $G = G_i^r$, then we can apply Theorem~\ref{label:bipartite:thm:critical} on the residual fractional matching $w_i^r$.

\subsection{Feasibility of the structures for levels $\{2, \ldots, K\}$}
\label{label:bipartite:sec:feasible:inv:2}

Recall Definition~\ref{label:bipartite:def:structure}. We will show that there is a way to satisfy all the invariants described so far. In other words, we will show that we can build the structures for levels $\{2, \ldots, K\}$ in a way that does not violate any invariant. We will prove this by induction. For the rest of this section, fix any $i \in \{2, \ldots, K\}$, and suppose that we have already built the structures for levels $\{i+1, \ldots, K\}$ without violating any invariant. It remains show how to build the structures for level $i$. Towards this end, we first construct the subgraph $G_i = (Z_i, E_i)$. \\

\begin{itemize}
\item Following Definition~\ref{label:bipartite:inv:node-partition:1}, we set $Z_i = Z_{i+1} \setminus T_{i+1}$ if $i < K$. Otherwise, we set $Z_i = Z_K = V$.\\
\item We define $E_i = \{ (u, v) \in E : u, v \in Z_i\}$ to be the set of edges in $E$ with endpoints in $Z_i$. \\
\end{itemize}

\noindent Next, we define the node-capacities $\{ b_i(v) \}, v \in Z_i$ as per Definition~\ref{label:bipartite:inv:normal:capacity}. In Lemma~\ref{label:bipartite:lm:positive:capacity}  we show that every node $v \in Z_i$ has a positive  capacity $b_i(v) > 0$. The complete proof of the lemma involves some case-analysis and is deferred to Section~\ref{label:bipartite:sec:lm:positive:capacity}. The intuition behind the proof, however,   can be summarized  as follows. \\

\begin{itemize}
\item If $i = K$, then we have $b_i(v) = 1-\delta > 0$. Else if $i < K$, then  $v \in Z_i$, and hence $v \in Z_{i+1} \setminus T_{i+1}$. Thus, Definition~\ref{label:bipartite:def:inv:node-partition} and Invariant~\ref{label:bipartite:inv:node-partition} (applied to level $i+1$) implies that $\W_v(w_{i+1}) < b_{i+1}(v)$. This positive gap between $b_{i+1}(v)$ and $\W_v(w_{i+1})$ translates into a positive value for $b_i(v)$ at level $i$.  
\end{itemize}

\begin{lemma}
\label{label:bipartite:lm:positive:capacity}
For  every node $v \in Z_i$, with $i \in \{2, \ldots, K\}$, we have $b_i(v) > 0$. 
\end{lemma}

Next, we build a fractional $b$-matching $w_i$ in the subgraph $G_i = (Z_i, E_i)$ with respect to the node-capacities $\{ b_i(v) \}, v \in Z_i$. We let $H_i = \{ e \in E_i : w_i(e) > 0\}$ denote the support of the fractional assignment $w_i$. We ensure that $w_i$ is uniform, in the sense that $w_i(e) = 1/d_i = 1/n^{(i-1)/K}$ for all $e \in E_i$. In other words, $w_i$ assigns the same weight to every edge in its support. This satisfies Invariant~\ref{label:bipartite:inv:normal:assignment}.

We now observe a crucial property. Since $\delta$ is a multiple of $\epsilon$ (see equation~\ref{label:bipartite:eq:parameter:2}), $1 = N \cdot \epsilon$ for some integer $N$ (see equation~\ref{label:bipartite:eq:parameter:1}), and the discretized node-weights $\{\W_v(w_j)\}$ for levels $j \in \{i+1, \ldots, K\}$ are also multiples of $\epsilon$ (apply Invariant~\ref{label:bipartite:inv:discretized} to levels $j > i$), Definition~\ref{label:bipartite:inv:normal:capacity} implies that all the node-capacities $\{b_i(v) \}$ are also multiples of $\epsilon$. 

\begin{observation}
\label{label:bipartite:ob:discretized:feasible}
For  every node $v \in Z_i$, with $i \in \{2, \ldots, K\}$, the capacity  $b_i(v)$ is an integral multiple of $\epsilon$. 
\end{observation}

Now, suppose that for every node $v \in Z_i$, we define $\W_v(w_i) = \lceil W_v(w_i)/\epsilon \rceil \cdot \epsilon$ to be the smallest multiple of $\epsilon$ that is no less than $W_v(w_i)$. Then the discretized weights $\{ \W_v(w_i) \}$  clearly satisfy: $\W_v(w_i) - 2 \epsilon \leq W_v(w_i) \leq \W_v(w_i)$. Further, since $b_i(v)$ is also an integral multiple of $W_v(w_i)$, we  get: $\W_v(w_i) \leq b_i(v)$ for all $v \in Z_i$. We conclude that the discretized weights $\{ \W_v(w_i) \}$ at level $i$  satisfy Invariant~\ref{label:bipartite:inv:discretized}. 

We now proceed towards verifying the consistency of Invariant~\ref{label:bipartite:inv:node-partition}. First, we need to show that the subsets $T_i$ and $S_i$  from Definition~\ref{label:bipartite:def:inv:node-partition} are mutually disjoint. This is shown in Lemma~\ref{label:bipartite:lm:partition}. The proof of Lemma~\ref{label:bipartite:lm:partition} appears in Section~\ref{label:bipartite:sec:lm:partition}. For an intuition behind the proof, recall the discussion immediately before Definition~\ref{label:bipartite:def:inv:node-partition}.

\begin{lemma}
\label{label:bipartite:lm:partition}
For  $i \in \{2, \ldots, K\}$, the subsets $T_i, S_i \subseteq Z_i$   are mutually disjoint.  
\end{lemma}

Next, we ensure that the fractional $b$-matching $w_i$ is {\em maximal with respect to the discretized capacities $\{\W_v(w_i) \}$}. To gain a better understanding of this concept, consider an edge $(u, v) \in E_i \setminus H_i$ that receives zero weight under $w_i$, and suppose that both its endpoints are non-tight (i.e., $u, v \in Z_i \setminus T_i$). By Definition~\ref{label:bipartite:def:inv:node-partition}, this means that $\W_u(w_i) < b_i(u)$ and $\W_v(w_i) < b_i(v)$. Since each of these quantities are integral multiples of $\epsilon$ (see Observation~\ref{label:bipartite:ob:discretized:feasible} and Invariant~\ref{label:bipartite:inv:discretized}), we infer that $W_u(w_i) \leq \W_u(w_i) \leq b_i(u) - \epsilon$ and $W_v(w_i) \leq \W_v(w_i) \leq b_i(v) - \epsilon$. Now, what will happen if we insert this edge $(u, v)$ into $H_i$ by setting $w_i(u,v) = 1/d_i$? To answer this question, we make the following observations. \\

\begin{itemize}
\item Since $1/d_i \leq 1/n^{1/K}$ for $i \geq 2$  (see Definiton~\ref{label:bipartite:def:inv:normal:assignment}) and since every edge in $H_i$ gets a weight of $1/d_i$, equation~\ref{label:bipartite:eq:parameter:4} implies that $W_u(w_i)$ (resp. $W_v(w_i)$) will increase by at most $\epsilon$. Accordingly, $\W_u(w_i)$ (resp. $\W_v(w_i)$) will increase by at most $\epsilon$. Hence, we will still  have $\W_u(w_i) \leq b_i(u)$ and $\W_v(w_i) \leq b_i(v)$. \\
\end{itemize}

\noindent In other words, if both the endpoints of an edge in $E_i \setminus H_i$ are non-tight, then we can insert that edge into $H_i$ without violating any of the invariants. But we cannot reach the same conclusion if at least one of the endpoints is tight. Accordingly, the second part of Invariant~\ref{label:bipartite:inv:node-partition} states that every edge in $E_i \setminus H_i$ will have at least one tight endpoint.  In this sense, the fractional assignment $w_i$ is ``maximal with respect to the discretized node-capacities''. 

Next, we define the residual node-capacities $\{b_i^r(v)\}, v \in Z_i$ as per Definition~\ref{label:bipartite:inv:residual:capacity}. Clearly, all these residual capacities are nonnegative. Finally, we compute a maximal fractional $b$-matching $w_i^r$ with respect to these residual node-capacities $\{ b_i^r(v)\}$ as per Invariant~\ref{label:bipartite:inv:residual:matching}. At this stage, we have finished constructing the structures for level $i$. 

\subsubsection{Proof of Lemma~\ref{label:bipartite:lm:positive:capacity}}
\label{label:bipartite:sec:lm:positive:capacity}

Fix any level $i \in \{2, \ldots, K\}$ and any node $v \in Z_i$. If $i = K$, then equation~\ref{label:bipartite:eq:level-small} implies that $s_i(v) = \infty$, and hence Definition~\ref{label:bipartite:inv:normal:capacity} guarantees that $b_i(v) = 1 - \delta > 0$. Thus, for the rest of the proof, we suppose that $i \in \{2, \ldots, K-1\}$. Further, by induction hypothesis, we assume that:
\begin{equation}
\label{label:bipartite:eq:runtime:000}
b_j(v) > 0 \text{ for all } j \in \{i+1, \ldots, K\}.
\end{equation}
We  now consider  three possible cases, depending on the value of $s_i(v)$. \\

\begin{itemize}
\item {\em Case 1. $s_i(v) = \infty$.} \\

In this case, equation~\ref{label:bipartite:eq:level-small} implies that $s_{i+1}(v) = \infty$. Accordingly, by Definition~\ref{label:bipartite:inv:normal:capacity} we get:
\begin{eqnarray}
\label{label:bipartite:eq:runtime:1}
b_{i+1}(v) & = & 1 - \delta - \sum_{j=i+2}^K \W_v(w_j) \\
b_i(v) & = & 1 - \delta - \sum_{j=i+1}^K \W_v(w_j) \label{label:bipartite:eq:runtime:4}
\end{eqnarray}
Since $v \in Z_i$, Definition~\ref{label:bipartite:inv:node-partition:1}  states that $v \in Z_{i+1} \setminus T_{i+1}$. Hence, applying Definition~\ref{label:bipartite:def:inv:node-partition}  we get:
\begin{eqnarray}
\label{label:bipartite:eq:runtime:2}
\W_v(w_{i+1}) < b_{i+1}(v)
\end{eqnarray}
From equations~\ref{label:bipartite:eq:runtime:1} and~\ref{label:bipartite:eq:runtime:2}, we get:
\begin{eqnarray*}
\W_v(w_{i+1}) < 1 - \delta - \sum_{j=i+2}^K \W_v(w_j)
\end{eqnarray*}
Rearranging the terms in the above inequality, we get:
\begin{equation}
\label{label:bipartite:eq:runtime:3}
\sum_{j=i+1}^K \W_v(w_j) < 1 - \delta
\end{equation}
From equations~\ref{label:bipartite:eq:runtime:3} and~\ref{label:bipartite:eq:runtime:4}, we get:
\begin{equation}
b_i(v) > 0 \label{label:bipartite:eq:runtime:5}
\end{equation}
The lemma follows from equation~\ref{label:bipartite:eq:runtime:5}. \\

\item {\em Case 2. $s_i(v) \in \{i+2, \ldots, K\}$.} \\

In this case, equation~\ref{label:bipartite:eq:level-small} implies that $s_i(v) = s_{i+1}(v)$. Accordingly, by Definition~\ref{label:bipartite:inv:normal:capacity} we get:
\begin{eqnarray}
\label{label:bipartite:eq:runtime:10}
b_{i+1}(v) & = & 1 - 5 \cdot 10^{(K-s_i(v)+1)} \cdot \delta - \sum_{j=i+2}^K \W_v(w_j) \\
b_i(v) & = & 1 - 5 \cdot 10^{(K-s_i(v) +1)} \cdot \delta - \sum_{j=i+1}^K \W_v(w_j) \label{label:bipartite:eq:runtime:11}
\end{eqnarray}
Since $v \in Z_i$, Definition~\ref{label:bipartite:inv:node-partition:1}  states that $v \in Z_{i+1} \setminus T_{i+1}$. Hence, applying Definition~\ref{label:bipartite:def:inv:node-partition}  we get:
\begin{eqnarray}
\label{label:bipartite:eq:runtime:12}
\W_v(w_{i+1}) < b_{i+1}(v)
\end{eqnarray}
From equations~\ref{label:bipartite:eq:runtime:10} and~\ref{label:bipartite:eq:runtime:12}, we get:
\begin{eqnarray*}
\W_v(w_{i+1}) < 1 - 5 \cdot 10^{(K-s_i(v) + 1)} \cdot \delta - \sum_{j=i+2}^K \W_v(w_j)
\end{eqnarray*}
Rearranging the terms in the above inequality, we get:
\begin{equation}
\label{label:bipartite:eq:runtime:13}
\sum_{j=i+1}^K \W_v(w_j) < 1 - 5 \cdot 10^{(K-s_i(v) + 1)} \cdot  \delta
\end{equation}
From equations~\ref{label:bipartite:eq:runtime:11} and~\ref{label:bipartite:eq:runtime:13}, we get:
\begin{equation}
b_i(v) > 0 \label{label:bipartite:eq:runtime:14}
\end{equation}
The lemma follows from equation~\ref{label:bipartite:eq:runtime:14}. \\

\item {\em Case 3. $s_i(v) = i+1$.} \\

In this case, applying Definition~\ref{label:bipartite:inv:normal:capacity} we get:
\begin{eqnarray}
b_i(v) & = & 1 - 5 \cdot 10^{(K-i)} \cdot \delta - \sum_{j=i+1}^K \W_v(w_j) \label{label:bipartite:eq:runtime:30}
\end{eqnarray}
Since $s_i(v) = i+1$, equation~\ref{label:bipartite:eq:level-small} implies that $v \in S_{i+1}$. Hence, applying Definition~\ref{label:bipartite:def:inv:node-partition} we get:
\begin{eqnarray}
\sum_{j=i+1}^K \W_v(w_j) \leq 8 \cdot 10^{(K-i-1)} \cdot \delta \label{label:bipartite:eq:runtime:31}
\end{eqnarray}
From equations~\ref{label:bipartite:eq:runtime:30} and~\ref{label:bipartite:eq:runtime:31} we get:
\begin{eqnarray}
b_i(v) & \geq & 1 - 5 \cdot 10^{(K-i)} \cdot \delta - 8 \cdot 10^{(K-i-1)} \cdot \delta \nonumber  \\
& \geq & 1 - 5 \cdot 10^{K} \cdot \delta - 8 \cdot 10^K \cdot \delta \nonumber \\
& = & 1 - 13 \cdot 10^{K} \cdot \delta \nonumber \\
& > & 0 \label{label:bipartite:eq:runtime:32}
\end{eqnarray}
Equation~\ref{label:bipartite:eq:runtime:32} follows from equation~\ref{label:bipartite:eq:parameter:3}. The lemma follows from equation~\ref{label:bipartite:eq:runtime:32}. 
\end{itemize}

\subsubsection{Proof of Lemma~\ref{label:bipartite:lm:partition}}
\label{label:bipartite:sec:lm:partition}

We first describe a small technical claim.

\begin{claim}
\label{label:bipartite:cl:bound:normal:capacity}
Consider any level $2 \leq i \leq K$ and any node $v \in Z_i$. We have:
$$b_i(v) \geq 1 - 5 \cdot 10^{(K-2)} \cdot \delta - \sum_{j=i+1}^K \mathcal{W}_v(w_j).$$
\end{claim}

\begin{proof}
Follows from Definition~\ref{label:bipartite:inv:normal:capacity} and the fact that $i \geq 2$. 
\end{proof}

We now proceed with the proof of Lemma~\ref{label:bipartite:lm:partition}.   For the sake of contradiction, suppose that there is a node $v \in T_i \cap S_i$. Since the node $v$ belongs to $T_i$,  Definition~\ref{label:bipartite:def:inv:node-partition} and Claim~\ref{label:bipartite:cl:bound:normal:capacity} imply that:
\begin{eqnarray}
\mathcal{W}_v(w_i) & = & b_i(v)  \nonumber \\
& \geq & 1 - 5 \cdot 10^{(K-2)} \cdot \delta - \sum_{j = i+1}^K \mathcal{W}_v(w_j)  \nonumber 
\end{eqnarray}
Rearranging the terms in the above inequality, we get:
\begin{equation}
\label{label:bipartite:eq:lm:partition:1}
\sum_{j=i}^K \mathcal{W}_v(w_j) \geq 1 - 5 \cdot 10^{(K-2)} \cdot \delta 
\end{equation}
On the other hand, since the node $v$ belongs to $S_i$ and since $2 \leq i \leq K$, Definition~\ref{label:bipartite:def:inv:node-partition} implies that:
\begin{equation}
\label{label:bipartite:eq:lm:partition:2}
\sum_{j=i}^K \mathcal{W}_v(w_j) \leq 8 \cdot 10^{K-2} \cdot \delta
\end{equation}
From equations~\ref{label:bipartite:eq:lm:partition:1} and~\ref{label:bipartite:eq:lm:partition:2} we derive that:
\begin{eqnarray}
8 \cdot 10^{K-2} \cdot \delta \geq 1 - 5 \cdot 10^{K-2} \cdot \delta  \nonumber
\end{eqnarray}
Rearranging the terms in the above inequality, we get:
\begin{eqnarray}
13 \cdot 10^{K-2} \cdot \delta \geq 1 \label{label:bipartite:eq:lm:partition:3}
\end{eqnarray}
But equation~\ref{label:bipartite:eq:lm:partition:3} cannot be true as long as equation~\ref{label:bipartite:eq:parameter:1} holds.  This leads to a contradiction. Thus, the sets $T_i$ and $S_i$ have to be mutually disjoint. This concludes the proof of the lemma.

\subsection{Invariants for level $i = 1$}
\label{label:bipartite:sec:inv:1}

Recall Definitions~\ref{label:bipartite:def:structure} and~\ref{label:bipartite:def:structure:one}. Suppose that we are given the structures for all the levels $j \in \{2, \ldots, K\}$. Given this input, in this section we show how to derive the structures for level one.  The definition below specifies the subgraph $G_1 = (Z_1, E_1)$. Note that this is exactly analogous to Definition~\ref{label:bipartite:inv:node-partition:1}. 

\begin{definition}
\label{label:bipartite:inv:node-partition:1:last}
We define $Z_1 = Z_2 \setminus T_2$. Further, the set $E_1 = \{ (u,v) \in E : u, v \in Z_1\}$ consists of all the edges with both endpoints in $Z_1$.
\end{definition}

The definition below shows how to derive the capacity $b_1^*(v)$ of any node $v \in Z_1$.

\begin{definition}
\label{label:bipartite:inv:node:capacity:one}
Consider any node $v \in Z_1$. The capacity $b_1^*(v)$ is defined as follows. 
\begin{eqnarray}
\label{label:bipartite:eq:kernel-capacity}
b_1^*(v) = \begin{cases} 1  - \sum_{j = 2}^K \mathcal{W}_v(w_{j}) & \text{ if } s_1(v) = \infty; \\ \\
1 -   10^{(K-s_1(v)+1)} \cdot \delta - \sum_{j = 2}^K \mathcal{W}_v(w_{j}) & \text{ else if } s_1(v) \in \{2, \ldots, K\}.
\end{cases}
\end{eqnarray}
\end{definition}

 In Section~\ref{label:bipartite:sec:property}, we show that $b_1^*(v) > 0$ for every node $v \in Z_1$ (see Lemma~\ref{label:bipartite:lm:positive:capacity:last}). 
 We now specify how to derive the fractional assignment $w_1^*$.

\begin{invariant}
\label{label:bipartite:inv:assignment:one}
The fractional assignment $w_1^*$ forms a fractional $b$-matching  in $G_1 = (Z_1, E_1)$ with respect to the node-capacities $\{ b_1^*(v) \}, v \in Z_1$. Further, the size of $w_1^*$ is at least $1/(1+\epsilon)$-times  the size of every other fractional $b$-matching in $G_1$ with respect to the same node-capacities. 
\end{invariant}

\subsection{Feasibility of the structures for level one}
\label{label:bipartite:sec:feasible:inv:1}

The next lemma states that if a node $v$ participates in the subgraph $G_1 = (Z_1, E_1)$, then  $b_1^*(v) > 0$. The proof of Lemma~\ref{label:bipartite:lm:positive:capacity:last} appears in Section~\ref{label:bipartite:sec:lm:positive:capacity:last}. The intuition behind the proof is similar to the one for Lemma~\ref{label:bipartite:lm:positive:capacity}.
\begin{lemma}
\label{label:bipartite:lm:positive:capacity:last}
For every  every node $v \in Z_1$, we have $b_i^*(v) > 0$. 
\end{lemma}

\subsubsection{Proof of Lemma~\ref{label:bipartite:lm:positive:capacity:last}}
\label{label:bipartite:sec:lm:positive:capacity:last}

Fix any node $v \in Z_1$. By Lemma~\ref{label:bipartite:lm:positive:capacity} we have
\begin{equation}
\label{label:bipartite:eq:positive:1}
b_i(v) > 0 \text{ for all } i \in \{2, \ldots, K\}.
\end{equation}
We consider three possible cases, depending on the value of $s_1(v)$.
\begin{itemize}
\item {\em Case 1. $s_1(v) = \infty$.} \\

In this case, equation~\ref{label:bipartite:eq:level-small} implies that $s_2(v) = \infty$. Hence, Definitions~\ref{label:bipartite:inv:normal:capacity} and~\ref{label:bipartite:inv:node:capacity:one} imply that:
\begin{eqnarray}
\label{label:bipartite:eq:positive:2}
b_2(v) & = & 1 - \delta - \sum_{j=3}^K \W_v(w_j) \\
b_1^*(v) & = & 1 - \sum_{j=2}^K \W_v(w_j) \label{label:bipartite:eq:positive:3}
\end{eqnarray}
Since $v \in Z_1$, Definition~\ref{label:bipartite:inv:node-partition:1:last} implies that $v \in Z_2 \setminus T_2$. Hence, applying Definition~\ref{label:bipartite:def:inv:node-partition} we get:
\begin{equation}
\label{label:bipartite:eq:positive:4}
\W_v(w_2) < b_2(v)
\end{equation}
From equations~\ref{label:bipartite:eq:positive:2} and~\ref{label:bipartite:eq:positive:4} we get:
\begin{equation}
\label{label:bipartite:eq:positive:5}
\W_v(w_2) < 1 - \delta - \sum_{j=3}^K \W_v(w_j)
\end{equation}
Rearranging the terms in the above inequality, we get:
\begin{equation}
\label{label:bipartite:eq:positive:6}
\sum_{j=2}^K \W_v(w_j) < 1 -\delta
\end{equation}
From equations~\ref{label:bipartite:eq:positive:3} and~\ref{label:bipartite:eq:positive:6} we infer that:
\begin{equation}
\label{label:bipartite:eq:positive:7}
b_1^*(v) > 0
\end{equation}
The lemma follows from equation~\ref{label:bipartite:eq:positive:7}.  \\

\item {\em Case 2. $s_1(v) \in \{3, \ldots, K\}$.} \\

In this case, equation~\ref{label:bipartite:eq:level-small} implies that $s_1(v) = s_2(v)$. Hence, applying Definitions~\ref{label:bipartite:inv:normal:capacity} and~\ref{label:bipartite:inv:node:capacity:one} we get:
\begin{eqnarray}
b_2(v) & =  & 1 - 5 \cdot 10^{(K-s_1(v) + 1)} \cdot \delta - \sum_{j=3}^K \W_v(w_j) \label{label:bipartite:eq:positive:10} \\
b_1^*(v) & = & 1 - 10^{(K-s_1(v) + 1)} \cdot \delta - \sum_{j=2}^K \W_v(w_j) \label{label:bipartite:eq:positive:11} 
\end{eqnarray}
Since $v \in Z_1$, Definition~\ref{label:bipartite:inv:node-partition:1:last} implies that $v \in Z_2 \setminus T_2$. Hence, applying Definition~\ref{label:bipartite:def:inv:node-partition} we get:
\begin{equation}
\label{label:bipartite:eq:positive:12}
\W_v(w_2) < b_2(v)
\end{equation}
From equations~\ref{label:bipartite:eq:positive:10} and~\ref{label:bipartite:eq:positive:12} we get:
\begin{eqnarray*}
\W_v(w_2) < 1 - 5 \cdot 10^{(K-s_1(v) + 1)} \cdot \delta - \sum_{j=3}^K \W_v(w_j)
\end{eqnarray*}
Rearranging the terms in the above inequality, we get:
\begin{eqnarray}
\label{label:bipartite:eq:positive:13}
\sum_{j=2}^K \W_v(w_j) < 1 - 5 \cdot 10^{(K-s_1(v)+1)} \cdot \delta
\end{eqnarray}
From equations~\ref{label:bipartite:eq:positive:11} and~\ref{label:bipartite:eq:positive:13} we get:
\begin{equation}
\label{label:bipartite:eq:positive:14}
b_1^*(v) > 4 \cdot 10^{(K- s_1(v) + 1)} \cdot \delta > 0
\end{equation}
The lemma follows from equation~\ref{label:bipartite:eq:positive:14}.  \\

\item {\em Case 3. $s_1(v) = 2$.} \\

In this case, applying Definition~\ref{label:bipartite:inv:node:capacity:one} we get:
\begin{equation}
\label{label:bipartite:eq:positive:20}
b_1^*(v) = 1 - 10^{(K-1)} \cdot \delta - \sum_{j=2}^K \W_v(w_j)
\end{equation}
Since $s_1(v) = 2$, equation~\ref{label:bipartite:eq:level-small} implies that $v \in S_2$. Hence, applying Definition~\ref{label:bipartite:def:inv:node-partition} we get:
\begin{equation}
\label{label:bipartite:eq:positive:21}
\sum_{j=2}^K \W_v(w_j) \leq 8 \cdot 10^{(K-2)} \cdot \delta
\end{equation}
From equations~\ref{label:bipartite:eq:positive:20} and~\ref{label:bipartite:eq:positive:21} we get:
\begin{eqnarray}
b_1^*(v) & \geq & 1 - 10^{(K-1)} \cdot \delta - 8 \cdot 10^{(K-2)} \cdot \delta \nonumber \\
& = & 1 - 18 \cdot 10^{(K-2)} \cdot \delta \nonumber \\
& > & 0 \label{label:bipartite:eq:positive:22}
\end{eqnarray}
Equation~\ref{label:bipartite:eq:positive:22} follows from equation~\ref{label:bipartite:eq:parameter:1}. 
\end{itemize}

\subsection{Some useful properties}
\label{label:bipartite:sec:property}

In this section, we derive some properties from the invariants. These properties  will be used later on to analyze the amortized update time of our algorithm. We start by noting that  both $1, \delta$ are integral multiples of $\epsilon$ (see equation~\ref{label:bipartite:eq:parameter:1},~\ref{label:bipartite:eq:parameter:2}), and that the discretized node-weights are also multiples of $\epsilon$ (see Invariant~\ref{label:bipartite:inv:discretized}). Hence,  Definitions~\ref{label:bipartite:inv:normal:capacity},~\ref{label:bipartite:inv:residual:capacity},~\ref{label:bipartite:inv:node:capacity:one} imply that all the node-capacities are also multiples of $\epsilon$.

\begin{observation}
\label{label:bipartite:ob:discretized}
For every level $i \in \{2, \ldots, K\}$ and every node $v \in Z_i$, both  $b_i(v)$ and $b_i^r(v)$ are integral multiples of $\epsilon$. Further, for every node $v \in Z_1$, the node-capacity $b_1^*(v)$ is an integral multiple of $\epsilon$. 
\end{observation}

From Lemmas~\ref{label:bipartite:lm:positive:capacity},~\ref{label:bipartite:lm:positive:capacity:last} and Observation~\ref{label:bipartite:ob:discretized}, we get the following observation. 
\begin{observation}
\label{label:bipartite:ob:positive}
For every level $i \in \{2, \ldots, K\}$ and every node $x \in Z_i$, we have $b_i(x) \geq \epsilon$. Further, for every node $x \in Z_1$, we have $b_1^*(x) \geq \epsilon$. 
\end{observation}

  The next observation follows immediately from Definition~\ref{label:bipartite:inv:node-partition:1}.

\begin{observation}
\label{label:bipartite:lm:tight}
Consider any  level $\ell(v) = i \in \{2, \ldots, K\}$. Then we have $T_i = \{ v \in V : \ell(v) = i\}$. 
\end{observation}

 Lemma~\ref{label:bipartite:lm:laminar} states that the edges in the subgraph $G_{i-1}$  is drawn from the support of the normal fractional assignment $w_i$ in the previous level $i$. This holds since only the non-tight nodes at level $i$ get demoted to  level $(i-1)$, and all the edges connecting two non-tight nodes are included in the support of $w_i$.

\begin{lemma}
\label{label:bipartite:lm:laminar}
For every level $i \in \{2, \ldots, K\}$, we have $E_{i-1} \subseteq H_{i}$.
\end{lemma}

\begin{proof}
Consider any edge $(u,v) \in E_{i-1}$ in the subgraph $G_{i-1} = (Z_{i-1}, E_{i-1})$. By Definitions~\ref{label:bipartite:inv:node-partition:1} and~\ref{label:bipartite:inv:node-partition:1:last}, we have $u, v \in Z_{i-1} = Z_{i} \setminus T_{i}$.  Since both $u, v \in Z_i \setminus T_i$, and since $E_i$ is the subset of edges in $G$ induced by the node-set $Z_i$, we have  $(u,v) \in E_i$. Accordingly,  Invariant~\ref{label:bipartite:inv:node-partition} implies that $(u,v) \in H_{i}$. To summarize, every edge $(u,v) \in E_{i-1}$ belongs to the set $H_i$. Thus, we have $E_{i-1} \subseteq H_i$. 
\end{proof}

The next lemma upper bounds the number of edges incident upon a node that can get nonzero weight under a normal fractional assignment $w_i$. This holds since a normal fractional assignment is {\em uniform}, in the sense that it assigns the same weight $1/d_i$ to every edge in its support. Hence, not too many edges incident upon a node can be in the support of $w_i$, for otherwise the total weight of that node will exceed one.

\begin{lemma}
\label{label:bipartite:lm:deg}
For each level $i \in \{2, \ldots, K\}$ and every node $v \in Z_i$, we have $\text{deg}_v(H_i) \leq d_i$. 
\end{lemma}

\begin{proof}
Invariant~\ref{label:bipartite:inv:normal:assignment} states that each edge $e \in H_i$ gets a weight of $w_i(e) = 1/d_i$ under the fractional assignment $w_i$. Thus, for every node $v \in Z_i$, it follows that $W_v(w_i) = (1/d_i) \cdot \text{deg}_v(H_i)$. For the sake of contradiction, suppose that the lemma is false. Then there must be a node $u \in Z_i$ with $\text{deg}_u(H_i) > d_i$. In that case, we would have:
$$\mathcal{W}_u(w_i) \geq W_u(w_i) = (1/d_i) \cdot \text{deg}_u(H_i) > 1 \geq b_i(u).$$
The first inequality follows from Invariant~\ref{label:bipartite:inv:discretized}. The last inequality follows from Definition~\ref{label:bipartite:inv:normal:capacity}. Thus, we get: $\mathcal{W}_u(w_i) > b_i(u)$, which contradicts Invariant~\ref{label:bipartite:inv:discretized}. Thus, our assumption was wrong, and this concludes the proof of the lemma. 
\end{proof}

Finally, we upper bound the maximum degree of a node in any subgraph $G_i = (Z_i, E_i)$. 

\begin{corollary}
\label{label:bipartite:cor:deg}
For each level $i \in \{1, \ldots, K\}$ and every node $v \in Z_i$, we have $\text{deg}_v(E_i) \leq n^{1/K} \cdot d_i$. 
\end{corollary}

\begin{proof}
Consider two possible cases, depending on the value of $i$.
\begin{itemize}
\item {\em $i = K$.}

In this case, we have $d_i = d_K = n^{1-1/K}$ (see Invariant~\ref{label:bipartite:inv:normal:assignment}), and hence, we get: 
$$\text{deg}_v(E_i) \leq |V| = n = n^{1/K} \cdot d_i.$$ 

\item {\em $i \in \{1, \ldots, K-1\}$.}

In this case, applying Lemmas~\ref{label:bipartite:lm:laminar},~\ref{label:bipartite:lm:deg} and Definition~\ref{label:bipartite:def:inv:normal:assignment}, we get:
\begin{eqnarray*}
\text{deg}_v(E_i) \leq \text{deg}_v(H_{i+1}) \leq d_{i+1} = n^{1/K} \cdot d_i.
\end{eqnarray*}
\end{itemize}
\end{proof}

\subsection{Feasibility of the solution}
\label{label:bipartite:sec:feasible}

We devote this section to the proof of Theorem~\ref{label:bipartite:th:feasible:main}, which states that the fractional assignment $(w+w_1^*+w^r)$ forms a fractional matching in the graph $G = (V, E)$. Specifically, we will show that $W_v(w+w^*_1+w^r) \leq 1$ for all nodes $v \in V$. The high level idea behind the proof can be explained as follows.

Consider any node $v \in V$ at level $i \in \{1, \ldots, K\}$.  We need to show that $W_v(w+w^*_1+w^r) \leq 1$. We start by noting that the quantity $W_v(w+w^*_1+w^r)$ is equal to the total weight received by the node $v$ from all the structures for levels $j \in \{i, \ldots, K\}$. This holds since by Corollary~\ref{label:bipartite:cor:level}, we have $v \notin Z_j$ for all $j \in \{1, \ldots, i-1\}$, and hence the node $v$ receives {\em zero} weight from all the structures for levels $j < i$. 

Consider the scenario where we have fixed the structures for all the levels $j \in \{i+1, \ldots, K\}$, and we are about to derive the structures for level $i$. We first want to upper bound the total weight received by the node $v$ from the residual fractional assignments in levels $j > i$. This is given by $\sum_{j = i+1}^K W_v(w^r_j)$. By Invariant~\ref{label:bipartite:inv:residual:matching}, each  $w_j^r$ is a fractional $b$-matching with respect to the capacities $\{b_j^r(x)\}$. Hence,  we infer that $W_v(w^r_j) \leq b_j^r(v)$ for all $j \in \{i+1, \ldots, K\}$. Thus, we get: $\sum_{j > i} W_v(w^r_j) \leq \sum_{j > i} b_j^r(v)$. So it suffices to upper bound the quantity $\sum_{j=i+1}^K b_j^r(v)$, which is done in Lemma~\ref{label:bipartite:lm:bound:residual:assignment}.  Next, we note that the total  weight received by  $v$ from the fractional assignments $\{w_j\}$ at levels $j > i$ is $\sum_{j > i} W_v(w_j)$. Thus, we get:
\begin{equation}
\label{label:bipartite:eq:total:weight:1}
\text{Total weight received by } v \text{ from all the structures for levels } j > i \text{ is at most } \sum_{j > i} b_j^r(v) + \sum_{j > i} W_v(w_j). 
\end{equation}
Now, we focus on the weights received by $v$ from the structures at levels $i$. Here, we need to consider two possible cases, depending on the value of $\ell(v) = i$. 
\begin{itemize}
\item {\em Case 1.} $\ell(v) = i \in \{2, \ldots, K\}$.

In this case, the total weight received by $v$ from the structures at levels $i$ is given by $W_v(w_i) + W_v(w_i^r)$. Since $w_i$ is a fractional $b$-matching with respect to the capacities $\{ b_i(x) \}$, we have $W_v(w_i) \leq b_i(v)$. Since $w_i^r$ is a fractional $b$-matching with respect to the capacities $\{ b_i^r(x) \}$, we have $W_v(w_i^r) \leq b_i^r(v)$. Thus, we get: $W_v(w_i) + W_v(w_i^r) \leq b_i(v) + b_i^r(v)$. Now, suppose  we can prove that:
\begin{equation}
\label{label:bipartite:eq:total:weight:2}
\left(b_i(v) + b_i^r(v) \right) + \sum_{j > i} b_j^r(v) + \sum_{j > i} W_v(w_j) \leq 1.
\end{equation}
Then equations~\ref{label:bipartite:eq:total:weight:1} and~\ref{label:bipartite:eq:total:weight:2} will together have the following implication:  the total weight received by  $v$ from all the structures for levels $j \geq i$ is at most one. Since $v$ receives {\em zero} weight from all the structures for levels $j < i$, it will follow that $W_v(w+w_1^*+w^r)$. Thus, Theorem~\ref{label:bipartite:th:feasible:main}  follows from equation~\ref{label:bipartite:eq:total:weight:2}, which is shown in Lemma~\ref{label:bipartite:lm:bound:l2}. 

\item {\em Case 2.} $\ell(v) = i = 1$.

In this case, the total weight received by $v$ from the structures at level $i = 1$ is given by $W_v(w_1^*)$. Since $w_1^*$ is a fractional $b$-matching with respect to the capacities $\{ b_1^*(x) \}$, we have $W_v(w_1^*) \leq b_1^*(v)$. Now, suppose we can prove that:
\begin{equation}
\label{label:bipartite:eq:total:weight:3}
b_1^*(v) + \sum_{j=2}^K b_1^r(v) + \sum_{j=2}^K W_v(w_j) \leq 1.
\end{equation}
Then equations~\ref{label:bipartite:eq:total:weight:1} and~\ref{label:bipartite:eq:total:weight:3} will together have the following implication:  the total weight received by  $v$ from all the structures for levels $j \in \{1, \ldots, K\}$ is at most one. Thus, Theorem~\ref{label:bipartite:th:feasible:main} follows from equation~\ref{label:bipartite:eq:total:weight:3}, which is shown in Lemma~\ref{label:bipartite:lm:bound:l1}. 
\end{itemize}
The rest of this section is organized as follows.
\begin{itemize}
\item In Section~\ref{label:bipartite:sec:lm:bound:residual:assignment} we prove Lemma~\ref{label:bipartite:lm:bound:residual:assignment}.
\item In Section~\ref{label:bipartite:sec:lm:bound:l2} we prove Lemma~\ref{label:bipartite:lm:bound:l2}.
\item In Section~\ref{label:bipartite:sec:lm:bound:l1} we prove Lemma~\ref{label:bipartite:lm:bound:l1}.
\item Finally, in Section~\ref{label:bipartite:sec:th:feasible:main} we apply Lemmas~\ref{label:bipartite:lm:bound:residual:assignment},~\ref{label:bipartite:lm:bound:l2} and~\ref{label:bipartite:lm:bound:l1} to prove Theorem~\ref{label:bipartite:th:feasible:main}. 
\end{itemize}

\begin{lemma}
\label{label:bipartite:lm:bound:residual:assignment}
Consider any node $v \in V$ at level $\ell(v) = i \in \{1, \ldots, K\}$. We have:
\begin{eqnarray*}
\sum_{j = i+1}^K b_j^r(v) \leq \begin{cases}
0 & \text{ if } s_i(v) = \infty; \\
 10^{K - s_i(v) + 1} \cdot \delta & \text{ else if } s_i(v) \in \{i+1, \ldots, K\}.
\end{cases}
\end{eqnarray*}
\end{lemma}

\begin{lemma}
\label{label:bipartite:lm:bound:l2}
Consider any node $v \in V$ at level $\ell(v) = i \in \{2, \ldots, K\}$. We have:
$$b_i(v) + b_i^{r}(v) \leq 1 - \sum_{j=i+1}^K b_j^r(v) - \sum_{j=i+1}^K W_v(w_j).$$
\end{lemma}

\begin{lemma}
\label{label:bipartite:lm:bound:l1}
Consider any node $v \in V$ at level $\ell(v) = 1$. We have:
$$b_1^*(v) \leq 1 - \sum_{j=2}^K b_j^r(v) - \sum_{j=2}^K W_v(w_j).$$
\end{lemma}

\subsubsection{Proof of Lemma~\ref{label:bipartite:lm:bound:residual:assignment}}
\label{label:bipartite:sec:lm:bound:residual:assignment}
We consider two possible cases, depending on equation~\ref{label:bipartite:eq:level-small}.
\begin{itemize}
\item {\em Case 1.} $s_i(v) = \infty$. 

Since $\ell(v) = i$, equation~\ref{label:bipartite:eq:label} and Definitions~\ref{label:bipartite:inv:node-partition:1},~\ref{label:bipartite:inv:node-partition:1:last} imply that $v \in Z_{j-1} = Z_j \setminus T_j$ for all $j \in \{i+1, \ldots, K\}$. Thus, we infer that:
\begin{equation}
\label{label:bipartite:eq:110} v \notin T_j \text{ for all } j \in \{i+1, \ldots, K\}.
\end{equation}
Since $s_i(v) = \infty$, by equation~\ref{label:bipartite:eq:level-small} we have $v \notin S_j$ for all $j \in \{i+1, \ldots, K\}$. This observation, along with equation~\ref{label:bipartite:eq:110}, implies that:
\begin{equation}
\label{label:bipartite:eq:111} v \notin S_j \cup T_j \text{ for all } j \in \{i+1, \ldots, K\}.
\end{equation}
Equation~\ref{label:bipartite:eq:111} and Definition~\ref{label:bipartite:inv:residual:capacity} imply that:
\begin{equation}
\label{label:bipartite:eq:112}
b^r_j(v) = 0 \text{ for all } j \in \{i+1, \ldots, K\}.
\end{equation}
The lemma follows from summing equation~\ref{label:bipartite:eq:112} over all $j \in \{i+1, \ldots, K\}$.

\item {\em Case 2.} $s_i(v) \in \{i+1, \ldots, K\}$. 

Let $s_i(v) = k$ for some $k \in \{i+1, \ldots, K\}$. Note that $k$ is the minimum index in $\{i+1, \ldots, K\}$ for which $v \in S_k$. Thus, we get:
\begin{equation}
\label{label:bipartite:eq:114}
v \notin S_j \text{ for all } j \in \{i+1, \ldots, k-1\}. 
\end{equation}
Since $\ell(v) = i$, equation~\ref{label:bipartite:eq:label} and Definitions~\ref{label:bipartite:inv:node-partition:1},~\ref{label:bipartite:inv:node-partition:1:last} imply that $v \in Z_{j-1} = Z_j \setminus T_j$ for all $j \in \{i+1, \ldots, K\}$. Thus, similar to Case 1, we infer that:
\begin{equation}
\label{label:bipartite:eq:115} v \notin T_j \text{ for all } j \in \{i+1, \ldots, K\}.
\end{equation}
Equations~\ref{label:bipartite:eq:114} and~\ref{label:bipartite:eq:115} imply that:
$$v \notin S_j \cup T_j \text{ for all } j \in \{i+1, \ldots, k-1\}.$$
Using the above inequality in conjunction with Definition~\ref{label:bipartite:inv:residual:capacity}, we conclude that:
\begin{equation}
\label{label:bipartite:eq:116}
b_j^r(v) = 0 \text{ for all } j \in \{i+1, \ldots, k-1\}. 
\end{equation}
Now, we focus on the interval $[k, K]$. From equation~\ref{label:bipartite:eq:115}, we get $v \notin T_j$ for all $j \in \{k, \ldots, K\}$. Thus, Definition~\ref{label:bipartite:inv:residual:capacity} implies that $b_j^r(v) \in \{0, 8 \cdot 10^{K-j} \cdot \delta\}$ for all $j \in \{k, \ldots, K\}$. Thus, we get:
\begin{equation}
\label{label:bipartite:eq:117}
b_j^r(v) \leq 8 \cdot 10^{K-j} \cdot \delta \text{ for all } j \in \{k, \ldots, K\}.
\end{equation}
Adding equations~\ref{label:bipartite:eq:116} and~\ref{label:bipartite:eq:117} we get:
\begin{eqnarray}
\sum_{j=i+1}^K b_j^r(v) & = & \sum_{j=i+1}^{k-1} b_j^r(v) + \sum_{j=k}^K b_j^r(v) \nonumber \\
& \leq & 0 + \sum_{j = k}^K 8 \cdot 10^{K-j} \cdot \delta \nonumber \\
& \leq & 10^{K-k+1} \cdot \delta \nonumber \\
& = & 10^{K-s_i(v)+1}\label{label:bipartite:eq:118}
\end{eqnarray}
This concludes the proof of the lemma in this case.
\end{itemize}

\subsubsection{Proof of Lemma~\ref{label:bipartite:lm:bound:l2}}
\label{label:bipartite:sec:lm:bound:l2}

Since $\ell(v) = i \geq 2$, Observation~\ref{label:bipartite:lm:tight} implies that $v \in T_i$. We now consider two possible cases, depending on the value of $s_i(v)$.

\begin{itemize}
\item {\em Case 1.} $s_i(v) = \infty$.

In this case, since $v \in T_i$, Definition~\ref{label:bipartite:inv:residual:capacity} states that:
\begin{equation}
\label{label:bipartite:eq:newplan:1}
b_i^r(v) = \delta
\end{equation}
Further, Definition~\ref{label:bipartite:inv:normal:capacity} states that:
\begin{equation}
\label{label:bipartite:eq:newplan:2000}
b_i(v) = 1 - \delta - \sum_{j=i+1}^K \mathcal{W}_v(w_j)
\end{equation}
Invariant~\ref{label:bipartite:inv:discretized} states that:
\begin{equation}
\label{label:bipartite:eq:newplan:2001}
\mathcal{W}_v(w_j) \geq W_v(w_j) \text{ for all } j \in \{i,  \ldots, K\}
\end{equation} 
From equations~\ref{label:bipartite:eq:newplan:2000} and~\ref{label:bipartite:eq:newplan:2001}, we get:
\begin{equation}
\label{label:bipartite:eq:newplan:2}
b_i(v) \leq 1 - \delta - \sum_{j=i+1}^K W_v(w_j)
\end{equation}
Finally, Lemma~\ref{label:bipartite:lm:bound:residual:assignment} states that:
\begin{equation}
\label{label:bipartite:eq:newplan:3}
\sum_{j=i+1}^K b_j^r(v) \leq 0
\end{equation}
From equations~\ref{label:bipartite:eq:newplan:1},~\ref{label:bipartite:eq:newplan:2} and~\ref{label:bipartite:eq:newplan:3}, we get:
$$b_i(v) + b_i^r(v) \leq 1 - \sum_{j=i+1}^K b_j^r(v) - \sum_{j=i+1}^K W_v(w_j).$$
This concludes the proof of the lemma in this case.

\item {\em Case 2.} $s_i(v) \in \{i+1, \ldots, K\}$. 

In this case, since $v \in T_i$,  Definition~\ref{label:bipartite:inv:residual:capacity} states that:
\begin{equation}
\label{label:bipartite:eq:newplan:4}
b_i^r(v) = 4 \cdot 10^{(K-s_i(v)+1)} \cdot \delta
\end{equation}
Further, Definition~\ref{label:bipartite:inv:normal:capacity} states that:
\begin{equation}
\label{label:bipartite:eq:newplan:5000}
b_i(v) = 1 - 5 \cdot 10^{(K-s_i(v)+1)} \cdot \delta - \sum_{j=i+1}^K \mathcal{W}_v(w_j)
\end{equation}
Invariant~\ref{label:bipartite:inv:discretized} states that:
\begin{equation}
\label{label:bipartite:eq:newplan:5001}
\mathcal{W}_v(w_j) \geq W_v(w_j) \text{ for all } j \in \{i, \ldots, K\}. 
\end{equation}
From equations~\ref{label:bipartite:eq:newplan:5000} and~\ref{label:bipartite:eq:newplan:5001}, we get:
\begin{equation}
\label{label:bipartite:eq:newplan:5}
b_i(v) \leq 1 - 5 \cdot 10^{(K-s_i(v)+1)} \cdot \delta - \sum_{j=i+1}^K W_v(w_j)
\end{equation}
Finally, since $s_i(v) \in \{i+1, \ldots, K\}$, Lemma~\ref{label:bipartite:lm:bound:residual:assignment} implies that:
\begin{equation}
\label{label:bipartite:eq:newplan:6}
\sum_{j=i+1}^K b_j^r(v)  \leq 10^{(K-s_i(v)+1)} \cdot \delta
\end{equation}
From equations~\ref{label:bipartite:eq:newplan:4},~\ref{label:bipartite:eq:newplan:5} and~\ref{label:bipartite:eq:newplan:6}, we get:
$$b_i(v) + b_i^r(v) \leq 1 - \sum_{j=i+1}^K b_j^r(v) - \sum_{j=i+1}^K W_v(w_j).$$
This concludes the proof of the lemma in this case. 
\end{itemize}

\subsubsection{Proof of Lemma~\ref{label:bipartite:lm:bound:l1}}
\label{label:bipartite:sec:lm:bound:l1}

We consider two possible cases, depending on the value of $s_1(v)$.
\begin{itemize}
\item {\em Case 1.} $s_1(v) = \infty$. 

In this case, Lemma~\ref{label:bipartite:lm:bound:residual:assignment} states that:
\begin{equation}
\label{label:bipartite:eq:newplan:10}
\sum_{j=2}^K b_j^r(v) \leq 0
\end{equation}
Definition~\ref{label:bipartite:inv:node:capacity:one} states that:
\begin{equation}
\label{label:bipartite:eq:newplan:11000}
b_1^*(v) = 1 - \sum_{j=2}^K \mathcal{W}_v(w_j)
\end{equation}
Invariant~\ref{label:bipartite:inv:discretized} states that:
\begin{equation}
\label{label:bipartite:eq:newplan:11001}
\mathcal{W}_v(w_j) \geq W_v(w_j) \text{ for all } j \in \{2, \ldots, K\}.
\end{equation}
From equations~\ref{label:bipartite:eq:newplan:11000} and~\ref{label:bipartite:eq:newplan:11001}, we get:
\begin{equation}
\label{label:bipartite:eq:newplan:11}
b_1^*(v) \leq 1 - \sum_{j=2}^K W_v(w_j)
\end{equation}
From equations~\ref{label:bipartite:eq:newplan:10} and~\ref{label:bipartite:eq:newplan:11}, we get:
$$b_1^*(v) \leq 1 - \sum_{j=2}^K b_j^r(v) - \sum_{j=2}^K W_v(w_j).$$
This concludes the proof of the lemma in this case. 

\item {\em Case 2.} $s_1(v) \in \{2, \ldots, K\}$.

In this case,  Lemma~\ref{label:bipartite:lm:bound:residual:assignment} implies that
\begin{equation}
\label{label:bipartite:eq:newplan:12}
\sum_{j=2}^K b_j^r(v) \leq 10^{(K-s_1(v)+1)} \cdot \delta
\end{equation}
Definition~\ref{label:bipartite:inv:node:capacity:one} states that:
\begin{equation}
\label{label:bipartite:eq:newplan:13000}
b_1^*(v) = 1 - 10^{(K-s_1(v)+1)} \cdot \delta - \sum_{j=2}^K \mathcal{W}_v(w_j)
\end{equation}
Invariant~\ref{label:bipartite:inv:discretized} states that:
\begin{equation}
\label{label:bipartite:eq:newplan:13001}
\mathcal{W}_v(w_j) \geq W_v(w_j) \text{ for all } j \in \{2, \ldots, K\}.
\end{equation}
From equations~\ref{label:bipartite:eq:newplan:13000} and~\ref{label:bipartite:eq:newplan:13001}, we get:
\begin{equation}
\label{label:bipartite:eq:newplan:13}
b_1^*(v) \leq 1 - 10^{(K-s_1(v)+1)} \cdot \delta - \sum_{j=2}^K W_v(w_j)
\end{equation}
From equations~\ref{label:bipartite:eq:newplan:12} and~\ref{label:bipartite:eq:newplan:13}, we get:
$$b_1^*(v) \leq 1 - \sum_{j=2}^K b_j^r(v) - \sum_{j=2}^K W_v(w_j).$$
This concludes the proof of the lemma in this case. 
\end{itemize}

\subsubsection{Proof of Theorem~\ref{label:bipartite:th:feasible:main}}
\label{label:bipartite:sec:th:feasible:main}

Fix any node $v \in V$. We will show that $W_v(w+w_1^*+w^r) \leq 1$. This will imply that $(w+w_1^*+w^r)$ is a fractional matching in the graph $G = (V, E)$. We consider two possible cases, depending on the level of $v$. 

\begin{itemize}
\item $\ell(v) = i \in \{2, \ldots, K\}$.

In this case, applying Lemma~\ref{label:bipartite:lm:bound:l2}, we get:
\begin{equation}
\label{label:bipartite:eq:newplan:50}
b_i(v) + \sum_{j=i+1}^K W_v(w_j) + \sum_{j=i}^K b_j^r(v) \leq 1
\end{equation}
Invariant~\ref{label:bipartite:inv:normal:assignment} states that $w_i$ is a fractional $b$-matching with respect to the node-capacities $\{b_i(x) \}$. Hence, we get:
\begin{equation}
\label{label:bipartite:eq:newplan:51}
W_v(w_i) \leq b_i(v)
\end{equation}
Invariant~\ref{label:bipartite:inv:residual:matching} states that for all $j \geq i$, $w^r_j$ is a fractional $b$-matching with respect to the node-capacities $\{b_j^r(x) \}$. Hence, we get:
\begin{equation}
\label{label:bipartite:eq:newplan:52}
W_v(w^r_j) \leq b_j^r(v) \text{ for all } j \in \{i, \ldots, K\}.
\end{equation}
From equations~\ref{label:bipartite:eq:newplan:50},~\ref{label:bipartite:eq:newplan:51} and~\ref{label:bipartite:eq:newplan:52}, we get:
\begin{equation}
\label{label:bipartite:eq:newplan:53}
\sum_{j=i}^K W_v(w_j) + \sum_{j=i}^K W_v(w_j^r) \leq 1
\end{equation}
Since $\ell(v) = i \in \{2, \ldots, K\}$, we have $v \notin Z_j$ for all $j \in \{1, \ldots, i-1\}$. Thus, we conclude that:
\begin{equation}
\label{label:bipartite:eq:newplan:54}
W_v(w_j) = W_v(w_j^r) = 0 \text{ for all } j \in \{2, \ldots, i-1\}; \text{ and } W_v(w_1^*) = 0.
\end{equation}
From equations~\ref{label:bipartite:eq:newplan:53} and~\ref{label:bipartite:eq:newplan:54}, we get:
\begin{eqnarray*}
W_v(w+w_1^*+w_r) & = & W_v(w_1^*) + W_v(w) + W_v(w^r) \\
& = & W_v(w_1^*) + \sum_{j=2}^{i-1} \left(W_v(w_j) + W_v(w^r_j) \right) + \sum_{j=i}^{K} \left(W_v(w_j) + W_v(w^r_j) \right) \\
& \leq & 0 + 0 + 1 \\
& = & 1
\end{eqnarray*}
This concludes the proof of the lemma in this case.

\item $\ell(v) = 1$.

In this case, applying Lemma~\ref{label:bipartite:lm:bound:l1}, we get:
\begin{equation}
\label{label:bipartite:eq:newplan:55}
b_1^*(v) + \sum_{j=2}^K W_v(w_j) + \sum_{j=2}^K b_j^r(v) \leq 1
\end{equation}
Invariant~\ref{label:bipartite:inv:assignment:one} states that $w_1^*$ is a fractional $b$-matching with respect to the capacities $\{b_1^*(x) \}$. Hence, we get:
\begin{equation}
\label{label:bipartite:eq:newplan:56}
W_v(w_1^*) \leq b_1^*(v)
\end{equation}
Invariant~\ref{label:bipartite:inv:residual:matching} states that for all $j \in \{2, \ldots, K\}$, $w_j^r$ is a fractional $b$-matching with respect to the capacities $\{b_j^r(x)\}$. Hence, we get:
\begin{equation}
\label{label:bipartite:eq:newplan:57}
W_v(w_j^r) \leq b_j^r(v) \text{ for all } j \in \{2, \ldots, K\}.
\end{equation}
From equations~\ref{label:bipartite:eq:newplan:55},~\ref{label:bipartite:eq:newplan:56} and~\ref{label:bipartite:eq:newplan:57}, we get:
\begin{eqnarray*}
W_v(w+w_1^*+w^r) & = & W_v(w_1^*) + W_v(w) + W_v(w^r) \\
& = & W_v(w_1^*) + \sum_{j=2}^K W_v(w_j) + \sum_{j=2}^K W_v(w^r) \\
& \leq & b_1^*(v) + \sum_{j=2}^K W_v(w_j) + \sum_{j=2}^K b_j^r(v) \\
& \leq & 1.
\end{eqnarray*}
This concludes the proof of the lemma in this case. 
\end{itemize}

\subsection{Approximation Guarantee of the Solution}
\label{label:bipartite:sec:approx}

We want to show that the size of the fractional assignment $(w+w_1^*+w^r)$ is {\em strictly} within a factor of two of the maximum cardinality matching in the input graph $G = (V, E)$. So we devote this section to the proof of Theorem~\ref{label:bipartite:th:approx:main}. We start by recalling some notations that will be used throughout the rest of this section.
\begin{itemize}
\item Given any subset of edges $E' \subseteq E$, we let $V(E') = \{ u \in V : \text{deg}_{E'}(v) > 0 \}$ be the set of endpoints of the edges in $E'$. Thus, for a  matching $M \subseteq E$, the set of matched nodes is given by $V(M)$. 
\item Consider any level $i \in \{2, \ldots, K\}$, and recall that Invariant~\ref{label:bipartite:inv:discretized} introduces the concept of a ``discretized'' node-weight $\W_v(w_i)$ for all $v \in Z_i$. We ``extend'' this  definition as follows.  If a node $v$ belongs to a level $\ell(v) = i \in \{2, \ldots, K\}$, then $v \notin Z_j$ for all $j \in \{1, \ldots, i-1\}$. In this case, we set $\W_v(w_j) =  W_v(w_j) = 0$ for all $j \in \{1, \ldots, i-1\}$. Further, for every node $u \in V$, we define $\W_u(w) = \sum_{j=2}^K \W_u(w_j)$.  
\end{itemize}
Instead of directly trying to bound the size of the fractional assignment $(w+w_1^*+w^r)$, we first prove the following theorem. Fix any  matching $M^* \subseteq E$ in the input graph $G$. The theorem below upper bounds the size of this matching in terms of the total discretized node-weight received by the matched nodes, plus the total weight received by all the nodes from $w_1^*$ and $w^r$.  The proof of Theorem~\ref{label:bipartite:th:approx:main:new} appears in Section~\ref{label:bipartite:sec:roadmap}. 
\begin{theorem}
\label{label:bipartite:th:approx:main:new}
Consider any  matching $M^* \subseteq E$ in the input graph $G = (V, E)$. Then we have:
$$\sum_{v \in V(M^*)}  \W_v(w) + \sum_{v \in V} \left( W_v(w_1^*) + W_v(w^r) \right) \geq (1+\epsilon)^{-1} \cdot (1+\delta/3) \cdot |M^*|.$$
\end{theorem}
To continue with the proof of  Theorem~\ref{label:bipartite:th:approx:main}, we need the observation that the discretized weight of a node is very close to its normal weight. This intuition is quantified in the following claim. 

\begin{claim}
\label{label:bipartite:cl:approx:main}
For every node $v \in V$, we have $W_v(w) \geq \W_v(w) - 2 K  \epsilon$. 
\end{claim}

\begin{proof}
Invariant~\ref{label:bipartite:inv:discretized} implies that:
$$W_v(w_j) \geq \W_v(w_j) - 2 \epsilon \text{ for all } j \in \{2, \ldots, K\}.$$
Summing the above inequality over all levels $j \in \{2, \ldots, K\}$, we get:
\begin{eqnarray*}
W_v(w) & = & \sum_{j =2}^K W_v(w_j) \\
& \geq  & \sum_{j=2}^K (\W_v(w_j) - 2 \epsilon) \\
& = & \sum_{j=2}^K \W_v(w_j) - 2 \epsilon (K-1)   \\
& = & \W_v(w) - 2   \epsilon (K-1) \geq \W_v(w) -  2 \epsilon K.
\end{eqnarray*} 
\end{proof}

We can now lower bound the total weight received by all the nodes under  $(w+w_1^*+w^r)$ as follows.

\begin{claim}
\label{label:bipartite:cl:approx:main:1}
Consider any matching $M^* \subseteq E$ in the input graph $G = (V, E)$. We have:
$$\sum_{v \in V} W_v(w+w_1^*+w^r) \geq \left(\frac{(1+\delta/3)}{(1+\epsilon)} - 4 K \epsilon\right) \cdot |M^*|.$$
\end{claim}

\begin{proof}
We infer that:
\begin{eqnarray}
\sum_{v \in V} W_v(w+w_1^*+w^r) & = & \sum_{v \in V} W_v(w) + \sum_{v \in V} W_v(w_1^*) + \sum_{v \in V} W_v(w^r) \nonumber \\
& \geq & \sum_{v \in V(M^*)} W_v(w) + \sum_{v \in V} W_v(w_1^*) + \sum_{v \in V} W_v(w^r) \label{label:bipartite:eq:run:1} \\
& \geq & \sum_{v \in V(M^*)} \left( \W_v(w) -  2 \epsilon K\right) + \sum_{v \in V} W_v(w_1^*) + \sum_{v \in V} W_v(w^r) \label{label:bipartite:eq:run:2} \\
& = & \sum_{v \in V(M^*)} \W_v(w) - 2 \cdot |M^*| \cdot (2\epsilon K) + \sum_{v \in V} W_v(w_1^*) + \sum_{v \in V} W_v(w^r) \label{label:bipartite:eq:run:3} \\
& = & \left( \sum_{v \in V(M^*)} \W_v(w) + \sum_{v \in V} W_v(w_1^*) + \sum_{v \in V} W_v(w^r)\right) - (4\epsilon K) \cdot |M^*| \nonumber \\
& \geq & \frac{(1+\delta/3)}{(1+\epsilon)} \cdot |M^*| - (4\epsilon K) \cdot |M^*| \label{label:bipartite:eq:run:4} \\
& = & \left(\frac{(1+\delta/3)}{(1+\epsilon)} - 4\epsilon K\right) \cdot |M^*| \label{label:bipartite:eq:run:5}
\end{eqnarray}
Equation~\ref{label:bipartite:eq:run:1} holds since $V(M^*) \subseteq V$. Equation~\ref{label:bipartite:eq:run:2} follows from Claim~\ref{label:bipartite:cl:approx:main}.  Equation~\ref{label:bipartite:eq:run:3} holds since $|V(M^*)| = 2 \cdot |M^*|$ as no two edges in $M^*$ share a common endpoint.
Equation~\ref{label:bipartite:eq:run:4} follows from Theorem~\ref{label:bipartite:th:approx:main:new}. Finally, the claim follows from equation~\ref{label:bipartite:eq:run:5}. 
\end{proof}

We now use Claim~\ref{label:bipartite:cl:approx:main:1} to lower bound the size of the fractional assignment $(w+w_1^*+w^r)$ by relating the sum of the node-weights with the sum of the edge-weights. Specifically, we  note that  each edge $e \in E$ contributes $2 \cdot (w(e) + w_1^*(e) + w^r(e))$ towards the sum $\sum_{v \in V} W_v(w+w_1^*+w^r)$. Hence, we infer that the sum of the edge-weights under $(w+w_1^*+w^r)$ is exactly $(1/2) \cdot \sum_{v \in V} W_v(w+w_1^*+w^r)$. 
\begin{equation}
\label{label:bipartite:eq:run:10}
\sum_{e \in E} (w(e) + w_1^*(e) + w^r(e) ) = (1/2) \cdot \sum_{v \in V} W_v(w+w_1^*+w^r)
\end{equation}

Accordingly, Claim~\ref{label:bipartite:cl:approx:main:1} and equation~\ref{label:bipartite:eq:run:10}  imply that the size of the fractional assignment $(w+w_1^*+w^r)$ is at least $f$ times the size of any matching $M^* \subseteq E$ in $G = (V, E)$, where:
$$f = (1/2) \cdot \left(\frac{(1+\delta/3)}{(1+\epsilon)} - 4 \epsilon K\right).$$
Setting $M^*$ to be the maximum cardinality matching in $G$, we derive the proof of Theorem~\ref{label:bipartite:th:approx:main}. To summarize, in order to prove Theorem~\ref{label:bipartite:th:approx:main}, it suffices to prove Theorem~\ref{label:bipartite:th:approx:main:new}. This is done in Section~\ref{label:bipartite:sec:roadmap}.

\subsubsection{Proof of Theorem~\ref{label:bipartite:th:approx:main:new}}
\label{label:bipartite:sec:roadmap}

Define the ``level'' of an edge to be the maximum level among its endpoints, i.e., we have:
\begin{equation}
\label{label:bipartite:eq:level:edge} 
\ell(u, v) = \max \left( \ell(u), \ell(v)\right) \text{ for every edge } (u,v) \in E.
\end{equation} 
Now, for every $i \in \{1, \ldots, K\}$, let $M^*_i = \{ (u,v) \in M^* : \ell(u, v) = i\}$ denote the subset of those  edges in $M^*$ that are at level $i$. It is easy to check that the subsets $M^*_1, \ldots, M^*_K$ partition the edge-set $M^*$. 

 By Observation~\ref{label:bipartite:lm:tight}, if a  node $v \in V$ is at level $\ell(v) = i \in \{2, \ldots, K\}$, then we also have $v \in T_i$. Consider any edge $(x, y) \in M^*_i$, where $i \geq 2$. Without any loss of generality, suppose that $\ell(x) = i$ and hence $x \in T_i$. Consider the other endpoint $y$ of this matched edge $(x,y)$. By definition, we have $\ell(y) \leq i$, which means that $y \in Z_i$. Now, Invariant~\ref{label:bipartite:inv:node-partition} states that the set $Z_i$ is partitioned into subsets $T_i, B_i$ and $S_i$. Thus, there are two mutually exclusive cases to consider: (1) either $y \in T_i \cup B_i$ or (2) $y \in S_i$. Accordingly, we  partition the edge-set $M^*_i$ into two subsets -- $M'_i$ and $M''_i$ -- as defined below.
\begin{enumerate}
\item $M'_i = \{ (u, v) \in M_i^* : u \in T_i \text{ and } v \in T_i \cup B_i\}$. This consists of the set of matched edges in $M_i^*$ with one endpoint in $T_i$ and the other endpoint in $T_i \cup B_i$.
\item $M''_i = \{ (u,v) \in M_i^* : u \in T_i \text{ and } v \in S_i \}$. This consists of the set of matched edges in $M_i^*$ with one endpoint in $T_i$ and the other endpoint in $S_i$. 
\end{enumerate}
To summarize, the set of  edges in $M^*$ is partitioned into the following subsets: 
$$M_1^*, \{ M'_2, M''_2\}, \{ M'_3, M''_3\}, \ldots, \{M'_K, M''_K\}.$$
We now define the matchings $M'$ and $M''$ as follows.
\begin{equation}
\label{label:bipartite:eq:newplan:100}
M' = \bigcup_{i=2}^K M'_i \text{ and } M'' = \bigcup_{i=2}^K M''_i
\end{equation}
Hence, the set of matched edges $M^*$ is partitioned by the subsets $M_1^*, M'$ and $M''$.  We will now consider the matchings $M_1^*$, $M'$ and $M''$ one after the other, and upper bound their sizes in terms of the node-weights.   

 Lemma~\ref{label:bipartite:lm:approx:last} helps us bound the size of the matching $M_1^*$. Note that each edge $(u,v) \in M_1^*$ has $\ell(u, v) = \max(\ell(u), \ell(v)) = 1$, and hence we must have $\ell(u) = \ell(v) = 1$. In other words, every edge in $M_1^*$ has both of its endpoints in level one. Since $E_1$ is the set of edges connecting the nodes at level one (see Definition~\ref{label:bipartite:inv:node-partition:1:last}), we infer that $M_1^* \subseteq E_1$. We will later apply Lemma~\ref{label:bipartite:lm:approx:last} by setting $M = M^*_1$. The proof of Lemma~\ref{label:bipartite:lm:approx:last} appears in Section~\ref{label:bipartite:sec:lm:approx:last}. 

\begin{lemma}
\label{label:bipartite:lm:approx:last}
Consider any matching $M \subseteq E_1$ in the subgraph $G_1 = (Z_1, E_1)$. Then we have:
$$\sum_{v \in V(M)} \W_v(w) + \sum_{v \in Z_1} W_v(w^*_1) \geq (1+\epsilon)^{-1} \cdot (1+\delta) \cdot |M|.$$
\end{lemma}

Lemma~\ref{label:bipartite:lm:approx:normal} helps us bound the size of the matching $M'_i$, for all $i \in \{2, \ldots, K\}$. Note that by definition every edge $(u, v) \in M'_i$ has one endpoint $u \in T_i$ and the other endpoint $v \in T_i \cup B_i$. Thus, we will later apply Lemma~\ref{label:bipartite:lm:approx:normal} by setting $M = M'_i$. The proof of Lemma~\ref{label:bipartite:lm:approx:normal} appears in Section~\ref{label:bipartite:sec:lm:approx:normal}.

\begin{lemma}
\label{label:bipartite:lm:approx:normal}
Consider any level $i \in \{2, \ldots, K\}$. Let $M \subseteq E_i$ be a matching in the subgraph $G_i = (Z_i, E_i)$ such that every edge $(u, v) \in M$ has one endpoint $u \in T_i$ and the other endpoint $v \in T_i \cup B_i$. Then we have:
$$\sum_{u \in V(M)} \W_u(w) \geq (1+3 \delta) \cdot |M|.$$
\end{lemma}

Lemma~\ref{label:bipartite:lm:approx:residual} helps us bound the size of the matching $M''_i$, for all $i \in \{2, \ldots, K\}$. Note that by definition every edge $(u, v) \in M''_i$ has one endpoint $u \in T_i$ and the other endpoint $v \in S_i$. Thus, we can apply Lemma~\ref{label:bipartite:lm:approx:residual} by setting $M = M''_i$. The proof of Lemma~\ref{label:bipartite:lm:approx:residual} appears in Section~\ref{label:bipartite:sec:lm:approx:residual}.

\begin{lemma}
\label{label:bipartite:lm:approx:residual}
Consider any level $i \in \{2, \ldots, K \}$. Let $M \subseteq E_i$ be a matching in the subgraph $G_i = (Z_i, E_i)$ such that every edge $(u, v) \in M$ has one endpoint  $u \in T_i$ and the other endpoint $v \in S_i$.  We have:
$$\sum_{u \in V(M)} \W_u(w) + \sum_{v \in Z_i} W_v(w_i^r) \geq (1+\delta/3) \cdot |M|.$$
\end{lemma}

To proceed with the proof of Theorem~\ref{label:bipartite:th:approx:main:new}, we first apply Lemma~\ref{label:bipartite:lm:approx:last} to get the following simple claim. This claim upper bounds the size of $M_1^*$ in terms of the total discretized weight received  by its endpoints, plus the total weight received from  $w_1^*$ by the nodes in $V$. 
\begin{claim}
\label{label:bipartite:cl:newplan:1}
We have:
$$\sum_{v \in V(M_1^*)} \W_v(w) + \sum_{v \in V} W_v(w_1^*) \geq (1+\epsilon)^{-1} \cdot (1+\delta) \cdot |M_1^*|.$$
\end{claim}

\begin{proof}
Since $M_1^*$ is a matching in  $G_1 = (Z_1, E_1)$,  from Lemma~\ref{label:bipartite:lm:approx:last}, we get:
$$\sum_{v \in V(M_1^*)} \W_v(w) + \sum_{v \in Z_1} W_v(w_1^*) \geq (1+\epsilon)^{-1} \cdot (1+\delta) \cdot |M_1^*|.$$
Since $Z_1 \subseteq V$, the claim follows from the above inequality. 
\end{proof}

Recall that the edge-set $M'$ is partitioned into subsets $M'_2, \ldots, M'_K$. Thus, we can upper  bound  $|M'|$ by applying Lemma~\ref{label:bipartite:lm:approx:normal} with $M = M'_i$, for $i \in \{2, \ldots, K\}$, and summing over the resulting inequalities. This is done in the claim below. This claim upper bounds the size of $M'$ in terms of the total discretized weight received by its endpoints.
\begin{claim}
\label{label:bipartite:cl:newplan:2}
We have:
$$\sum_{v \in V(M')} \W_v(w) \geq (1+3\delta) \cdot |M'|.$$
\end{claim}

\begin{proof}
For every $i \in \{2, \ldots, K\}$, we can apply Lemma~\ref{label:bipartite:lm:approx:normal} on the matching $M'_i$ to get:
\begin{eqnarray}
 \sum_{v \in V(M'_i)} \W_v(w) \geq  (1+3\delta) \cdot |M'_i| \label{label:bipartite:eq:newplan:101}
\end{eqnarray}
Now, note that the set of edges $M'$ has been partitioned into subsets $M'_2, \ldots, M'_K$.  Hence, summing equation~\ref{label:bipartite:eq:newplan:101} over all levels $2 \leq i \leq K$, we get:
\begin{eqnarray*}
\sum_{v \in V(M')} \W_v(w) & = & \sum_{i=2}^K \sum_{v \in V(M'_i)} \W_v(w) \\
& \geq & \sum_{i=2}^K (1+3\delta) \cdot |M'_i| \\
& = & (1+3\delta) \cdot |M'|
\end{eqnarray*}
\end{proof}

Recall that the edge-set $M''$ is partitioned into subsets $M''_2, \ldots, M''_K$. Further, we have $w^r = \sum_{j=2}^K w^r_j$. Thus, we can  apply Lemma~\ref{label:bipartite:lm:approx:residual} with $M = M''_i$, for $i \in \{2, \ldots, K\}$, and  sum over the resulting inequalities to get the following claim. This claim upper bounds the size of $M''$ in terms of the total discretized weight received by its endpoints,  plus the total weight received from $w^r$ by the nodes in $V$. 
\begin{claim}
\label{label:bipartite:cl:newplan:3}
We have:
$$\sum_{v \in V(M'')} \W_v(w) + \sum_{v \in V} W_v(w^r) \geq (1+\delta/3) \cdot |M''|.$$
\end{claim}

\begin{proof}
For every $i \in \{2, \ldots, K\}$, since $M''_i$ is a matching in $G_i = (Z_i, E_i)$ such that each edge $(u, v) \in M''_i$ has one endpoint in $T_i$ and the other endpoint in $S_i$, we can apply Lemma~\ref{label:bipartite:lm:approx:residual} on $M''_i$ to get:
$$\sum_{v \in V(M''_i)} \W_v(w) + \sum_{v \in Z_i} W_v(w_i^r) \geq (1+\delta/3) \cdot |M''_i|.$$
Since $Z_i \subseteq V$, we have $\sum_{v \in V} W_v(w_i^r) \geq \sum_{v \in Z_i} W_v(w_i^r)$.  Hence, the above inequality implies that:
\begin{equation}
\label{label:bipartite:eq:newplan:200}
\sum_{v \in V(M''_i)} \W_v(w) + \sum_{v \in V} W_v(w_i^r) \geq (1+\delta/3) \cdot |M''_i|  \text{ for all } i \in \{2, \ldots, K\}.
\end{equation}
Since the set of edges $M''$ is partitioned into subsets $M''_2, \ldots, M''_K$, and since $w^r = \sum_{i=2}^K w_i^r$, we get:
\begin{eqnarray}
\sum_{v \in V(M'')} \W_v(w) + \sum_{v \in V} W_v(w^r) & = & \sum_{i = 2}^K \sum_{v \in V(M''_i)} \W_v(w) + \sum_{v \in V} \sum_{i=2}^K W_v(w^r_i) \nonumber \\
& = & \sum_{i=2}^K \left( \sum_{v \in V(M''_i)} \W_v(w) + \sum_{v \in V} W_v(w^r_i) \right) \nonumber \\
& \geq & \sum_{i=2}^K (1+\delta/3) \cdot |M''_i| \label{label:bipartite:eq:newplan:201} \\
& = & (1+\delta/3) \cdot |M''|. \nonumber
\end{eqnarray}
Equation~\ref{label:bipartite:eq:newplan:201} follows from equation~\ref{label:bipartite:eq:newplan:200}. This concludes the proof of the claim.
\end{proof}

Since the set of matched edges $M^*$ is partitioned into subsets $M_1^*, M'$ and $M''$, we can now add the inequalities in Claims~\ref{label:bipartite:cl:newplan:1},~\ref{label:bipartite:cl:newplan:2} and~\ref{label:bipartite:cl:newplan:3} to derive Theorem~\ref{label:bipartite:th:approx:main:new}. The primary observation is that the sum in the left hand side of Theorem~\ref{label:bipartite:th:approx:main:new} upper bounds the sum of the left hand sides of the inequalities stated in these three claims. This holds since  no two edges in $M^*$ share a common endpoint, and, accordingly, no amount of node-weight is counted twice if we sum the left hand sides of these claims. To bound the sum of the right hand sides of these claims, we use the fact that $|M^*| = |M_1^*| + |M'| + |M''|$. Specifically, we get:
\begin{eqnarray}
& & \sum_{v \in V(M)}  \W_v(w) +  \sum_{v \in V} \left(W_v(w_1^*) + W_v(w^r)\right) \nonumber \\ 
& = & \sum_{v \in V(M)} \W_v(w) + \sum_{v \in V} W_v(w_1^*) + \sum_{v \in V} W_v(w^r) \nonumber \\
& = & \left(\sum_{v \in V(M_1^*)} \W_v(w) + \sum_{v \in V(M')} \W_v(w) + \sum_{v \in V(M'')} \W_v(w)\right) + \sum_{v \in V} W_v(w_1^*) + \sum_{v \in V} W_v(w^r) \label{label:bipartite:eq:revised:1} \\
& = & \left(\sum_{v \in V(M_1^*)} \W_v(w) + \sum_{v \in V} W_v(w^*_1)\right) + \left(\sum_{v \in V(M'')} \W_v(w) + \sum_{v \in V} W_v(w^r) \right) + \sum_{v \in V(M')} \W_v(w) \nonumber \\
& \geq & (1+\epsilon)^{-1} \cdot  (1+\delta) \cdot |M_1^*| + (1+\delta/3) \cdot |M''| + (1+3\delta) \cdot |M'| \label{label:bipartite:eq:revised:2} \\
& \geq & (1+\epsilon)^{-1} \cdot (1+\delta/3) \cdot \left( |M_1^*| + |M''| + |M'| \right) \nonumber \\
& = & (1+\epsilon)^{-1} \cdot (1+\delta/3) \cdot |M^*| \label{label:bipartite:eq:revised:3}
\end{eqnarray}

Equation~\ref{label:bipartite:eq:revised:1} holds  since the edges in the matching $M^*$ are partitioned into subsets $M_1^*, M', M''$, and hence, the node-set $V(M^*)$ is also partitioned into subsets $V(M_1^*), V(M'), V(M'')$. Equation~\ref{label:bipartite:eq:revised:2} follows from Claims~\ref{label:bipartite:cl:newplan:1},~\ref{label:bipartite:cl:newplan:2} and~\ref{label:bipartite:cl:newplan:3}. Finally, Theorem~\ref{label:bipartite:th:approx:main:new} follows from equation~\ref{label:bipartite:eq:revised:3}.

\subsubsection{Proof of Lemma~\ref{label:bipartite:lm:approx:last}}
\label{label:bipartite:sec:lm:approx:last}

We will carefully construct a fractional $b$-matching $w'_1$ in the graph $G_1 = (Z_1, E_1)$ with respect to the node-capacities $\{b_1^*(u) \}, u \in Z_1$. Then we will show that:
\begin{equation}
\label{label:bipartite:eq:negative:600}
\sum_{v \in V(M)} \W_v(w) +\sum_{v \in Z_1} W_v(w'_1) \geq (1+\delta) \cdot |M|
\end{equation}
By Invariant~\ref{label:bipartite:inv:assignment:one}, we know that $w_1^*$ is a fractional $b$-matching  in the graph $G_1$ with respect to the node-capacities $\{b_1^*(u) \}, u \in Z_1$. Further, the size of $w_1^*$ is a $(1+\epsilon)$-approximation to the size of the maximum fractional $b$-matching with respect to the same node-capacities. Thus, we have:
\begin{equation}
\label{label:bipartite:eq:negative:601}
 \sum_{v \in Z_1} W_v(w_1^*) \geq (1+\epsilon)^{-1} \cdot \sum_{v \in Z_1} W_v(w'_1)
\end{equation}
From equations~\ref{label:bipartite:eq:negative:600} and~\ref{label:bipartite:eq:negative:601}, we can derive that:
\begin{eqnarray*}
\sum_{v \in V(M)} \W_v(w) +\sum_{v \in Z_1} W_v(w^*_1) & \geq & \sum_{v \in V(M)} \W_v(w) + (1+\epsilon)^{-1} \cdot \sum_{v \in Z_1} W_v(w'_1) \\
& \geq & (1+\epsilon)^{-1} \cdot (1+\delta) \cdot |M|
\end{eqnarray*}
Thus, the lemma follows from equations~\ref{label:bipartite:eq:negative:600} and~\ref{label:bipartite:eq:negative:601}. \\

It remains to prove equation~\ref{label:bipartite:eq:negative:600}. Towards this end, we first define the fractional assignment $w'_1$ below. 
\begin{equation}
\label{label:bipartite:eq:negative:602}
w'_1(u,v)  = \begin{cases}
\min(b_1^*(u), b_1^*(v)) & \text{ for all edges } (u, v) \in M; \\
0 & \text{ for all edges } (u, v) \in E_1 \setminus M.\end{cases}
\end{equation}
Since $M \subseteq E_1$ and since no two edges in $M$ share a common endpoint, $w'_1$ as defined above forms a fractional $b$-matching in the graph $G_1 = (Z_1, E_1)$ with respect to the node-capacities $\{b_1^*(u)\}, u \in Z_1$.  We will show that:
\begin{equation}
\label{label:bipartite:eq:negative:603}
\W_x(w) +W_x(w'_1) + \W_y(w) + W_y(w'_1) \geq (1+\delta) \text{ for every edge } (x, y) \in M.
\end{equation}
By equation~\ref{label:bipartite:eq:negative:602} we have $W_v(w'_1) = 0$ for all nodes $v \in Z_1 \setminus V(M)$. Hence, equation~\ref{label:bipartite:eq:negative:603} implies that:
\begin{eqnarray*}
\sum_{v \in V(M)} \W_v(w) + \sum_{v \in Z_1} W_v(w'_1)  \geq \sum_{(x,y) \in M} \left\{\W_x(w) +W_x(w'_1) + \W_y(w) + W_y(w'_1) \right\} \geq (1+\delta) \cdot |M|
\end{eqnarray*}
Hence, equation~\ref{label:bipartite:eq:negative:600} follows from equation~\ref{label:bipartite:eq:negative:603}. 
We thus focus on showing that equation~\ref{label:bipartite:eq:negative:603} holds. 

\paragraph{Proof of equation~\ref{label:bipartite:eq:negative:603}.} \ \\

\noindent Fix any edge $(x, y) \in M$, and  suppose that $b_1^*(x) \leq b_1^*(y)$. Then equation~\ref{label:bipartite:eq:negative:602} implies that:
\begin{equation}
\label{label:bipartite:eq:negative:605}
W_x(w'_1) = W_y(w'_1) = w'_1(x, y) = b_1^*(x)
\end{equation}
We claim that:
\begin{equation}
\label{label:bipartite:eq:negative:606}
\W_x(w) + 2 \cdot b_1^*(x) \geq 1+\delta
\end{equation}
From equations~\ref{label:bipartite:eq:negative:605} and~\ref{label:bipartite:eq:negative:606}, we can make the following deductions.
\begin{eqnarray*}
\W_x(w)+W_x(w'_1) + \W_y(w)+W_y(w'_1) & \geq &  \W_x(w) + W_x(w'_1) + W_y(w'_1) \\
& = & \W_x(w) + 2 \cdot b_1^*(x) \\
& \geq & 1+\delta
\end{eqnarray*}
Thus, in order to prove equation~\ref{label:bipartite:eq:negative:603}, it suffices to prove equation~\ref{label:bipartite:eq:negative:606}. For the rest of this section, we focus on proving equation~\ref{label:bipartite:eq:negative:606}. There are three possible cases to consider, depending on the value of $s_1(x)$. 
\begin{itemize}
\item {\em Case 1.} $s_1(x) = \infty$. 

Since $s_1(x) = \infty$, we have $x \notin S_i$ for all $i \in \{2, \ldots, K\}$. In particular, we have $s_2(x) = \infty$. Hence, Definition~\ref{label:bipartite:inv:normal:capacity} implies that:
\begin{equation}
\label{label:bipartite:eq:negative:620}
b_2(x) = 1 - \delta - \sum_{j=3}^K \mathcal{W}_x(w_j)
\end{equation}
Invariant~\ref{label:bipartite:inv:discretized} implies that:
\begin{equation}
\label{label:bipartite:eq:negative:621}
\mathcal{W}_x(w_2) \leq b_2(x)
\end{equation}
From equations~\ref{label:bipartite:eq:negative:620} and~\ref{label:bipartite:eq:negative:621}, we get:
\begin{equation}
\label{label:bipartite:eq:negative:622}
\mathcal{W}_x(w_2) \leq 1 - \delta - \sum_{j=3}^K \mathcal{W}_x(w_j) 
\end{equation}
Rearranging the terms in the above inequality, we get:
\begin{equation}
\label{label:bipartite:eq:negative:622000}
\sum_{j=2}^K \mathcal{W}_x(w_j) \leq 1 - \delta
\end{equation}
Since $s_1(x) = \infty$, Definition~\ref{label:bipartite:inv:node:capacity:one}  imply that:
\begin{equation}
\label{label:bipartite:eq:negative:624}
b_1^*(x)  = 1 - \sum_{j=2}^K \mathcal{W}_x(w_j)
\end{equation}
From equations~\ref{label:bipartite:eq:negative:622000} and~\ref{label:bipartite:eq:negative:624}, we get:
\begin{equation}
\label{label:bipartite:eq:negative:624000}
b_1^*(x) \geq \delta
\end{equation}
From equations~\ref{label:bipartite:eq:negative:624} and~\ref{label:bipartite:eq:negative:624000}, we get:
\begin{equation}
\label{label:bipartite:eq:negative:624001}
 2 \cdot b_1^*(x) + \sum_{j=2}^K \mathcal{W}_x(w_j) \geq (1+\delta).
\end{equation}
Equation~\ref{label:bipartite:eq:negative:606} follows from equation~\ref{label:bipartite:eq:negative:624001} and the fact that $\W_x(w) = \sum_{j=2}^K \W_x(w_j)$. 

\item {\em Case 2.} $s_1(x) \in \{3, \ldots, K\}$. 

Since $s_1(x) \in \{3, \ldots, K\}$, we infer that $s_2(x) = s_1(x)$ (see equation~\ref{label:bipartite:eq:level-small}). Hence, Definition~\ref{label:bipartite:inv:normal:capacity} implies that:
\begin{equation}
\label{label:bipartite:eq:negative:630}
b_2(x) = 1 - 5 \cdot 10^{(K-s_1(x)+1)} \cdot \delta - \sum_{j=3}^K \W_x(w_j)
\end{equation}
Invariant~\ref{label:bipartite:inv:discretized} states that:
\begin{equation}
\label{label:bipartite:eq:negative:631}
\W_x(w_2) \leq b_2(x)
\end{equation}
From equations~\ref{label:bipartite:eq:negative:630} and~\ref{label:bipartite:eq:negative:631}, we get:
\begin{equation}
\label{label:bipartite:eq:negative:632}
\W_x(w_2) \leq 1 - 5 \cdot 10^{(K-s_1(x)+1)} \cdot  \delta - \sum_{j=3}^K \W_x(w_j)
\end{equation}
Rearranging the terms in the above inequality, and noting that $\W_x(w) = \sum_{j=2}^K \W_x(w_j)$, we get:
\begin{equation}
\label{label:bipartite:eq:negative:633}
\W_x(w) \leq 1 - 5 \cdot 10^{(K-s_1(x)+1)} \cdot \delta
\end{equation}
Since $s_1(x) \in \{3, \ldots, K\}$, Definition~\ref{label:bipartite:inv:node:capacity:one}  imply that:
\begin{equation}
\label{label:bipartite:eq:negative:634}
b_1^*(x)  = 1 - 10^{(K-s_1(x)+1)} \cdot \delta - \W_x(w)
\end{equation}
From equations~\ref{label:bipartite:eq:negative:633} and~\ref{label:bipartite:eq:negative:634}, we infer that:
\begin{equation}
\label{label:bipartite:eq:negative:635}
b_1^*(x) \geq 4 \cdot 10^{(K-s_1(x)+1)} \cdot \delta
\end{equation}
From equations~\ref{label:bipartite:eq:negative:634} and~\ref{label:bipartite:eq:negative:635}, we get:
\begin{eqnarray*}
\W_x(w) + 2 \cdot b_1^*(x) & = & \left(\W_x(w) + b_1^*(x) \right) + b_1^*(x) \\
& = & \left(1 - 10^{(K-s_1(x)+1)} \cdot \delta\right) + b_1^*(x) \\
& \geq & \left(1 - 10^{(K-s_1(x)+1)} \cdot \delta\right) + 4 \cdot 10^{(K-s_1(x)+1)} \cdot \delta \\
& = & 1 + 3 \cdot 10^{(K-s_1(x)+1)} \cdot \delta \\
& \geq & (1+\delta). 
\end{eqnarray*}
Thus, equation~\ref{label:bipartite:eq:negative:606} holds in this case.

\item {\em Case 3.} $s_1(x) = 2$. 

Since $s_1(x) = 2$, we have $x \in S_2$. Further, since $\W_x(w) = \sum_{j=2}^K \W_x(w_j)$, Definition~\ref{label:bipartite:def:inv:node-partition} gives:
\begin{equation}
\label{label:bipartite:eq:negative:625}
\W_x(w) \leq 8 \cdot 10^{(K-2)} \cdot \delta
\end{equation}
On the other hand, since $\W_x(w) = \sum_{j=2}^K \W_x(w_j)$ and $s_1(x) = 2$, Definition~\ref{label:bipartite:inv:node:capacity:one} states that:
\begin{equation}
\label{label:bipartite:eq:negative:626}
b_1^*(x) = 1 - 10^{(K-1)} \cdot \delta - \W_x(w)
\end{equation}
From equations~\ref{label:bipartite:eq:negative:625} and~\ref{label:bipartite:eq:negative:626}, we get:
\begin{eqnarray}
2 \cdot b_1^*(x) & = & 2 - 2 \cdot 10^{(K-1)} \cdot \delta - 2 \cdot \W_x(w) \nonumber \\
& \geq & 2 - 2 \cdot 10^{(K-1)} \cdot \delta - 16 \cdot 10^{(K-2)} \cdot \delta \nonumber \\
& = & 2 - 36 \cdot 10^{(K-2)} \cdot \delta \nonumber \\
& \geq & 1 + \delta \label{label:bipartite:eq:negative:627}
\end{eqnarray}
Equation~\ref{label:bipartite:eq:negative:627} holds due to equation~\ref{label:bipartite:eq:parameter:1}, which implies that  $\delta < 1/(1 + 36 \cdot 10^{(K-2)})$. From equation~\ref{label:bipartite:eq:negative:627}, we get:
$$\W_x(w) + 2 \cdot b_1^*(x) \geq 2 \cdot b_1^*(x) \geq (1+\delta).$$
Thus, equation~\ref{label:bipartite:eq:negative:606} holds in this case.

\end{itemize}

\subsubsection{Proof of Lemma~\ref{label:bipartite:lm:approx:normal}}
\label{label:bipartite:sec:lm:approx:normal}

Consider any edge $(u, v) \in M$ with $v \in T_i$ and $u \in T_i \cup B_i$. Since $v \in T_i$, Definition~\ref{label:bipartite:def:inv:node-partition} implies that:
\begin{equation}
\label{label:bipartite:eq:negative:500}
\W_v(w_i) = b_i(v) 
\end{equation}
Next, since $i \in \{2, \ldots, K\}$, and either $s_i(v) = \infty$ or $i+1 \leq s_i(v) \leq K$, Definition~\ref{label:bipartite:inv:normal:capacity} implies that:
\begin{equation}
\label{label:bipartite:eq:negative:501}
b_i(v) \geq 1 - 5 \cdot 10^{(K-i)} \cdot \delta - \sum_{j=i+1}^K \W_v(w_j)
\end{equation}
From equations~\ref{label:bipartite:eq:negative:500} and~\ref{label:bipartite:eq:negative:501}, we infer that:
$$\W_v(w_i) \geq 1 - 5 \cdot 10^{(K-i)} \cdot \delta - \sum_{j=i+1}^K \W_v(w_j).$$
Rearranging the terms in the above inequality, we get:
\begin{equation}
\label{label:bipartite:eq:negative:502}
\W_v(w) \geq \sum_{j = i}^{K} \W_v(w_j) \geq 1 - 5 \cdot 10^{(K-i)} \cdot \delta 
\end{equation}
Since $u \in T_i \cup B_i$, and since the node-set $Z_i$ is partitioned by the sets $T_i, B_i, S_i$ (see Invariant~\ref{label:bipartite:inv:node-partition}), we infer that $u \notin S_i$. Hence, applying Definition~\ref{label:bipartite:def:inv:node-partition}, we get:
\begin{equation}
\label{label:bipartite:eq:negative:503}
\W_u(w) \geq \sum_{j=i}^K \W_u(w_j) > 8 \cdot 10^{(K-i)} \cdot \delta
\end{equation}
Adding equations~\ref{label:bipartite:eq:negative:502} and~\ref{label:bipartite:eq:negative:503}, we get:
\begin{equation}
\label{label:bipartite:eq:negative:504}
\W_u(w) + \W_v(w) \geq 1+3 \cdot 10^{(K-i)} \cdot \delta  \geq 1+3 \delta 
\end{equation}
Summing equation~\ref{label:bipartite:eq:negative:504} over all the edges in the matching $M$, we get:
$$\sum_{x \in V(M)} \W_x(w) = \sum_{(x, y) \in M} \left(\W_x(w) + \W_y(w)\right) \geq (1+3\delta) \cdot |M|.$$
This concludes the proof of the lemma. 

\subsubsection{Proof of Lemma~\ref{label:bipartite:lm:approx:residual}}
\label{label:bipartite:sec:lm:approx:residual}

Set $\lambda = 4 \cdot 10^{K-i} \cdot \delta$. For every node $v \in T_i$, either $s_i(v) = \infty$ or $s_i(v) \in \{i+1, \ldots, K\}$. Thus, Definition~\ref{label:bipartite:inv:residual:capacity} states that for every node $v \in T_i$, either $b_i^r(v) = \delta \leq \lambda$ or $b_i^r(v) = 4 \cdot 10^{K-s_i(v)+1} \cdot \delta \leq \lambda$. We infer that $b_i^r(v) \in [0, \lambda]$ for every node $v \in T_i$. In contrast, for every node $v \in S_i$, Definition~\ref{label:bipartite:inv:residual:capacity} states that $b_i^r(v) = 2\lambda$.  By Invariant~\ref{label:bipartite:inv:node-partition}, the node-sets $T_i, S_i \subseteq Z_i$ are mutually disjoint. Finally, by Invariant~\ref{label:bipartite:inv:residual:matching}, the fractional assignment $w_i^r$ is a maximal fractional $b$-matching in the residual graph $G_i^r = (T_i \cup S_i, E^r_i)$ with respect to the capacities $\{ b_i^r(v) \}, v \in T_i \cup S_i$. The residual graph $G_i^r$ is  bipartite since $S_i \cap T_i = \emptyset$ and since  every edge $e \in E_i^r$ has one endpoint in $T_i$ and the other endpoint in $S_i$.
 Further,  $M$ is a matching in the graph $G^r_i$ since $E_i^r$ contains all the edges in $E_i$ with one tight and one small endpoints. Thus,  Theorem~\ref{label:bipartite:thm:critical} implies that:
\begin{equation}
\label{label:bipartite:eq:230}
\sum_{v \in T_i \cup S_i} W_v(w_i^r) \geq (4/3) \cdot  \sum_{v \in T_i \cap V(M)} b_i^r(v)
\end{equation}

We can now infer that:
\begin{eqnarray}
\sum_{v \in V(M)} \W_v(w) + \sum_{v \in Z_i} W_v(w_i^r) & \geq &  \sum_{v \in V(M)} \W_v(w) + \sum_{v \in T_i \cup S_i} W_v(w_i^r) \label{label:bipartite:eq:231} \\
& \geq &  \sum_{v \in V(M)} \W_v(w) + (4/3) \cdot  \sum_{v \in T_i \cap V(M)} b_i^r(v) \label{label:bipartite:eq:232} \\
& \geq & \sum_{v \in T_i \cap V(M)} \W_v(w) + (4/3) \cdot  \sum_{v \in T_i \cap V(M)} b_i^r(v) \nonumber \\
& = & \sum_{v \in T_i \cap V(M)} \left( \W_v(w) + (4/3) \cdot b_i^r(v) \right) \nonumber \\
& \geq & \sum_{v \in T_i \cap V(M)} (1+\delta/3) \label{label:bipartite:eq:233} \\
& = &  (1+\delta/3) \cdot |M| \label{label:bipartite:eq:234}
\end{eqnarray}
Equation~\ref{label:bipartite:eq:231} holds since $Z_i \supseteq T_i \cup S_i$ (see Invariant~\ref{label:bipartite:inv:node-partition}). Equation~\ref{label:bipartite:eq:232} follows from equation~\ref{label:bipartite:eq:230}. Equation~\ref{label:bipartite:eq:233} follows from Claim~\ref{label:bipartite:cl:100} stated below. Equation~\ref{label:bipartite:eq:234} holds since every edge in $M$ has exactly one endpoint in $T_i$ and no two edges in $M$ share a common endpoint. Thus, we have $|T_i \cap V(M)| = |M|$.  This leads to the proof of the lemma.

Thus, in order to prove Lemma~\ref{label:bipartite:lm:approx:residual}, it suffices to  prove Claim~\ref{label:bipartite:cl:100}. This is done below. 

\begin{claim}
\label{label:bipartite:cl:100}
For every node $v \in T_i$, we have: $\W_v(w) + (4/3) \cdot b^r_i(v) \geq  1 + \delta/3$. 
\end{claim}

\begin{proof}
We consider two possible cases, depending on the value of $s_i(v)$. 
\begin{itemize}
\item {\em Case 1.} $s_i(v) = \infty$.

In this case, Definition~\ref{label:bipartite:inv:residual:capacity} states that:
\begin{equation} 
\label{label:bipartite:eq:200}
b_i^r(v) = \delta
\end{equation}
On the other hand, since $v \in T_i$, Definition~\ref{label:bipartite:def:inv:node-partition} states that:
\begin{equation}
\label{label:bipartite:eq:201}
\W_v(w_i) = b_i(v) 
\end{equation} 
Finally, since $v \in Z_i$ and $s_i(v) = \infty$, Definition~\ref{label:bipartite:inv:normal:capacity} states that:
\begin{equation}
\label{label:bipartite:eq:202}
b_i(v) = 1 - \delta - \sum_{j=i+1}^K \W_v(w_j)
\end{equation}
From equations~\ref{label:bipartite:eq:201} and~\ref{label:bipartite:eq:202} we get:
$$\W_v(w_i) = 1 - \delta - \sum_{j=i+1}^K \W_v(w_j).$$
Rearranging the terms in the above inequality, we get:
\begin{eqnarray}
\W_v(w) \geq \sum_{j=i}^K \W_v(w_j) = 1 - \delta  \label{label:bipartite:eq:203}
\end{eqnarray}
From equations~\ref{label:bipartite:eq:200} and~\ref{label:bipartite:eq:203}, we infer that:
$$\W_v(w) + (4/3) \cdot b_i^r(v) \geq (1+\delta/3).$$
This concludes the proof of the claim in this case.

\item {\em Case 2.} $s_i(v) = k$ for some $k \in \{i+1, \ldots, K\}$. 

In this case, Definition~\ref{label:bipartite:inv:residual:capacity} states that:
\begin{equation}
\label{label:bipartite:eq:205}
b_i^r(v) = 4 \cdot 10^{(K- k+1)} \cdot \delta
\end{equation}
On the other hand, since $v \in T_i$, Definition~\ref{label:bipartite:def:inv:node-partition} states that:
\begin{equation}
\label{label:bipartite:eq:206}
\W_v(w_i) = b_i(v)
\end{equation} 
Finally, since $v \in Z_i$ and $s_i(v) = k \in \{i+1, \ldots, K\}$, Definition~\ref{label:bipartite:inv:normal:capacity} states that:
\begin{equation}
\label{label:bipartite:eq:207}
b_i(v) = 1 - 5 \cdot 10^{(K-k+1)} \cdot \delta - \sum_{j=i+1}^K \W_v(w_j)
\end{equation}
From equations~\ref{label:bipartite:eq:206} and~\ref{label:bipartite:eq:207} we get:
$$\W_v(w_i) = 1 - 5 \cdot 10^{(K-k+1)} \cdot \delta - \sum_{j=i+1}^K \W_v(w_j).$$
Rearranging the terms in the above inequality, we get:
\begin{eqnarray}
\sum_{j=i}^K \W_v(w_j) = 1 - 5 \cdot 10^{(K-k+1)} \cdot \delta  \label{label:bipartite:eq:208}
\end{eqnarray}
Since $\W_v(w) \geq \sum_{j=i}^K \W_v(w_j)$, equation~\ref{label:bipartite:eq:208} implies that:
\begin{equation}
\label{label:bipartite:eq:209}
\W_v(w) \geq 1 - 5 \cdot 10^{(K-k+1)} \cdot \delta 
\end{equation}
From equations~\ref{label:bipartite:eq:205} and~\ref{label:bipartite:eq:209} we get:
\begin{eqnarray}
\label{label:bipartite:eq:2090}
\W_v(w) + (4/3) \cdot b_i^r(v) \geq 1 + (1/3) \cdot 10^{(K-k+1)} \cdot \delta  \geq 1 + (\delta/3)  
\end{eqnarray}
The last inequality holds since $2 \leq k \leq K$.  The claim follows from equation~\ref{label:bipartite:eq:2090} in this case. 
\end{itemize}
\end{proof}

\section{The algorithm}
\label{label:bipartite:sec:algo}

In this section, we will present a dynamic algorithm for maintaining a $(1+\epsilon)^2$-approximation to the size of the fractional matching $(w+w^r+w^*_1)$. First, we recall a standard assumption used in dynamic algorithms literature on the sequence of edge insertions/deletions in the input graph.

\begin{assumption}
\label{label:bipartite:assume:standard}
The input graph $G = (V, E)$ is empty (i.e., $E = \emptyset$) in the beginning of the sequence of edge insertions/deletions.
\end{assumption}

However, we will make an (apparently) stronger assumption on the input sequence, as stated below.

\begin{assumption}
\label{label:bipartite:assume:new}
The input graph $G = (V, E)$ is empty (i.e., $E = \emptyset$)  in the beginning  of the sequence of edge insertions/deletions. Further, the  input  graph $G = (V, E)$ is also empty (i.e., $E = \emptyset$) at the end of the sequence of edge insertions/deletions.
\end{assumption}

In the theorem below, we show that Assumption~\ref{label:bipartite:assume:new} is without any loss of generality. 

\begin{theorem}
\label{label:bipartite:th:assume}
Suppose that a dynamic algorithm has an amortized update time of $\kappa(n)$, where $n = |V|$ is the number of nodes in the input graph, under Assumption~\ref{label:bipartite:assume:new}. Then the same dynamic algorithm has an amortized update time of $O(\kappa(n))$ under Assumption~\ref{label:bipartite:assume:standard}. 
\end{theorem}

\begin{proof}
Consider a dynamic algorithm that has an amortized update time of $\kappa(n)$ under Assumption~\ref{label:bipartite:assume:new}. Now, consider a sequence of $t$ edge insertions/deletions  that satisfy Assumption~\ref{label:bipartite:assume:standard}. Hence, the graph $G$ is empty in the beginning of this sequence. Let $G^{(t)} = (V, E^{(t)})$ denote the status of the input graph at the end of this sequence. Note that $G^{(t)}$ need not be empty.  Define $|E^{(t)}| = m^{(t)}$, and  note that  $m^{(t)} \leq t$. 

We now delete all the edges from $G^{(t)}$ one after the other. This leads to a new sequence of $(t+m^{(t)})$ edge insertions/deletions where the input graph is empty in the end. In other words, the new sequence satisfies Assumption~\ref{label:bipartite:assume:new}.  Hence, the total time spent by the algorithm under this new sequence is at most $(t+m^{(t)}) \cdot \kappa(n) \leq 2 t \cdot \kappa(n) = O(t \cdot \kappa(n))$. Since the old sequence is a prefix of the new sequence, we infer that the total time spent by the algorithm under the old sequence is at most the time spent under the new sequence. Thus, the total time spent under the old sequence (which consisted of only $t$ edge insertions/deletions) is $O(t \cdot \kappa(n))$. So the amortized update time under the old sequence is $O(\kappa(n))$. This concludes the proof. 
\end{proof}

 Before proceeding any further, we recall the properties derived in Section~\ref{label:bipartite:sec:property}, as they will be crucial in  analyzing  our algorithm. We also define the phrase ``$w$-structures for a level $i \in \{2, \ldots, K\}$''.
\begin{definition}
\label{label:bipartite:def:w:structure}
Consider any level $i \in \{2, \ldots, K\}$. The phrase ``$w$-structures for level $i$'' refers to:  The subgraph $G_i = (Z_i, E_i)$,  the normal-capacities $\{b_i(x)\}, x \in Z_i,$   the  fractional assignment $w_i$ with support $H_i \subseteq E_i$, and the node-weights $\{ \W_x(w_i), W_x(w_i) \}$, $x \in Z_i$. 
\end{definition}

From the invariants in Section~\ref{label:bipartite:sec:inv:2}, it is  apparent that the $w$-structures for a level $i \in \{2, \ldots, K\}$ depend {\em only on} the $w$-structures for levels $j \in \{i+1, \ldots, K\}$. In particular, the $w$-structures for a level $i$ do not in any way depend on the residual fractional assignment $w^r$ (see equation~\ref{label:bipartite:eq:new:w}) or the fractional assignment $w_1^*$. For the time being, we will focus on designing an algorithm that only maintains the $w$-structures for levels $\{2, \ldots, K\}$. This will immediately give us the size of the fractional assignment $w = \sum_{j=2}^K w_j$. Later on we will show how to extend our algorithm to maintain good approximations to the sizes of the fractional assignments $w^r$ and $w_1^*$ as well. The rest of this section is organized as follows. \\
\begin{enumerate}
\item In Section~\ref{label:bipartite:sec:data:structures}, we present the data structures that will be used by our algorithm to maintain the $w$-structures for levels $\{2, \ldots, K\}$. \\
\item In Section~\ref{label:bipartite:sec:handle:update}, we show how to update the above data structures after an edge insertion/deletion in $G$. This gives us a dynamic algorithm for maintaining the $w$-structures for levels $\{2, \ldots, K\}$.
\item In Section~\ref{label:bipartite:sec:maintain:normal}, we analyze the amortized update time of our algorithm from Section~\ref{label:bipartite:sec:handle:update}. See Theorem~\ref{label:bipartite:th:maintain:normal}.
\item In Section~\ref{label:bipartite:sec:maintain:one}, we show how to maintain the size of the fractional matching $w_1^*$. See Theorem~\ref{label:bipartite:th:main:one}.
\item Finally, in Section~\ref{label:bipartite:sec:maintain:residual}, we show how to maintain a $(1+\epsilon)^2$-approximation to the size of the residual fractional matching $w^r$. See Theorem~\ref{label:bipartite:th:maintain:residual}.  \\
\end{enumerate}
\noindent Our main result is summarized in the theorem below.

\begin{theorem}
\label{label:bipartite:th:runtime:main:main}
We can maintain a $(1+\epsilon)^2$-approximation to the size of the fractional matching $(w+w^r+w_1^*)$ in $O((10/\epsilon)^{K+8} \cdot n^{2/K})$ update time.
\end{theorem}

\begin{proof}
The theorem follows if we sum over the update times given in Theorems~\ref{label:bipartite:th:maintain:normal},~\ref{label:bipartite:th:main:one},~\ref{label:bipartite:th:maintain:residual}, and then apply equation~\ref{label:bipartite:eq:parameter:7}.
\end{proof}

\subsection{Data structures}
\label{label:bipartite:sec:data:structures} 
Our algorithm  keeps the following data structures. \\
\begin{itemize}
\item The input graph $G = (V, E)$ using adjacency lists. \\
\item For every level $i \in \{2, \ldots, K\}$;  \\
\begin{itemize} 
\item The subgraph $G_i = (Z_i, E_i)$ using adjacency lists.  \\
\item For every node $v \in Z_i$, the capacity $b_i(v)$,  and the weights $\W_v(w_i)$, $W_v(w_i)$. We ``extend'' the capacities $\{ b_i(v) \}$ and the node-weights $\{ \W_v(w_i), W_v(w_i) \}$ in a natural to the input graph $G = (V, E)$. Thus, for every node $v \in V \setminus Z_i$,  we set $b_i(v)  = \W_v(w_i) = W_v(w_i)  = 0$.   \\
\item The support $H_i = \{ e \in E_i : w_i(e) \}$ of the fractional assignment $w_i$ using a doubly linked list.   We  assume that each edge $e \in H_i$ gets a weight $w_i(e) = 1/d_i = 1/n^{(i-1)/K}$, without {\em explicitly storing} these edge-weights anywhere in our data structures.  \\
\item The subgraph $(Z_i, H_i)$ using adjacency lists.  \\
\item The partition of the node-set $Z_i$ into subsets $T_i, B_i, S_i \subseteq Z_i$. The sets $T_i, B_i, S_i$ are maintained as  doubly linked lists. Given any node $v \in Z_i$, the data structure can report in constant time whether the node is tight (i.e., $v \in T_i$) or big (i.e., $v \in B_i$) or small (i.e., $v \in S_i$).  \\
\end{itemize}
\item Whenever a node (resp. an edge) appears in a doubly linked list described above, we store a pointer from that node (resp. edge) to its position in the linked list. Using these pointers, we can insert (resp. delete) any element into (resp. from) any list in constant time. \\
\end{itemize}

\subsection{Handling the insertion/deletion of an edge $(u,v)$ in the input graph $G = (V, E)$}
\label{label:bipartite:sec:handle:update}

Suppose that an edge $(u,v)$ is either inserted into or deleted from $G$. To handle this edge insertion/deletion, we first update the adjacency lists  in $G = (V, E)$. Next, we update the $w$-structures for the levels  in a ``top down'' manner, as described in Figure~\ref{label:bipartite:fig:initial}. The set $L_j$ consists of  the subset of nodes $x \in V$ whose discretized weight $\W_x(w_j)$ changes its value while fixing the $w$-structures for level $i$ as per Section~\ref{label:bipartite:sec:maintain:normal}.

\begin{figure}[htbp]
\centerline{\framebox{
\begin{minipage}{5.5in}
\begin{tabbing}
1. \ \ \ \ \ \ \= For $i = K$ to $2$; \\
2. \> \ \ \ \ \ \ \ \ \= Set $L_i \leftarrow \emptyset$. \\
3. \> \> Call the subroutine FIX-STRUCTURES($i$). See Section~\ref{label:bipartite:sec:maintain:normal}. 
\end{tabbing}
\end{minipage}
}}
\caption{\label{label:bipartite:fig:initial}}
\end{figure}

In Section~\ref{label:bipartite:sec:maintain:normal}, whenever in set $H_i$ we  insert/delete  an edge incident upon a node $z$, we will call the subroutine UPDATE($i, z$)
 described below. This subroutine updates the value of $\W_z(w_i)$ {\em in a lazy manner} so as to ensure that Invariant~\ref{label:bipartite:inv:discretized} is satisfied. 
 
 To be more specific, just before  changing the value of $W_z(w_i)$, Invariant~\ref{label:bipartite:inv:discretized} was satisfied. Now, an edge insertion/deletion changes the value of $W_z(w_i)$ by $1/d_i < \epsilon$ (see Definition~\ref{label:bipartite:def:inv:normal:assignment}, Invariant~\ref{label:bipartite:inv:normal:assignment} and equation~\ref{label:bipartite:eq:parameter:4}). Hence, we need to change the value of $\W_z(w_i)$ by at most $\epsilon$. This is done in Figure~\ref{label:bipartite:fig:update} in a lazy manner.
 \begin{enumerate}
 \item If  the current value of $\W_z(w_i)$ is too small to satisfy Invariant~\ref{label:bipartite:inv:discretized}, then we increase  $\W_z(w_i)$ by $\epsilon$.
 \item Else if  the current value of $\W_z(w_i)$ is too large to satisfy Invariant~\ref{label:bipartite:inv:discretized}, then we decrease  $\W_z(w_i)$ by $\epsilon$. 
 \end{enumerate}

\begin{figure}[htbp]
\centerline{\framebox{
\begin{minipage}{5.5in}
\begin{tabbing}
1. \ \ \ \ \ \ \= If $\W_v(w_i) < W_v(w_i)$, then \\
2. \> \ \ \ \ \ \ \ \ \= Set $\W_v(w_i) \leftarrow \W_v(w_i) + \epsilon$. \\
3. \>  Else if $W_v(w_i) \leq \W_v(w_i) - 2 \epsilon$, then \\
4. \> \> Set $\W_v(w_i) \leftarrow \W_v(w_i) - \epsilon$. 
\end{tabbing}
\end{minipage}
}}
\caption{\label{label:bipartite:fig:update} UPDATE($i, z$).}
\end{figure}

\subsubsection{The subroutine FIX-STRUCTURES($i$), where $i \in \{2, \ldots, K\}$.}
\label{label:bipartite:sec:maintain:normal}

Consider the scenario where we have fixed the $w$-structures for all levels $j > i$. At this stage, for  $j \in \{i+1, \ldots, K\}$, the set $L_j$ consists of all the nodes $x \in V$ such that the value of $\W_x(w_j)$ has been changed while fixing the $w$-structures for level $j$. Further, at this stage we have $L_j = \emptyset$ for all $j \in \{2, \ldots, i\}$. The procedure  in this section will serve two objectives: (a) it will update the $w$-structures for level $i$, and (b) whenever  the value of  $\W_x(w_i)$ is changed for any node $x$, it will set $L_i \leftarrow L_i \cup \{x\}$. Thus, at any point in time, the set $L_i$ will  consist of exactly those  nodes whose $\W_x(w_i)$ values have been changed till now.  Note that to fix the $w$-structures for level $i$, we first have to update  the subgraph $G_i = (Z_i, E_i)$ and the support $H_i \subseteq E_i$ of the fractional assignment $w_i$. Specifically, we have to address two types of issues if $i < K$, as described below. \\
\begin{enumerate}
\item Since we have already fixed the $w$-structures for levels $j > i$, it is plausible that there has been a change in the node-set $Z_{i+1}$ and the partition of $Z_{i+1}$ into subsets $T_{i+1}, B_{i+1}, S_{i+1}$. To satisfy Definition~\ref{label:bipartite:inv:node-partition:1}, we might have to   insert (resp. delete) some nodes into (resp. from) the set $Z_i$. 
\begin{itemize}
\item (a) According to Definition~\ref{label:bipartite:inv:node-partition:1} and Lemma~\ref{label:bipartite:lm:laminar}, whenever we insert a node $x$ into the subset $Z_i$ (this happens if $x \in Z_{i+1} \setminus T_{i+1}$ and $x \notin Z_i$), we have to ensure that for all nodes $y \in Z_i$ with $(x, y) \in E_i \subseteq H_{i+1}$, the edge $(x, y)$ is inserted into $E_i$. Accordingly, we have to scan through all the edges $(x, y) \in H_{i+1}$, and for every such edge $(x, y)$, if we find that $y \in Z_i$, then we need to insert the edge $(x, y)$ into the subgraph $G_i = (Z_i, E_i)$.
\item (b) According to Definition~\ref{label:bipartite:inv:node-partition:1}, whenever we delete a node $x$ from the subset $Z_i$ (this happens if $x \in Z_i$ and $x \notin Z_{i+1} \setminus T_{i+1}$), we also have to delete all its incident edges $(x, y) \in E_i$ from the subgraph $G_i = (Z_i, E_i)$. Further, when deleting an edge $(x, y)$ from $E_i$, we have to check if $(x, y) \in H_i$, and if the answer is yes, then we have to delete the edge $(x, y)$ from the set $H_i$ as well (see Invariant~\ref{label:bipartite:inv:normal:assignment}). The last step  will decrease each of the weights $W_x(w_i), W_y(w_i)$ by $1/d_i$ (since $w_i(e) = 1/d_i$ for all $e \in H_i$), and hence it might also change  the values of $\W_x(w_i)$ and $\W_y(w_i)$, which, in turn, might lead us to reassign $y$ in the partition of the node-set $Z_i$ into  subsets $T_i, S_i, B_i$. \\
\end{itemize} 
\item There might be some node $x \in V$ that remains in the set $Z_i$, but whose normal-capacity $b_i(x)$ has to be changed (for we have already changed the $w$-structures for levels $j > i$, and they determine the value of $b_i(x)$ as per Definition~\ref{label:bipartite:inv:normal:capacity}). In this event, we might encounter a situation where the value of $\W_x(w_i)$  exceeds the actual value of $b_i(x)$. To fix this issue, for such nodes $x$, as a precautionary measure we  ``turn off''   all its incident edges $(x, y) \in H_i$  with nonzero weights under $w_i$. Specifically, we  visit every edge $(x, y) \in H_i$, and delete the edge $(x, y)$ from $H_i$. This last step  decreases the weights $W_x(w_i), W_y(w_i)$ by $1/d_i$, and accordingly, we might need to change the values of $\W_x(w_i)$ and $\W_y(w_i)$ as well. As a result, we might need to reassign the nodes $x, y$ in the partition of  $Z_i$ into subsets $T_i, S_i, B_i$. \\
\end{enumerate}
Let us say that a ``bad event'' happens for a node $x$ at level $i$ iff either the node   gets  inserted into (resp. deleted from) the set $Z_i$ or its normal-capacity $b_i(x)$ gets changed. From Definitions~\ref{label:bipartite:inv:normal:capacity},~\ref{label:bipartite:def:inv:node-partition},~\ref{label:bipartite:inv:node-partition:1} and equation~\ref{label:bipartite:eq:level-small}, we conclude that  such a bad event  happens only if the discretized weight $\W_x(w_j)$ of the node at some level $j > i$ gets changed. In other words, a bad event happens for a node $x$ at level $i$ only if $x \in \bigcup_{j = i+1}^K L_j$. It follows that all the operations described above are implemented in Figure~\ref{label:bipartite:fig:subgraph}. To be more specific, Steps (02) -- (07) deal with Case 1(a), Steps (08) -- (19) deal with Case 1(b), and Steps (20) -- (29) deal with Case 2.  Finally, we note that if $i = K$, then we do not need to worry about any of these issues since then  $\bigcup_{j=i+1}^K L_j = \emptyset$. The node-set $Z_K$ and the capacities $\{b_K(x)\}$ remain unchanged (see Definitions~\ref{label:bipartite:inv:node-partition:1},~\ref{label:bipartite:inv:normal:capacity}).

We now analyze the total time it takes to perform the operations in Figure~\ref{label:bipartite:fig:subgraph}. The main For loop in line 01 runs for $\left| \bigcup_{j > i} L_j \right|$ iterations. During each such iteration, it takes time proportional to either  $\text{deg}_x(H_{i+1})$  (see the For loop in line 06) or $\text{deg}_x(E_i)$  (see the For loops in lines 10, 22). Since $d_{i+1} = d_i \cdot n^{1/K}$ (see Definition~\ref{label:bipartite:def:inv:normal:assignment}), Lemma~\ref{label:bipartite:lm:deg} and Corollary~\ref{label:bipartite:cor:deg} imply that $\text{deg}_x(H_{i+1}) \leq d_i \cdot n^{1/K}$ and $\text{deg}_x(E_i) \leq d_i \cdot n^{1/K}$. Hence, the time taken by each iteration of the main For loop is $O(d_i \cdot n^{1/K})$. Accordingly, the total runtime becomes $O(\left| \bigcup_{j > i} L_j \right| \cdot d_i \cdot n^{1/K})$, as stated in the lemma below.

\begin{figure}[htbp]
\centerline{\framebox{
\begin{minipage}{5.5in}
\begin{tabbing}
01.  \= For all nodes $x \in \bigcup_{j=i+1}^K L_{j}$; \\
02. \> \ \ \ \ \ \ \ \= If $x \in Z_{i+1} \setminus T_{i+1}$ and $x \notin Z_i$, then \ \ \ \ \ {\em (Definition~\ref{label:bipartite:inv:node-partition:1} is violated)}\\
03. \> \> \ \ \ \ \ \ \= Insert the node $x$ into the subsets $Z_i$ and $S_i$. \\
04. \> \> \> Update the value of  $b_i(x)$, based on the structures for levels $j > i$ {\em (see Definition~\ref{label:bipartite:inv:normal:capacity}).} \\
05. \> \> \> Set $\W_x(w_i) = \epsilon$, $W_x(w_i) = 0$. {\em (Note that $\W_x(w_i) = \epsilon \leq b_i(x)$ by Observation~\ref{label:bipartite:ob:positive}.)}   \\
06. \> \> \> For every edge  $(x, y) \in H_{i+1}$; \\
07. \> \> \> \ \ \ \ \ \ \= If $y \in Z_i$, then insert the edge $(x, y)$ into the subgraph $G_i = (Z_i, E_i)$. \\
08. \> \> Else if $x \notin Z_{i+1} \setminus T_{i+1}$ and $x \in Z_i$, then \ \ \ \ \  {\em (Definition~\ref{label:bipartite:inv:node-partition:1} is violated)} \\
09. \> \> \ \ \ \ \ \ \= Delete the node $x$ from the subset $Z_i$ and from $T_i \cup S_i \cup B_i$. \\
10. \> \> \> For every edge  $(x, y) \in E_{i}$; \\
11. \> \> \> \ \ \ \ \ \ \= Delete the edge $(x, y)$ from the subgraph $G_i = (Z_i, E_i)$. \\ 
12. \> \> \> \> If $(x, y) \in H_i$, then \\
13. \> \> \> \> \ \ \ \ \ \ \ \ \ \  \= Delete the edge $(x,y)$ from the subset $H_i$.  \\ 
14. \> \> \> \> \> For each node $z \in \{x, y\}$; \\
15. \> \> \> \> \>  \ \ \ \ \ \ \ \= Decrease  the value of $W_z(w_i)$ by $1/d_i$.    \\
16. \> \> \> \> \> \> Update the value of $\W_z(w_i)$ by calling the subroutine $\text{UPDATE}(i, z_i)$. \\
17. \> \> \> \> \>   \>  If the value of $\W_z(w_i)$ changes in the previous step, then set $L_i \leftarrow L_i \cup \{z\}$.  \\
18. \> \> \> \> \>    Reassign $y$ to one of the subsets $T_i, S_i, B_i$ according to Definition~\ref{label:bipartite:def:inv:node-partition}. \\
19. \> \> \> Set $\W_x(w_i) = b_i(x) = 0$. \\
20. \> \> Else if $x \in Z_{i+1} \setminus T_{i+1}$ and $x \in Z_i$, then \ \ \ \ {\em (the value of $b_i(x)$ might have changed)} \\
21. \> \> \> Update the value of  $b_i(x)$, based on the structures for levels $j > i$ {\em (see Definition~\ref{label:bipartite:inv:normal:capacity}).} \\
22. \> \> \> For every edge  $(x, y) \in E_{i}$; \\
23. \> \> \> \ \ \ \ \ \ \=  If $(x, y) \in H_i$, then \\
24. \> \> \> \> \ \ \ \ \ \ \ \   \= Delete the edge $(x, y)$ from the subset $H_i$.  \\ 
25. \> \> \> \> \> For each $z \in \{x, y\}$; \\
26. \> \> \> \> \>  \ \ \ \ \ \ \ \ \ \= Decrease $W_z(w_i)$ by $1/d_i$. \\
27. \> \> \> \> \>  \> Update the value of $\W_z(w_i)$ by   calling the subroutine $\text{UPDATE}(i, z)$. \\
28. \> \> \> \> \>  \> If the value of $\W_z(w_i)$ changes in the previous step, then set $L_i \leftarrow L_i \cup \{z\}$.  \\
29. \> \> \> \> \> \> Reassign $z$ to one of the subsets $T_i, S_i, B_i$ according to Definition~\ref{label:bipartite:def:inv:node-partition}. 
\end{tabbing}
\end{minipage}
}}
\caption{\label{label:bipartite:fig:subgraph} Updating  the subgraph $G_i = (Z_i, E_i)$.} 
\end{figure}

\begin{lemma}
\label{label:bipartite:lm:runtime:fig:subgraph}
It takes $O\left(\left| \bigcup_{j=i+1}^K L_j \right| \cdot d_i \cdot n^{1/K}\right)$ time to perform the operations in Figure~\ref{label:bipartite:fig:subgraph}. 
\end{lemma}

\paragraph{Remark.} Note that for $i = K$ the set $\bigcup_{j=i+1}^K L_j$ is empty.  So we do not execute any of the steps in Figure~\ref{label:bipartite:fig:subgraph}. \\

We have not yet updated the adjacency list data structures of $G_i = (Z_i, E_i)$ to reflect the insertion/deletion of the edge $(u, v)$. This is done in Figure~\ref{label:bipartite:fig:update:edge}. We only need to update the data structures if both the endpoints $u, v$ belong to $Z_i$, for otherwise by Definition~\ref{label:bipartite:inv:node-partition:1} the edge $(u, v)$ does not participate in the subgraph $G_i = (Z_i, E_i)$. Note that if the edge $(u, v)$ is to inserted into $G_i$ (lines 01-02 in Figure~\ref{label:bipartite:fig:update:edge}), then we leave the edge-set $H_i$ unchanged. Else if the edge is to be deleted from $G_i$ (lines 03-11 in Figure~\ref{label:bipartite:fig:update:edge}), then we first delete it from $E_i$, and then check if the edge also belonged to  $H_i$. If the answer is yes, then we delete  $(u, v)$ from $H_i$ as well. Thus, the weights $W_u(w_i)$ and $W_v(w_i)$ decrease by $1/d_i$ due to this operation. Accordingly, the discretized weights $\W_u(w_i)$ and $\W_v(w_i)$ might also undergo some changes. If any of these discretized weights does get modified, then we insert the corresponding node into $L_i$. All these operations take constant time. 

\begin{lemma}
\label{label:bipartite:lm:runtime:fig:update:edge}
It takes $O(1)$ time to perform the operations in Figure~\ref{label:bipartite:fig:update:edge}. 
\end{lemma}

\begin{figure}[htbp]
\centerline{\framebox{
\begin{minipage}{5.5in}
\begin{tabbing}
01. \ \= If $u, v \in Z_i$ and the edge $(u,v)$ has been inserted into $G = (V, E)$, then \\
02. \> \ \ \ \ \ \ \ \= Insert the edge $(u,v)$ into the subgraph $G_i = (Z_i, E_i)$. \\
03. \> Else if $u, v \in Z_i$ and the edge $(u,v)$ has been deleted from $G = (V, E)$, then \\
04. \> \> Delete the edge $(u,v)$ from the subgraph $G_i = (Z_i, E_i)$. \\
05. \> \> If $(u,v) \in H_i$, then \\
06. \> \> \ \ \ \ \ \ \ \ \= Delete the edge $(u,v)$ from the subset $H_i$. \\
07. \> \> \> For each node $x \in \{u, v\}$; \\ 
08. \> \> \> \ \ \ \ \ \ \ \= Decrease the value of $W_x(w_i)$ by $1/d_i$. \\
09. \> \> \> \> Update the value of $\W_x(w_i)$ by calling the subroutine UPDATE($i, z$). \\
10. \> \> \> \> If the value of $\W_x(w_i)$ changes in the previous step, then set $L_i \leftarrow L_i \cup \{x\}$. \\
11. \> \> \> \> Reassign the node $x$ into one of the subsets $T_i, S_i, B_i$ according to Definition~\ref{label:bipartite:def:inv:node-partition}. 
\end{tabbing}
\end{minipage}
}}
\caption{\label{label:bipartite:fig:update:edge} Inserting/deleting the edge $(u,v)$ in  $G_i = (Z_i, E_i)$.  }
\end{figure}

By this time, we have fixed the adjacency lists for the subgraph $G_i = (Z_i, E_i)$, and we have also ensured that $\W_x(w_i) \leq b_i(x)$ for every node $x \in Z_i$. However, due to the deletions of multiple edges from $H_i$, we might end up in a situation where we find two nodes $x, y \in Z_i \setminus T_i$ connected by an edge $(x, y) \in E_i$, but the edge $(x, y)$ is not part of $H_i$ (thereby violating Invariant~\ref{label:bipartite:inv:node-partition}). This can happen only if at least one of the nodes $x, y$ are part of $\bigcup_{j = i}^K L_j$. For if both  $x, y \notin \bigcup_{j = i}^K L_j$, then it means that there have been no changes in their normal capacities $\{b_j(z)\}_{z \in \{x, y\}}$ and discretized weights $\{\W_z(w_j) \}_{z \in \{x, y\}}$ at levels $j \geq i$. It follows that both  $x, y$ belonged to $Z_i \setminus T_i$ just prior to the insertion/deletion of $(u, v)$ in $G$, and hence the edge $(x, y)$ was also part of $H_i$ at that instant (to satisfy Invariant~\ref{label:bipartite:inv:node-partition}). Further, the edge $(x, y)$ was also not deleted from $H_i$ in Figure~\ref{label:bipartite:fig:subgraph}, for otherwise either $x$ or $y$ would have been part of $\bigcup_{j = i+1}^K L_j$. We conclude that if both $x, y \notin \bigcup_{j \geq i} L_j$ and $x, y \in Z_i \setminus T_i$, then the edge $(x, y)$ must belong to $H_i$ at this moment. 

Accordingly, we perform the operations in Figure~\ref{label:bipartite:fig:handle:capacities:2}, whereby we go through the list of nodes in $\bigcup_{j = i}^K L_j$, and for each node $x$ in this list, we check if $x \in Z_i$. If the answer is yes, then we visit all the edges $(x, y) \in E_i \setminus H_i$ incident upon $x$. If we find that both $x, y$ do not belong to $T_i$, then to satisfy Invariant~\ref{label:bipartite:inv:node-partition} we insert the edge $(x, y)$ into the set $H_i$. The time taken for these operations is $\left| \bigcup_{j \geq i} L_j \right|$ times the maximum degree of a node in $E_i$, which is $d_i \cdot n^{1/K}$ by Corollary~\ref{label:bipartite:cor:deg}. Thus, we get the following lemma. 

\begin{lemma}
\label{label:bipartite:lm:runtime:fig:handle:capacities:2}
The time taken for the operations in Figure~\ref{label:bipartite:fig:handle:capacities:2} is $O\left(\left| \bigcup_{j=i}^K L_j \right| \cdot d_i \cdot n^{1/K}\right)$.
\end{lemma}

\begin{figure}[htbp]
\centerline{\framebox{
\begin{minipage}{5.5in}
\begin{tabbing}
1. \ \ \ \= Set $L \leftarrow \bigcup_{j=i}^K L_j$. \\
2. \> For each node $x \in L$; \\
3. \> \ \ \ \ \ \ \ \= If $x \in Z_i$, then call the subroutine FIX($i, x$). See Figure~\ref{label:bipartite:fig:fix}. 
\end{tabbing}
\end{minipage}
}}
\caption{\label{label:bipartite:fig:handle:capacities:2} Handling the nodes rendered non-tight so far. }
\end{figure}

\begin{figure}[htbp]
\centerline{\framebox{
\begin{minipage}{5.5in}
\begin{tabbing}
 1. \ \ \ \ \ \ \ \= For each edge $(x, y) \in E_i$; \\
 2. \> \ \ \ \ \ \ \ \= If $(x, y) \notin H_i$ and $x \notin T_i$ and $y \notin T_i$, then \\
 3. \> \> Insert the edge $(x, y)$ to the subset $H_i$.  \\
 4. \> \> For each node $z \in \{x, y\}$; \\
 5. \> \> \ \ \ \ \ \ \= Increase the value of $W_z(w_i)$ by $1/d_i$. \\
 6. \> \> \> Update the value of $\W_z(w_i)$ by calling the subroutine UPDATE($i, z$).\\
 7. \> \> \>  If the value of $\W_z(w_i)$ changes in the previous step, then set $L_i \leftarrow L_i \cup \{z\}$. \\
 8.  \> \> \> Reassign the node $z$ into one of the subsets $T_i, S_i, B_i$ according to Invariant~\ref{label:bipartite:inv:node-partition}. 
 \end{tabbing}
\end{minipage}
}}
\caption{\label{label:bipartite:fig:fix} FIX($i, x$).}
\end{figure}

At this stage, we have fixed the $w$-structures for level $i$, and hence we finish the execution of the subroutine FIX-STRUCTURES($i$).  We conclude this section by deriving a bound on the running time of this subroutine. This bound follows from Lemmas~\ref{label:bipartite:lm:runtime:fig:subgraph},~\ref{label:bipartite:lm:runtime:fig:update:edge} and~\ref{label:bipartite:lm:runtime:fig:handle:capacities:2}.

\begin{lemma}
\label{label:bipartite:lm:runtime:main}
The subroutine $\text{FIX-STRUCTURES}(i)$ runs for $O\left(1 +  \left| \bigcup_{j=i}^K L_{j} \right| \cdot d_i \cdot n^{1/K}\right)$ time. 
\end{lemma}

\subsection{Analyzing the amortized update time}
\label{label:bipartite:sec:framework}

We  analyze the amortized update time of our algorithm (as described in Section~\ref{label:bipartite:sec:handle:update}) over a sequence of edge insertions/deletions in the input graph $G = (V, E)$. We  prove an amortized update time of $O((10/\epsilon)^{K+2} \cdot n^{1/K})$. Since the algorithm in Section~\ref{label:bipartite:sec:handle:update} maintains the $w$-structures (see Definition~\ref{label:bipartite:def:w:structure}) for all the levels $j \in \{2, \ldots, K\}$, we can easily augment it to maintain the size of the fractional matching $w = \sum_{j=2}^K w_j$ without incurring any additional overhead in the update time. This leads to the following theorem.

\begin{theorem}
\label{label:bipartite:th:maintain:normal}
We can maintain the size of $w$ in $O((10/\epsilon)^{K+2} \cdot n^{1/K})$ amortized update time.
\end{theorem}

Following Assumption~\ref{label:bipartite:assume:new} and Theorem~\ref{label:bipartite:th:assume}, we  assume that the input graph is empty (i.e., $E = \emptyset$) in the beginning {\em and} at the end of this sequence of edge insertions/deletions. We  use a charging scheme to analyze the amortized update time. Specifically, we create  the following ``bank-accounts''.  \\

\begin{itemize}
\item For each node $v \in V$ and level $i \in \{2, \ldots, K\}$, there are  two accounts, namely:
\begin{itemize}
\item $\text{{\sc Normal-Account}}[v, i]$.
\item $\text{{\sc Work-Account}}[v, i]$.
\end{itemize}
\item For level $i = 1$, each node $v \in V$ has exactly one account, namely:
\begin{itemize}
\item $\text{{\sc Work-Account}}[v, 1]$. 
\end{itemize}
\item For every unordered pair of nodes $(u,v)$ and level $i \in \{1, \ldots, K\}$, there is one bank account, namely:
\begin{itemize}
\item $\text{{\sc Work-Account}}[(u,v), i]$. 
\end{itemize}
\item 
From this point onwards, the phrase ``{\em work-account'}'  will  refer to any bank account of the form $\text{{\sc Work-Account}}[v, i]$ or $\text{{\sc Work-Account}}[(u,v), i]$.  Similarly, the phrase ``{\em normal-account}'' will  refer to any bank account of the form $\text{{\sc Normal-Account}}[v, i]$.  \\
\end{itemize}


\noindent  While handling an edge insertion/deletion in $G$,  we will sometimes  ``transfer'' money from one bank account to another.  Further, we will allow the bank accounts to have negative balance. For example, if an account has $x$ dollars, and we transfer $y > x$ dollars from it to some other account, the balance in the first account will become $(x-y)$ dollars, which is a negative amount. We will ensure that the following properties hold. 
\begin{property}
\label{label:bipartite:pro:initial}
In the beginning (before the first edge insertion in $G$), each bank account has zero balance.
\end{property}
\begin{property}
\label{label:bipartite:pro:edge}
For each edge insertion/deletion in the input graph $G$, we deposit $O((10/\epsilon)^{K+2} \cdot n^{1/K})$ dollars  into   bank-accounts.
\end{property}
\begin{property}
\label{label:bipartite:pro:positive}
Money is never withdrawn from a work-account. So it  always has a nonnegative balance.
\end{property}
\begin{property}
\label{label:bipartite:pro:end:normal}
In the end (when $G$ becomes empty again), each normal-account has a nonnegative balance.
\end{property}
\begin{property}
\label{label:bipartite:pro:end}
The  total balance in the work-accounts is at least the total update time. 
\end{property}

Properties~\ref{label:bipartite:pro:positive},~\ref{label:bipartite:pro:end:normal} and~\ref{label:bipartite:pro:end} imply that at the end of  the sequence of edge insertions/deletions,  the sum of the balances in the bank accounts is at least  the total update time of the algorithm. By Properties~\ref{label:bipartite:pro:initial} and~\ref{label:bipartite:pro:edge}, the  former quantity is  $O(10/\epsilon)^{K+2} \cdot n^{1/K})$ times the number of edge insertions/deletions in $G$. Hence,  we get an amortized update time of $O((10/\epsilon)^{K+2} \cdot n^{1/K})$. We  now specify the exact rules that govern the functioning of all the bank accounts. We will show that these rules satisfy all the properties described above. This will conclude the proof of the amortized update time. 

\subsubsection{Rules governing the bank accounts}
\label{label:bipartite:sec:rules}

The first rule  describes the initial conditions.
\begin{Rule}
\label{label:bipartite:rule:charge:0}
In the beginning (before the first edge insertion in $G$), every bank account has zero balance.
\end{Rule}

The next rule states how to deposit money into the bank accounts after an edge insertion/deletion in $G$.

\begin{Rule}\label{label:bipartite:rule:charge:3}
When an edge $(u,v)$ is  inserted into (resp. deleted from) $G$, we make the following deposits. 
\begin{itemize}
\item For each $i \in \{1, \ldots, K\}$; 
\begin{itemize}
\item We deposit one dollar into the account $\text{{\sc Work-Account}}[(u,v), i]$. 
\end{itemize}
\item For each $i \in \{2, \ldots, K\}$;
\begin{itemize}
\item We deposit $(10/\epsilon)^{i+1} \cdot n^{1/K}$ dollars into  $\text{{\sc Normal-Account}}[x, i]$ for all $x \in \{u, v\}$. 
\end{itemize}
\end{itemize}
\end{Rule}

\noindent Before proceeding any further, we  need to make a simple observation. Consider any level $i \in \{2, \ldots, K\}$, and the set of edges $H_i$ that form the support of the fractional assignment $w_i$.  Our algorithm in Section~\ref{label:bipartite:sec:handle:update}  deletes an edge $(x,y) \in H_i$  from the set $H_i$    due to one of the following reasons:
\begin{itemize}
\item The edge $(x,y)$ is getting deleted from the input graph $G = (V, E)$, and hence we have to delete the edge from $H_i$ (see Steps~04,~06 in Figure~\ref{label:bipartite:fig:update:edge}). We say that this is a ``{\em natural deletion}'' from the set $H_i$.
\item Some other edge $(x', y')$ is getting deleted from the input graph $G = (V, E)$. While handling this edge deletion, we have to fix the $w$-structures for level $i$, and as a result the edge $(x, y)$ gets deleted from $H_i$ (see Steps~13, 24 in Figure~\ref{label:bipartite:fig:subgraph}). In this event, the edge $(x, y)$ does {\em not} get deleted from the input graph $G$ itself. We say that this is an ``{\em artificial deletion}'' from the set $H_i$. 
\end{itemize}
From our algorithm in Section~\ref{label:bipartite:sec:handle:update}, we make the following observation. 

\begin{observation}
\label{label:bipartite:ob:deletion}
No artificial deletion takes place from the edge-set $H_K$. Further, at any level $2 \leq i < K$, an edge $(x, y)$ gets artificially deleted from the set $H_i$ only when it has at least one endpoint in $\bigcup_{j > i} L_j$. 
\end{observation}

We  now describe the rule that governs the transfer of money from one normal account to another.

\begin{Rule}\label{label:bipartite:rule:charge:1}
Consider any level $2 \leq i \leq K$. By Observation~\ref{label:bipartite:ob:deletion},   an edge   $(x,y)$ gets artificially  deleted from  $H_i$ only if $i < K$, and when such an event occurs, the edge must have at least one endpoint (say $x$) in $L_j$ for some $i < j  \leq K$. At that instant, for all $z \in \{x, y\}$, we transfer  $(10/\epsilon)^{i+1}  \cdot n^{1/K}$  dollars from  $\text{{\sc Normal-Account}}[x, j]$  to  $\text{{\sc Normal-Account}}[z, i]$. 
\end{Rule}

Finally, we state the rule governing  the transactions between  the normal-accounts and the work-accounts.

\begin{Rule}\label{label:bipartite:rule:charge:2}
Suppose that while handling an edge insertion/deletion in $G$  some node $x \in V$ becomes part of the set $L_i$. At that instant,  for each $1 \leq j \leq i$, we transfer $d_j \cdot n^{1/K}$ dollars from $\text{{\sc Normal-Account}}[x, i]$ to $\text{{\sc Work-Account}}[x, j]$.
\end{Rule}

Rule~\ref{label:bipartite:rule:charge:0} ensures that Property~\ref{label:bipartite:pro:initial} is satisfied. According to Rule~\ref{label:bipartite:rule:charge:3}, when an edge insertion/deletion takes places in $G$, the total amount of money deposited to the bank accounts is at most:
$$K + \sum_{i = 2}^K (10/\epsilon)^{i+1} \cdot n^{1/K} = K + O((10/\epsilon)^{K+2} \cdot n^{1/K}) = O((10/\epsilon)^{K+2} \cdot n^{1/K}).$$ 
The last equality holds since $K < n^{1/K}$ (see equation~\ref{label:bipartite:eq:parameter:5}). Thus,  Property~\ref{label:bipartite:pro:edge} is also satisfied. We also note that the four rules described above immediately imply Property~\ref{label:bipartite:pro:positive}. Next, we focus on proving Property~\ref{label:bipartite:pro:end}. Towards this end, we focus on any given edge insertion/deletion in $G$. We handle this event as per the procedure in Figure~\ref{label:bipartite:fig:initial}. Thus, according to Lemma~\ref{label:bipartite:lm:runtime:main}, the total time taken to handle this edge insertion/deletion in $G$ is given by:
\begin{equation}
\label{label:bipartite:eq:work:100}
\sum_{i = 2}^K \left( 1 +  \left| \bigcup_{j=i}^K L_j \right| \cdot d_i \cdot n^{1/K} \right)
\end{equation}
Due to this edge insertion/deletion in $G$, as per Rule~\ref{label:bipartite:rule:charge:3} the sum $\sum_{(x,y)} \sum_{i=1}^K \text{{\sc Work-Account}}[(x,y), i]$ increases by:
\begin{equation}
\label{label:bipartite:eq:work:101}
\sum_{i=1}^K 1 
\end{equation}
According to Rule~\ref{label:bipartite:rule:charge:2},  for every level $i \in \{1,  \ldots,  K\}$, the sum $\sum_{x \in V} \text{{\sc Work-Account}}[x, i]$ increases by:
\begin{equation}
\label{label:bipartite:eq:work:102}
\left| \bigcup_{j=i}^K L_j \right| \cdot d_i \cdot n^{1/K}
\end{equation}
Thus, due the edge insertion/deletion in $G$, the sum  of all the work-accounts increases by at least:
\begin{equation}
\label{label:bipartite:eq:work:103}
\sum_{i=1}^K \left( 1 +  \left| \bigcup_{j=i}^K L_j \right| \cdot d_i \cdot n^{1/K} \right)
\end{equation}
From equations~\ref{label:bipartite:eq:work:100} and~\ref{label:bipartite:eq:work:103}, we reach the following conclusion. \\
\begin{itemize}
\item Due to each edge insertion/deletion in $G$, the total balance in the work-accounts increases by an amount that is at least  the total time spent  to handle that edge insertion/deletion. \\
\end{itemize}
The above statement, along with Property~\ref{label:bipartite:pro:initial} and~\ref{label:bipartite:pro:positive} (which we have proved already), implies Property~\ref{label:bipartite:pro:end}. Thus, it remains to prove Property~\ref{label:bipartite:pro:end:normal}, which is done by the lemma below. Its proof appears in Section~\ref{label:bipartite:sec:lm:runtime:critical}. As there is no normal-account at level one, the lemma  takes care of all the normal-accounts. This concludes the proof of Theorem~\ref{label:bipartite:th:maintain:normal}.

\begin{lemma}
\label{label:bipartite:lm:runtime:critical}
Consider any node $x \in V$ and any level $i \in \{2, \ldots, K\}$. At the end of the sequence of edge insertions/deletions (when $G$ becomes empty), we have a nonnegative balance in $\text{{\sc Normal-Account}}[x, i]$.
\end{lemma}

\subsubsection{Proof of Lemma~\ref{label:bipartite:lm:runtime:critical}}
\label{label:bipartite:sec:lm:runtime:critical}

Throughout this section, we will use the phrase ``time-horizon'' to refer to the time-interval that begins just before the first edge insertion in $G$ (when $G$ is empty), and ends after the last edge deletion from $G$ (when $G$ becomes empty again). We will also use the following notations. \\
\begin{itemize}
\item Let $\Lambda^+$  be the total amount of money deposited into $\text{{\sc Normal-Account}}[x, i]$ during the entire time-horizon. Similarly, let $\Lambda^-$ be the total amount of money withdrawn from $\text{{\sc Normal-Account}}[x, i]$ during the entire time-horizon. We will show that $\Lambda^+ \geq \Lambda^-$. Since $\text{{\sc Normal-Account}}[x, i]$ has zero balance before the first edge insertion in $G$ (see Rule~\ref{label:bipartite:rule:charge:0}), this will imply Lemma~\ref{label:bipartite:lm:runtime:critical}. \\
\item Consider any given edge insertion/deletion in $G$. We say that this edge insertion/deletion is ``{\em critical}'' iff the following property holds: While handling the edge insertion/deletion in $G$ using the algorithm  from Section~\ref{label:bipartite:sec:handle:update}, the node $x$ becomes part of the set $L_i$. 

\medskip
Let  $\Delta[x, i]$  be the total number of {\em critical} edge insertions/deletions in $G$ during the time-horizon. \\
\item Let $m^-[x, i]$ (resp. $m^+[x,i]$) denote the total number of times an edge incident upon $x$ is deleted from (resp. inserted into) the set $H_i$ during the entire time-horizon. Since the input graph $G$ is empty both before and after the time-horizon, we have $m^+[x, i] = m^-[x,i]$. \\
\end{itemize}

\paragraph{Roadmap for the proof.} In Claim~\ref{label:bipartite:cl:work:critical}, we lower bound $\Lambda^+$ in terms of $m^-[x, i]$. In Claim~\ref{label:bipartite:cl:work:critical:1}, we upper bound $\Lambda^-$ in terms of $\Delta[x, i]$. In Claim~\ref{label:bipartite:cl:epoch:1}, we related the two quantities $\Delta[x, i]$ and $m^-[x, i]$. Next, in Corollary~\ref{label:bipartite:cor:new:epoch:10}, we use Claims~\ref{label:bipartite:cl:work:critical:1} and~\ref{label:bipartite:cl:epoch:1} to upper bound $\Lambda^-$ in terms of $m^-[x, i]$. Finally, Claim~\ref{label:bipartite:cl:work:critical} and Corollary~\ref{label:bipartite:cor:new:epoch:10} imply that $\Lambda^+ \geq \Lambda^-$. This concludes the proof of the lemma.

\begin{claim}
\label{label:bipartite:cl:work:critical}
We have: $\Lambda^+ \geq m^-[x, i] \cdot (10/\epsilon)^{i+1} \cdot n^{1/K}$. 
\end{claim}

\begin{proof}
 Rule~\ref{label:bipartite:rule:charge:3} and~\ref{label:bipartite:rule:charge:1} implies that each time an edge incident upon $x$ gets deleted from $H_i$, we deposit $(10/\epsilon)^{i+1} \cdot n^{1/K}$ dollars into $\text{{\sc Normal-Account}}[x,i]$. Since $m^-[x,i]$ denotes the total number of edge deletions incident upon $x$ that take place in $H_i$, the claim follows. 
\end{proof}

\begin{claim}
\label{label:bipartite:cl:work:critical:1}
We have: $\Lambda^- \leq \Delta[x, i] \cdot 5 \cdot (10/\epsilon)^i \cdot d_i \cdot n^{1/K}$. 
\end{claim}

\begin{proof}
Consider an edge insertion/deletion in the input graph. Suppose that the node $x$ becomes   part of the set $L_i$ while we handle this edge insertion/deletion using the algorithm from Section~\ref{label:bipartite:sec:handle:update}.  In other words, this edge insertion/deletion in $G$ contributes one towards the value of $\Delta[x, i]$. We will show that while handling this edge insertion/deletion, we withdraw at most $5 \cdot (10/\epsilon)^i \cdot d_i \cdot n^{1/K}$ dollars  from $\text{{\sc Normal-Account}}[x, i]$.

We first bound the total amount of money that are withdrawn from $\text{{\sc Normal-Account}}[x,i]$  due  to  Rule~\ref{label:bipartite:rule:charge:1}. Consider any level $j \in \{2, \ldots, i-1\}$. While handling the edge insertion/deletion in $G$, the algorithm from Section~\ref{label:bipartite:sec:handle:update} deletes at most $\text{deg}_x(H_j)$ edges incident upon $x$ from the set $H_j$. For each such deletion of an edge $(x, y)$, we withdraw $2 \cdot (10/\epsilon)^{j+1} \cdot n^{1/K}$ dollars from $\text{{\sc Normal-Account}}[x, i]$, and distribute  this amount evenly between $\text{{\sc Normal-Account}}[x, j]$ and $\text{{\sc Normal-Account}}[y, j]$.  Hence, the total amount of money withdrawn from $\text{{\sc Normal-Account}}[x, i]$ due to Rule~\ref{label:bipartite:rule:charge:1} is at most:
\begin{eqnarray}
\sum_{j=2}^{i-1} \text{deg}_x(H_j) \cdot 2 \cdot (10/\epsilon)^{j+1} \cdot n^{1/K} & \leq & \sum_{j=2}^{i-1} d_j  \cdot 2 \cdot (10/\epsilon)^{j+1} \cdot n^{1/K} \label{label:bipartite:eq:final:1} \\
& \leq & \sum_{j=2}^{i-1} d_i  \cdot 2 \cdot (10/\epsilon)^{j+1} \cdot n^{1/K} \label{label:bipartite:eq:final:2}  \\
& \leq & d_i \cdot 4 \cdot (10/\epsilon)^{i} \cdot n^{1/K} \label{label:bipartite:eq:final:3}
\end{eqnarray}
Equation~\ref{label:bipartite:eq:final:1} follows from Lemma~\ref{label:bipartite:lm:deg}. Equation~\ref{label:bipartite:eq:final:2} holds since $d_j \leq d_i$ for all $j \leq i$ (see Definition~\ref{label:bipartite:def:inv:normal:assignment}). 

Next, we bound the total amount of money withdrawn from $\text{{\sc Normal-Account}}[x, i]$ due to Rule~\ref{label:bipartite:rule:charge:2}. For each level $j \in \{1, \ldots, i\}$, we withdraw $d_j \cdot n^{1/K}$ dollars from $\text{{\sc Normal-Account}}[x, i]$ and transfer this amount to $\text{{\sc Work-Account}}[x, j]$. Hence, the total amount withdrawn from $\text{{\sc Normal-Account}}[x, i]$ due to Rule~\ref{label:bipartite:rule:charge:2} is given by:
\begin{eqnarray}
\sum_{j=1}^i d_j \cdot n^{1/K} & \leq & 2 \cdot d_i \cdot n^{1/K} \label{label:bipartite:eq:final:000} \\
& \leq & (10/\epsilon)^{i} \cdot d_i \cdot n^{1/K} \label{label:bipartite:eq:final:0001}
\end{eqnarray}
Equation~\ref{label:bipartite:eq:final:000} follows from Definition~\ref{label:bipartite:def:inv:normal:assignment} and equation~\ref{label:bipartite:eq:parameter:6}. Finally, we note that Rules~\ref{label:bipartite:rule:charge:2} and~\ref{label:bipartite:rule:charge:1} are the only rules that govern the withdrawal of money from $\text{{\sc Normal-Account}}[x, i]$. Hence, adding equations~\ref{label:bipartite:eq:final:3} and~\ref{label:bipartite:eq:final:0001}, we reach the following conclusion. \\
\begin{itemize}
\item At most $5 \cdot (10/\epsilon)^i \cdot d_i \cdot n^{1/K}$ are withdrawn from $\text{{\sc Normal-Account}}[x, i]$ while handling the insertion/deletion of an edge in $G$ that results in $x$ becoming part of $L_i$. \\
\end{itemize}
Since $\Delta[x, i]$ is the number of times  $x$ becomes part of $L_i$ during the entire time-horizon, the claim follows. 
\end{proof}

\begin{claim}
\label{label:bipartite:cl:epoch:1}
We have:
$(\epsilon d_i/2) \cdot \Delta[x, i] \leq  m^-[x,i].$
\end{claim}

\begin{proof}
Recall the notion of a ``{\em critical}'' edge insertion/deletion in $G$. Such an edge insertion/deletion is characterized by the following property: While handling such an edge insertion/deletion using the algorithm from Section~\ref{label:bipartite:sec:handle:update}, the node $x$ becomes part of the set $L_i$. Also recall that $\Delta[x, i]$ denotes the total number of critical edge insertions/deletions during the entire time-horizon.

Accorinngly, by definition, between any two critical edge insertions/deletions in $G$, the discretized weight $\W_x(w_i)$ (which is an integral multiple of $\epsilon$) must have changed by at least $\epsilon$, since the node $x$ becomes part of $L_i$ only when its discretized weight $\W_x(w_i)$ changes. Recall that the discretized weight $\W_x(w_i)$ is updated in a {\em lazy manner} after a change in the weight $W_x(w_i)$ (see Figure~\ref{label:bipartite:fig:update}). Further, in one step the weight $W_x(w_i)$ changes by $1/d_i$ (since each edge in the support of $w_i$ has weight $1/d_i$), and we have $1/d_i \leq 1/n^{1/K} \ll \epsilon$ (see Definition~\ref{label:bipartite:def:inv:normal:assignment} and equation~\ref{label:bipartite:eq:parameter:4}). Accordingly,  we infer that between any two critical edge insertions/deletions in $G$, the weight $W_x(w_i)$ also changes by at least $\epsilon$. Next, note that the weight $W_x(w_i)$ changes by $1/d_i$ only when there is an edge insertion/deletion incident upon $x$ in $H_i$. Hence, an $\epsilon$ change in the weight $W_x(w_i)$ corresponds to $\epsilon d_i$ edge insertions/deletions in $H_i$ incident upon $x$. 

Let $m[x, i] = m^+[x, i] + m^-[x, i]$ denote the total number of edge insertions/deletions in $H_i$ incident upon $x$ during the entire time-horizon. The above discussion implies that $\Delta[x, i] \cdot (\epsilon d_i) \leq m[x, i]$. By Assumption~\ref{label:bipartite:assume:new} (also see Theorem~\ref{label:bipartite:th:assume}), we have $m^+[x, i] = m^-[x, i]$ and hence $m[x, i] = 2 \cdot m^-[x, i]$. Accordingly, we get: $(\epsilon d_i) \cdot \Delta[x, i] \leq m[x, i] = 2 \cdot m^-[x, i]$. The claim follows. 
\end{proof}

\begin{corollary}
\label{label:bipartite:cor:new:epoch:10}
We have: $\Lambda^- \leq m^-[x, i] \cdot (10/\epsilon)^{i+1} \cdot n^{1/K}$.
\end{corollary}

\begin{proof}
From Claims~\ref{label:bipartite:cl:work:critical:1} and~\ref{label:bipartite:cl:epoch:1}, we infer that:
\begin{eqnarray*}
\Lambda^- \leq \Delta[x, i] \cdot 5 \cdot (10/\epsilon)^i \cdot d_i \cdot n^{1/K} = (\epsilon d_i/2) \cdot \Delta[x, i] \cdot (10/\epsilon)^{i+1} \cdot n^{1/K} \leq m^-[x,i] \cdot (10/\epsilon)^{i+1} \cdot n^{1/K}
\end{eqnarray*}
\end{proof}

From Claim~\ref{label:bipartite:cl:work:critical} and Corollary~\ref{label:bipartite:cor:new:epoch:10}, we infer that $\Lambda^+ \geq \Lambda^-$.

\newcommand{\Z}{\mathcal{Z}}
\newcommand{\M}{\mathcal{M}}

\subsection{Maintaining the size of the fractional assignment $w_1^*$.}
\label{label:bipartite:sec:maintain:one}

From Observations~\ref{label:bipartite:ob:discretized} and~\ref{label:bipartite:ob:positive}, we infer that for every node $v \in Z_1$ the capacity $b_1^*(v)$ is a positive integral multiple of $\epsilon$. Further, Definition~\ref{label:bipartite:inv:node:capacity:one} ensures that $b_1^*(v) \leq 1$ for every node $v \in Z_1$. Thus, we get:
\begin{observation}
\label{label:bipartite:ob:level:one}
For every node $v \in Z_1$, we have $b_1^*(v) = \kappa^*_v \cdot \epsilon$ for some integer $\kappa^*_v \in [1, 1/\epsilon]$.
\end{observation}

We now define a ``meta-graph'' $\G_1 = (\Z_1, \E_1)$. For clarity of exposition, we will refer to the nodes in $\Z_1$ as ``meta-nodes'' and the edges in $\E_1$ as ``meta-edges''.  The meta-graph is constructed as follows. For each node $v \in Z_1$, we create $\kappa^*_v$ meta-nodes $v(1), \ldots, v(\kappa^*_v)$, where $\kappa^*_v$ is defined as in Observation~\ref{label:bipartite:ob:level:one}. Next, for each edge $(u, v) \in E_1$, we create $\kappa^*_u \cdot \kappa^*_v$ meta-edges $\{ (u(i), v(j)) \}, 1 \leq i \leq \kappa^*_u, 1 \leq j \leq \kappa^*_v$. The next theorem bounds the size of the maximum matching in the meta-graph.

\begin{theorem}
\label{label:bipartite:th:meta-graph:one}
The size of the maximum (integral) matching in  $\G_1 = (\Z_1, \E_1)$ equals $(1/\epsilon)$ times the maximum size of a fractional $b$-matching in $G_1 = (Z_1, E_1)$ with respect to the node-capacities $\{ b_1^*(v) \}, v \in Z_1$. 
\end{theorem}

\begin{proof} (Sketch)
The maximum possible size of a fractional $b$-matching in $G_1$ is given by  the following LP.
\begin{eqnarray}
& & \text{Maximize } \qquad  \sum_{(u, v) \in E_1} x^*(u, v) \label{label:bipartite:eq:lp:one} \\
\text{s. t. } \qquad & & \sum_{(u, v) \in E_1} x^*(u, v)  \leq  \kappa^*_v \cdot \epsilon \text{ for all nodes } v \in Z_1. \label{label:bipartite:eq:lp:one:1} \\
& & x^*(u, v) \geq 0 \text{ for all edges } (u, v) \in E_1^*.
\end{eqnarray}
Since the input graph $G = (V, E)$ is bipartite, the subgraph $G_1 = (Z_1, E_1)$ is also bipartite. Hence, the constraint matrix of the above LP is totally unimodular. Thus, in the optimal solution to the above LP, each variable $x^*(u, v)$ is going to take a value that is an integral multiple of $\epsilon$. We can map such a solution $\{x^*(u, v)\}$ in a natural way to a matching $\M_1 \subseteq \E_1$ in the meta-graph: For every edge $(u, v) \in E$, include $x^*(u,v)/\epsilon$  meta-edges from the collection $\{ (u(i), v(j)) \}, 1 \leq i \leq \kappa^*_u, 1 \leq j \leq \kappa^*_v$, into the matching $\M_1$. The size of the matching $\M_1$ will be exactly $(1/\epsilon)$ times the LP-objective. 

Similarly, given any matching $\M_1 \subseteq \E_1$ in the meta-graph, we can construct a feasible solution to the above LP in a natural way. Intially, set $x^*(u, v) = 0$ for all edges $(u, v) \in E_1$. Next, scan through the meta-edges in $\M_1$, and for every meta-edge of the form $(u(i), v(j)) \in \M_1$, set $x^*(u, v) \leftarrow x^*(u,v) + \epsilon$. It is easy to check that at the end of this procedure, we will get a feasible solution to the above LP whose objective value is exactly $\epsilon$ times the size of $\M_1$. The theorem follows.
\end{proof}

\begin{lemma}
\label{label:bipartite:lm:meta-graph:one:maintain}
In a dynamic setting, we can maintain the meta-graph $\G_1 = (\Z_1, \E_1)$ in $O( (10/\epsilon)^{K+4} \cdot n^{1/K})$ update time, amortized over the number of edge insertions/deletions in the input graph $G = (V, E)$.
\end{lemma}

\begin{proof} (Sketch)
We first describe our algorithm for maintaining the meta-graph $\G_1$.

Whenever an edge $(u, v)$ is inserted into (resp. deleted from) the input graph $G = (V, E)$, we call the procedure described in Section~\ref{label:bipartite:sec:handle:update}. This updates the $w$-structures for levels $\{2, \ldots, K\}$ in a top down manner. When the procedure finishes updating the $w$-structures for level $2$, we consider the set of nodes $L = \bigcup_{j=2}^K L_j$. These are only nodes whose capacities in level one need to be updated, for if a node $z\notin L$, then none of its discretized weights $\W_z(w_j)$ were changed, and hence its capacity $b_1^*(z)$ also does not change.

Accordingly, we scan through the nodes in $L$. For each  node $z \in L$, we first update the value of $b_1^*(z)$, and then check all its incident edges  $(z, z') \in H_2$ (since $E_1 \subseteq H_2$ by Lemma~\ref{label:bipartite:lm:laminar}). For each such edge $(z, z') \in H_2$, we insert/delete $\kappa^*_z \cdot \kappa^*_{z'} \leq (1/\epsilon^2)$ meta-edges in $\G_1$, depending on whether the nodes $z, z'$ belong to $Z_1$ or not. Since $\text{deg}_{z}(H_2) \leq d_2 = n^{1/K}$ for all nodes $z$ (see Lemma~\ref{label:bipartite:lm:deg}, Definition~\ref{label:bipartite:def:inv:normal:assignment}), this procedure takes $O(|L| \cdot (1/\epsilon)^2 \cdot n^{1/K})$ time.  But note that according to our charging scheme in Section~\ref{label:bipartite:sec:framework} (see Rule~\ref{label:bipartite:rule:charge:2}), each node $z \in L$ has $d_1 \cdot n^{1/K} = n^{1/K}$ dollars deposited into $\text{{\sc Work-Account}}[z, 1]$. Thus, we reach the following conclusion:
\begin{itemize}
\item The total time spent in maintaining the meta-graph $\G_1$ is at most $(1/\epsilon)^2$ times the amount of dollars deposited into the work-accounts at level one. 
\end{itemize}
From our framework in Section~\ref{label:bipartite:sec:framework}, it follows that the amortized update time for maintaining the meta-graph $\G_1$ is at most $(1/\epsilon^2)$ times the amount of dollars deposited into the bank accounts per edge insertion/deletion in $G$. Accordingly, we get an amortized update time of $O((1/\epsilon)^2 \cdot (1/\epsilon)^{K+2} \cdot n^{1/K}) = O((1/\epsilon)^{K+4} \cdot n^{1/K})$. 
\end{proof}

\begin{corollary}
\label{label:bipartite:cor:lm:meta-graph:one:maintain}
The number of edge insertions/deletions in  $\G_1$ is at most $O((10/\epsilon)^{K+4} \cdot n^{1/K})$ times the number of edge insertions/deletions in the input graph $G$. 
\end{corollary}

\begin{proof}
Let $t$ be the ratio between the number of edge insertions/deletions in $\G_1$ and the number of edge insertions/deletions in $G$. The corollary follows from Lemma~\ref{label:bipartite:lm:meta-graph:one:maintain} and the fact that $t$ cannot be larger than the amortized update time for maintaining $\G_1$. 
\end{proof}

\begin{theorem}
\label{label:bipartite:th:main:one}
We can maintain the value of $w_1^*$ in $O((10/\epsilon)^{K+8} \cdot n^{2/K})$ update time, amortized over the number of edge insertions/deletions in the input graph $G = (V, E)$.
\end{theorem}

\begin{proof}
By Invariant~\ref{label:bipartite:inv:assignment:one}, the size of $w_1^*$ gives a $(1+\epsilon)$-approximation to the maximum possible size of a fractional $b$-matching in $G_1 = (Z_1, E_1)$ with respect to the node-capacities $\{ b_1^*(z) \}$. We will maintain a matching $\M'_1 \subseteq \E_1$ in the meta-graph $\G_1 = (\Z_1, \E_1)$ whose size will be a $(1+\epsilon)$-approximation to the size of the maximum matching in $\G_1$. By Theorem~\ref{label:bipartite:th:meta-graph:one}, the quantity $\epsilon \cdot |\M'_1|$ will serve as an accurate estimate for  the size of $w_1^*$. 

Corollary~\ref{label:bipartite:cor:deg} and Definition~\ref{label:bipartite:def:inv:node-partition} guarantee that $\text{deg}_z(E_1) \leq d_1 \cdot n^{1/K} = n^{1/K}$ for all nodes $z \in Z_1$. Since each edge $(u, v) \in E_1$  leads to $\kappa^*_u \cdot \kappa^*_v \leq (1/\epsilon)^2$ meta-edges, we infer that the maximum degree of a meta-node in $\G_1$ is $d^* = (1/\epsilon^2) \cdot n^{1/K}$. Hence, using a recent result of Gupta and Peng~\cite{GuptaP13}, we can maintain the matching $\M'_1$ in $O(d^*/\epsilon^2)$ update time, amortized over the number of edge insertions/deletions in $\G_1$. Using Corollary~\ref{label:bipartite:cor:lm:meta-graph:one:maintain}, we get an update time of $O((d^*/\epsilon^2) \cdot (10/\epsilon)^{K+4} \cdot n^{1/K}) = O((10/\epsilon)^{K+8} \cdot n^{2/K})$, amortized over the number of edge insertions/deletions in the input graph. Finally, note that this subsumes the update time for maintaining the meta-graph $\G_1$, as derived in Lemma~\ref{label:bipartite:lm:meta-graph:one:maintain}.
\end{proof}

\subsection{Maintaining a $(1+\epsilon)$-approximation to the size of $w^r$}
\label{label:bipartite:sec:maintain:residual}

Consider any level $i \in \{2, \ldots, K\}$. From Observation~\ref{label:bipartite:ob:discretized} and Definition~\ref{label:bipartite:inv:residual:capacity}, we infer that for every node $v \in T_i \cup S_i$ the residual capacity $b_i^r(v)$ is a positive integral multiple of $\epsilon$. Further,  equation~\ref{label:bipartite:eq:parameter:3} and Definition~\ref{label:bipartite:inv:residual:capacity} ensures that $b_i^r(v) \leq 1$ for every node $v \in T_i \cup S_i$. Thus, we get:
\begin{observation}
\label{label:bipartite:ob:level:middle}
For every node $v \in T_i \cup S_i$, with $i \in \{2, \ldots, K\}$, we have the residual capacity $b_i^r(v) = \kappa^r_v \cdot \epsilon$ for some integer $\kappa^r_v \in [1, 1/\epsilon]$.
\end{observation}

Similar to Section~\ref{label:bipartite:sec:maintain:one}, we now define a ``meta-graph'' $\G_i = (\Z_i, \E_i)$ at each level $i \in \{2, \ldots, K\}$. For clarity of exposition, we will refer to the nodes in $\Z_i$ as ``meta-nodes'' and the edges in $\E_i$ as ``meta-edges''.  The meta-graph is constructed as follows. For each node $v \in T_i \cup S_i$, we create $\kappa^r_v$ meta-nodes $v(1), \ldots, v(\kappa^r_v)$, where $\kappa^r_v$ is defined as in Observation~\ref{label:bipartite:ob:level:middle}. Next, for each edge $(u, v) \in E^r_i$ (see Invariant~\ref{label:bipartite:inv:residual:matching}), we create $\kappa^r_u \cdot \kappa^r_v$ meta-edges $\{ (u(j), v(l)) \}, 1 \leq j \leq \kappa^r_u, 1 \leq l \leq \kappa^r_v$. The next theorem relates  the fractional matchings in the meta-graph $\G_i$ with the fractional $b$-matchings in the residual graph $G_i^r = (T_i \cup S_i, E_i^r)$.

\begin{theorem}
\label{label:bipartite:th:meta-graph}
Consider any level $i \in \{2, \ldots, K\}$.
\begin{enumerate}
\item Every fractional matching $w''_i$ in the meta-graph $\G_i$ corresponds to a unique fractional $b$-matching $w'_i$ in the residual graph $G_i^r$ with respect to the node-capacities $\{b_i^r(z)\}$. The size of $w''_i$ is exactly $(1/\epsilon)$-times the size of $w'_i$. 
\item A maximal fractional matching  in $\G_i$ corresponds to a maximal fractional $b$-matching  in $G_i^r$ with respect to the node-capacities $\{ b_i^r(z)\}$.
\end{enumerate}
\end{theorem}

\begin{proof} (Sketch)
Given a fractional matching $w''_i$ in  $\G_i$, we construct the corresponding fractional $b$-matching $w'_i$ in $G_i^r$ as follows. Initially, we set $w_i(e) \leftarrow 0$ for every edge $e \in E_i^r$. Next, we scan through the list of meta-edges in $\E_i$. For every meta-edge of the form $(u(j), v(l)), 1 \leq j \leq \kappa^r_u, 1 \leq l \leq \kappa^r_v$, we set $w'_i(u, v) \leftarrow w'_i(u, v) + \epsilon$.  It is easy to check that this construction satisfies both the properties stated in the theorem. 
\end{proof}

\begin{lemma}
\label{label:bipartite:lm:meta-graph:maintain}
In a dynamic setting, we can maintain all the meta-graphs $\{ \G_i = (T_i \cup S_i, \E_i) \}, 2 \leq i \leq K$, in $O((10/\epsilon)^{K+4} \cdot n^{1/K})$ update time, amortized over the number of edge insertions/deletions in  $G = (V, E)$.
\end{lemma}

\begin{proof} (Sketch)
Similar to the proof of Lemma~\ref{label:bipartite:lm:meta-graph:one:maintain}.
\end{proof}

\begin{corollary}
\label{label:bipartite:cor:lm:meta-graph:maintain}
The number of edge insertions/deletions in $\bigcup_{i=2}^K \E_i$ is at most $O((10/\epsilon)^{K+4} \cdot n^{1/K})$ times the number of edge insertions/deletions in $G$. 
\end{corollary}

\begin{proof} (Sketch)
Similar to the proof of Corollary~\ref{label:bipartite:cor:lm:meta-graph:one:maintain}.
\end{proof}

\begin{theorem}
\label{label:bipartite:th:maintain:residual}
There is a dynamic algorithm with the following properties:
\begin{enumerate}
\item It maintains a $(1+\epsilon)^2$-approximation to the size of $w^r$. 
\item Its update time is $O((10/\epsilon)^{K+6} \cdot n^{1/K} \cdot \log n)$, amortized over the edge insertions/deletions in $G$. 
\end{enumerate}
\end{theorem}

\begin{proof} (Sketch)
If we could maintain a maximal fractional matching in $\G_i$ for all $i \in \{2, \ldots, K\}$, then we would have had an exact estimate of the size of $w_i^r$ for all $i \in \{2, \ldots, K\}$ (see Theorem~\ref{label:bipartite:th:meta-graph}, Invariant~\ref{label:bipartite:inv:residual:matching}). Instead, we use the dynamic algorithm of Bhattacharya, Henzinger and Italiano~\cite{Arxiv} to maintain a $(1+\epsilon)^2$-maximal fractional matching in each $\G_i$.\footnote{In such a fractional matching $w''_i$,  for every meta-edge $(z, z') \in \E_i$, either $W_{z}(w''_i) \geq 1/(1+\epsilon)^2$ or $W_{z'}(w''_i) \geq 1/(1+\epsilon)^2$. } This means that we get a $(1+\epsilon)^2$-approximation to the size of each $w_i^r$, and hence a $(1+\epsilon)^2$-approximation to the size of $w^r = \sum_{i=2}^K w_i^r$. 

Using the analysis  of Bhattacharya, Henzinger and Italiano~\cite{Arxiv}, the update time of this dynamic algorithm is $O(\log n/\epsilon^2)$, amortized over the number of edge insertions/deletions in $\bigcup_{i=2}^K \E_i$. Applying Corollary~\ref{label:bipartite:cor:lm:meta-graph:maintain}, we get a update time of $O((10/\epsilon)^{K+6} \cdot n^{1/K} \cdot \log n)$, amortized over the number of edge insertions/deletions in $G$. Note that this subsumes the update time for maintaining the meta-graphs $\{\G_i\}, 2 \leq i \leq K$, as derived in Lemma~\ref{label:bipartite:lm:meta-graph:maintain}. 
\end{proof}

\end{document}